\documentclass[10pt]{amsart}

\usepackage[margin=1in]{geometry}
\usepackage{amssymb,comment,hyperref,tensor,slashed,float}
\usepackage{mathrsfs, todonotes}
\usepackage{wasysym}
\usepackage{tikz}
\usepackage{bbold}
\usepackage{enumitem}
\usetikzlibrary{intersections}
\usetikzlibrary{decorations.pathmorphing}
\usetikzlibrary{shapes,arrows}
\usepackage{arydshln}

\newtheorem{theorem}{Theorem}[section]

\newtheorem{lemma}[theorem]{Lemma}
\newtheorem{proposition}[theorem]{Proposition}

\newtheorem{assumpt}[theorem]{Assumption}
\newtheorem{corollary}[theorem]{Corollary}

\theoremstyle{definition}
\newtheorem{definition}[theorem]{Definition}

\theoremstyle{remark}
\newtheorem{remark}[theorem]{Remark}

\numberwithin{equation}{section}
\newcommand{\R}{\mathbb{R}}

\newcommand{\N}{\mathbb{N}}
\newcommand{\de}{\mbox{ d}}
\newcommand{\der}{\mbox{d}}
\newcommand{\norm}[1]{\left\|#1\right\| }
\newcommand{\desude}[2]{\frac{\partial #1}{\partial #2}}
\newcommand{\lun}{\textbf{F}}
\newcommand{\ldu}{\textbf{G}}
\newcommand{\fara}{\mathcal{F}}
\newcommand{\farad}{\,{^\star\!}\mathcal{F}}
\newcommand{\gara}{\mathcal{G}}
\newcommand{\garad}{{^\star\!}\mathcal{G}}
\newcommand{\bara}{\mathcal{B}}
\newcommand{\barad}{{^\star\!}\mathcal{B}}
\newcommand{\mfara}{\mathcal{M}}
\newcommand{\mfarad}{{^\star\!}\mathcal{M}}

\newcommand{\snabla}{\slashed{\nabla}}
\newcommand{\sdelta}{\slashed{\Delta}}

\newcommand{\lbar}{{\underline{L}}}
\newcommand{\curl}{\slashed{\mbox{curl }}}
\newcommand{\dive}{\slashed{\mbox{div }}}

\newcommand{\alphabar}{{\underline{\alpha}}}
\newcommand{\gbar}{\slashed{g}}

\newcommand{\svol}{\slashed{\varepsilon}}
\newcommand{\duu}{\mathfrak{D}_{u_1}^{u_2}}

\newcommand{\ellmbi}{\ell_{(\text{MBI})}}
\newcommand{\redsh}{V_{\text{red}}}
\newcommand{\lie}{\mathcal{L}}

\newcommand{\rie}{\text{Rm}\,}
\newcommand{\ipsi}{\mathcal{Y}}
\newcommand{\sipsi}{{^\star\!}\mathcal{Y}}

\newcommand{\regio}[2]{\mathcal{R}_{#1}^{#2}}
\newcommand{\oregio}[2]{\overline{\mathcal{R}}_{#1}^{#2}}
\newcommand{\tregio}[2]{\widetilde{\mathcal R}_{#1}^{#2}}
\newcommand{\nonl}[1]{\textbf{NL}_{#1}}
\newcommand{\sphint}[1]{\int_{\mathbb{S}^2} #1 \desphere}
\newcommand{\hdelta}{{H_{\Delta}}}

\newcommand{\oscr}{{\leo}}
\newcommand{\ono}{{\leo}_{\text{no}}}

\newcommand{\fgiusto}{F_{\checked}}
\newcommand{\fgiustor}{F_{\checked}^{\mathcal{R}}}
\newcommand{\fgiustordeg}{F_{\checked, \text{deg}}^{\mathcal{R}}}

\newcommand{\qtot}{Q_\text{tot}}
\newcommand{\iindicess}{\mathscr{I}_{\mathcal{K}}}
\newcommand{\iindiceo}{\mathscr{I}_{\leo}}
\newcommand{\iindicesg}{\mathscr{I}_{\mathscr{S}}}
\newcommand{\faram}{\dot{\fara}}
\newcommand{\faramd}{{}^\star \! \dot{\fara}}
\newcommand{\otm}{{\text{OT}}}
\newcommand{\mqtensor}{\dot{Q}}
\newcommand{\desphere}{\de {\mathbb{S}^2}}
\newcommand{\slie}{\slashed{\lie}}
\newcommand{\conplus}{\overline{C}}
\newcommand{\conminus}{\underline{C}}
\newcommand{\p}{\partial}
\newcommand{\eps}{\varepsilon}
\newcommand{\enn}{\mathbb{N}_{\geq 0}}

\allowdisplaybreaks

\newsavebox\MBox

\newcommand{\chih}{\chi_{\mathcal{H}^+}}
\newcommand{\redeps}{\varepsilon_{\text{red}}}

\newcommand{\qnl}{\dot Q^{(NL)}}
\newcommand{\marad}{{}^\star \! \mathcal{M}}
\newcommand{\fackip}{Fackerell--Ipser }
\begin{document}

\title[Nonlinear stability for the MBI system on a Schwarzschild background]
{Nonlinear stability for the Maxwell--Born--Infeld system on a Schwarzschild background}

\author{Federico Pasqualotto}
\address{\small Princeton University, Department of Mathematics, Fine~Hall,~Washington~Road,~Princeton, NJ 08544,~United~States\vskip.2pc \small University of Cambridge, Department of Pure Mathematics and Mathematical
	Statistics, Wilberforce~Road,~Cambridge~CB3~0WA,~United~Kingdom\vskip.2pc }
\email{fp2@princeton.edu}

\begin{abstract}
	In this paper we prove small data global existence for solutions to the Maxwell--Born--Infeld (MBI) system on a fixed Schwarzschild background. This system has appeared in the context of string theory and can be seen as a nonlinear model problem for the stability of the background metric itself, due to its tensorial and quasilinear nature. The MBI system models nonlinear electromagnetism and does not display birefringence. The key element in our proof lies in the observation that there exists a first-order differential transformation which brings solutions of the spin $\pm 1$ Teukolsky Equations, satisfied by the extreme components of the field, into solutions of a ``good'' equation (the \fackip Equation).
	This strategy was established in~\cite{masthes} for the linear Maxwell field on Schwarzschild. We show that analogous \fackip equations hold for the MBI system on a fixed Schwarzschild background, which are however nonlinearly coupled. To essentially decouple these right hand sides, we set up a bootstrap argument.
	We use the $r^p$ method of Dafermos and Rodnianski in \cite{newmihalis} in order to deduce decay of some null components, and we infer decay for the remaining quantities by integrating the MBI system as transport equations. 
\end{abstract}

\maketitle

\setcounter{tocdepth}{1}

\maketitle

%\tableofcontents

\section{Introduction and motivation}
In this paper we consider the \emph{Maxwell--Born--Infeld} (MBI) system, which is a system of partial differential equations for nonlinear electromagnetism, on a fixed Schwarzschild background. This system\footnote{To be precise, its formulation on $3+1$-dimensional Minkowski spacetime.} was first considered by Born and Infeld in~\cite{born33}, and interest in it has been revived in relatively recent times by a connection with string theory, where a higher-dimensional version of the MBI lagrangian appears. Our interest in this theory arises from two considerations: first, the MBI system can be viewed as a \emph{nonlinear model problem for the stability of the Schwarzschild metric} as a solution of the vacuum Einstein equations. This model problem has the additional feature of being tensorial and quasilinear. Furthermore, MBI
 is a \emph{natural} theory of nonlinear electromagnetism, in a sense which will be explained in Section~\ref{sec:specstru}, and as such it has been proposed as a candidate for a well-defined theory of point charge motion (see Section~\ref{sec:motivation}). We prove that, when the initial data are sufficiently small in a weighted Sobolev space, there exists a unique global-in-time solution, decaying with inverse polynomial rates at infinity. This can be interpreted as the nonlinear stability of the trivial solution to the MBI system on a Schwarzschild background.

\subsection{Overview of the result} The \emph{Maxwell--Born--Infeld} (MBI) system is a hyperbolic system of partial differential equations in $3+1$ dimensions, a higher-dimensional version of which has been widely studied in the string theory literature (see e.g.~\cite{stringy}). In general, the MBI system can be formulated as follows. Let $(\mathcal{M},g)$ be a smooth, Lorentzian, $(3+1)$-dimensional spacetime. Let $\fara$ be a smooth two-form on $\mathcal M$. We say that $\fara$ satisfies the MBI system on $\mathcal M$ if the following tensorial equations hold true in $\mathcal M$:
\begin{equation}\label{eq:mbiintro}
	\nabla^\mu \farad_{\mu\nu} = 0, \qquad \nabla^\mu \left[\ellmbi^{-1}\left(\fara_{\mu\nu} - \ldu \farad_{\mu\nu}\right) \right] = 0,
\end{equation}
here, $\ellmbi^2:= 1+\lun -\ldu^2$, with the invariants defined as
\begin{align*}
\lun = \frac 1 2 \fara_{\mu\nu} \fara^{\mu\nu}, \qquad
\ldu = \frac 1 4 \fara_{\mu\nu} \farad^{\mu\nu}.
\end{align*}
Furthermore, $\nabla$ denotes the Levi-Civita connection on $(\mathcal M, g)$ and $\farad_{\mu\nu}$ indicates the Hodge dual of $\fara_{\mu\nu}$, which is $\farad_{\mu\nu} = \frac 1 2 \varepsilon_{\mu\nu\alpha\beta}\fara^{\alpha\beta}$, where $\varepsilon$ is the standard volume form on $(\mathcal M, g)$.

The MBI system~\eqref{eq:mbiintro} shows features which are similar to the vacuum Einstein equations, and therefore can be used as a model problem to understand the stability problem of black holes to the vacuum Einstein equations. In this context, the interest in the MBI system stems from the fact that it is both quasilinear and tensorial, and moreover its linearization around the trivial solution is exactly the linear Maxwell theory.

In this work, we solve the MBI system on a \textit{fixed Schwarzschild background}. Recall the Schwarzschild spacetime $(\mathcal{S}, g)$ with mass $M>0$ (for the precise definition, see Definition~\ref{def:schwreg}). Recall also that that the exterior region (which we denote $(\mathcal{S}_e, g)$) of the Schwarzschild spacetime can be parametrized by the usual coordinates $(t, r, \theta, \varphi) \in \R \times (2M, \infty) \times (0, \pi) \times (0, 2 \pi)$. Upon setting $\mu = \frac{2M} r$, we have the following expression for the metric tensor in this coordinate system:
\begin{equation}
	g = -(1-\mu) \de t \otimes \mbox{d} t + (1-\mu)^{-1} \de r \otimes \mbox{d} r + r^2 (\mbox{d}\theta \otimes \mbox{d} \theta + \sin^2 \theta \de \varphi \otimes \mbox{d} \varphi).
\end{equation}

Here is an informal version of our main theorem. The reader may find the precise statement of the main theorem in Section~\ref{sec:statement} (Theorem~\ref{thm:gwp}). We refer the reader to the relevant sections for the precise definitions of the objects in this informal discussion.

\begin{theorem}[Informal version of Theorem~\ref{thm:gwp}]\label{thm:informal} Let $(\mathcal{S},g)$ be the Schwarzschild spacetime of mass $M>0$. Let furthermore $\widetilde{\Sigma}_{t_0^*}$ be a spacelike hypersurface to be defined in Section~\ref{sec:not:regfol} (see also Figure~\ref{fig:statem}). Consider the MBI system \eqref{eq:mbiintro} on $(\mathcal{S},g)$, with initial data $\fara_0$ on $\widetilde{\Sigma}_{t_0^*}$ such that
\begin{enumerate}
	\item $\fara_0$ is smooth,
	\item $\fara_0$ is small in a weighted, higher order Sobolev space,
	\item $\fara_0$ has asymptotically vanishing charge at spacelike infinity on $\widetilde{\Sigma}_{t_0^*}$,
	\item $\fara_0$ satisfies the constraint equations~\eqref{eq:constreq} on $\widetilde{\Sigma}_{t_0^*}$.
\end{enumerate}
Then, the Maxwell--Born--Infeld system \eqref{eq:mbiintro} on the fixed Schwarzschild background admits a global-in-time solution $\fara$ defined in the exterior region of Schwarzschild, which we called $\mathcal{S}_e$, intersected with the causal future of $\widetilde{\Sigma}_{t_0^*}$ (the shaded region in Figure~\ref{fig:statem}), and having $\fara_0$ as initial data.

Furthermore, $\fara$ decays with quantitative rates and the charge of $\fara$ vanishes at null infinity (denoted by $\mathcal{I}^+$ in Figure~\ref{fig:statem}).
\end{theorem}

\begin{figure}[H]
	\centering
	\begin{tikzpicture}	

\node (I)    at ( 0,0) {};

\path
  (I) +(90:4)  coordinate[label=90:$i^+$]  (top)
       +(-90:4) coordinate[label=-90:$i^-$] (bot)
       +(0:4)   coordinate                  (right)
       +(180:4) coordinate[label=180:$i^0$] (left)
       +(180:4.9) coordinate[label=180:$ $] (left1)
       ;
 \path (top) +(180:0.01) coordinate (tar);

\path 
	(top) + (180:4) coordinate[label=90:$i^+$]  (acca)
		+ (-45: 2) coordinate (nulluno)
		+ (-45: 3) coordinate (nulldue)
	;

\path 
	(left) + (-45: 2) coordinate (correspuno)
		+ (-45: 3) coordinate (correspdue)
	;

\path (top) + (-150: 1.7) coordinate (horiuno)
			+ (-150: 6) coordinate (horidue);

\draw[name path = stre, opacity = 0] (left1) to (tar);
\draw [name path = ciao, thick, opacity = 0] (horiuno) to [bend left = 15] node[midway, above]{} (right);
\draw [name path = bye, thick, opacity = 0] (horidue) to [bend left = 10] node[midway, above]{} (right);
\path [name intersections={of=ciao and stre,by=E}];
\path [name intersections={of=bye and stre,by=D}];
\draw[fill=gray!20] (D) to [bend left = 10] (right) to (top) to [bend left = 6] cycle;
\path (top) + (-135:1.42) coordinate (cancuno)
			+ (-135:4.1) coordinate (cancdue);
\draw [thick] (D) to [bend left = 10] node[midway, above]{$\widetilde{\Sigma}_{t^*_0}$} (right);
\draw [thick] (top) to [bend left = 6] node[midway, above, sloped]{\small $r = r_{\text{in}}$} (D);
\draw[fill=white] (D) to (cancdue) to (top) to [bend left =6] cycle;

\draw (left) to 
          node[midway, above left, sloped] {}
      (top) to 
          node[midway, above, sloped]{$\mathcal{I}^+$}
      (right) to 
          node[midway, above, sloped]{$\mathcal{I}^-$}
      (bot) to
          node[midway, above, sloped]{$\mathcal{H}^-$}    
      (left) -- cycle;
\draw[thick] (top) to (right);
\draw[decorate,decoration=zigzag] (top) -- (acca)
      node[midway, above, inner sep=2mm] {$r=0$};

\end{tikzpicture}
	\caption{Penrose diagram depicting $\widetilde{\Sigma}_{t_0^*}$ and the shaded region of existence of our solution.  We need $\widetilde{\Sigma}_{t_0^*}$ to extend slightly in the interior region only as a technical requirement.}\label{fig:statem}
\end{figure}
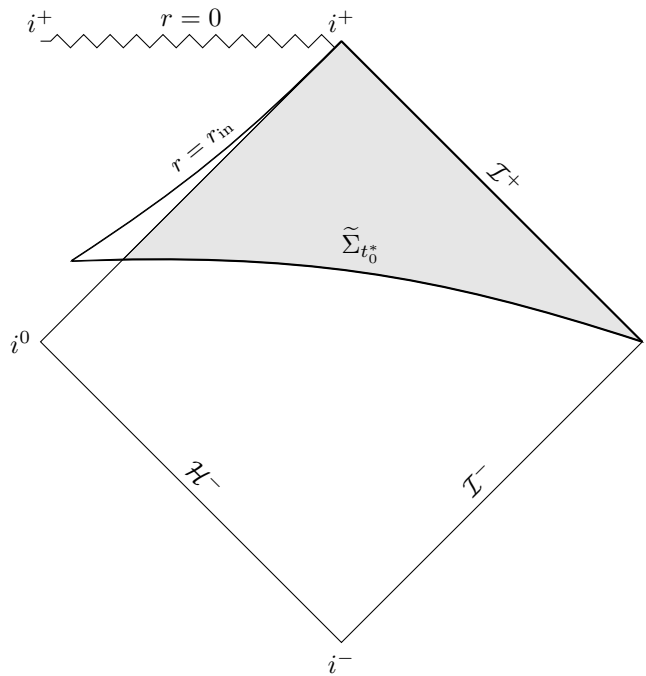

\begin{remark}
	It would be of interest, in view of the analogies with the Kerr stability conjecture, to remove the assumption $(3)$ regarding absence of charge. In fact, it is widely believed that a general gravitational perturbation of the Kerr metric with parameters $(a, M)$ (where $a$ is the angular momentum, and $M$ is the mass of the black hole) evolves into a different member of the Kerr family (i.e., the ``final'' parameters $(a,M)$ will be different). Therefore, there exists a map which connects the initial perturbation with the ``final'' values of $a$ and $M$: understanding such map is essential in a proof of the full Kerr stability conjecture. This is the so-called \emph{final state problem}. Similarly, in the context of the MBI system on Schwarzschild, we have a two-parameter family of static solutions to the MBI system (described in detail in Proposition~\ref{prop:statsols}). The two parameters, in this case, are the  ``magnetic charge'' $\sigma_0$ and the ``electric charge'' $\rho_0$. Consider now a perturbation of the MBI field around one of those static solutions. The natural problem, in this case, would be to determine the ``asymptotic value'' of those parameters as the system evolves in time, as a function of the initial perturbation: this has a clear parallel with the \emph{final state} question outlined in the case of gravitational perturbations. We elaborate more on this in Appendix~\ref{sec:heuristic}, in which we provide a heuristic argument giving an explicit expression of the asymptotics. In particular, the argument presented in Appendix~\ref{sec:heuristic} indicates a way to solve the \emph{final state problem} in the MBI case which circumvents the use of modulation techniques.
	\end{remark}

The key ingredients to establish our result are quantitative decay estimates in a bootstrap setting. First, let us introduce some notation: we consider an unknown two-form $\fara$, and we define the \emph{middle components} ($\rho$, $\sigma$) and the \emph{extreme components} ($\alpha, \alphabar$) of the field as follows:
\begin{equation*}
	\rho:= \frac 1 {2(1-\mu)} \fara(\lbar, L), \ \  \sigma:= \frac 1 {2(1-\mu)} \farad(\lbar, L), \ \  \alpha_A := \fara (\partial_{\theta^A}, L), \ \  \alphabar_A:= \fara(\partial_{\theta^A}, \lbar),
\end{equation*}
where $L := \partial_t + \partial_{r^*}$, $\lbar := \partial_t - \partial_{r^*}$, $\mu := \frac{2M}{r}$. Here, the vector fields $\p_t$ and $\p_{r^*}$ are the coordinate vector fields induced by the usual Regge--Wheeler coordinates (see Definition~\ref{def:rwcoords} for a precise definition). Furthermore, $\farad$ denotes the Hodge dual of $\fara$: $\farad_{\mu\nu} := \frac 12 \varepsilon_{\mu\nu\alpha\beta}\fara^{\mu\nu}$, where $\varepsilon$ is the usual volume form on Schwarzschild. Finally, $\theta^A$ is a local coordinate system on the conformal sphere $\mathbb{S}^2$.
Having introduced this notation, we can proceed to a brief description of our strategy:
\begin{enumerate}
	\item We first establish a local existence statement. We then proceed to set up a bootstrap argument: the bootstrap assumptions are a set of decay estimates for the various null components of the field $\fara$.
	\item We then show that the middle components ($\rho$, $\sigma$) of the MBI field, once differentiated in the angular direction, satisfy nonlinear scalar \fackip Equations.
	The analogous strategy in the linear case appeared in our earlier paper~\cite{masthes}, in the context of estimates on the spin $\pm 1$ Teukolsky Equations. See also Blue~\cite{blue}.
	The two so-called \fackip Equations, which are completely decoupled in the linear case, are now (for MBI) coupled through a cubic, nonlinear right hand side. We remark here that all calculations are performed at the scalar level, hence the equations are slightly different from those appearing in \cite{masthes}.
	\item In order to effectively decouple the two \fackip Equations, we need to control the cubic right hand sides. To accomplish that, we use the bootstrap assumptions to prove conservation of higher-order weighted energy at the level of the full system.
	\item Having established (3), we proceed to prove decay of solutions to these \fackip Equations employing the strategy of Dafermos and Rodnianski (\cite{newmihalis}).
	\item Finally, we deduce decay of the various components of the field only at lower derivative order, closely following~\cite{masthes}, using the MBI equations as transport equations. This enables us to close the bootstrap argument.
\end{enumerate}

For a more detailed overview of the proof, see Section~\ref{sec:pfoverview}.

\subsection{The black hole stability problem: recent advances}\label{sec:literature}
Motivation to study our problem arises from the so-called \textit{black hole stability problem}. The crucial question of the \textit{full nonlinear stability of Schwarzschild} as a solution to the vacuum Einstein equations, and hence the model's physical relevance, remains open to date.

In this work, we regard the MBI system as a toy model to understand some aspects of the nonlinear stability properties of the Schwarzschild geometry. We remark in particular that there have been advances in understanding other nonlinear models on black hole backgrounds (see point $(3)$ of the list below). In addition, more recently, Klainerman and Szeftel proved the first statement establishing the nonlinear stability of the Schwarzschild family as a solution to the Einstein vacuum equations, under a restricted class of perturbations (see point~$(6)$ of the list below).

We now give an outline of some recent research efforts towards the black hole stability problem. This will help us put our analysis into context. Stability problems in General Relativity received great attention following the monumental work of Christodoulou and Klainerman on the stability of Minkowski spacetime~\cite{globalnon}. Since then, considerable research efforts have been focussed on understanding the stability problem of nontrivial solutions. There has been important progress, of which we make an incomplete list here.

\begin{enumerate}
	\item A good amount of effort has been devoted to the analysis of the decay properties of the linear wave equation on black hole spacetimes ($\square_g \psi = 0$), as this is the most basic problem to study. 
	On a fixed Schwarzschild background, we cite the fundamental results of Blue--Sterbenz and Dafermos--Rodnianski (see resp.~\cite{Blue2006,redshift}).
	
	Then, research effort was focussed on understanding the picture for Kerr black holes with parameters $|a|~\ll~M$. We refer to the work of Andersson--Blue, Dafermos--Rodnianski and Tataru--Tohaneanu (see resp.~\cite{abannals,unodue,ttkerr}).
	
	We finally remark that, for the full subextremal range of parameters of the Kerr black hole $|a| < M$, the problem has additional difficulties. Only recently has there been a complete proof of decay of solutions to the wave equation by Dafermos--Rodnianski--Shlapentokh-Rothman in~\cite{paperiii}. A good introduction to the research field can be found in~\cite{lecturenotes}.
	
	Furthermore, various authors have provided methods to obtain $L^\infty$ decay estimates from $L^2$ decay estimates, with varying degrees of generality of the black hole spacetime considered. Let us cite here the work by Dafermos--Rodnianski, Metcalfe--Tataru--Tohaneanu, Moschidis, and Tataru, resp.~in~\cite{newmihalis, metatoprice, Moschidis2016, tatlocal}.
	
	\item As a second important thread, researchers have been focusing on understanding the decay properties of the linear Maxwell field on black hole spacetimes. In this context, the relevance to the black hole stability problem lies in the tensorial nature of the Maxwell equations.
	
	Let us cite the important papers of Blue and Andersson--Blue~\cite{blue, blueand}, in which the decay properties of the Maxwell field respectively on Schwarzschild and on slowly rotating Kerr are proved. In \cite{masthes}, we approach the problem from a different point of view. We start from the spin $\pm 1$ Teukolsky Equations, apply a differential transformation to the extreme components, and hence obtain a \fackip Equation for the resulting transformed quantity. This equation is then employed to prove decay estimates. See also~\cite{mettattoh, stetat, sari}.
	
	\item A third thread has also been studied, namely the analysis of nonlinear model problems. Several nonlinear equations have been analyzed on curved spacetimes in order to understand how the geometry of the manifold influences the behaviour of these models. For example, in the context of semilinear equations, global existence for nonlinear wave equations satisfying the null condition has been established by Luk in~\cite{Luk2013} on slowly rotating Kerr spacetimes. In the context of quasilinear problems, there has been an advance by Lindblad and Tohaneanu in~\cite{Lindblad2016}. In that paper, the authors prove global existence for to a class of quasilinear wave equations $\Box_{g(u, t, x)} u =0$, so that the metric $g(u,t,x)$ asymptotically approaches the Schwarzschild metric.
	
	\item A fourth thread concerns the full linearized picture of black hole stability. In this area, there has been important work by Dafermos, Holzegel and Rodnianski in~\cite{linearized}. In that paper, the authors prove the linear stability of the Schwarzschild black hole as a first step in the program to solve the Kerr stability conjecture. See also the subsequent paper by Hung, Keller and Wang \cite{Hung2017}, and the work by Johnson~\cite{Johnson2018}. More recently, Dafermos--Holzegel--Rodnianski in~\cite{DHR2017} proved decay for the spin $\pm 2$ Teukolsky equations on slowly rotating Kerr ($|a| \ll M$), a crucial step to address the nonlinear stability of Kerr in the slowly rotating case. See also~\cite{Ma2017}, in which the author proves an integrated local energy decay statement for the spin $\pm 2$ Teukolsky equations on slowly rotating Kerr. Finally, in the context of electrovacuum, the linearized stability of Reissner--Nordstr\"om under the Einstein--Maxwell system has been established in the PhD thesis of Giorgi~\cite{Giorgi2019}.
	
	In the present work, we employ a strategy similar to that used in \cite{linearized}. In fact, our paper \cite{masthes} is based on finding a differential transformation on the extreme components, which is directly analogous to the differential transformation in \cite{linearized}, used in that context to deduce the Regge--Wheeler Equation. We can then view the present paper as a \emph{nonlinear analogue} of \cite{linearized} and \cite{masthes}. This also points to the fact that the physical-space techniques employed seem to be robust for application to nonlinear problems. We will elaborate more on this point in Section~\ref{sec:motivation}.
	\item A fifth thread is focused on the solution of problems with nonzero cosmological constant. In the case $\Lambda < 0$, one expects slow rate of decay for solutions to the scalar wave equation. See the work of Holzegel and Smulevici~\cite{hsads} for a proof that, on Kerr--AdS spacetimes, solutions to the Klein--Gordon equation decay only logarithmically. Due to these slow decay rates, this family of spacetimes is conjectured not to be stable under gravitational perturbations.
	
	On the other hand, in the case $\Lambda > 0$, one can prove much stronger decay rates for solutions to the linear wave equation. This suggests nonlinear stability, and accordingly Hintz and Vasy have indeed proved the global nonlinear stability of the Kerr--de Sitter solution in the remarkable recent~\cite{hv2016}. They make essential use of the exponential decay properties of linear fields on such backgrounds, and the structure of the nonlinearities does not play a crucial role.
	
	\item Finally, the sixth thread is focused on solving the full nonlinear stability problem for the Einstein equations, when perturbing the metric in the neighborhood of a known nontrivial solution. A seminal result in the nonlinear case was obtained by Holzegel in his PhD thesis~\cite{holzegelthesis}. More recently, a solution to a restricted case of the problem of stability of Schwarzschild under gravitational perturbations has appeared. In the work~\cite{KS2017}, Klainerman and Szeftel established the nonlinear stability of the Schwarzschild black hole under axially symmetric polarized perturbations, using the linear theory of~\cite{linearized}. In particular, these perturbations ensure that the resulting evolution will converge to the Schwarzschild black hole. The full, finite codimension nonlinear stability of Schwarzschild, without symmetry assumptions, has been announced in~\cite{tayloroberwolfach}.
\end{enumerate}

\subsection{Motivation to study the MBI system on Schwarzschild}\label{sec:motivation}
 
 The MBI theory appeared for the first time in a 1933 paper by Born and Infeld~\cite{born33}. The original version of the theory was described by a $3+1$ dimensional Lagrangian, and it was formulated on flat spacetime. Only decades later, in the works of Boillat~\cite{boillat1970} and Bialynicki-Birula~\cite{BialynickiBirula:1984tx}, a crucial property of this system was discovered. These works effectively proved that MBI is \emph{distinguished} among all nonlinear theories of electromagnetism as, on Minkowski spacetime, it is the only theory which satisfies the following properties: it is gauge-invariant, it is Lorentz-invariant, it gives rise to static point charge solutions which have finite self-energy, its linearization is the linear Maxwell theory, and the MBI light cone is comprised of a single conical sheet. We will comment more on these properties in Section~\ref{sec:specstru}. Such requirements are, in some sense, minimal for a nonlinear theory of electromagnetism, thereby making the study of its stability properties a natural question.

In addition to the interest from these general considerations, the MBI theory received considerable attention in the last 20 years, since connections with string theory were unveiled in the physics community. We will not focus on these aspects, but let us just refer to the introduction of~\cite{Gibbons:1997xz}, and to the paper~\cite{Callan1997} for an example of the connections between MBI and string theory.

A third motivation to study the MBI system comes from Born's program to devise an electromagnetic theory free of ``divergence problems''. In this direction, Kiessling, in the papers~\cite{kie1} and~\cite{kie2}, formulated a well-defined initial value problem describing point charges in flat spacetime, where the electric and magnetic field obey the MBI equations. He avoids the issue of infinite Lorentz self-force employing Hamilton--Jacobi theory, and furthermore shows global existence to the joint initial value problem only in the static case.\footnote{The well-posedness of solutions to the equations of motion of MBI point charges in the general case still remains an interesting open problem.} Moreover, he does not show uniqueness of such static solutions.

Following these developments, some stability results were established in the context of the MBI theory. In the paper~\cite{speck1}, Speck proved the global nonlinear stability for the MBI system on Minkowski spacetime, in the small data regime. Subsequently, again Speck, in a follow-up paper~\cite{speck2} showed global existence and nonlinear stability of solutions to the MBI system (and other nonlinear models of electromagnetism) coupled with the Einstein equations, for initial data close to trivial data.

With the present work, we would like to contribute to the analysis of black hole stability by considering a quasilinear problem on the Schwarzschild background. In addition, we contribute to the stability analysis of the MBI system \emph{per se}, in continuity with the research efforts described in this section. The present work can therefore be put in the context of points $(2)$ and $(3)$ of the list in Section~\ref{sec:literature}, and can be seen as a proxy to study the nonlinear stability problem of Schwarzschild under gravitational perturbations. We expect an analogous result to hold for MBI on slowly rotating Kerr, even though the techniques developed here do not immediately apply to that case. Furthermore, we are chiefly interested in the issues arising from the presence of a nonlinearity.

Let us also remark that the present work is fundamentally based on estimates on the \fackip Equation, similar to those established in the linear case in \cite{masthes}. We may therefore understand the present paper as a nonlinear application of ideas originated in the work of Dafermos--Holzegel--Rodnianski on linearized gravity on Schwarzschild. In such work \cite{linearized}, the authors find a first-order differential transformation which brings the extreme components of the field ($\alpha$ and $\alphabar$) into quantities which satisfy decoupled Regge--Wheeler Equations. Thinking ahead towards the nonlinear stability of Schwarzschild (seen as part of the larger Kerr family), in principle one will need a strategy to handle the fact that the corresponding Regge--Wheeler equations will not decouple. 

We believe that the main contribution of this paper is how to deal with such coupling, in the context of the MBI system, which we view as a model problem. The core of our argument are indeed the estimates of the right hand sides of the nonlinear \fackip Equations, in Section~\ref{sec:nonlstruct} and Section~\ref{sec:improvedrs}.

We now review, for the benefit of the reader who might not be familiar with these concepts, the geometry of the Schwarzschild spacetime, as well as some of the difficulties associated with proving decay of linear waves on black hole backgrounds. This will help us motivate our strategy, and it will shed light on some of the key issues.

\subsection{Aside: the geometry of Schwarzschild and decay of linear waves}\label{sub:aside}

In this section, we will first give a broad overview of the Schwarzschild geometry. We will then touch upon the issue of decay of solutions to the linear wave equation on a Schwarzschild background, and we will give a brief description of the various difficulties associated with showing decay. The reader familiar with these concepts may skip this section. The interested reader should be warned, however, that we included a minimal amount of material, and we refer the reader to the lecture notes of Dafermos and Rodnianski~\cite{lecturenotes} for a comprehensive introduction to the subject. 

Fix $M > 0$ (the ``mass'' of the spacetime). The Schwarzschild spacetime with mass $M$ can be described by different sets of coordinates. The first such set of coordinates, which we call \emph{regular coordinates}, describes the Schwarzschild manifold $(\mathcal{S}, g)$ as the set  
$$
\mathcal{S} := (t^*,r_1,\theta, \varphi) \in (-\infty, + \infty) \times (0, + \infty) \times (0,\pi) \times (0, 2\pi),
$$
endowed with the metric given by the expression
	\begin{equation*}
	\begin{aligned}
	&g := -\left(1-\frac{2M}{r_1}\right) \de t^* \otimes \de t^* +  \frac{2M}{r_1} \de t^* \otimes \de r_1+  \frac{2M}{r_1}  \de r_1 \otimes \de t^*+\left(1+\frac{2M}{r_1}\right) \de r_1 \otimes \de r_1\\
	& \qquad  + r_1^2 \der \theta \, \otimes \, \der \theta + r_1^2 \sin^2 \theta \der \varphi \, \otimes \, \der \varphi.
	\end{aligned}
	\end{equation*}
These coordinates are suitable to describe both the \emph{exterior Schwarzschild} region (the region ``outside the black hole'') and the \emph{interior Schwarzschild} region (or the ``black hole region''), which are resp.~defined as the subsets
	\begin{equation*}
		\mathcal{S}_e := \{(t^*,r_1,\theta, \varphi) \in \mathcal{S}: r_1 > 2M\}, \qquad \mathcal{S}_i := \{(t^*,r_1,\theta, \varphi) \in \mathcal{S}: 0 < r_1 < 2M\}.
	\end{equation*}
The region $\mathcal{S}_i$ has the following property: every null geodesic $\gamma(s)$ ($s \in \R$, $s \geq 0$) originating from a point inside $\mathcal{S}_i$ is such that the $r_1$-coordinate of the point $\gamma(s)$ will remain less than $2M$ for all $s \geq 0$.

The two regions $\mathcal{S}_e$ and $\mathcal{S}_i$ are separated by the \emph{future event horizon} $\mathcal{H}^+$, which is defined as
	\begin{equation*}
	\mathcal{H}^+ := \{(t^*, r_1, \theta, \varphi) \in \mathcal{S}: r_1 = 2M \}.
	\end{equation*}
Finally, as can be read off easily from the expression of the metric in these coordinates, the coordinate vector field $\p_{r_1}$ remains everywhere spacelike, whereas the coordinate vector field $\p_{t^*}$, goes from timelike to spacelike as $\mathcal{H}^+$ is crossed. In particular, this means that, in the interior region, the hypersurfaces given by the level sets of the function $t^*$ are uniformly spacelike. As such, in view of the positive energy condition, they are suitable as boundary surfaces of regions where to perform energy estimates.

Furthermore, the symmetries of the Schwarzschild spacetime (i.~e.~its Killing fields) are exactly the vector field $\p_t$, as well as any vector field induced by rotations about the center of symmetry ($\p_\varphi$ is one such example).

\begin{remark}
The strategy to prove decay for the wave equation on Minkowski uses the whole algebra of Killing fields, which is much larger in that case (in particular, the presence of Lorentz boosts is very helpful). In our work, we need to adopt a different method, which is robust for applications on a black hole background: the $r^p$ method of Dafermos--Rodnianski~\cite{newmihalis}. In that respect, our work is different from the nonlinear stability of the MBI system on Minkowski by Speck in~\cite{speck1}. 
\end{remark}

We now proceed to describe a different set of coordinates, which are useful to describe the interior region and the exterior region~\emph{separately}. We let
\begin{align*}
&t:= t^*- 2M \log(r_1-2M), \qquad \text{and} \qquad 
r := r_1 \qquad \text{if} \qquad  r_1 > 2M,\\
&t:= t^* + 2M \log(2M-r_1), \qquad \text{and} \qquad 
r := r_1 \qquad \text{if} \qquad 0 <r_1 < 2M.
\end{align*}
In these coordinates, the metric on $\mathcal{S} \setminus \mathcal{H}^+$ becomes:
\begin{equation}\label{eq:irregpre}
g = -\left(1-\frac{2M}{r}\right) dt \otimes dt + \left(1-\frac{2M}{r}\right)^{-1} dr \otimes dr +r^2( \der \theta \, \otimes \, \der \theta +  \sin^2 \theta \der \varphi \, \otimes \, \der \varphi).
\end{equation}
First of all, we notice a coordinate singularity at $\mathcal{H}^+$. Secondly, let us note that, as the event horizon is crossed, the vector field $\p_t$ and $\p_r$ exchange roles: $\p_t$ is timelike in the exterior and spacelike in the interior, whereas $\p_r$ is spacelike in the exterior and timelike in the interior. Furthermore, $\p_t = \p_{t^*}$ is a regular vector field everywhere on $\mathcal{S}$, whereas $\p_r$ is singular at $\mathcal{H}^+$. Indeed, we have the expression:
$$
\p_{r} = \frac{2M}{r}(1-\mu)^{-1}\p_{t} + \p_{r_1},
$$
which clearly shows that $\p_{r}$ is singular at $r = 2M$. Furthermore, the coordinate hypersurfaces $t = \text{const.}$ are spacelike in the exterior, and timelike in the interior, whereas the coordinate hypersurfaces $r = \text{const.}$ are timelike in the exterior, and spacelike in the interior.

We now describe yet another coordinate system, namely the \emph{Regge--Wheeler coordinates} $(t_2, r^*, \theta, \varphi)$, which arise from the following change of coordinates:
\begin{align*}
&t_2 := t, \qquad r^* := r + 2M \log (r-2M), \qquad \text{if} \ r > 2M,\\
&t_2 := t, \qquad r^* := r - 2M \log (2M-r), \qquad \text{if} \ 0 < r < 2M.
\end{align*}
Note that these coordinates are well-defined everywhere, except at the future event horizon $\mathcal{H}^+$.
In these coordinates, the metric on $\mathcal{S}\setminus \mathcal{H}^+$ takes the form:
\begin{equation*}
g = -(1-\mu) \de t_2 \otimes \de t_2 +  (1-\mu) \de r^* \otimes \de r^* + r^2( \der \theta \, \otimes \, \der \theta +  \sin^2 \theta \der \varphi \, \otimes \, \der \varphi).
\end{equation*}
From this definition, the following relations are evident: $\p_{t_2} = \p_t, \p_{r^*} = (1-\mu)\p_r$, and furthermore it holds that the coordinate vector field $\p_{r^*}$ admits a smooth extension across the event horizon $\mathcal{H}^+$.

We now come to the last set of coordinates we will use, namely the \emph{null coordinates}, which have the property that the corresponding coordinate vector fields are always null:
\begin{equation*}
\begin{aligned}
u := t - r^*, \qquad
v := t + r^*.
\end{aligned}
\end{equation*}
The metric, in this case, takes the following form:
\begin{equation*}
g =-2 (1-\mu) \de u \otimes \de v -2 (1-\mu) \de v \otimes \de u + r^2( \der \theta \, \otimes \, \der \theta +  \sin^2 \theta \der \varphi \, \otimes \, \der \varphi).
\end{equation*}
\begin{remark}\label{rmk:drstar}
Due to the simple form of the metric in these coordinate systems, it is often convenient to use vector fields as $\p_{r^*}$ and $\p_u$ as commutators. Nevertheless, one must bear in mind that these vector fields are, in a certain sense, degenerate at the event horizon. Indeed, writing $\p_{r^*}$ and $\p_u$ in terms of the regular coordinates, we have the formulae:
$$
\p_{r^*} = \frac{2M}{r}\p_t + (1-\mu)\p_{r_1}, \qquad \p_u = (1-\mu)(\p_t - \p_{r_1}).
$$
Thus both vector fields have degenerate control over ``derivatives transversal to the event horizon'' in the direction of $\p_{r_1}$, as the coefficients in front of such vector field are going to zero as $r \to 2M$ in these expressions. Hence, the natural vector field used to control non-degenerate transversal derivatives at $\mathcal{H}^+$ will be $(1-\mu)^{-1} \p_u$. It is furthermore evident from the expression above that such vector field admits a smooth extension across the event horizon $\mathcal{H}^+$.
\end{remark}

Having given a description of the geometry of the Schwarzschild spacetime, we now turn to a discussion of decay of linear waves on Schwarzschild itself. Although the (recent) history of the subject is quite rich, for the purposes of this introduction we will focus only on a particular type of decay statement, with no claim of completeness. We refer the reader to the discussion in the lecture notes~\cite{lecturenotes} for a more detailed overview. 

The issue of decay of linear waves on a non-flat manifold (such as Schwarzschild) is a delicate one, as the non-trivial geometry may determine the presence of regions in which light rays get ``trapped'', leading to a possible obstruction to decay. Nevertheless, at least in the specific example of the Schwarzschild manifold, it is possible to prove that solutions to the linear wave equation decay with quantitative rates. To make our discussion more concrete, we give an example of such a statement:

\begin{theorem}[\cite{newmihalis}]\label{thm:decaylinear}
Let $(\mathcal{S}_e, g)$ be the exterior Schwarzschild manifold, endowed with the metric $g$. Let furthermore $t^*_0 \in \R$, and let $\phi: \mathcal{S}_e \cap \{t^* \geq t^*_0\} \to \R$ be a smooth solution to the following initial value problem, with initial data posed on the hypersurface $\Sigma_{t^*_0}$:
\begin{equation}\label{eq:waveqschw}
\begin{aligned}
&\Box_g \phi =0,\\
&\phi|_{t^*=t^*_0} = \phi_0,\\
&\p_t \phi|_{t^*=t^*_0} = \phi_1.
\end{aligned}
\end{equation}
Suppose furthermore that $\phi_0$ and $\phi_1$ are smooth and compactly supported functions on $\Sigma_{t_0^*}$. Then, $\phi$ satisfies the following decay properties on $\mathcal{S}_e  \cap \{t^* \geq t^*_0\} $:
\begin{equation}
\begin{aligned}
&|\p \phi| \leq C \frac 1 {(1+r) (1+|u|)^{\frac 12}}, \qquad \text{and} \qquad  |\p \phi| \leq C \frac 1 {(1+r)^{\frac 12} (1+|u|)} \quad \text{if} \ \ r \geq 3M,\\
&|\p \phi| \leq C \frac 1 {(1+v)} \quad \text{if} \ \ 2M < r < 3M.\\
\end{aligned}
\end{equation}
Here, $C$ is a positive constant which depends on a high-order, weighted Sobolev norm of the initial data $(\phi_0, \phi_1)$.
\end{theorem}

\begin{remark}
Note that these decay rates are sharp ``along the light cone'' (this means that the power of $r$ which appears in the estimates cannot be improved). Nevertheless, the decay rate in terms of the $u$-variable is not sharp.
\end{remark}

As we already pointed out, the proof of a decay statement like Theorem~\ref{thm:decaylinear} requires careful analysis, due to the various complications introduced by the background. With this in mind, we now turn to a description of the main difficulties.

The first object one needs to consider when studying the decay of linear waves on a non-flat manifold is the geodesic flow associated to the manifold itself. In the case at hand, the main feature which distinguishes Schwarzschild from flat space is the presence of \emph{trapped null geodesics}. These are particular geodesics which live in the exterior region $\mathcal{S}_e$, and have the property that their $r$-coordinate always remains bounded. On Schwarzschild, it is possible to give an explicit characterization of these geodesics, they are null geodesics which start at $r = 3M$ and ``spin around'' the black hole in the angular direction.

Since light travels on null geodesics, this feature will have repercussions on the decay of linear waves on Schwarzschild: waves traveling in the angular direction around $r = 3M$ stay, in some sense, localized around $r = 3M$, which could be an obstruction to decay. This is best seen looking at the Morawetz (or integrated local energy decay) estimate for the linear wave equation on Schwarzschild, which was first derived in~\cite{redshift}. Here, we refer to the paper~\cite{DR2007}. Suppose that $\phi$ is a solution to the linear wave equation on Schwarzschild:
$$
\square_g \phi = 0.
$$
There exists a constant $C > 0$ such that the following holds. Let $t^*_0 > 0$, and $\mathcal{R}^{\infty}_{t^*_0} := \{(t^*, r_1, \theta, \varphi) \in \mathcal{S}_e: t^* \geq t^*_0\}$, and let $\Sigma_{t^*_0} :=  \{(t^*, r_1, \theta, \varphi) \in \mathcal{S}_e: t^* = t^*_0\}$. Then, we have
\begin{equation}\label{eq:morawaveeq}
\begin{aligned}
&\int_{\mathcal{R}^\infty_{t^*_0}} \Big[ r^{-3}  (\p_{r^*}\phi)^2 +  \Big(1-\frac{3M}{r} \Big)^2 (r^{-1}|\snabla \phi|^2 + r^{-4}(\p_t \phi)^2) \Big]  r^2 \sin \theta \de t^* \de r_1 \de \theta \de \varphi \\
&\qquad \leq C \int_{\Sigma_{t^*_0}} \Big[ (\p_{r^*}\phi)^2 +  |\snabla \phi|^2 +(\p_t \phi)^2\Big] r^2 \sin \theta \de r_1 \de \theta \de \varphi 
\end{aligned}
\end{equation}
Here, we denoted $|\snabla \phi|^2 := |\p_\theta \phi|^2 + \frac 1 {\sin^2 \theta} |\p_\varphi \phi|^2$.
This estimate is already telling us that the field $\phi$ is decaying, albeit in an integrated sense. If we want to obtain quantitative decay estimate, then our goal is to translate this integrated decay into a pointwise one.

From inequality~\eqref{eq:morawaveeq} we can already infer some of the key obstructions. First, we note that, at $r = 3M$, some of the terms in the LHS of display~\eqref{eq:morawaveeq} are degenerate. This is due to the presence of trapping at the ``photon sphere'' ($r = 3M$). Fortunately, this issue is resolved by commuting the equation with $\p_t$. Indeed, it is possible to show that the following estimate holds true:

\begin{equation}\label{eq:morawaveeqnond}
\begin{aligned}
&\int_{\mathcal{R}^\infty_{t^*_0}} \Big[ r^{-3}  (\p_{r^*}\phi)^2 + r^{-1}|\snabla \phi|^2 + r^{-4}(\p_t \phi)^2 \Big] (1-\mu) r^2 \sin \theta \de t^* \de r_1 \de \theta \de \varphi \\
&\qquad \leq C \sum_{k = 0,1}\int_{\Sigma_{t^*_0}} \Big[ (\p_{r^*}(\p_t)^k\phi)^2 +  |\snabla (\p_t)^k \phi|^2 +(\p_t (\p_t)^k \phi)^2\Big] r^2 \sin \theta \de r_1 \de \theta \de \varphi.
\end{aligned}
\end{equation}

This estimate ``loses derivatives'', in a sense that the RHS requires one additional derivative when compared to the LHS, but now the degeneracy at $r = 3M$ has been removed. For a more detailed description of this issue in Schwarzschild, see Section~4.1 of the lecture notes~\cite{lecturenotes}.

There is yet another degeneracy which is manifest from~\eqref{eq:morawaveeq}, namely the one in the term $(\p_{r^*} \phi)^2$ at the event horizon, as previously discussed in Remark~\ref{rmk:drstar}. In fact, it is clear from said remark that no derivative transversal to~$\mathcal{H}^+$ is controlled by estimate~\eqref{eq:morawaveeq}. The fundamental insight, which was introduced in the pioneering work~\cite{redshift}, to overcome such degeneracy is the use of a specific multiplier vector field, called the \emph{redshift multiplier}, which is future directed, transverse to the future event horizon, and which gives rise to a positive bulk term close to the event horizon itself. In our argument, we will make essential use of a version of the redshift multiplier. We define the redshift multiplier in Section~\ref{sub:redshift}, and we prove the associated crucial bulk positivity property in Lemma~\ref{lem:coercivity}.

For the purposes of this introduction, let us just mention that, using the redshift multiplier, one can improve estimate~\eqref{eq:morawaveeq} as follows:
\begin{equation}\label{eq:redshiftedmora}
\begin{aligned}
&\int_{\mathcal{R}^\infty_{t^*_0}} \Big[ r^{-3}  (\p_{r^*}\phi)^2 +  \Big(1-\frac{3M}{r} \Big)^2 (r^{-1}|\snabla \phi|^2 + r^{-4}(\p_t \phi)^2) \\
& \qquad \qquad \qquad + \chi_{\mathcal{H}^+}(r) ((1-\mu)^{-1}\p_{u}\phi)^2\Big]  r^2 \sin \theta \de t^* \de r_1 \de \theta \de \varphi \\
&\qquad \leq C \int_{\Sigma_{t^*_0}} \Big[ (\p_{r^*}\phi)^2 +  |\snabla \phi|^2 +(\p_t \phi)^2+ \chi_{\mathcal{H}^+}(r) ((1-\mu)^{-1}\p_{u}\phi)^2\Big](1-\mu) r^2 \sin \theta \de r_1 \de \theta \de \varphi 
\end{aligned}
\end{equation}
Here, $\chi_{\mathcal{H}^+}$ denotes a non-negative radial  cut-off function which is equal to $1$ in a neighborhood of $r=2M$ and vanishes outside a larger neighborhood of $r = 2M$. We note that now, since $(1-\mu)^{-1}\p_u = \p_t - \p_{r_1}$ in regular coordinates, the estimate controls a full set of non-degenerate derivatives at $r = 2M$. 

Combining the previous techniques used to obtain inequalities~\eqref{eq:morawaveeqnond} and~\eqref{eq:redshiftedmora}, we can finally obtain a Morawetz estimate with no degeneracies. A much more detailed introduction to the redshift effect can be found in~\cite{lecturenotes}, Section~3.3.

Since in our argument we will need to estimate higher-order derivatives of the field \emph{transversal to the event horizon}, we will also need to commute the equation with a non-degenerate vector field transversal to the event horizon: we will choose the vector field
$$
Y := (1-\mu)^{-1}\p_u.
$$ 
In Lemma~\ref{lem:pointpos} of Section~\ref{sub:redshcomm}, we will prove that this vector field enjoys good commutation properties.

We conclude this overview noting that, in the Schwarzschild case, one can then use the integrated local energy decay statement, along with the ``physical space method'' of Dafermos--Rodnianski~\cite{newmihalis} in order to obtain pointwise decay rates. 

Finally, a few remarks are in order concerning some subtle points in our argument where the geometry of the Schwarzschild spacetime plays a crucial role.

\begin{remark}\label{rmk:rinclose}
In the course of our argument, we will need to perform energy estimates ``slightly inside the black hole region''. In other words, we will need to choose, in certain estimates, the hypersurface $r = r_{\text{in}}$, with $0<r_{\text{in}}<2M$, and $r_\text{in}$ close to $2M$ as a boundary hypersurface (this can be visualized, for instance, in Figure~\ref{fig:statem}). An example in which we will need to use this technique is in the proof of  inequality~\eqref{eq:mixfinal} in Lemma~\ref{lem:qtot}. The first reason for this choice is that the hypersurface $r = r_\text{in}$ is \emph{uniformly spacelike}, and hence, in view of the positive energy condition, it will schematically give a positive future boundary term when performing energy estimates using the multiplier $\p_t$. 

Schematically, the reason why $r_{\text{in}}$ should be chosen close to $2M$ is connected with the redshift estimate~\eqref{eq:redshiftedmora}, since in general such an estimate (i.e., the positivity of the corresponding bulk term) is only valid in a certain region close to the event horizon. The crucial part of the argument in which this mechanism plays a role is Lemma~\ref{lem:qtot}. See also Remark~\ref{rmk:afterqtot}.
\end{remark}

\begin{remark}
One might argue that, in view of the positive energy condition, null hypersurfaces give positive boundary terms when performing energy estimates: hence, why not choose $\mathcal{H}^+$ as a future boundary hypersurface? As we already mentioned, the causal structure of the MBI theory differs from the underlying Schwarzschild causal structure (MBI is a \emph{quasilinear} system). Furthermore, when performing high-order energy estimates, we have to consider the ``commuted'' equations, and in that case we use the \emph{canonical stress} to perform our estimates. Note that this tensor (see Equation~\eqref{eq:dotqdef}) differs from the usual stress--energy--momentum tensor of the MBI theory (see Equation~\eqref{eq:semtdef}). Due to the discrepancy between the MBI geometry and the background Schwarzschild geometry, it is unclear how to show positivity of the resulting boundary fluxes on null hypersurfaces, when considering estimates arising from the~\emph{canonical stress}. In view of this, it is not clear a priori that the boundary term which would arise from the energy identity for higher--order derivatives of the MBI using $\mathcal{H}^+$ as a future bounding hypersurface would be positive. Hence the need to consider a ``nearby'' uniformly spacelike hypersurface to bound our region in the future.
\end{remark}

\begin{remark}
We focus on the study of the MBI system only in the \emph{exterior region} of Schwarzschild as our motivation comes from the black hole stability problem, and the exterior region of Schwarzschild is conjectured to be stable under gravitational perturbations. Furthermore, even the linear wave equation is very badly behaved in the interior of the Schwarzschild black hole, see for instance the recent work~\cite{FS2018}. Note that, in such paper, blow-up results are shown where the blow-up occurs at $r = 0$. Here, on the other hand, we always focus on the region $r \geq r_{\text{in}}$.
\end{remark}

We now give an outline of the strategy adopted in the present work.

\subsection{Structure of the proof}\label{sec:pfoverview}

Let us assume, in general, that we are considering a $4$-dimensional Lorentzian spacetime $(\mathcal M,g)$ whose canonical volume form is $\varepsilon$. Assume $\fara$ is an antisymmetric covariant two-form (the field tensor). Define the Lagrangian density
\begin{equation} \label{eq:densuno}
{}^*\mathscr{L}_{\text{MBI}} := \left(1 - \sqrt{1 +  \frac 1 2 \fara_{\mu\nu} \fara^{\mu\nu} - \left(\frac 1 4 \fara_{\mu\nu} \farad^{\mu\nu}\right)^2}\right)\varepsilon.
\end{equation}
Here, $\farad_{\mu\nu} := \frac 1 2 \varepsilon_{\mu\nu\kappa\lambda} \fara^{\kappa\lambda}$ denotes the Hodge dual.

The {\bf Euler--Lagrange equations for the Maxwell--Born--Infeld theory} are obtained by imposing that $\fara$ be closed as a two-form and be a critical point for closed and compactly supported variations of the action functional arising from the density ${}^*\mathscr{L}_{\text{MBI}}$:
\begin{equation}\label{eq:formmbiprel}
	d \fara = 0, \qquad H^{\mu\nu\kappa\lambda}[\fara] \  \nabla_\mu \fara_{\kappa\lambda} = 0.
\end{equation}
Here, $H^{\mu\nu\kappa\lambda}[\fara]$ is a $4$-contravariant tensor field which depends at least quadratically on $\fara$. See Section~\ref{MBIdef} for the precise definition.

Let us fix $M>0$, and consider the Schwarzschild spacetime $(\mathcal{S}, g)$. For a precise definition, see Section~\ref{schwasec}. Recall the coordinates $(t^*, r_1, \theta, \varphi)$ as in Definition~\ref{def:schwreg}.  We furthermore consider $r_{\text{in}} < 2M$, a number sufficiently close to $2M$.\footnote{The reason for this choice is technical, and is explained in Remark~\ref{rmk:rinclose} in Section~\ref{sub:aside}.}
Given $s \in \R$, we define $\widetilde{\Sigma}_{s} := \{(t^*, r_1, \theta, \varphi), t^* = s, r_1 \geq r_{\text{in}} \}$ as seen in the $(t^*, r_1, \theta, \varphi)$ coordinates. We also let $\tregio{s_1}{s_2}$ the spacetime region defined by $\{(t^*, r_1, \theta, \varphi), s_1 \leq t^* \leq  s_2, r_1 \geq r_{\text{in}} \}$. Let $t^*_0 >0$, and recall that we are imposing initial data $\fara_0$ on $\widetilde{\Sigma}_{t^*_0}$.

Define $\rho$ and $\sigma$ to be the \emph{middle components} of the MBI field, and $\alpha$, $\alphabar$ to be the \emph{extreme components} of the MBI field:
\begin{equation}\label{eq:rhodefprel}
	\rho:= \frac 1 {2(1-\mu)} \fara(\lbar, L), \ \  \sigma:= \frac 1 {2(1-\mu)} \farad(\lbar, L), \ \  \alpha_A := \fara( \partial_{\theta^A},  L), \ \  \alphabar_A:= \fara(\partial_{\theta^A}, \lbar),
\end{equation}
where $L := \partial_t + \partial_{r^*}$, $\lbar := \partial_t - \partial_{r^*}$, $\mu := \frac{2M}{r}$ (see Definition~\ref{def:rwcoords}), and finally $\theta^A$ is a local coordinate system for the conformal sphere $\mathbb{S}^2$.

Moreover, we define the sets of vector fields:
\begin{equation*}
\mathcal{V} := \{L, (1-\mu)^{-1}\lbar\}, \qquad
{\leo}:= \{\Omega_{i}\}_{i = 1, 2,3}, \qquad {\leo}_{\text{no}}:=  \{\Omega_{i}/r\}_{i = 1, 2,3}.
\end{equation*}
Here, $\Omega_{i}$ are three rotation (Killing) vector fields on the conformal sphere $\mathbb{S}^2$ whose linear span is 2-dimensional everywhere.

Furthermore, given $k \in \enn$\footnote{$\enn$ denotes the set of all non-negative integers: $\enn = \{0,1, \ldots\}$.}, and given a $m_1$-covariant tensor field $T$, we define the following norms:
\begin{equation*}
|T| := \sum_{V_1, \ldots, V_{m_1} \in \mathcal{V} \cup \ono}|T(V_1, \ldots, V_{m_1})|, \qquad 
|\partial^k T| := \sum_{\substack{V_1,\ldots, V_{m_1} \in \mathcal{V} \cup \ono\\
W_1, \ldots, W_k \in \mathcal{V} \cup \oscr}}|(\nabla_{W_1} \cdots \nabla_{W_k}T)(V_1, \ldots, V_{m_1})|.
\end{equation*}
Here, $\nabla$ denotes the Levi-Civita connection on $\mathcal{S}$. Furthermore, given $j \in \N_{\geq 0}$, we let
$$
|\partial^{\leq j} T| := \sum_{k =0}^j |\partial^k T|, \qquad  \qquad
	\norm{\fara}^2_{H^j(\widetilde{\Sigma}_{t^*})} := \int_{\widetilde{\Sigma}_{t^*}}|\partial^{\leq j}\fara|^2 \de \widetilde{\Sigma}_{t^*}.
$$
Here, $\de \widetilde{\Sigma}_{t^*} = r^2 \sin \theta \de r_1 \de \theta \de \varphi$ is the induced volume form on the surface $\widetilde{\Sigma}_{t^*}$.

Finally, let $\tau := (1-\chi_1(r))(u^+ +1)+ \chi_1(r)t^*,$  $\chi_1$ being a smooth radial cutoff function such that
\begin{equation*}
\chi_1(r) = 1 \text{ for } r \in [3/2 M, 3 M], \qquad \chi_1(r) = 0 \text{ for } r \in (0, M] \cup [4M, \infty).
\end{equation*}
Here, $u^+ := \max\{u, 0\}$ is the positive part.

We will setup a \textbf{bootstrap argument} in $t^*$ (see Figure~\ref{fig:init}): we assume some decay estimates in $L^\infty$ for the null components of the field in the region $\tregio{t_0^*}{t_1^*}$, and we seek to then improve those estimates (by proving that in fact they hold with a better constant). There are three key ingredients to carry out this program, which we describe here.

	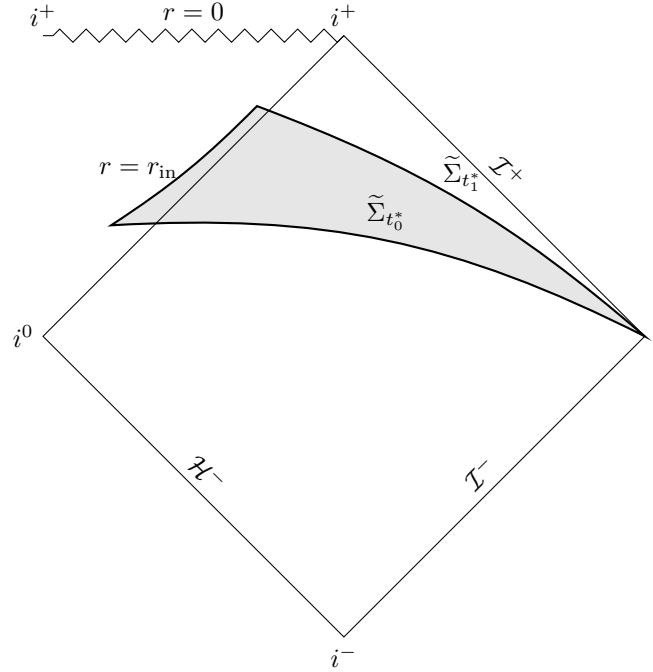
\begin{figure}[H]
		\centering
		\begin{tikzpicture}	

\node (I)    at ( 0,0) {};

\path
  (I) +(90:4)  coordinate[label=90:$i^+$]  (top)
       +(-90:4) coordinate[label=-90:$i^-$] (bot)
       +(0:4)   coordinate                  (right)
       +(180:4) coordinate[label=180:$i^0$] (left)
       +(180:4.9) coordinate[label=180:$ $] (left1)
       ;
 \path (top) +(180:0.01) coordinate (tar);

\path 
	(top) + (180:4) coordinate[label=90:$i^+$]  (acca)
		+ (-45: 2) coordinate (nulluno)
		+ (-45: 3) coordinate (nulldue)
	;

\path 
	(left) + (-45: 2) coordinate (correspuno)
		+ (-45: 3) coordinate (correspdue)
	;

\path (top) + (-150: 1.7) coordinate (horiuno)
			+ (-150: 5) coordinate (horidue);
\draw[name path = stre, opacity = 0] (left1) to (tar);

\draw [name path = ciao, thick, opacity = 0] (horiuno) to [bend left = 15] node[midway, above]{} (right);
\draw [name path = bye, thick, opacity = 0] (horidue) to [bend left = 10] node[midway, above]{} (right);
\path [name intersections={of=ciao and stre,by=E}];
\path [name intersections={of=bye and stre,by=D}];
\draw [thick, fill=gray!20] (D) to [bend left = 15] node[midway, above]{$\widetilde{\Sigma}_{t^*_0}$} (right) to [bend right = 10] node[midway, above]{$\widetilde{\Sigma}_{t^*_1}$} (E) to [bend left = 6] node[midway, left]{ $r = r_{\text{in}}$} (D) to cycle;
\draw (left) -- 
          node[midway, above left, sloped]    {}
      (top) --
          node[midway, above, sloped] {$\mathcal{I}^+$}
      (right) -- 
          node[midway, above, sloped] {$\mathcal{I}^-$}
      (bot) --
          node[midway, above, sloped]    {$\mathcal{H}^-$}    
      (left) -- cycle;
      
\draw[decorate,decoration=zigzag] (top) -- (acca)
      node[midway, above, inner sep=2mm] {$r=0$};
      
\end{tikzpicture}
		\caption{Penrose diagram depicting $\widetilde{\Sigma}_{t_0^*}$ and the bootstrap region $\tregio{t^*_0}{t_1^*}$ highlighted in gray color.}\label{fig:init}
	\end{figure}

	{\bf (1)} We first prove a \textbf{local existence statement}:
	\begin{theorem}[Informal version of Theorem~\ref{thm:local}]
		Let $\fara_0$ be initial data for the MBI system on $\widetilde{\Sigma}_{t_0^*}$. We suppose that
		\begin{itemize}
			\item $\fara_0$ is smooth,
			\item $\fara_0$ is sufficiently small in a higher-order Sobolev norm,
			\item $\fara_0$ satisfies the constraint Equations~\eqref{eq:constreq} on $\widetilde{\Sigma}_{t_0^*}$.
		\end{itemize}
		Then, there exist a time $t_1^*$ and a smooth solution $\fara$ to the MBI system~\eqref{eq:formmbiprel} on $\tregio{t_0^*}{t_1^*} = \cup_{t^* \in [t_0^*,t_1^* ]} \widetilde{\Sigma}_{t^*}$ having $\fara_0$ as initial data.
	\end{theorem} 
	
	\begin{remark}
	The issue of local well-posedness here is a nonstandard one. This is primarily due to the fact that, in view of the presence of nontrivial geometry, it is unclear how to write the MBI system as a symmetric hyperbolic system. Indeed, in the work~\cite{Brenier2004}, Brenier shows local existence for the MBI system on Minkowski, defining additional quantities which, coupled with the MBI system itself, complete the system to a $10 \times 10$ symmetric hyperbolic system. It is unclear how to generalize this to our setting. Here, on the other hand, we use an alternative manifestation of the hyperbolic character of the system, namely the fact that it admits a \emph{canonical stress}. This is the nonlinear analogue of the stress--energy--momentum tensor, and it satisfies the crucial property that its divergence is of lower order in terms of numbers of derivatives (see Lemma~\ref{lem:qdivecan}, in the linearized case). This enables energy estimates on which the Gr\"onwall inequality can be used.
	\end{remark}
	
	\begin{remark}
	We also remark that, when performing such energy estimates, it is important that the future boundary terms be positive. This essentially amounts to the fact that the hypersurfaces bounding the region of integration to the future should be spacelike. We already noted, nevertheless, that the causal structure of the MBI theory differs from the underlying one in Schwarzschild, as MBI is a \emph{quasilinear} theory, hence posing the problem of how to choose those hypersurfaces. There are at least two ways of resolving this issue. The first is to exactly solve the eikonal equation associated to MBI and thus determine its causal geometry. The second is to assume that the initial data is small, so that the causal structure of MBI be sufficiently close to the one of Schwarzschild: it will then suffice to take uniformly spacelike hypersurfaces relative to the Schwarzchild metric. We follow this second route, since our end goal is anyway a global stability result for small data, and since it is way more straightforward.
	\end{remark}

	The local existence theorem is the constructive element of our proof. We obtain it by linearizing the MBI system first (in Section~\ref{lintheory}). The choice of foliation $\widetilde \Sigma_{t^*}$ leads naturally to a splitting of the MBI system into a set of $6$ equations of hyperbolic character, which we call the \emph{evolution part}, and a remaining set of $2$ equations which involve derivatives only intrinsic to the foliation $\widetilde \Sigma_{t^*}$, which we call the \emph{MBI constraint equations}. We then proceed to solve the evolution part of the system by means of linearization: our method of proof relies on a-priori estimates (Section~\ref{sub:apriorilocal}), and the Cauchy--Kowalevskaya Theorem to solve the linearized problem. We use approximation by polynomials and then take limits, noting that the energy estimates will provide the necessary compactness (see Section~\ref{sub:locallin}). This shows existence for the evolution part of the system. Finally, we prove in Section~\ref{sub:proofloc} that \emph{only at the nonlinear level}, the constraints are propagated (this is finally done in the proof of Theorem~\ref{thm:local} in Section~\ref{sub:proofloc}). This shows local existence of solutions to the MBI system.
	
	\vspace{10pt}
	
	{\bf (2)} We then proceed to set up the \textbf{bootstrap assumptions} (BA) (this is done in Section~\ref{sec:bootstrap}). For simplicity, in this outline we focus only on the exterior region $\mathcal{S}_e$. We say that the MBI field $\fara$ satisfies the bootstrap assumptions $BA\left(\mathcal{S}_e,A ,j,\varepsilon\right)$ if on $\mathcal{S}_e$ there holds:
	\begin{equation}\label{eq:bootstrapprev}
	\boxed{
		\begin{aligned}
		|\partial^{\leq j} \rho|, |\partial^{\leq j} \sigma|,|\partial^{\leq j} \alpha| \leq A \varepsilon^{\frac 3 4} \min \{\tau^{-1} r^{-3/2}, \tau^{-1/2} r^{-2}\} ,\qquad
		|\partial^{\leq j} \alphabar| \leq A \varepsilon^{\frac 3 4} \tau^{-1} r^{-1}.
		\end{aligned}
	}
	\end{equation}
	\begin{remark}
		We remark that $\alphabar$ enjoys the worst decay rate of all components. Furthermore, we obtain the sharp decay rate in $r$ for $\alphabar, \rho$ and $\sigma$. This is the optimal decay rate which can be obtained from the $r^p$ method of Dafermos--Rodnianski (compare with the linear case in~\cite{blue, masthes}, in which exactly the same decay rates for said components are proved). Even though these decay rates are optimal in $r$, the $u$ decay, in reality, is expected to be much stronger for all components. Moreover, for $\alpha$, we do not get the sharp $r$-decay as, in the linear case, it should be possible to show $\alpha \sim r^{-3}$ for fixed $u$-coordinate, under strong conditions on initial data (compare again with~\cite{blue, masthes}). Here, on the other hand, we only show $\alpha \sim r^{-2}$ for fixed $u$-coordinate. The reason why we are not able to obtain sharp decay rates in $u$ lies intrinsically in the limitations of the method we use. Furthermore, we do not obtain sharp decay rates in $r$ for $\alpha$ as we do not need them to close the argument (even though we expect to be able to obtain sharp $r$-decay for $\alpha$, upon application of the methods in~\cite{masthes}).
		
		Moreover, as we already noted, the $u$-decay is conjectured to be much stronger, at least in the linear case. In fact, in the paper~\cite{Price2004}, the authors give a heuristic argument by which the decay rate on a region of bounded $r$ in Schwarzschild should be $(t^*)^{-3}$. This has yet to be proved in the linear case on Schwarzschild. On Minkowski spacetime, on the other hand, Christodoulou and Klainerman~\cite{CK1990} were able to prove $t^{-\frac 52}$ decay in the region where $r \leq \frac t 2$.
		
		Finally, we note that, in the case of Minkowski spacetime, again in the paper~\cite{CK1990} (Theorem~3.2), Christodoulou and Klainerman were able to prove the following decay rates for the Maxwell field, in the region $r \geq \frac  t 2$:
		$$
		|\alphabar| \leq C r^{-1} u^{- \frac 32}, \qquad |\rho|, |\sigma| \leq C r^{-2} u^{-\frac 12}, \qquad |\alpha| \leq C v^{-\frac 52}.
		$$
		These rates carry over to the MBI system, in the small data regime. Indeed, Speck was able to show the same decay rates, in~\cite{speck1} (see Proposition~10.1), in the context of a nonlinear bootstrap argument.
	\end{remark}
	Under the bootstrap assumptions outlined above, we deduce uniform $L^2$ bounds for higher--order Sobolev norms in Proposition~\ref{prop:energy}.
	\begin{proposition}[Informal version of Proposition~\ref{prop:energy}]
	Let $N \in \N_{\geq 0}$. There exists $C >0$ such that the following holds. Assume that $\fara$ is a solution to the MBI system on $\mathcal{S}_e$.\footnote{This is schematic. In practice, we will assume the solution to exist and satisfy the bootstrap assumptions only in a subset of $\mathcal{S}_e$.} We also assume
	\begin{itemize}
		\item the bootstrap assumptions $BA\left(\mathcal{S}_e, 1 ,N,\varepsilon \right)$,
		\item smallness in $L^2$ of $2N$ derivatives of $\fara_0$: $\norm{\fara_0}_{H^{2N}(\widetilde{\Sigma}_{t_0^*})} \leq \varepsilon^2$.
	\end{itemize}
	Then, for all $t^* \geq t_0^*$ we have, under these conditions,
	\begin{equation}\label{eq:energyprel}
	\norm{\fara}^2_{H^{2N}(\widetilde{\Sigma}_{t^*})} \leq C \varepsilon^2.
	\end{equation}	
	\end{proposition}
	In order to prove this Proposition, we use the canonical stress tensor $\dot Q$, contracted once with the Killing vector field $\partial_t$. The tensor $\dot Q$ is the analogue of the stress--energy--momentum tensor for higher derivatives, and has the property that its covariant divergence does not contain top-order derivative terms.
	For the purposes of this outline, one can think of $\dot Q$ as being exactly the same as its Maxwell counterpart, applied to the field $\dot \fara$. $\dot \fara$ is, schematically, the tensor field $\fara$ commuted $N$ times with unit length partial derivatives. We therefore have, schematically,
	$$
	\dot Q_{\mu\nu} \approx Q^{\text{(MW)}}_{\mu\nu}[\dot \fara] =\dot \fara\indices{_\mu^\alpha}\dot \fara_{\nu\alpha} - \frac 14 g_{\mu\nu} \dot \fara^{\alpha \beta}\dot \fara_{\alpha \beta}.
	$$
	Given a vector field $X$, we then define the current as the contraction: $\dot J^X_\nu := Q^{\text{(MW)}}_{\mu\nu}[\dot \fara] X^\mu$. Recall also that we denote $T := \p_t$. Schematically, one obtains the following estimate (this is essentially the content of Lemma~\ref{lem:degen} in Section~\ref{sub:degen}):
	\begin{equation}\label{eq:schematic1}
	\begin{aligned}
	&\int_{\Sigma_{t^*_1}}g(\dot J^T, n_{\Sigma_{t^*}}) \de \Sigma_{t^*} - \int_{\Sigma_{t^*_0}} g(\dot J^T, n_{\Sigma_{t^*}}) \de \Sigma_{t^*}\leq \int_{t^*_0}^{t^*_1}\Big(\sup_{\Sigma_s}|\p^N \fara|^2 \Big)\Vert \fara \Vert^2_{H^{2N}(\Sigma_s)}  \de s \\
	&\qquad \leq  C\eps^2 \int_{t^*_0}^{t^*_1}s^{-2}\Vert \fara \Vert^2_{H^{2N}(\Sigma_s)}  \de s
	\end{aligned}
	\end{equation}
	Here, we used the bootstrap assumptions on the lower order terms. Note also that $n_{\Sigma_{t^*}}$ is defined as the future-directed Lorentz unit normal to the foliation $\Sigma_{t^*}$. If this estimate had been carried out in Minkowski space, we would have had, schematically (due to the fact that $\p_t$ is uniformly spacelike), that $g(\dot J^T, n_{\Sigma_{t^*}}) \geq C|\dot \fara|^2$. This would have allowed for an immediate application of the Gr\"onwall inequality in the case of Minkowski space.
	
	In the Schwarzschild case, on the other hand, we have degeneracy at the event horizon, since the quantity $g(\dot J^T, n_{\Sigma_{t^*}})$ only controls $(1-\mu)^2 \dot \fara(\p_{r_1}, \p_{\theta^A})$ (we are using $(t^*, r_1)$ regular coordinates). We therefore need to prove an estimate which will have non-degenerate control on the quantity $\fara(\p_{r_1}, \p_{\theta^A})$ near $\mathcal{H}^+$. This is precisely the estimate arising from using the redshift vector field $\redsh$ as a multiplier. $\redsh$ is timelike everywhere and furthermore satisfies good bulk positivity properties near $\mathcal{H}^+$, and coincides with $\p_t$ for $r$ large. We prove such estimate in Section~\ref{sub:redshcomm}. We obtain, schematically:
	\begin{equation}\label{eq:schematic2}
	\begin{aligned}
	&\int_{\Sigma_{t^*_1} \cap \{r \leq r_{\text{out}}\}} g(\dot J^{\redsh}, n_{\Sigma_{t^*}}) \de \Sigma_{t^*} - \int_{\Sigma_{t^*_0}\cap \{r \leq r_{\text{out}}\}} g(\dot J^{\redsh}, n_{\Sigma_{t^*}}) \de \Sigma_{t^*}\\
	&  +  \int_{t^*_0}^{t^*_1} \int_{\Sigma_s \cap \{r \leq r_{\text{out}}\}}  g(\dot J^{\redsh}, n_{\Sigma_{t^*}}) \de \Sigma_{t^*} \de s \\
	&\qquad \leq C \int_{t^*_0}^{t^*_1}\Big(\sup_{\Sigma_s}|\p^{N} \fara|^2 \Big)\Vert \fara \Vert^2_{H^{2N}(\Sigma_s \cap \{r \geq r_{\text{out}}\})}  \de s + C \int_{t^*_0}^{t^*_1} \int_{\Sigma_s \cap \{r \geq r_{\text{out}}\}}  g(\dot J^{T}, n_{\Sigma_{t^*}}) \de \Sigma_{t^*} \de s,	
	\end{aligned}
	\end{equation}
	for some $r_{\text{out}} > 2M$. Note that now $g(\dot J^{\redsh}, n_{\Sigma_{t^*}})$ controls all components of the field $\dot \fara$, with no degeneracy.  Summing the term 
	$$\int_{t^*_0}^{t^*_1} \int_{\Sigma_s \cap \{r \geq r_{\text{out}}\}}  g(\dot J^{\redsh}, n_{\Sigma_{t^*}}) \de \Sigma_{t^*} \de s$$
	 on both sides of inequality~\eqref{eq:schematic2}, we obtain a differential inequality (see \eqref{togron}) which schematically looks like:
	\begin{equation}\label{eq:schematic3}
	\begin{aligned}
	&\int_{\Sigma_{t^*_1}}g(\dot J^{\redsh}, n_{\Sigma_{t^*}}) \de \Sigma_{t^*} - \int_{\Sigma_{t^*_0}}g(\dot J^{\redsh}, n_{\Sigma_{t^*}}) \de \Sigma_{t^*}+  \int_{t^*_0}^{t^*_1} \int_{\Sigma_s}  g(\dot J^{\redsh}, n_{\Sigma_{t^*}}) \de \Sigma_{t^*} \de s\\
	&\qquad \leq C \eps^2\int_{t^*_0}^{t^*_1}\frac 1 {s^2}\Vert \fara \Vert^2_{H^{2N}(\Sigma_s \cap \{r \geq r_{\text{out}}\})}  \de s + \int_{t^*_0}^{t^*_1} \int_{\Sigma_s \cap \{ r \geq r_{\text{out}}\}}  g(\dot J^{T}, n_{\Sigma_{t^*}}) \de \Sigma_{t^*} \de s \\
	&\qquad \leq C\eps^2 \int_{t^*_0}^{t^*_1}\frac 1 {s^2}\Vert \fara \Vert^2_{H^{2N}(\Sigma_s \cap \{r \geq r_{\text{out}}\})}  \de s + (t^*_1-t^*_0)\underbrace{\sup_{s \in [t_0^*, t_1^*]} \int_{\Sigma_s \cap \{ r \geq r_{\text{out}}\}}  g(\dot J^{T}, n_{\Sigma_{t^*}}) \de \Sigma_{t^*} \de s}_{(i)} ,
	\end{aligned}
	\end{equation}
	We note that we can control term $(i)$ by the estimate~\eqref{eq:schematic1}. We then obtain a differential inequality which, in virtue of the presence of a positive bulk term on the LHS (the third term in the first line of display~\eqref{eq:schematic3}), can be integrated to obtain the required boundedness. This trick is familiar from work on the wave equation on Schwarzschild, and is inspired by the treatment in the lecture notes~\cite{lecturenotes},~Section 3.3.3. The trick is carried out in full detail in Section~\ref{sub:concludingl2}, specifically in the proof of Proposition~\ref{prop:energy}.
	\begin{remark}
	As a technical remark, recall that we are not allowed to choose the event horizon as a bounding hypersurface for our estimates, due to the causal structure of MBI being different from that of Schwarzschild. Therefore, here we need to integrate ``slightly inside'' the black hole region, i.~e.~up to the surface $r = r_{\text{in}}$, with $r_{\text{in}} < 2M$. This property is used crucially in the proof of Lemma~\ref{lem:qtot}.
	\end{remark}
	Having established uniform $L^2$ boundedness, our next task will be to recover the $L^\infty$ estimates at the level of low derivatives from the $L^2$ bounds.
	
	\vspace{10pt}
	
	{\bf (3)} The \textbf{key observation} is then that, if $\fara$ satisfies the MBI system, $\rho$ and $\sigma$ satisfy the \textbf{nonlinear \fackip equations}:
	\begin{equation}\label{eq:spinref}
	\boxed{
		- r^{-2} L \lbar (r^2 \rho) +(1-\mu) \slashed{\Delta}  \rho = \text{\bf NL}_1 + \text{\bf NL}_2, \qquad	- r^{-2} L \lbar (r^2 \sigma) + (1-\mu)\slashed{\Delta}  \sigma = \text{\bf NL}_3.}
	\end{equation}
	The right hand sides of these equations are nonlinear in $\fara$ and its derivatives up to order $2$. For the exact form of such nonlinear terms see Section~\ref{sec:spinredu}, in particular Equations~\eqref{eq1.3} and~\eqref{eq:wavesigma}. 
	
	\begin{remark}
		Let us notice that the corresponding scalar \fackip Equations for the Maxwell linear field on Schwarzschild hold equating the nonlinear terms to zero  $\text{\bf NL}_1 + \text{\bf NL}_2= 0$ and $\text{\bf NL}_3 = 0$ in the previous display~\eqref{eq:spinref}.
	\end{remark}
	
	\begin{remark} These \fackip equations exhibit stationary solutions, as shown in Proposition~\ref{prop:statsols} in the Appendix. This is clearly a potential obstruction to decay. To deal with this issue, we need to control the spherical averages of $\sigma$ and $\rho$. We prove the following statement. Under the assumption that, on initial data, the charge (spherical average of $\sigma$ and $\rho$ on spheres of constant $r$-coordinate) vanishes at spacelike infinity, and under the bootstrap assumptions, we show that the spherical average of $\sigma$ is always zero, and the spherical average of $\rho$ decays sufficiently fast along the evolution. The proof follows directly from the conservation laws provided by the MBI system integrated on spheres of constant $r$-coordinate (see Equation~\eqref{sphavg}). Finally, it is not difficult to show (at least at the linear level) that solutions of Equations~\eqref{eq:spinref}, once projected away from the $l = 0$ spherical harmonic mode, decay with quantitative decay rates to the trivial solution. This reasoning (more precisely, its suitable adaptation to the MBI system), together with the decay of spherical averages of $\sigma$ and $\rho$, will prove that $\sigma$ and $\rho$ are in fact decaying in time to the trivial solution. In particular, this will exclude the presence of nontrivial stationary solutions.
	\end{remark}
	
	We would now like to use the uniform estimates~\eqref{eq:energyprel} in order to deduce that there exists a constant $C$ and an integer $j \in \N_{\geq 0}$ independent of the data, such that the bootstrap assumptions $BA\left(\mathcal{S}_e,1 ,j,\varepsilon\right)$ imply the statement $BA\left(\mathcal{S}_e,C ,j,\varepsilon^2 \right)$, with $C$ independent of $\eps$.
	This essentially amounts to showing improved decay for the null components of the field $\fara$. 
	
	Let us first focus on how we prove decay for the middle components $\sigma$ and $\rho$. For simplicity of exposition, as the equations~\eqref{eq:spinref} have a similar structure in this outline we focus on the component $\rho$, and we rename $Z := r^2 \rho$. We have the equation:
	\begin{equation}\label{eq:prez}
	L \lbar Z -(1-\mu) \slashed{\Delta} Z = R_\rho,
	\end{equation}
	where $R_\rho :=- r^2( \text{\bf NL}_1 + \text{\bf NL}_2)$.
	
	We wish to follow the strategy by Dafermos and Rodnianski in~\cite{newmihalis} to prove decay for $Z$. The $r^p$ method requires three ingredients:
	\begin{enumerate}
	\item[{\bf(a)}] a Morawetz (integrated local energy decay) estimate,
	\item[{\bf(b)}] the $r^p$ multiplier estimates, with $p = 0, 1, 2$, and
	\item[{\bf(c)}] an energy inequality.
	\end{enumerate}
	We now describe points {\bf (a), (b), (c)} in some detail.
	\vspace{10pt}
	
	{\bf (a)} Let us first focus on the {\bf Morawetz estimate}, which is obtained in Section~\ref{sec:premora}, Section~\ref{sec:nonlstruct}, and Section~\ref{sec:morawetz}. This inequality is obtained essentially multiplying Equation~\eqref{eq:prez} by $f(r)\p_{r^*} Z$, with a suitable choice of a smooth, radial function $f(r)$, and adding a lower-order correction. The future (and horizon) boundary terms are then estimated using a $\p_t$ multiplier estimate on the same Equation~\eqref{eq:prez}.
	
	The main difficulty here consists in estimating the nonlinear errors which appear on the right hand side of our estimates. The analysis of such error terms is carried out in Section~\ref{sec:nonlstruct}. Schematically, the worst term (in terms of decay) in the expression for $R_\rho$ behaves like:
	$$
	|R_\rho| \sim r |\dot \alpha| |\alphabar|^2 \sim r^{-1} \tau^{-2} |\dot \alpha|.
	$$
	A dotted quantity here signifies a quantity which is differentiated a ``top order'' number of times, and thereby cannot be estimated using the $L^\infty$ estimates arising from the bootstrap assumptions. Recall furthermore that $\alphabar$ is the worst decaying of all terms.
	\begin{remark}
	The fact that a cubic term in $\alphabar$ never appears follows from the fact that the MBI system satisfies a version of the ``null condition'', i.~e.~the nonlinearities never displays three worst-decaying terms ($\alphabar$) multiplied together.
	\end{remark}
	Thus, when estimating the error terms arising in the Morawetz estimate, we have that the worst term is
	\begin{equation}\label{eq:schematicmor}
	\int_{\regio{t_0^*}{t_1^*}} r^2 |R_\rho|^2 \de u \de v \desphere \sim \int_{\regio{t_0^*}{t_1^*}} \tau^{-2} |\dot \alpha|^2 \de u \de v \desphere \sim \int_{t_0^*}^{t_1^*} s^{-2} \int_{\Sigma_{s}} |\dot \alpha|^2 \de \Sigma_{t^*} \de s \leq \varepsilon^2.
	\end{equation}
	Here, we used the uniform energy bound on ``high'' derivatives from point $(2)$ of this outline, as well as the bootstrap assumptions, and the fact that $\tau r \geq C t^*$, for some constant $C > 0$. We therefore see how it is possible to bound the nonlinear error terms arising from the Morawetz estimate.
	\begin{remark}
	As a technical remark, we also need to prove a redshift estimate in order to control derivatives transversal to the future event horizon $\mathcal{H}^+$ (see Proposition~\ref{prop:redsh}). Furthermore, this has to be carried out in a commutation setting, since we need higher order estimates (see Section~\ref{sec:commutrw}).
	\end{remark}
	\vspace{15pt}
	
	{\bf (b)} We now outline how we deal with the $\boldsymbol{r^p}$ {\bf multiplier estimates}. The $r^p$ estimates for $Z$ (and the similar ones for $\sigma$), with $p=1$,  are carried out in Section~\ref{sec:pweighted1}. Note that the estimates in said section give control only on fluxes through the surfaces $\Sigma_{t^*}$.\footnote{The reader familiar with the $r^p$-multiplier method may be confused at this stage, as the usual application of said method requires the foliation to be outgoing null ``close to future null infinity''. Indeed, the $r^p$ estimates of Section~\ref{sec:pweighted1} are preliminary and their purpose is uniquely to give additional control on the nonlinear error terms, which in turn allows us to set $p=2$ in the $r^p$ estimates of Section~\ref{sec:previsited}. It is in this section that we will use a foliation which is outgoing null ``close to future null infinity''.} We start by multiplying Equation~\eqref{eq:prez} by $r^p L Z$ and integrating by parts. The LHS of the resulting estimate then, schematically, will give us control over the following spacetime integral:
	$$
	\int_{\regio{t_0^*}{t_1^*}} r^{p-1} |L Z|^2 \de u \de v \desphere
	$$
	(recall that $L = \p_v$ in the null $(u,v)$ coordinates). On the other hand, the RHS will contain the error terms
	$$
	\int_{\regio{t_0^*}{t_1^*}} r^p |R_\rho| |LZ| \de u \de v \desphere.
	$$
	Hence, using the Cauchy--Schwarz inequality, we see that, in order to close the $r^p$ estimates, it suffices to control
	$$
	\int_{\regio{t_0^*}{t_1^*}} r^{p+1} |R_\rho|^2 \de u \de v \desphere.
	$$ 
	If $p =1$ (see Equation~\eqref{eq:schematicmor}), we obtain the same error term as in the Morawetz estimate, which we are able to bound. On the other hand, if $p =2$, the resulting expression barely fails to be integrable, as the resulting power of $t^*$ is effectively $-1$.\footnote{We have, schematically, that $$ \int_{\regio{t_0^*}{t_1^*}} r^{3} |\alphabar|^4 |\dot \alpha|^2 \de u \de v \desphere \leq C  \int_{\regio{t_0^*}{t_1^*}} r \tau^{-4}|\dot \alpha|^2 \de u \de v \desphere \leq C \int_{\regio{t_0^*}{t_1^*}} (t^*)^{-1} |\dot \alpha|^2 \de \Sigma_{t^*} \de t^*,$$ which is not integrable using the uniform $L^2$ estimates arising from {\bf Step (2)}.} Hence, we have to improve our way of estimating the term $|R_\rho| \sim r |\dot \alpha| |\alphabar|^2$. As we previously noted, even though we cannot put the $\dot \alpha$ term in $L^\infty$, we can improve its $L^2$ estimates in terms of $r$-weights, as we expect this term to have very good $r$-decay. This is carried out in Section~\ref{sec:l2improved}, in the following way.
	
	We note that $\alpha$ satisfies the following transport equation:
	\begin{equation}\label{eq:alphatrprel}
	\boxed{
		\snabla_\lbar(r \alpha_A) = r(1-\mu)(\snabla_A \rho + \svol_{AB}\snabla^B \sigma ) + r(1-\mu)(\hdelta)_{\mu A \kappa \lambda} \nabla^{\mu} \fara^{\kappa\lambda}.
	}
	\end{equation}
	For the precise definition of each term in this equation, see Section~\ref{sec:altmbi}, and Section~\ref{sec:not:geo}. Here, $\snabla_\lbar$ is the covariant derivative in direction $\lbar$ projected on spheres of constant $r$-coordinate, whereas $\snabla_A$ is the covariant derivative in direction $\p_{\theta^A}$, again projected on spheres of constant $r$-coordinate. The term 
	$$
	r(1-\mu)(\hdelta)_{\mu A \kappa \lambda} \nabla^{\mu} \fara^{\kappa\lambda},
	$$ finally, is a nonlinear term which contains all components of $\fara$.
	
	By multiplying an appropriately commuted version of equation~\eqref{eq:alphatrprel} by $\dot \alpha$, and integrating by parts, we schematically obtain the following estimate, valid for all $s \geq t_0^*$ (recall that $\dot \alpha$ is differentiated by the top number of derivatives):
	\begin{equation}
	\label{eq:impralpha}
	\int_{\Sigma_{s}}  r|\dot \alpha|^2 \de \Sigma_{t^*} \leq C \eps^2.
	\end{equation}
\begin{remark}
To prove~\eqref{eq:impralpha}, we expect to be able to control the nonlinear error terms, since such estimate is schematically like the $r^p$ estimate with $p=1$, which can be controlled by the above. This $p=1$ estimate for the $\alpha$ component is the content of Proposition~\ref{prop:impalpha}.
\end{remark}
	We will also need the corresponding $L^2$ improvement for the middle components $\sigma$ and $\rho$:
	\begin{equation}\label{eq:imprrs}
	\int_{\Sigma_{s}}  r(|\dot \sigma |^2 +|\dot \rho |^2 )\de \Sigma_{t^*} \leq C \eps^2.
	\end{equation}
	This is essentially obtained by the $r^p$ estimate, with $p=1$, on Equation~\eqref{eq:prez} (and the corresponding equation for $\sigma$). Said $r^p$ estimate (with the correct spacelike boundary terms) was the content of Section~\ref{sec:pweighted1}, and the resulting bounds for $\sigma$ and $\rho$ are recorded in Proposition~\ref{prop:imprhosigma}. To conclude, armed with the improved $L^2$ estimates~\eqref{eq:impralpha},~\eqref{eq:imprrs} we can now bound the integral (corresponding to the $r^p$ estimate with $p =2$):
	$$
	\int_{\regio{t_0^*}{t_1^*}} r^{3} |R_\rho|^2 \de u \de v \desphere \leq C \eps^2.
	$$
	This bound is achieved in Section~\ref{sec:improvedrs}. The reasoning we just outlined shows that we are allowed to choose $p=1,2$ in the $r^p$ estimates for $Z$: we collect those estimates in Proposition~\ref{prop:previsited}, which is the main proposition in Section~\ref{sec:previsited}, where the $r^p$ estimates are recorded, with $p = 2, 1$. These estimates are suitable to be used in the $r^p$ method as they give control of energies on a foliation which is outgoing null ``close to future null infinity'', as opposed to the estimates of Section~\ref{sec:pweighted1}, which are on the foliation $\Sigma_{t^*}$.
\vspace{15pt}

	{\bf (c)} The last ingredient of the $r^p$ multiplier method is the {\bf energy inequality}. This is obtained by multiplying both sides of Equation~\eqref{eq:prez} by $\p_t Z$, and integrating by parts. We require the future and past boundary terms be strictly spacelike for $r \leq R$, and outgoing null for $r > R$, where $R$ is some big radius which is fixed at the beginning of our proof. We denote such fluxes by $\fgiustor$. The main challenge here is to control the error terms arising in the bulk of the resulting estimates. This is the content of Proposition~\ref{prop:encfgiusto}, whose statement not only contains a form of energy conservation for fluxes of the form $\fgiusto$, but also contains the associated Morawetz estimate having only terms of the form $\fgiusto$ on the RHS.
	\begin{remark}
	As a technical remark, let us note that in order to close said energy inequality, we also need to prove a Morawetz estimate whose boundary terms are of the form $\fgiustor$. This is carried out in Proposition~\ref{prop:morafgiusto}.
\end{remark}\vspace{15pt}

After having concluded the $r^p$ estimates, we run through standard procedure to get decay of the fluxes $\fgiustor$ in Section~\ref{sec:pdecay}.

By Sobolev embedding, we can immediately deduce, using the decay of the $\fgiustor$ fluxes, pointwise decay for $\rho$ (and the corresponding decay for $\sigma$): this is the content of Section~\ref{sec:sobrs}. We are therefore only left with proving $L^\infty$ decay for $\alpha$ and $\alphabar$. 

In order to achieve decay for $\alpha$, in the spirit of~\cite{masthes}, we will use the transport equation satisfied by $\alpha$ itself (Equation~\eqref{eq:alphatrprel}). Note, however, that the reasoning in this work is rough, as we \emph{only} use the $L^\infty$ estimates on $\rho$ and $\sigma$ we already obtained, in conjunction with the uniform boundedness of $\alpha$ and $\alphabar$ resulting from the uniform energy estimates from point {\bf (1)} and Sobolev embedding. This reasoning is wasteful in terms of loss of derivatives (we could instead perform $L^2$-based estimates), nevertheless it spares us some technical difficulties. We integrate Equation~\eqref{eq:alphatrprel} in the $\lbar$ direction, and this gives decay for $\alpha$ (see Proposition~\ref{prop:linfalpha}).

Finally, we have to obtain pointwise estimates for $\alphabar$. These are achieved similarly (in Proposition~\ref{prop:decayalphabar}) by considering the transport equation satisfied by $\alphabar$ itself:
\begin{equation}\label{eq:alphabartrapre}
	\boxed{
		\snabla_L(r \alphabar_A) = r(1-\mu)(- \snabla_A \rho + \svol_{AB}\snabla^B \sigma ) +r(1-\mu)(\hdelta)_{\mu A \kappa \lambda} \nabla^{\mu} \fara^{\kappa\lambda}.
	}
\end{equation}
Upon integration of such equation in the $L$ direction, and using the $L^\infty$ estimates obtained so far, we are able to show the $L^\infty$ estimates for $\alphabar$. 

We finally close the bootstrap argument in Section~\ref{sec:close}.

\begin{remark}
We also note that, in addition to these structural identities involving $\alpha$ and $\alphabar$ (Equations~\eqref{eq:alphatrprel} and~\eqref{eq:alphabartrapre}), we will need, in the course of our argument, the other transport equations satisfied by $\rho$ and $\sigma$, which follow from writing the MBI system in null form. Due to their importance, we include them here:
\begin{align}\label{eq:transprho1}
&\boxed{
-L(r^2 \rho) + r^2 \dive \alpha = - r^2{H_{\Delta}} \indices{^\mu _L ^\kappa ^\lambda} \nabla_\mu \fara_{\kappa \lambda}.
}\\ \label{eq:transprho2}
&\boxed{
\hat \lbar(r^2 \rho) + r^2 (1-\mu)^{-1}\dive \alphabar = - r^2{H_{_\Delta}} \indices{^\mu _{\hat \lbar} ^\kappa ^\lambda} \nabla_\mu \fara_{\kappa \lambda}.}
\end{align}

Also, we have
\begin{align}\label{eq:transpsigma1}
&\boxed{
L(r^2 \sigma)+r^2 \curl \alpha = 0.
}\\ \label{eq:transpsigma2}
&\boxed{
- \hat \lbar(r^2 \sigma) + (1-\mu)^{-1} r^2 \curl \alphabar= 0.}
\end{align}
Equations~\eqref{eq:alphatrprel} and~\eqref{eq:alphabartrapre}, together with~\eqref{eq:transprho1}--\eqref{eq:transpsigma2} correspond to the full MBI system written in null form. Furthermore, those are 8 scalar equation in 6 scalar unknowns, and, for instance, the choice~\eqref{eq:alphatrprel} and~\eqref{eq:alphabartrapre} together with~~\eqref{eq:transprho1} and~\eqref{eq:transpsigma1} leads to a closed system (the other two equations just need to be satisfied initially).
\end{remark}

\subsection{Outline of the paper}

We begin by introducing all the necessary notations, definitions and preliminary lemmas in Section~\ref{sec:defs}. We proceed to state the main result of this work, Theorem~\ref{thm:gwp} of Section~\ref{sec:statement}. In Section~\ref{sec:local}, we establish the local existence statement for the MBI system. In Section~\ref{sec:bootstrap}, we will formulate the $L^\infty$ bootstrap assumptions. Section~\ref{sec:commut} is dedicated to deriving commuted versions of the system at the tensorial level, for all the components simultaneously. In Section~\ref{sec:canstress} we then derive the crucial properties enjoyed by the canonical stress. We then proceed to deduce $L^2$ estimates from the $L^\infty$ assumptions in Section~\ref{sec:l2fromlinf}. We deal with the issue of spherical averages of $\rho$ and $\sigma$ in Section~\ref{sec:charge}. We consider the nonlinear \emph{\fackip Equations} satisfied by the middle components $\sigma$ and $\rho$ in Section~\ref{sec:spinredu}, and commute derivatives with them in Section~\ref{sec:commutrw} in order to obtain higher-order estimates. We then prove that the middle components, once commuted with angular operators, ($\Omega \sigma$ and $\Omega \rho$) satisfy a Morawetz estimate in Section~\ref{sec:premora}, with no bounds on the nonlinear right hand sides yet. We proceed to analyze the structure of the nonlinearities arising in the \fackip Equations in Section~\ref{sec:nonlstruct}. We use the estimates in Section~\ref{sec:nonlstruct} to close the Morawetz estimates in Section~\ref{sec:morawetz}. We then apply the $r^p$-method of Dafermos and Rodnianski (\cite{newmihalis}) in Section~\ref{sec:pweighted1}, where we can only choose $p =1$. This allows us to obtain (in Section~\ref{sec:l2improved}) a first improvement of the $r$-weights on the spacelike fluxes of all the components. In turn, these improved estimates let us infer better control on the nonlinearities of the \fackip Equations (Section~\ref{sec:improvedrs}). With these improved estimates, a second application of the $r^p$-method, now with $p = 2$ (Sections~\ref{sec:previsited}, \ref{sub:moragiustoests}, and~\ref{sec:pdecay}) lets us deduce the correct decay of the fluxes in order to close the $L^\infty$ decay. Finally, in Sections~\ref{sec:sobrs},~\ref{sec:soba}, and~\ref{sec:sobb} we recover the bootstrap assumptions with better constants. We close the argument in Section~\ref{sec:close}.

\subsection{Acknowledgements}
I would like to thank Prof.~Jonathan Luk for his patience and guidance, for suggesting the problem to me, and for inviting me to Stanford University to finish the project. I would also like to thank my advisor, Prof.~Mihalis Dafermos, for his patience, his encouragement and his comments on preliminary versions of the manuscript. Moreover, I thank Jan Sbierski for the suggestion to look at Fritz John's approach to local existence. Moreover, I thank John Anderson and Yakov Shlapentokh-Rothman for very valuable discussions.

\section{Definitions and preliminary facts}\label{sec:defs}
In this Section, we set up the framework of our study. We formulate the MBI Equations first.
\subsection{The MBI system: Lagrangian formulation and equations} \label{MBIdef}
In order to proceed, let us introduce the system we are analyzing in a more formal fashion. Let $\fara$ be an antisymmetric two-form, the field tensor. We formulate the MBI system from the Lagrangian point of view. Let us define the Lagrangian density for the MBI model as
\begin{equation*}
{}^*\mathscr{L}_{\text{MBI}} := 1- \ellmbi,
\end{equation*}
where 
\begin{equation*}
\ellmbi^2 = 1 + \lun -\ldu^2,
\end{equation*}
and the invariants $\lun$ and $\ldu$ are defined as follows:
\begin{align*}
\lun = \frac 1 2 \fara_{\mu\nu} \fara^{\mu\nu}, \qquad
\ldu = \frac 1 4 \fara_{\mu\nu} \farad^{\mu\nu}.
\end{align*}
\begin{remark}
Unless otherwise specified, $\ellmbi$, $\lun$ and $\ldu$ will all depend on $\fara$. If they depend on another tensor, we will indicate this with square brackets. For instance, if $\mathcal{B}$ is a two-form: 
$$
\lun[\mathcal{B}] := \frac 1 2 \mathcal{B}_{\mu\nu}\mathcal{B}^{\mu\nu}.
$$
\end{remark}
We then postulate that $\fara$ is closed:
\begin{equation}\label{mbiuno}
d \fara = 0 \iff \nabla_{[\mu}\fara_{\nu\kappa]} = 0.
\end{equation}
Moreover, we define a tensor $\mathcal{M}$ in the following way:
\begin{equation*}
{}^*\mathcal{M}_{\mu\nu} := \frac{ \partial {}^*\mathscr{L}_{\text{MBI}}} {\partial \fara_{\mu\nu}}.
\end{equation*}
The Euler--Lagrange equations for the MBI theory can then be formulated requiring that $\mathcal{M}$ be closed:
\begin{equation}\label{mbidue}
d \mathcal{M} = 0.
\end{equation}
Together, (\ref{mbiuno}) and (\ref{mbidue}) are the equations of motion for the theory arising from the Lagrangian density ${}^*\mathscr{L}_{\text{MBI}}$.
\begin{remark}
	By Taylor-expanding the square root, we obtain, omitting higher order terms,
	\begin{equation*}
	1- \ellmbi \approx 1 - 1 - \lun / 2  = -\fara^{\mu\nu} \fara_{\mu\nu}.
	\end{equation*}
	The latter is (up to a sign) the Lagrangian corresponding to the linear Maxwell theory.
\end{remark}

\subsection{The special structure of the MBI theory}\label{sec:specstru}

As already briefly noted, the Lagrangian density ${}^*\mathscr{L}_{\text{MBI}}$ can be uniquely determined by imposing some natural requirements on the theory. We specify more precisely what makes the MBI theory special.

In~\cite{orangebook} it is shown that the only Lorentz-invariant and gauge-invariant Lagrangians for a nonlinear theory of electromagnetism without sources on a general Lorentzian spacetime are of the following form:
\begin{equation}\label{eq:formlag}
\mathscr{L}_{\text{NLE}} = L_{\text{NLE}}(\lun, \ldu) \varepsilon,
\end{equation}
where $L_{\text{NLE}}(x,y)$ is a smooth function of two real variables, $\lun$ and $\ldu$ are the invariants as previously defined, and $\varepsilon$ is the standard volume form of the considered spacetime.

In the works by Boillat~\cite{boillat1970} and Plebanski~\cite{plebanski1970}, the main observation is that the MBI theory is the only nonlinear electromagnetic theory whose Lagragian density is of the form (\ref{eq:formlag}), and which furthermore satisfies the following properties: \begin{itemize}
\item its linearization around the trivial solution is the linear Maxwell theory,
\item solutions corresponding to point charges have finite self-energy (the spatial $L^2$ norm of the field is bounded), and
\item {\bf it does not give rise to birefringence}.
\end{itemize}

We elaborate a little on the last property, starting from the paper of Boillat~\cite{boillat1970}. Roughly speaking, birefringence, from a physical point of view, means that there are multiple directions of light propagation basing on the polarization of the light wave itself. From a mathematical point of view, as proved by Boillat in~\cite{boillat1970} (see also Section 2 of~\cite{stringy}), birefringence means that cone of directions of light propagation at each point comprises only one conical sheet (3-dimensional). In~\cite{boillat1970}, this requirement is imposed by looking at the local propagation of weak discontinuities. Otherwise, it can be formulated by saying that the characteristic surface of the considered theory, at each point in spacetime, is made only of one conical sheet.

Here is the formal mathematical description of the absence of birefringence. Consider a nonlinear electromagnetic theory with field tensor $\fara$ on a 4-dimensional Lorentzian spacetime $(\mathcal{M},g)$. Let the equations of our field theory be
\begin{equation}\label{eq:hform1}
d \fara = 0, \qquad H^{\mu\nu\kappa\lambda}[\fara] \ \nabla_\mu \fara_{\kappa\lambda} = 0.
\end{equation}
Here, $\nabla$ is the Levi-Civita connection of $(\mathcal{M}, g)$, and $H^{\mu\nu\kappa\lambda}[\fara]$ is a tensor depending on the components of $\fara$ as in \eqref{eq:hform1}. Define the \emph{characteristic set at point $p$} of our theory to be
\begin{equation}
C^*_p := \{\xi \neq 0 \in T^*_pM, \text{ such that } N(\chi(\xi))/\text{span}(\xi) \neq 0\},
\end{equation}
with $\chi(\xi):= H^{\mu \kappa \nu \lambda} \xi_\kappa \xi_\lambda$, and $N(\chi)$ being the null space of $\chi$ at $p$. This set describes the directions of light propagation. Then, \emph{birefringence} can be reformulated saying that the conical set $C^*_p$ is composed of only one (3-dimensional) conic sheet.

\begin{remark}
	Notice that the directions of light propagation are connected to positivity properties of the energy--momentum--stress tensor. The absence of birefringence lets us define a metric, the so-called \emph{Boillat metric}, whose null directions are precisely the directions of light propagation in the Born--Infeld theory, and which furthermore gives good positivity properties in order to obtain a-priori estimates.
\end{remark}
It now holds that the MBI theory is, in fact, not birefringent. For a more detailed analysis, see for example Section 5 in~\cite{BialynickiBirula:1984tx}, or the introduction of~\cite{kie1}.

We also remark that the MBI theory lies in a class of nonlinear electromagnetic field theories which have been considered by Tahvildar-Zadeh in the paper~\cite{TahvildarZadeh2011}. This class is comprised of those theories which, when coupled to the Einstein equations, give rise to spherically symmetric electrovacuum solutions (analogous to the Reissner--Nordstr\"om solutions to linear Einstein--Maxwell) which display the mildest singularity possible, i.~e.~a conical singularity on the time axis. This feature makes these solutions particularly suited for a possible well-defined description of point charges coupled to gravity, as one hopes that the mild nature of the singularity allows for a possible well-defined formulation of the law of motion of point charges in curved spacetime. Compare with the works by Kiessling~\cite{kie1, kie2}, in which such a law of motion is formulated in the context of flat spacetime.

We proceed to introduce the versions of the MBI system which will be useful for our calculations.

\subsection{Formulation of the MBI system (I)}\label{sec:altmbi}
We notice that the  MBI Equations \eqref{mbiuno} and \eqref{mbidue} can be written in the form
\begin{equation*}
\left\{
\begin{array}{ll}
\nabla^\mu {}^\star \! \mathcal{M}_{\mu\nu} = 0, \\
\nabla^\mu \farad_{\mu\nu} = 0.
\end{array}
\right.
\end{equation*}
Here, we have defined
\begin{equation}\label{eq:mfaradef}
\mfarad_{\mu\nu} := -\ellmbi^{-1}(\fara_{\mu\nu}- \ldu \farad_{\mu\nu}),\qquad \ellmbi^2:= 1+\lun -\ldu^2.
\end{equation}

\subsection{Formulation of the MBI system (II)}\label{sec:mbidef}
In this section, we formulate a second version of the MBI system, which we will use in some calculations. The invariants for the MBI system, $\lun$ and $\ldu$, are as before. The form of the MBI system is:
\begin{equation}\label{MBI}\boxed{
	\left\{
	\begin{array}{c}
	\nabla_{[\mu}\fara_{\nu \lambda]} = 0,\\
	H^{\mu\nu\kappa\lambda}[\fara] \ \nabla_\mu \fara_{\kappa\lambda} = 0.
	\end{array}
	\right.}
\end{equation}
Here,
\begin{align}\label{hform1}
H^{\mu\nu\kappa\lambda}[\fara] := \frac1 2 \left[g^{\mu\kappa} g^{\nu\lambda}-g^{\mu\lambda}g^{\nu\kappa}\right] + H_{\Delta}^{\mu\nu\kappa\lambda}[\fara]
\end{align}
with
\begin{align}
\hdelta^{\mu\nu\kappa\lambda}[\fara] := \frac 1 2 \left\{-\ellmbi^{-2}\fara^{\mu\nu}\fara^{\kappa\lambda} + \ldu \ellmbi^{-2}\left(\fara^{\mu\nu}\farad^{\kappa\lambda}+\farad^{\mu\nu}\fara^{\kappa\lambda}\right)-\left(1+\ldu^2\ellmbi^{-2}\right)\farad^{\mu\nu}\farad^{\kappa\lambda}\right\}\label{hform2}
\end{align}
We can rewrite the MBI system also highlighting the ``linear part'', which corresponds to the linear Maxwell system:
\begin{equation}\label{eq:MBI2}
\left\{
\begin{array}{c}
\nabla_{[\mu}\fara_{\nu \lambda]} = 0,\\
\nabla_\kappa \fara\indices{^\kappa^\nu} +\hdelta^{\mu\nu\kappa\lambda}[\fara] \ \nabla_\mu \fara_{\kappa\lambda} = 0.
\end{array}
\right.
\end{equation}
\begin{remark}
 Unless otherwise specified, we will write the tensors
 $$
 H\indices{^\mu^\nu^\kappa^\lambda} = H\indices{^\mu^\nu^\kappa^\lambda}[\fara], \qquad \hdelta\indices{^\mu^\nu^\kappa^\lambda} = \hdelta\indices{^\mu^\nu^\kappa^\lambda}[\fara].
 $$ If this is not so, we will specify this by the use of square brackets, cf.~Equation~\eqref{hformg}.
\end{remark}

\subsection{The Schwarzschild spacetime} \label{schwasec}

\begin{definition}\label{def:schwreg}
	Fix a number $M > 0$. The \emph{Schwarzschild spacetime} $(\mathcal{S}, g)$ of mass $M$ is the $4$-dimensional Lorentzian manifold which, seen as a set, is 
	\begin{equation*}
	\mathcal{S} := (t^*,r_1,\theta, \varphi) \in (-\infty, + \infty) \times (0, + \infty) \times (0,\pi) \times (0, 2\pi),
	\end{equation*}
	(note that this coordinate chart does not include the poles).
	In these coordinates, the metric is given by
	\begin{equation*}
	\begin{aligned}
	&g := -\left(1-\frac{2M}{r_1}\right) \de t^* \otimes \de t^* +  \frac{2M}{r_1} \de t^* \otimes \de r_1+  \frac{2M}{r_1}  \de r_1 \otimes \de t^*+\left(1+\frac{2M}{r_1}\right) \de r_1 \otimes \de r_1\\
	& \qquad  + r_1^2 \der \theta \, \otimes \, \der \theta + r_1^2 \sin^2 \theta \der \phi \, \otimes \, \der \phi.
	\end{aligned}
	\end{equation*}
	Furthermore, we define the \emph{exterior Schwarzschild} region, as a set, by 
	\begin{equation*}
		\mathcal{S}_e := (t^*,r_1,\theta, \varphi) \in (-\infty, + \infty) \times (2M, + \infty) \times (0,\pi) \times (0, 2\pi),
	\end{equation*}
	with the same metric $g$ as before.
	
	We further define
	\begin{equation*}
	\mathcal{H}^+ := \{(t^*, r_1, \theta, \varphi) \in \mathcal{S}: r_1 = 2M \}
	\end{equation*}
	to be the \emph{future horizon} of the black hole region.
	\begin{remark}
	Note that, in this system of coordinates, $r_1 = 2M$ does not include what is usually referred to as the past event horizon, nor it includes the bifurcation sphere.
	\end{remark}
	
	Let $\nabla$ be the Levi--Civita connection of the metric $g$. Finally, let
	$$
	\mu := \frac{2M}{r_1}.
	$$
\end{definition}
It is useful to describe the exterior region with different systems of coordinates. 
\begin{definition}[Irregular coordinates]\label{def:irregcoord}
We set
\begin{align*}
&t:= t^*- 2M \log(r_1-2M), \qquad \text{and} \qquad 
r := r_1 \qquad \text{if} \qquad  r_1 > 2M,\\
&t:= t^* + 2M \log(2M-r_1), \qquad \text{and} \qquad 
r := r_1 \qquad \text{if} \qquad 0 <r_1 < 2M.
\end{align*}
In these coordinates, the metric on $\mathcal{S} \setminus \mathcal{H}^+$ becomes:
\begin{equation*}
g = -\left(1-\frac{2M}{r}\right) dt \otimes dt + \left(1-\frac{2M}{r}\right)^{-1} dr \otimes dr + r^2 g_{\mathbb{S}^2},
\end{equation*}
where $g_{\mathbb{S}^2}$ is the standard metric on the conformal sphere $\mathbb{S}^2$, which, in the usual coordinates $(\theta, \varphi)$, is defined as
\begin{equation*}
g_{\mathbb{S}^2} := \der \theta \, \otimes \, \der \theta + \sin^2 \theta \, \der \varphi \, \otimes \, \der \varphi.
\end{equation*}
\end{definition}

\begin{remark}
This coordinate system (restricted to $\mathcal{S}_e$) is the one originally introduced by Schwarzschild. These coordinates are evidently singular at the event horizon $\mathcal{H}^+$
\end{remark}

\begin{definition}[Regge--Wheeler coordinates]\label{def:rwcoords}
We introduce the irregular $(t_2, r^*, \theta, \varphi)$ coordinates, which are given by
\begin{align*}
&t_2 := t, \qquad r^* := r + 2M \log(r-2M), \qquad \text{if} \ r > 2M,\\
&t_2 := t, \qquad r^* := r - 2M \log(2M-r), \qquad \text{if} \ 0 < r < 2M.
\end{align*}
These coordinates are irregular at the horizon $\mathcal{H}^+$. The metric is diagonalized, and it is represented as follows:
\begin{equation*}
g = -(1-\mu) \de t_2 \otimes \de t_2 +  (1-\mu) \de r^* \otimes \de r^* + r^2 g_{\mathbb{S}^2}.
\end{equation*}
\end{definition}
Finally, from the $(t_2, r^*, \theta, \varphi)$ coordinates, we can introduce the \emph{null coordinates}.
\begin{definition}[Null coordinates]
\begin{equation}\label{uvcoord}
\begin{aligned}
u := t_2 - r^*, \qquad
v := t_2 + r^*.
\end{aligned}
\end{equation}
The metric is represented as follows:
\begin{equation*}
g =-2 (1-\mu) \de u \otimes \de v -2 (1-\mu) \de v \otimes \de u + r^2 g_{\mathbb{S}^2}.
\end{equation*}
\end{definition}

\begin{remark}
	The coordinate field $\partial_t = \partial_{t^*}$, as well as any rotation vector field $\Omega$ (see Definition in~\ref{sec:not:geo}) are Killing fields for the metric $g$. This means that
	\begin{equation*}
	(\mathcal{L}_{\partial_t} g)_{\mu\nu} = (\nabla_\mu \partial_t)_\nu + (\nabla_\nu \partial_t)_\mu = 0, \qquad (\mathcal{L}_{\Omega} g)_{\mu\nu} = (\nabla_\mu \Omega)_\nu + (\nabla_\nu \Omega)_\mu = 0.
	\end{equation*}
	In particular, the Lie derivative in their direction commutes with the covariant derivatives (see~\cite{globalnon}),
	\begin{equation*}
	[\nabla, \lie_{\partial_t}] = [\nabla, \lie_{\Omega}] = 0.
	\end{equation*}
\end{remark}
\begin{remark}
We notice that the vector field $\left(1-\frac{2M}{r}\right)^{-1} \partial_u$ is a regular nonvanishing vector field up to and including $\mathcal{H}^+$. It is a geodesic vector field. We also compute the covariant derivatives:
\begin{align*}
&\nabla_{\partial_{u}} \partial_v = \nabla_{\partial_{v}} \partial_u = 0, \\
&\nabla_{\partial_{u}} \partial_u = -\nabla_{\partial_{v}} \partial_v = - \frac{2M}{r^2}, \\
& \nabla_{\partial_u}( \partial_\phi/r) = \nabla_{\partial_u}( \partial_\theta/r) = 0.
\end{align*}
\end{remark}

Throughout this paper, we will denote $L := \partial_v$, and $\lbar := \partial_u$.

\subsection{A dictionary to relate different coordinate systems}

Since we consider several coordinate systems, it may be beneficial for the reader to include a calculation of the derivatives of coordinate functions belonging to one such coordinate system with respect to coordinate vector fields induced by another of such coordinate systems.

We consider the coordinate systems: $(t^*, r_1, \theta, \varphi)$, $(t_2, r^*, \theta, \varphi)$, $(t,r,\theta,\varphi)$, and all the calculations which follow are performed restricting to the exterior region $\mathcal{S}_e$.

First of all, we just restrict to the functions $t^*, r_1,t_2, r^*, t,r$, as the angular functions $\theta$ and $\varphi$ are always the same throughout all the coordinate systems considered.

Furthermore, we note that the following relations hold in $\mathcal{S}_e$:
$$
r_1 = r, \qquad t_2 = t,
$$
hence we will restrict only to the coordinate functions $t,r,t^*, r^*$.

Moreover, we have that
\begin{align*}
&\p_{t^*} = \p_t,  \qquad \p_{r_1} = \p_r + \frac{2M}{r} (1-\mu)^{-1}\p_t,\\
&\p_{t_2} = \p_t, \qquad \p_{r^*} = (1-\mu)\p_r.
\end{align*}
Hence, we will restrict only to the vector fields $\p_t, \p_r, \p_{r^*}, \p_{r_1}$. We have the following table: \vspace{30pt}
\begin{center}
\begin{tabular}{ |c|cccc| } 
 \hline 
  & $t$ & $r$ & $t^*$ & $r^*$  \\ 
\hline 
 $\p_t$ & 1 & 0 & 1 & 0 \\ \hdashline[0.7pt/3pt]
 $\p_r$ & 0 & 1 & $-\mu(1-\mu)^{-1}$ & $(1-\mu)^{-1}$ \\ \hdashline[0.7pt/3pt]
 $\p_{r^*}$ & 0 & $(1-\mu)$ & $-\mu$ & 1 \\  \hdashline[0.7pt/3pt]
 $\p_{r_1}$ & $\mu (1-\mu)^{-1}$ & 1 & 0 & $(1-\mu)^{-1}$ \\
 \hline
\end{tabular}
\end{center}
\subsection{Notation for the geometry and connections}\label{sec:not:geo}

\begin{itemize}[wide, labelwidth=!, labelindent=0pt]
	\item The Schwarzschild spacetime is defined, as usual, using the $(t^*, r_1, \theta, \varphi)$ coordinates. See Definition~\ref{def:schwreg}.
	\item Let $\mathcal{S}_e$ be the exterior region of the Schwarzschild spacetime. In the $(t^*, r_1, \theta, \varphi)$ coordinates, this is the region $S_e := \{(t^*, r_1, \theta, \varphi): r_1 \geq 2M\}$.
	\item Let $\mathcal{H}^+ \subset \mathcal{S}$ be the future event horizon, i.~e.~$\mathcal{H}^+ := \{(t^*, r_1, \theta, \varphi): r_1 = 2M\}$.
	\item Let $(t,r, \theta, \varphi)$ be the irregular coordinates for Schwarzschild, as introduced in Definition~\ref{def:irregcoord}. 
	\item Let $(t^*, r_1, \theta, \varphi)$ such that $t^* = t +2M \log (r-2M)$, $r_1 = r$.
	\item The coordinates $(t,r,\theta, \varphi)$ can be used to parametrize the region $\{(t^*, r_1, \theta, \varphi): r_1 \in (0, 2M)\}$. The coordinate transformation then becomes $t^* = t + 2M \log(2M -r)$, $r_1 = r$.
	\item Let $(t, r^*, \theta, \varphi)$ such that $r^* = r+ 2M \log(r-2M)$ if $r > 2M$, and $r^* = r- 2M \log(2M-r)$ if $0 < r < 2M$.
	\item Let $(u,v,\theta, \varphi)$ such that $u = t-r^*$, $v = t+r^*$.
	\item Let $\chi_1$ be a smooth cutoff function such that $\chi_1(r) = 1$ for $r \in \left[\frac 3 2 M, 3M\right]$, and finally $\chi_1(r) = 0$ for $r \geq 4M$ and $0 < r \leq M$. We define
	\begin{equation*}
	\tau := (1-\chi_1(r)) (u^++1) + \chi_1(r)t^*.
	\end{equation*}
	Here, $u^+$ is the positive part of $u$: $u^+ := \max\{u, 0\}$.
	\item Let $L:= \partial_v$, $\lbar := \partial_u$, $\hat{\lbar} := (1-\mu)^{-1}\lbar$, $T := \p_t = \frac 1 2 (L + \lbar)$.
	\item Let $\mu := 2M/r$.
	\item Let $\mathcal{V}$ be the set of vector fields:
	\begin{equation*}
	\mathcal{V} := \{L, (1-\mu)^{-1}\lbar\}.
	\end{equation*}
	\item Let  ${\leo}$ be the set of vector fields
	$$
	{\leo}:= \{\Omega_{i}\}_{i = 1, 2,3}, \qquad {\leo}_{\text{no}}:=  \{\Omega_{i}/r\}_{i = 1, 2,3},
	$$
	where $\Omega_{i}$ are three rotation (Killing) vector fields whose linear span is bidimensional everywhere.
	\item Let $(\theta^A, \theta^B)$ be local coordinates for the sphere $\mathbb{S}^2$.
	\item Let $\partial_{\theta^A}$ and $\partial_{\theta^B}$ the associated vector fields to $(\theta^A, \theta^B)$.
	\item  Let
	\begin{equation}\label{def:killing}
	\mathcal{K} := \{T, \Omega_1, \Omega_2, \Omega_3\}
	\end{equation}
	be a set of Killing fields on Schwarzschild.
	\item Let $\nabla$ be the Levi-Civita connection with respect to the metric $g$, $\lie$ the associated Lie derivative.
	\item Let $\snabla$ be the Levi-Civita connection projected on spheres of constant $r$. $\snabla$ accepts only capital indices, which indicate tensors tangent to the spheres of constant $r$. $\snabla$ can be extended to the derived bundles in the canonical way, i.e. asking it to satisfy the Leibniz rule. 
	\item Let $\snabla_L$ and $\snabla_\lbar$ be defined on vector fields as $$\snabla_L X:= (\nabla_L X)^p, \qquad \snabla_\lbar X:= (\nabla_\lbar X)^p.$$
	Here, $(\cdot )^p$ indicates projection on the spheres of constant $r$. Extend these connections to the derived bundles in the usual way.
	\item Let $\slashed g$ be the restriction of the ambient metric $g$ to the spheres of constant $r$-coordinate. Furthermore, let $g_{\mathbb{S}^2}$ be the metric on the conformal sphere $\mathbb{S}^2$, which has the expression, in the usual coordinate system:
	\begin{equation*}
	g_{\mathbb{S}^2} = \de \theta^2 + (\sin \theta)^2 \de \varphi^2.
	\end{equation*}
	Furthermore, the following relation holds true:
	$$
	\slashed g = r^2 g_{\mathbb{S}^2}.
	$$
	\item Let $\snabla_{\mathbb{S}^2}$ be the Levi-Civita connection associated to the metric $g_{\mathbb{S}^2}$ on the conformal sphere $\mathbb{S }^2$. We have that also $\snabla_{\mathbb{S}^2}$ accepts only capital indices, and furthermore:
	$$
	\snabla_{\mathbb{S}^2} = r \snabla.
	$$
\end{itemize}

\subsection{Notation for derivatives, multi-indices, and variations of functions and tensors}\label{sec:not:der}

\begin{itemize}[wide, labelwidth=!, labelindent=0pt] 
\item By the symbol $\enn$, we denote the set of all natural numbers including $0$:
$$
\enn := \{0, 1, \ldots \}.
$$
\item Let $b \in \N_{\geq 0}$, we define $\iindicess^b$ to be the set composed of multi-indices consisting only of Killing fields:
\begin{equation*}
\iindicess^b :=\{\text{Ordered strings of } b \text{ elements of }\mathcal{K}\}.
\end{equation*}

\item Let furthermore
\begin{equation*}
	\iindiceo^b := \{\text{Ordered strings of }b\text{ elements of }\leo\}.
\end{equation*}
\item Let $\mathscr{S}$ be either $\leo$ or $\mathcal{K}$. Then, we let
\begin{equation*}
\iindicesg^{\leq j} := \bigcup_{\substack{m \in \N_{\geq 0}\\ m \leq j}} \iindicesg^m.
\end{equation*}

\item We now define the sum of multi-indices. Let $I \in \iindicesg^b$. Then
$$
J+K =I = (K_1. \ldots, K_b)
$$
if and only if there exist $b_1, b_2 \in \N_{\geq 0}$, $b_1+ b_2 = b$ and increasing sequences $\{i_1 , \ldots, i_{b_1}\} \subset \{1, \ldots, b\}$, $\{j_1 , \ldots, j_{b_2}\} \subset \{1, \ldots, b\}$ such that
\begin{equation}\label{eq:defsumind}
\begin{aligned}
\{i_1 , \ldots, i_{b_1}\} \cup \{j_1 , \ldots, j_{b_2}\} = \{1, \ldots, b\},\\
I = (K_{i_1}, \ldots, K_{i_{b_1}}) \in  \iindicesg^{b_1}, \\
J = (K_{j_1}, \ldots, K_{j_{b_2}}) \in \iindicesg^{b_2}.
\end{aligned}
\end{equation}

\item If $g: \mathcal{S} \to \R$ is a smooth function and $I \in \iindicesg^b$, $I = (K_1, \ldots, K_b)$ we define
\begin{equation*}
	\partial^I_{\mathscr{S}} g := K_1( \cdots (K_b g)),
\end{equation*}
the iterated derivative.

\begin{definition}[Notation for derivatives of functions]\label{def:dots}
Let $b, k, j \in \N_{\geq 0}$, let $I \in \mathscr{I}^b_{\leo}$, and let $g:\mathcal{S} \to \R$ be a smooth function. Let us recall the definition of $\hat L:= (1-\mu)^{-1}\lbar$. Then, we define
\begin{equation}
\hat g_{,I,j,k} := \partial^I_{\leo}\hat\lbar^j T^k g.
\end{equation}
Furthermore, we define
\begin{equation}
\dot g_{,I,j,k} := \partial^I_{\leo}\lbar^j T^k g.
	\end{equation}
\end{definition}

\item If $I \in \iindicesg^b$, and $I = (K_1, K_2, \ldots, K_b)$, we define
\begin{equation*}
\lie^I_{\mathscr{S}} \fara := \lie_{K_1} \cdots \lie_{K_b} \fara.
\end{equation*}

\item Recall: $\chi_1(r)$ is a smooth cutoff function such that $\chi_1(r) = 1$ for $r \in \left[\frac 3 2 M, 3M\right]$, and finally $\chi_1(r) = 0$ for $r \geq 4M$ and $0 < r \leq M$. Let $Y$ be the vector field
\begin{equation*}
Y:= \chi_1(r) (1-\mu)^{-1} \lbar.
\end{equation*}

\item Given an integer $a \in \N_{\geq 0}$, let 
$$
\nabla^a_Y \fara := \underbrace{\nabla_Y \ldots\nabla_Y}_{a \text{-times}}\fara.
$$
\begin{remark}
Note that the last expression is different from $(\nabla^a \fara )(\underbrace{Y, \ldots, Y}_{a\text{-times}})$.
\end{remark}
\item Let 
\begin{equation*}
\faram_{,I,a} := \lie^I_{\mathcal{K}} \nabla_Y^a \fara.
\end{equation*}

\end{itemize}

\begin{definition}[Horizon cutoff function] \label{def:cutoff}
Let $r_{\text{in}}$ be a number such that $r_\text{in} \in (0,2M)$, and let $I$ be an open interval such that $[r_\text{in}, 4M - r_{\text{in}}] \subset I \subset [M, 3M]$. We define a smooth radial function (in $(t^*, r_1, \omega, \varphi)$-coordinates) $\chi_{\mathcal{H}^+}(r)$, such that
\begin{equation}
\chih(r) = \left\{
\begin{array}{ll}
1 & \text{for } r \in [r_\text{in}, 4M - r_{\text{in}}]\\
0 & \text{for } r \in (0, \infty)\setminus I.
\end{array}
\right.
\end{equation}
\end{definition}

\subsection{Redshift vector field}\label{sub:redshift}

\begin{itemize}[wide, labelwidth=!, labelindent=0pt]
\item Let the redshift vector field be defined as follows:
\begin{equation*}
	\redsh :=2 \desude{}{t^*}+ \{(1-\mu)(1+\mu)+ 5\chi_1(r)(1-\mu)+\chi_1(r)\}\desude{}{t^*} + \{(1-\mu)^2 - 5\chi_1(r)(1- \mu)-\chi_1(r)\}\desude{}{r_1}.
\end{equation*}
Here, $\chi_1$ is a smooth cutoff function such that $\chi_1(r) = 1$ for $r \in \left[\frac 3 2 M, 3M\right]$, and finally $\chi_1(r) = 0$ for $r \geq 4M$ and $0 < r \leq M$.

With respect to the null vector fields $L$ and $\lbar$, we have up to $\mathcal{H}^+$,
\begin{equation*}
V_{\text{red}} = L + \lbar + 5 \chi_1(r) \left( (1-\mu)L +\lbar\right)+\chi_1(r) (1-\mu)^{-1}\lbar.
\end{equation*}
\end{itemize}

\begin{remark}
	We remark that, since $\partial_{t^*}$ is a Killing vector field, the deformation tensor relative to $\redsh$ satisfies the following property, when $r \in [3/2 M, 3M]$:
	\begin{equation}
		{}^{(\redsh)}\pi_{\mu\nu} = {({\lie}_{\redsh} g)}_{\mu\nu} = {}^{(V_1)}\pi_{\mu\nu},
	\end{equation}
	with $V_1 =  - 5 \mu L + (1-\mu)^{-1}\lbar$. 
\end{remark}
\begin{remark}
Notice that the considered vector field is regular everywhere, including the horizon. Also, it is future-directed and strictly timelike on the set $r \in [15/8 M, 17/8 M]$, in particular 
\begin{equation*}
g(V_{\text{red}}, V_{\text{red}}) < - 1/2
\end{equation*}
in that region.
\end{remark}

\subsection{Spacetime regions and foliation}\label{sec:not:regfol}
\begin{itemize}[wide, labelwidth=!, labelindent=0pt]
\item Let $t_2^* \geq t_1^* \in \R$. Let
\begin{equation*}
\Sigma_{t_1^*} := \{ t^* = t^*_1\} \subset \mathcal{S}_e.
\end{equation*}
\item Let
\begin{equation*}
\regio{t_1^*}{t_2^*} := \cup_{s \in [t_1^*, t_2^*]} \Sigma_s.
\end{equation*}
\item Let $n_{\Sigma_{t^*}}$ be the future-directed unit normal to the foliation $\Sigma_{t^*}$.
\item Ler $0 < r_{\text{in}} < 2M$. Let $t_2^* \geq t_1^* \in \R$. Let
\begin{equation*}
\widetilde{\Sigma}_{t_1^*} := \{ t^* = t^*_1\} \cap \{r_1 \geq r_{\text{in}}\}.
\end{equation*}
\item Let
\begin{equation*}
\tregio{t_1^*}{t_2^*} := \cup_{s \in [t_1^*, t_2^*]} \widetilde{\Sigma}_s.
\end{equation*}
\item Let
\begin{equation*}
\hat{\Sigma}_{t_1^*} := \{ t^* = t^*_1\} \cap \{r_1 \geq M\}.
\end{equation*}
\item Let ${n}_{ \widetilde{\Sigma}_{t^*}}$ be the future-directed unit normal to the foliation $\widetilde{\Sigma}_{t^*}$.
\item Let $\de\Sigma_{t^*}$ be the natural volume form induced by the foliation $\Sigma_{t^*}$.
\item Let $\desphere$ the volume form on the conformal sphere $\mathbb{S}^2$.
\begin{equation*}
\desphere := \sin \theta \de \theta \wedge \de \varphi.
\end{equation*}
\item Let $\text{dVol}$ the natural volume form on the Lorenztian manifold $\mathcal{S}_e$. In $(u,v,\omega)$ coordinates, this has the form
\begin{equation*}
\text{dVol} = 2(1-\mu)r^2 \de u \wedge \de v \wedge \desphere.
\end{equation*}
\item If $\regio{}{} \subset \mathcal{S}$, we denote by $\mathcal{J}^+(\regio{}{})$ the causal future of $\regio{}{}$. This is, the set of all points $p \in \mathcal{S}_e$ such that there exists a $\mathcal{C}^1$ curve $\gamma: [0,1] \to \mathcal{S}$ such that $\gamma(0) \in \regio{}{}$, $\gamma(1) = p$, and furthermore $\gamma'(s)$ is a future directed causal vector for all $s \in [0,1]$.
\item We adopt the following notation for outgoing and ingoing cones. Let $\tilde u$ and $\tilde v$ be real numbers:
\begin{equation}
\overline{C}_{\tilde u} := \{u = \tilde u\}, \qquad \underline{C}_{\tilde v} := \{v = \tilde v\}.
\end{equation}
\item Let $u_2 \geq u_1$ real numbers, and let $R > 2M$. We define the following spacetime regions:
\begin{align*}
 \mathfrak{D}_{u_1}^{u_2} :&= \left\{r \geq R, u \in [u_1, u_2] \right\},\\
\mathcal{Z}_{u_1} :&= \left(\{t^* = u_1 + R + 4M \log(R-2M)\} \cap \{2M \leq r \leq R \}\right) \cup (\{u = u_1\} \cap \{r \geq R\}), \\
\mathfrak{Z}_{u_1}^{u_2} :&= \cup_{u \in [u_1, u_2]} \mathcal{Z}_u.
\end{align*}
It may be helpful to visualize these regions on a Penrose diagram, and we refer to Figure~\ref{fig4}, in which these regions are depicted. The region $ \mathfrak{D}_{u_1}^{u_2}$ is the region bounded by the timelike curve $r = R$ and the two outgoing null hypersurfaces $\{u = u_1\}$ and $\{u = u_2\}$. The region $\mathcal{Z}_{u_1}$ coincides with a constant-$t^*$ hypersurface for $r \leq R$, and coincides with the outgoing null cone $\{u = u_1\}$ for $r \geq R$. The value 
$$
t^* = u_1 + R + 4M \log(R-2M)
$$
is chosen so that the two pieces match at the sphere in which the null cone $\{u = u_1\}$ meets the timelike surface $r = R$. In fact, we have, for a general point in $\mathcal{S}_e$:
$$
t^* = t + 2M \log(r-2M) = u +r^* + 2M \log (r-2M) = u +r+4M \log (r-2M),
$$
and the condition follows by plugging in $u = u_1, r = R$. Finally, $\mathfrak{Z}_{u_1}^{u_2}$ is the region bounded above by $\mathcal{Z}_{u_2}$ and below by $\mathcal{Z}_{u_1}$. 

\end{itemize}
\subsection{Raising and lowering indices}
\begin{itemize}[wide, labelwidth=!, labelindent=0pt]
\item Let $X^\mu$ be a vector field in $\Gamma(\mathcal{S})$. Then $X_\nu$ is defined as the one-form $X^\mu g_{\mu\nu}$. A similar remark holds for derived bundles.
\end{itemize}

\subsection{Sobolev norms and their equivalences}

\begin{definition}[Pointwise norm for tensors] \label{def:gennormder}
Let $T$ be a covariant tensor with $m_1$ indices. Given $k \in \enn$, we define
\begin{align}
|T| &:= \sum_{V_1, \ldots, V_{m_1} \in \mathcal{V} \cup \ono}|T(V_1, \ldots, V_{m_1})|,\\
|\partial^k T| &:= \sum_{\substack{V_1,\ldots, V_{m_1} \in \mathcal{V} \cup \ono\\
W_1, \ldots, W_k \in \mathcal{V} \cup \oscr}}|(\nabla_{W_1} \cdots \nabla_{W_k}T)(V_1, \ldots, V_{m_1})|,\\
|\lie^k T| &:= \sum_{\substack{V_1,\ldots, V_{m_1} \in \mathcal{V} \cup \ono\\
W_1, \ldots, W_k \in \mathcal{V} \cup \oscr}}|(\lie_{W_1} \cdots \lie_{W_k} T)(V_1, \ldots, V_{m_1})|.
\end{align}

Furthermore, if $T$ belongs to a tensor bundle intrinsic to the spheres of constant $r$, we define
\begin{align}
|\snabla^k T| &:= \sum_{\substack{V_1,\ldots, V_{m_1} \in \mathcal{V} \cup \ono\\
W_1, \ldots, W_k \in \mathcal{V} \cup \oscr}}|(\snabla_{W_1} \cdots \snabla_{W_k}T)(V_1, \ldots, V_{m_1})|.
\end{align}

Finally, letting $h \in \enn$, we define 
\begin{align}
	&|\partial^{\leq h} T| := \sum_{k \leq h} |\partial^k T|,\qquad |\lie^{\leq h} T| := \sum_{k \leq h} |\lie^k T|,\\
	&|\snabla^{\leq h} T| := \sum_{k \leq h} |\snabla^k T|,\qquad |\nabla^{\leq h} T| := \sum_{k \leq h} |\nabla^k T|.
\end{align}
\end{definition}

\begin{remark}
Note that the angular derivatives automatically carry a factor of $r$.
\end{remark}

\begin{remark} Note that, in view of the definition above, viewing $\nabla^k T$ as a tensor,
\begin{align*}
|\nabla^k T| &:= \sum_{\substack{V_1,\ldots, V_{m_1} \in \mathcal{V} \cup \ono\\
W_1, \ldots, W_k \in \mathcal{V} \cup \ono}}|(\nabla^k T)(W_1, \ldots, W_k, V_1, \ldots, V_{m_1})|.
\end{align*}
Note that derivatives in this particular case carry no $r$-weights.
\end{remark}

\begin{proposition}\label{norm.equiv}
Given $k \in \enn$ and $T$ a tensor field on $\mathcal{S}$, we have the following pointwise inequalities between norms:
\begin{equation*}
  |\nabla^{\leq k} T| \lesssim |\partial^{\leq k} T| \simeq |\lie^{\leq k} T|.
\end{equation*}
Here, the symbol $A \simeq B$ indicates that there exist constants $C_2 > C_1 >0$ such that $C_1 A \leq B \leq C_2 A$. The symbol $A \lesssim B$ indicates that there exists a constant $C_1 >0$ such that $A \leq C_1 B$. Furthermore, $C_1$ and $C_2$ depend only on the mass parameter $M>0$ of Schwarzschild.
If, in addition, $T$ is a section of a tensor bundle intrinsic to the spheres of constant $r$, we have
\begin{equation*}
  |\nabla^{\leq k} T| \lesssim |\partial^{\leq k} T| \simeq |\lie^{\leq k} T| \simeq  |\snabla^{\leq k} T|.
\end{equation*}
\end{proposition}

\begin{definition}[Weighted Sobolev norm]\label{def:sobolevnorm} Let $k \in \enn$, and $p \in \R$, $p \geq 0$, let $T$ be a tensor field on $\mathcal{S}$. The $p$-weighted $k$\textsuperscript{th}-order Sobolev norm of $T$ on $\Sigma_{t^*_1}$ is
\begin{equation*}
\norm{\fara}_{H^{k,p}(\Sigma_{t_1})} := \int_{\Sigma_{t^*_1}}|\partial^{\leq k} \fara| r^p  d\Sigma_{t^*}.
\end{equation*}
We define:
$$
\norm{\fara}_{H^{k}(\Sigma_{t^*_1})} := \norm{\fara}_{H^{k,0}(\Sigma_{t^*_1})}.
$$
\end{definition}

\begin{proposition}[Pointwise bounds for tensors]
There holds
\begin{equation*}
|A^{\mu_1 \ldots \mu_n}| = |A_{\mu_1 \ldots \mu_n}|.
\end{equation*}
Furthermore,
\begin{equation*}
|A_{\mu_1 \ldots \mu_n}B^{\mu_1 \ldots \mu_n}| \lesssim |A_{\mu_1 \ldots \mu_n}| |B_{\mu_1 \ldots \mu_n}|,
\end{equation*}
where the implicit constant is uniform on the Schwarzschild spacetime $\mathcal{S}$.
\end{proposition}

\begin{proof}
The first claim follows from the definition. The second claim follows writing the product 
\begin{equation*}
A_{\mu\nu}B^{\mu\nu}
\end{equation*}
in the frame $\mathcal{X} := \{L, (1-\mu)^{-1} \lbar, \partial_\theta/r, \partial_\varphi/(r \sin \theta)\}$, and noting that the only nonzero numbers among $g(X,Y)$, with $X,Y \in \mathcal{X}$ are
\begin{equation*}
g(L, (1-\mu)^{-1}\lbar) = -2, \qquad g(\partial_\theta/r, \partial_\theta/r)= 1,\qquad g(\partial_\varphi/(r \sin \theta), \partial_\varphi/(r \sin \theta)) = 1, 
\end{equation*}
which is a regular frame field except at $\theta = 0, \pi$. By spherical symmetry, we see that there is nothing special with the points for which $\theta = 0, \pi$, and we conclude.
\end{proof}

\begin{proposition}[Control of all derivatives via the variation $\dot \fara$]\label{prop:liecontrol}
 Recall the following definition, valid for $a, b \in \enn$, and $I \in \iindicess^b$ :
\begin{equation*}
\faram_{,I,a} := \lie^I_{\mathcal{K}} \nabla^a_Y \fara.
\end{equation*}
Let $j \in \N_{\geq 0}$. Let $r_{\text{in}}$ such that $\frac 3 2 M \leq r_{\text{in}} < 2M$ (recall that the definition of $\tregio{}{}$ depends on $r_{\text{in}}$). Then, there exist positive constants $C_1$ and $C_2$ and a number $\varepsilon_{\text{var}}>0$ such that the following holds. Let $t_1^* > t_0^*$, and let $\fara$ be a smooth solution to the MBI system~(\ref{MBI}) on $\tregio{t_0^*}{t_1^*}$. Assume the boostrap assumptions $\text{BA}(\tregio{t_0^*}{t_1^*},1,\left \lfloor \frac j 2 \right \rfloor, \varepsilon_{\text{var}})$. We have the pointwise inequality, valid on $\tregio{t_0^*}{t_1^*}$:
\begin{equation}\label{eq:liecontrol}
C_1 |\partial^{\leq j} \fara|^2 \leq \chih(r) \sum_{a = 0}^j \sum_{I \in \iindicess^{\leq j-a}} |\dot \fara_{,I,a}|^2 + \sum_{I \in \iindicess^{\leq j}} |\dot \fara_{,I,0}|^2 \leq C_2 |\partial^{\leq j} \fara|^2,
\end{equation}
where $\chih(r)$ the smooth cutoff in Definition~\ref{def:cutoff}.
\end{proposition}

\begin{proof}[Sketch of proof]
	The right hand inequality in~\eqref{eq:liecontrol} trivially follows from Proposition~\ref{norm.equiv}. Furthermore, for all $r \in [r_{\text{in}}, 4M - r_{\text{in}}]$, the left hand inequality in~\eqref{eq:liecontrol} also follows trivially from Proposition~\ref{norm.equiv}. It remains to show the left hand inequality for $r \geq 4M-r_{\text{in}}$. We wish to show, in that case, the existence of $C_1$ such that
	\begin{equation}
	C_1 |\partial^{\leq j} \fara|^2 \leq  \sum_{I \in \iindicess^{\leq j}} |\dot \fara_{,I,0}|^2.
	\end{equation}
The proof proceeds by induction on the number of derivatives falling on $\fara$, which we call $m$ ($m \leq j$). We start with the induction base case $m=1$. We seek to control 
$$
\sum_{U, V \in \ono \cup \{\p_t , \p_r\}}|(\nabla_{\p_r} \fara)(U,V)|.
$$
If both $U$ and $V$ are different from $\p_r$, then, using the closure equation $\nabla_{[\alpha} \fara_{\beta \gamma]} = 0$, we obtain
$$
|(\nabla_{\p_r} \fara)(U,V)| \leq |(\nabla_{U} \fara)(\p_r,V)|+ |(\nabla_{V} \fara)(\p_r,U)| \leq 2\sum_{I \in \iindicess^{\leq 1}} |\dot \fara_{,I,0}|.
$$
If either $U$ or $V$ is $\p_r$, we use the Euler--Lagrange equation for MBI (see~\eqref{MBI}):
$$
\nabla_\kappa \fara\indices{^\kappa^\nu} +\hdelta^{\mu\nu\kappa\lambda}[\fara] \ \nabla_\mu \fara_{\kappa\lambda}=0,
$$
and we obtain, for some positive number $C(r_{\text{in}})$ depending on $r_{\text{in}}$:
\begin{equation*}
\begin{aligned}
&C(r_\text{in})|\nabla_{\p_r} \fara(\p_r, U)| \\
&\quad \leq (1-\mu) |\nabla_{\p_r} \fara(\p_r, U)|\\
&\quad \leq C  \big( |\nabla_{\p_t} \fara(\p_t, U)| +\sum_{\Omega \in \ono} |\nabla_{\Omega} \fara(\Omega, U)| + |U_\nu \hdelta^{\mu\nu\kappa\lambda}[\fara] \ \nabla_\mu \fara_{\kappa\lambda}|\big) \\
& \quad  \leq C \sum_{I \in \iindicess^{\leq 1}} |\dot \fara_{,I,0}| + C\varepsilon^2_{\text{var}} |\nabla_{\p_r} \fara(\p_r, U)|,
\end{aligned}
\end{equation*}
where in the last inequality we have decomposed the tensors in the frame $\{\p_t, \p_r, \Omega_1, \Omega_2\}$, with $\Omega_1, \Omega_2 \in \ono$, and we have used the bootstrap assumptions $\text{BA}(\tregio{t_0^*}{t_1^*},1,\left \lfloor \frac j 2 \right \rfloor, \varepsilon_{\text{var}})$ to bound the terms arising from $\hdelta$. The induction base case is then concluded taking $\varepsilon_{\text{var}}$ small depending on $r_{\text{in}}$.

We now show the induction step $m-1 \to m$. We wish to bound the following expression:
\begin{equation}
\sum_{\substack{U, V \in \{\p_t , \p_r\} \cup  \ono
	\\ W_i \in \{\p_t, \p_r\} \cup \leo}}|(\nabla_{W_1} \nabla_{W_2}\cdots \nabla_{W_m} \fara)(U,V)|.
\end{equation}
We focus only on one term:
\begin{equation}\label{eq:inductstep}
|(\nabla_{W_1} \nabla_{W_2}\cdots \nabla_{W_m} \fara)(U,V)|.
\end{equation}
If all of the $W_i$ belong to $\{\p_t\} \cup \leo$, we are done by Proposition~\ref{prop:liecontrol}. Let's call $m_1$ the number of indices $i$ for which $W_i = \p_r$ in expression~\eqref{eq:inductstep}. 
If $m_1 =0$, we are done. Let's assume that $m_1>0$. Due to boundedness of the curvature coefficients (see Appendix~\ref{sec:riemann}), we can assume, without loss of generality, possibly after commuting the derivatives in expression~\eqref{eq:inductstep}, that $W_m = \p_r$. We then distinguish two cases, as in the induction base step:
\begin{enumerate}
\item If both $U, V \in \{\p_t\} \cup \ono$, we use the closure equation to get 
$$
|(\nabla_{W_1} \nabla_{W_2}\cdots \nabla_{W_m} \fara)(U,V)| \leq |(\nabla_{W_1} \nabla_{W_2}\cdots \nabla_{U} \fara)(W_m,V)|+|(\nabla_{W_1} \nabla_{W_2}\cdots \nabla_{V} \fara)(W_m,U)|.
$$
This decreases the number $m_1$ by one.

\item If either $U$ of $V$ is $\p_r$, we use the MBI Euler--Lagrange equation to get 
\begin{equation*}
\begin{aligned}
&|(\nabla_{W_1} \nabla_{W_2}\cdots \nabla_{\p_r} \fara)(\p_r,U)|\\
&\quad \leq C  \big( | \nabla_{W_1} \nabla_{W_2}\cdots \nabla_{\p_t} \fara(\p_t, U)| +\sum_{\Omega \in \ono} |\nabla_{W_1} \nabla_{W_2}\cdots \nabla_{\Omega} \fara(\Omega, U)| \\
&\qquad + |\nabla_{W_1} \nabla_{W_2}\cdots (U_\nu \hdelta^{\mu\nu\kappa\lambda}[\fara] \ \nabla_\mu \fara_{\kappa\lambda})|\big) \\
& \quad  \leq C \sum_{I \in \iindicess^{\leq m-1}} |\dot \fara_{,I,0}| + C\varepsilon^2_{\text{var}}|(\nabla_{W_1} \nabla_{W_2}\cdots \nabla_{\p_r} \fara)(\p_r,U)|,
\end{aligned}
\end{equation*}
Here, we have used the bootstrap assumptions, the Leibnitz rule, and the induction hypothesis. This decreases the number $m_1$ by one.
\end{enumerate}
By repeated application of steps $(1)$ and $(2)$ above, we keep decreasing $m_1$ until $m_1 = 0$. This concludes the inductive step.
\end{proof}

\subsection{The null decomposition}\label{sec:not:null}
\begin{itemize}[wide, labelwidth=!, labelindent=0pt]
\item Let $\theta^A$, $\theta^B$ local coordinates for $\mathbb{S}^2$.
\item Let $\svol_{AB}$ be the induced volume form on the spheres of constant $r$.
\item Middle components:
\begin{equation}\label{eq:middledef}
\rho := \frac 1 2 (1-\mu)^{-1} \fara(\lbar , L), \qquad \sigma := \frac 12 \fara_{AB} \svol^{AB}
\end{equation}
\item Extreme components:
\begin{equation}\label{eq:extdef}
\alpha_A := \fara_{\mu\nu} (\partial_{\theta^A})^\mu L^\nu, \qquad \alphabar_A:= \fara_{\mu\nu}(\partial_{\theta^A})^\mu \lbar^\nu.
\end{equation}
\item Expression for the invariants in terms of the null decomposition:
\begin{equation*}
\lun = - \rho^2 + \sigma^2 - (1-\mu)^{-1}\alpha^A \alphabar_A , \qquad \ldu = - \rho \sigma - \frac 12 (1-\mu)^{-1} \svol^{AB}\alpha_A \alphabar_B.
\end{equation*}

\item Nomenclature for the weighted middle components:
\begin{equation*}
Z := r^2 \rho, \qquad W:= r^2 \sigma.
\end{equation*}
\end{itemize}

\subsection{Hodge dual}\label{sec:hodge}
\begin{itemize}[wide, labelwidth=!, labelindent=0pt]
\item Let $\varepsilon_{\alpha \beta \gamma \delta}$ be the standard volume form on the Schwarzschild exterior.
\item The Hodge dual of $\fara$ is
\begin{equation*}
{}^\star \! \fara_{\kappa\lambda} := \frac 1 2 \fara^{\mu\nu} \varepsilon_{\kappa\lambda\mu\nu}.
\end{equation*}
\item 
	We define the null decomposition of the dual of $\fara$ as follows:
	\begin{equation}
	\begin{aligned}
&^\odot\rho := \frac 1 2 (1-\mu)^{-1} \farad(\lbar , L), & ^\odot \sigma := \frac 12 \farad_{AB} \svol^{AB},\\
& ^\odot\alpha_A := \fara_{\mu\nu} (\partial_{\theta^A})^\mu L^\nu, & ^\odot \alphabar_A:= \fara_{\mu\nu}(\partial_{\theta^A})^\mu \lbar^\nu.
	\end{aligned}
	\end{equation}
	There are straightforward identities connecting the null decomposition of $\farad$ to the null decomposition of $\fara$. See Appendix~\ref{sec:dual}.
\end{itemize}

\subsection{Stress--energy--momentum tensor and canonical stress}\label{sec:not:stresses}
Given $\gara$ a smooth 2-form on $\mathcal{S}$, define the Maxwell (linear) stress--energy--momentum tensor as:
\begin{equation*}
\dot Q^{(\text{MW})}_{\mu\nu}[\gara] := {\gara}_{\mu\alpha} \gara \indices{_\nu^\alpha} - \frac 1 4 g_{\mu\nu} \gara^{\alpha \beta} \gara_{\alpha \beta}.
\end{equation*}
Given smooth two-forms $\fara$ and $\gara$ on $\mathcal{S}$, we define the canonical stress:
\begin{equation}\label{eq:dotqdef}
\dot Q \indices{^\mu _\nu}[\gara] := H^{\mu\zeta\kappa\lambda}[\fara] {\gara}_{\kappa\lambda} {\gara}_{\nu\zeta} - \frac 1 4 \delta_\nu^\mu  H^{\eta\zeta\kappa\lambda}[\fara] {\gara}_{\eta\zeta}{\gara_{\kappa\lambda}}.
\end{equation}
\begin{remark}
	In the above, we suppressed the dependence of $\dot Q$ on $\fara$, as  defined in Equation~\eqref{eq:dotqdef}, since \emph{every occurrence of the tensor field $\dot Q$ will have $\fara$ as argument of the tensor field $H$}.
\end{remark}

\section{Statement of the result} \label{sec:statement}
In this work, we prove the following Theorem. Let $t_0^*$ be a positive real number and $r_{\text{in}} \in (M, 2M)$. We let $i$ be the inclusion of $\widetilde{\Sigma}_{t_0^*} = \{t^* = t_0^*\} \cap \{r \geq r_{\text{in}}\}$ into $\mathcal{S}$. Let $\Lambda^2$ be the space of smooth two-forms on a vector bundle. We consider the tensor field $\fara_0 \in \Lambda^2(i^*(T\mathcal{S}))$ as initial data for the MBI system.

Let us also introduce the \emph{initial charge} of the MBI system on Schwarschild. We define the following two functions of $r$ on the initial hypersurface $\widetilde{\Sigma}_{t_0^*}$:
\begin{equation}\label{eq:defchar}
MC(r) := \int_{\mathbb{S}^2} (1-\mu)^{-1}r^2\farad_0(\lbar, L) \desphere, \qquad
EC(r) := \int_{\mathbb{S}^2} (1-\mu)^{-1}r^2{}^\star \!\mfara_0(\lbar, L) \desphere.
\end{equation}
Recall the definition of the tensor field $\mfarad$ in display~\eqref{eq:mfaradef} ($\mfarad_0$ is just what appears in expression~\eqref{eq:mfaradef}, where we replace every occurrence of $\fara$ with $\fara_0$). The function $MC(r)$ is the \emph{magnetic charge}, whereas the function $EC(r)$ is the \emph{electric charge}. 
\begin{remark}
We briefly comment on the physical interpretation of these two quantities. Let us for a moment restrict our attention to flat spacetime. In Minkowski spacetime, consider the linear Maxwell theory (so that the equality $\mfarad = -\fara$ holds), under the electric-magnetic decomposition of $\fara$ induced by usual system of coordinates $(t,x_1,x_2,x_3)$. Let us furthermore restrict our attention to the hypersurface $t = t_0$. In this setting, the analogous expression for $MC(r)$ (resp.~$EC(r)$) is, up to a multiplicative constant, the magnetic flux (resp.~the electric flux) through the sphere of radius $r$ centered at the origin. Furthermore, under the linear Maxwell evolution, these are conserved quantities as functions of $t$ and $r$.
\end{remark}

It is an easy consequence of the MBI system on Schwarzschild that both $MC(r)$ and $EC(r)$ are invariants under the MBI evolution. Furthermore, we say that \emph{the magnetic and electric charge vanish initially} if there exists a constant $C$ and a number $\eta > 0$ such that the following inequality holds:
\begin{equation}\label{eq:vanishcharge}
MC(r) \leq Cr^{-\eta}, \qquad EC(r) \leq C r^{-\eta}.
\end{equation}

For the MBI system on Schwarzschild, there exist static solutions with non-vanishing initial charge (for instance, such that $EC(r)$ and $MC(r)$ are positive constant functions of $r$). These solutions are constructed in Propositon~\ref{prop:statsols} in Appendix~\ref{sec:heuristic}. 

\begin{remark}
Furthermore, these are regular solutions in the exterior region $\mathcal{S}_e$. This behavior should be contrasted with the case of MBI on Minkowski spacetime, where static solutions with positive constant charge display a point singularity. 
\end{remark}

We are now ready to state the main result of this work:

\begin{theorem}[Global well-posedness for MBI on Schwarzschild exterior for small initial data]\label{thm:gwp}
There exist $\varepsilon_0 > 0$ sufficiently small, and $r_{\text{in}} \in (M, 2M)$ sufficiently close to $2M$ such that, letting $\fara_0 \in \Lambda^2(i^*(T\mathcal{S}))$, we have the following. 
Let $W$, $Z$ be defined as Section~\ref{sec:not:null}. Furthermore, define the variations $\dot Z_{I,j,k}$ and $\dot W_{,I,j,k}$ as in Definition~\ref{def:dots}. Denote by a subscript $0$ the quantities derived from initial data, i.e. derived from $\fara_0$.

Assume the following boundedness of the initial energy, in terms of the null decomposition, for $0 <\varepsilon < \varepsilon_0$:
\begin{align}
\label{eq:ensmallf}&\sum_{|I| + j + k \leq l+4} \int_{\widetilde \Sigma_{t_0^*}} \left[r^2 |L \dot{(Z_0)}_{,I,j,k}|^2 + r|\snabla \dot{(Z_0)}_{,I,j,k} |^2\right] r^{-2} \de \widetilde \Sigma_{t^*} \leq \varepsilon^2, \\
\label{eq:ensmalls}&\sum_{|I|+j + k \leq l+4} \int_{\widetilde \Sigma_{t_0^*}} \left[r^2 |L \dot{ (W_0)}_{,I,j,k}|^2 + r|\snabla \dot{ (W_0)}_{,I,j,k} |^2\right]r^{-2}\de \widetilde \Sigma_{t^*} \leq \varepsilon^2,\\
\label{eq:ensmallt}&\int_{\widetilde \Sigma_{t_0^*}} r^2 |\partial^{\leq l+5} \alpha_0|^2 \de \Sigma_{t^*} \leq \varepsilon^2, \qquad \int_{\widetilde \Sigma_{t_0^*}} r^4 |\partial^{\leq l+4} \snabla_{\p_{r_1}} \alpha_0|^2 \de \Sigma_{t^*} \leq \varepsilon^2,\\
&\norm{\fara_0}^2_{H^{l+9}(\widetilde \Sigma_{t_0^*})} \leq \varepsilon^2.\label{eq:ensmall}
\end{align}
Here, we make the choice $l = 18$.
We furthermore suppose that the charge vanishes initially, i.~e.~, that \eqref{eq:vanishcharge} holds true. We assume also that $\fara_0$ satisfies the constraint equations on $\widetilde \Sigma_{t_0^*}$:
\begin{equation}\label{eq:constreq}
\nabla^\mu (\farad_0)_{\mu \lambda}  n^\lambda_{\widetilde \Sigma_{t_0^*}} = 0, \qquad H_{\mu \nu \kappa \lambda}[\fara_0] \ \nabla^\mu (\fara_0)^{\kappa \lambda}  n^\nu_{\widetilde \Sigma_{t_0^*}} = 0.
\end{equation}
Here, $ n^\nu_{\widetilde \Sigma_{t_0^*}}$ is the future-directed Lorentz unit normal to the hypersurface $\widetilde \Sigma_{t_0^*}$.

Under these assumptions, $\fara_0$ launches a unique smooth globally-defined solution $\fara \in \Lambda^2(\{t^* \geq t^*_0\}\cap \{r \geq r_{\text{in}}\} )$ to the MBI system~\eqref{MBI} on Schwarzschild, whose null components furthermore decay according to the decay rates in display~\eqref{bootstrap}. More precisely, $13$ derivatives of the null components decay in $L^\infty$, as the field $\fara$ in particular satisfies the decay $BA\left(\mathcal{S}_e,\frac 1 2,13, \varepsilon \right)$ (see Assumption~\eqref{as:bootstrap} for the definition).

\end{theorem}

\begin{remark}
	Let us remark that, in view of Proposition~\ref{prop:constraints}, the set of nontrivial initial data considered here is nonempty.
\end{remark}

\begin{remark}

	We note that, in the statement of Theorem~\ref{thm:gwp}, we needed to assume $l+9 = 27$ \emph{unweighted} derivatives of the field to be bounded, whereas only $l+5 = 23$ \emph{weighted} derivatives are required to be bounded. This discrepancy can be explained in the following way. In our argument, to be able to choose $p=2$ in the $r^p$ estimates (ultimately aimed at deducing decay), we have to show bounds on spacelike \emph{weighted} energies of the ``good'' components $\rho, \sigma$ and $\alpha$, which are essentially just a version of estimates~\eqref{eq:ensmallf}--\eqref{eq:ensmallt} propagated in time (on the foliation $\Sigma_{t^*}$). Such estimates are obtained in Propositions~\ref{prop:impalpha} and~\ref{prop:imprhosigma} by performing an $r^p$-type estimate between two spacelike hypersurfaces belonging to the foliation $\Sigma_{t^*}$ (note: these $r^p$ estimates are different from the $r^p$ estimate we employ to obtain $L^\infty$ decay). However, nonlinear error terms arise in the RHS of the weighted estimates we just mentioned. Bounds for such nonlinear error terms are obtained in turn by using the \emph{unweighted} $L^2$ bounds. The loss of derivatives connected to this process results in the discrepancy between the number of weighted derivatives versus unweighted derivatives which we see here.
\end{remark}

\subsection{Particular types of solutions to the constraint Equations \eqref{eq:constreq}}

Here, we prove the existence of compactly supported solutions to the constraint equations. Recall that $i$ is the inclusion of $\widetilde{\Sigma}_{t_0^*}$ into $\mathcal{S}$.

\begin{proposition}\label{prop:constraints}
	There exists $\fara_0 \in \Lambda^2(i^*(T\mathcal{S}))$, smooth and compactly supported in $r$, which furthermore satisfies the constraint equations \eqref{eq:constreq}:
	\begin{equation}\label{eq:constref}
	\nabla^\mu (\farad_0)_{\mu \lambda}  n^\lambda_{\widetilde \Sigma_{t_0^*}} = 0, \qquad H_{\mu \nu \kappa \lambda}[\fara_0] \ \nabla^\mu (\fara_0)^{\kappa \lambda}  n^\nu_{\widetilde \Sigma_{t_0^*}} = 0.
	\end{equation}
\end{proposition}

\begin{proof}[Proof of Proposition~\ref{prop:constraints}]
	We begin by noticing that we have the following expression, with respect to $L$ and $\lbar$:
	\begin{equation}\label{eq:nform}
		n^\nu_{\widetilde{\Sigma}_{t_0^*}} = f_1(r) \frac 1 2 \left(L + \frac{1+\mu}{1-\mu} \lbar \right),
	\end{equation}
	with $f_1(r)$ a nonvanishing smooth function in $r$ which ensures that the right hand side has the right normalization. Dividing Equations \eqref{eq:constref} by $f_1(r)$, we obtain that the constraints are satisfied if the following equations are. For ease of notation, let us drop the subscript $0$ in $\fara_0$, and furthermore notice we are using the formulation of the MBI system in subsection~\ref{sec:altmbi}.
	\begin{equation}\label{eq:constrchng}
		\nabla^\mu \farad_{\mu L}  + \frac{1+\mu}{1-\mu}	\nabla^\mu \farad_{\mu \lbar} = 0, \qquad \nabla^\mu \marad_{\mu L}  + \frac{1+\mu}{1-\mu}	\nabla^\mu \marad_{\mu \lbar} = 0.
	\end{equation}
	Recall the following notation concerning Hodge duals introduced in Section~\ref{sec:hodge}:
	\begin{equation}
	\begin{aligned}
&^\odot\rho := \frac 1 2 (1-\mu)^{-1} \farad(\lbar , L), & ^\odot \sigma := \frac 12 \farad_{AB} \svol^{AB},\\
& ^\odot\alpha_A := \fara_{\mu\nu} (\partial_{\theta^A})^\mu L^\nu, & ^\odot \alphabar_A:= \fara_{\mu\nu}(\partial_{\theta^A})^\mu \lbar^\nu.
	\end{aligned}
	\end{equation}
	We now make use of the following facts, which follow by computation and from the Hodge dual calculations in Section~\ref{sec:dual} of the Appendix:
	\begin{equation}\label{eq:calculations}
	\begin{aligned}
		\nabla_\mu \fara\indices{^\mu_L} &= - r^{-2}L(r^2 \rho) + \dive \alpha,\\
		\nabla_\mu \fara\indices{^\mu_\lbar} &=  r^{-2}\lbar(r^2 \rho) + \dive \alphabar, \\
		^\odot\alphabar &= - \alphabar^B \svol_{BA}, \\
		^\odot\alpha &= \alpha^B \svol_{BA}, \\
		^\odot\rho &= \sigma, \\
		^\odot \sigma &= - \rho.
	\end{aligned}
	\end{equation}
	By definition and using the null decomposition, we have
	\begin{equation}\label{eq:simpleinva}
	\begin{aligned}
		&\ellmbi^2:= 1+\lun -\ldu^2,\\
		&\lun = - \rho^2 + \sigma^2 - (1-\mu)^{-1}\alpha^A \alphabar_A , \qquad \ldu = - \rho \sigma - \frac 12 (1-\mu)^{-1} \svol^{AB}\alpha_A \alphabar_B.
	\end{aligned}
	\end{equation}
	Then, the first equation in \eqref{eq:constrchng} reduces to
	\begin{equation}
		r^{-2}\left( - L + \frac{1+\mu}{1-\mu}\lbar\right) (r^2 \sigma) - \curl \alpha - \frac{1+\mu}{1-\mu}\curl \alphabar = 0.
	\end{equation}
	Rewriting the last display in the ``regular'' coordinate system $(t^*, r_1, \theta^A, \theta^B)$, recalling that
	$$
	\p_{r_1} = \frac 12 \big(L - \frac{1+\mu}{1-\mu}\lbar\big) ,
	$$
	we get
	\begin{equation}\label{eq:forsigmacons}
		\p_{r_1} (r_1^2 \sigma) = - \frac{r_1^2}2 \curl \alpha - \frac{r_1^2(1+\mu)}{2(1-\mu)}\curl \alphabar.
	\end{equation}
	
	We now make the ansatz
	\begin{equation}\label{eq:ansatz}
		\alpha_A = \kappa f(r_1) \svol_{AB} \snabla^B h, \qquad  \alphabar_A = 0.
	\end{equation}
	with $f$, a smooth radial function, and $h: \mathbb{S}^2 \to \R$ a smooth function on the sphere $\mathbb{S}^2$, and $\kappa >0$ a small number.
	\begin{remark}
	Note that, in this case, $\alpha$ is not an angular gradient. A slight variant of this reasoning will give solutions in which $\alpha$ is an angular gradient. In that case, $\alphabar$ would have to be nonzero, though.
	\end{remark}
	Under this ansatz, the invariants and $\ellmbi$ simplify as follows, due to expression~\eqref{eq:simpleinva}:
	\begin{equation}
	\lun = - \rho^2 + \sigma^2 , \qquad \ldu = - \rho \sigma, \qquad \ellmbi^2= (1+\sigma^2)(1-\rho^2).
	\end{equation}
	
	We then seek solutions to the second equation in display \eqref{eq:constrchng}. We again make use of the facts \eqref{eq:calculations} to deduce, after a somewhat lengthy calculation, that the second equation in display \eqref{eq:constrchng}, under the ansatz~\eqref{eq:ansatz}, reduces to
	\begin{equation}\label{eq:forrhocons}
		\p_{r_1}\Big(r_1^2 \big(\rho - \frac {\sigma \ldu}{\ellmbi} \big) \Big) - \frac{r_1^2}{2} \alpha_A \snabla^A(\ellmbi^{-1}) + \frac{r_1^2}{2} \alpha^B \svol_{AB} \snabla^A(\ldu \ellmbi^{-1}) = \frac {r_1^2} 2\frac{\ldu \,  \curl \alpha }{\ellmbi}.
	\end{equation}
	We then wish to find solutions of Equations~\eqref{eq:forsigmacons} and~\eqref{eq:forrhocons} on the set $[a,b] \times \mathbb{S}^2$ (i.e.~$r_1$ is allowed to vary between $a$ and $b$), where $ [a,b]$ is a small interval such that $a < b$, $2M \in (a,b)$, and all null components vanish for $r_1 \geq b$.
	
	We now choose $f(r_1)$ to be smooth, non-increasing and compactly supported in $[a,b)$, such that $f(a) = 1$. We also choose $h$ to be a smooth function on $\mathbb{S}^2$.
	
	It is straightforward to solve Equation~\eqref{eq:forsigmacons} for $\sigma$, for all $\omega \in \mathbb{S}^2$:
	\begin{equation*}
	\sigma(r_1, \omega) = (r_1)^{-2} \int_{r_1}^b \frac{s_1^2}2 \curl \alpha (s_1, \omega ) \de s_1
	\end{equation*}	
	for all $r_1 \in [a,b]$ (recall that $\alpha$ has already been chosen).
	
	An explicit computation also gives that Equation~\eqref{eq:forrhocons} can be recast in the following form:
	\begin{equation}\label{eq:transport}
	\p_{r_1} \rho(\omega, r_1) + F^A(\rho, r_1, \omega)\snabla_A \rho = G(\rho,r_1, \omega).
	\end{equation}
	Here, $F^A(\rho, \omega, r_1)$ is a smooth vector field on $\mathbb{S}^2$. The components of $F^A$, furthermore, when written locally in a regular coordinate system, and the function $G$, enjoy the following property: there exists a constant $C>0$ such that derivatives up to order $10$ of the functions ($F^A, G$) in all their arguments $(\rho, r_1, \omega)$ are bounded in absolute value by $C \kappa$ on the set $\rho \in [-\kappa, \kappa]$, $r_1 \in [a,b]$, $\omega \in \mathbb{S}^2$. Since~\eqref{eq:transport} is now a nonlinear transport equation on the product manifold $[a,b] \times  \mathbb{S}^2$, and in view of the smallness condition above, we can invoke a local existence statement which ensures the existence of a smooth $\rho(r_1, \omega)$, numbers $\kappa$, $a < 2M < b$ in the above conditions such that Equation~\eqref{eq:transport} is verified. This concludes the proof.

\end{proof}
\begin{remark}
In this remark, we show how to obtain special solutions of the MBI constraint equations for which either $\lun$ or $\ldu$ does not vanish.

Upon choosing $h$ to be an eigenfunction of the spherical Laplacian, we have that the zeros of $\curl \alpha$ are a set of co-dimension $1$ on the sphere. Furthermore, it follows from our construction that, in this case, $\sigma = g(r_1) h$, where $g(r_1)$ is a smooth radial function. A calculation then shows that setting $\rho$ identically zero ensures that equation~\eqref{eq:forrhocons} is satisfied. We have therefore found a solution of the constraint equations with  invariant $\lun$ non identically vanishing.

On the other hand, let us consider the usual spherical coordinates $(\theta, \varphi)$ on $\mathbb{S}^2$. It is possible to choose $h(\theta, \varphi)$ such that it coincides with the function $\varphi^2+\varphi$ in a neighborhood of the point $(\theta =\pi/4, \varphi = 0) \in \mathbb{S}^2$. Let us now suppose, by way of contradiction, that $\rho$ and all its derivatives are zero at the point $p$ of coordinates $(r_1 = 2M, \theta = \pi/4, \varphi = 0)$. Since $\slashed \Delta h = \frac 1 {\sin^2 \theta} \p^2_{\varphi} (\varphi^2 +\varphi) \neq 0$ when calculated at the point $(\theta =\pi/4, \varphi = 0) \in \mathbb{S}^2$, it follows that $\sigma$ is nonzero in a neighborhood of $p$ in $(r_1, \theta, \varphi)$ coordinates. From Equation~\eqref{eq:forrhocons}, we now deduce that the only term that survives at $p$ satisfies
$$
\alpha_A \snabla^A \big(1/\sqrt{1+\sigma^2}\big) = 0.
$$
Due to the form of $h$, this implies that $\svol_{AB} \snabla^B h \snabla^A \slashed{\Delta} h = 0$ at $p$. But we calculate
$$
\svol_{AB} \snabla^B h \snabla^A \slashed{\Delta} h = - \frac 1 {\sin \theta} \p_\varphi(\varphi^2 + \varphi) \p_\theta\Big(\frac 1 {\sin^2 \theta} \Big),
$$
which is nonzero at $p$. Therefore, either $\rho$ is nonzero at $p$, in which case $\ldu$ is nonzero at $p$, or the gradient of $\rho$ is nonzero at $p$, in which case $\rho$ is nonzero at a point nearby, but $\sigma$ is nonvanishing in a neighborhood of $p$, which proves that $\ldu$ is nonzero somewhere.

Finally, it is clear that we can localize the two constructions above and achieve both $\lun$ and $\ldu$ nonzero at different points. Choose two points $\omega_1, \omega_2$ on the sphere and set $h$ to be locally an eigenfunction of the spherical Laplacian near $\omega_1$, such that the resulting $\sigma$ is positive in a neighborhood of $\omega_1$. Regarding $\omega_2$, follow exactly the same construction for non-vanishing $\ldu$ just described. This gives the desired solution.
\end{remark}

\section{Local existence}\label{sec:local}
This is the first step in our reasoning. Here we establish a local-in-time existence result for the MBI system on Schwarzschild.

\begin{theorem}[Local existence for small data]\label{thm:local}
	
	For all $N_0 \in \N_{\geq 0}$, $N_0 \geq 5$, there exist $\varepsilon_0 > 0$,  and $r_{\text{in}} \in (M, 2M)$ such that, letting $\fara_0 \in \Lambda^2(i^*(TS_e))$, we have the following. 
	
	Assume the following boundedness of the initial energy
	\begin{align}
	&\norm{\fara_0}^2_{H^{N_0}(\widetilde \Sigma_{t_0^*})} \leq \varepsilon^2.\label{eq:ensmall1}
	\end{align}
	We furthermore suppose the vanishing of charge (\ref{eq:decchar}) on $\Sigma_{t_0^*}$. We assume also the $\fara_0$ satisfies the constraint equations on $\widetilde \Sigma_{t_0^*}$:
	\begin{equation}\label{eq:constreq1}
	\nabla^\mu (\fara_0)_{\mu \lambda}  n^\lambda_{\widetilde \Sigma_{t_0^*}} = 0, \qquad H_{\mu \nu \kappa \lambda}[\fara_0] \ \nabla^\mu (\fara_0)^{\kappa \lambda}  n^\nu_{\widetilde \Sigma_{t_0^*}} = 0.
	\end{equation}
	Under these assumptions, there exists $t_1^* > t_0^*$ such that $\fara_0$ launches a unique smooth local-in-time solution $\fara \in \Lambda^2(\{t_1^* > t^* \geq t^*_0\}\cap \{r \geq r_{\text{in}}\} )$ to the MBI system (\ref{MBI}) on Schwarzschild.
\end{theorem}

We will prove this statement later, in Section~\ref{sub:proofloc}. We now introduce a suitable linearization of the system \eqref{MBI}, which will be used in order to prove the local existence statement.

\subsection{The linearized form of the MBI system} \label{lintheory}
We now investigate the evolution part of the linearized MBI system.

\begin{definition} Let $r_{\text{in}} \in (0,2M)$. Consider smooth two-forms $\fara$ and $\mathcal{B}$ defined on $\tregio{t^*_0}{t^*_1}$. We say that the form $\fara$ satisfied the $\mathcal{B}$-linearized MBI system on $\mathcal{R} \subset \tregio{t^*_0}{t^*_1}$ with initial data $\fara_0$ on $\widetilde \Sigma_{\tau_0}$ if there holds
	
	\begin{equation}\label{linmbi}
	\left\{
	\begin{array}{ll}
	\Pi_{\nu \alpha} \nabla_\mu \farad^{\mu \nu} = 0, & \text{on } \regio{}{}\\
	H^{\mu\nu\kappa\lambda} [\mathcal{B}] \ \Pi_{\nu \alpha}  \nabla_\mu \fara_{\kappa\lambda} = 0, & \text{on } \regio{}{}\\
	\fara =\fara_0, & \text{on } \regio{}{} \cap \widetilde \Sigma_{t_0^*}
	\end{array}\right.
	\end{equation}
	
	Here, $\Pi$ is the orthogonal projection on each of the $\Sigma_{t^*}$'s, and $H[\mathcal{B}]$ is such that
	\begin{align}\label{hformg}
	H^{\mu\nu\kappa\lambda}[\mathcal{B}] :&= \frac 1 2 \left[g^{\mu\kappa} g^{\nu\lambda}-g^{\mu\lambda}g^{\nu\kappa}\right] + H_{\Delta}^{\mu\nu\kappa\lambda}[\mathcal{B}] \\ \nonumber
	\hdelta^{\mu\nu\kappa\lambda}[\mathcal{B}] :&= \frac 1 2 \left\{-\ellmbi^{-2}[\mathcal{B}] \ \mathcal{B}^{\mu\nu}\mathcal{B}^{\kappa\lambda} + \ldu[\mathcal{B}] \  \ellmbi^{-2}[\mathcal{B}] \ \left(\mathcal{B}^{\mu\nu}\barad^{\kappa\lambda}+\barad^{\mu\nu}\bara^{\kappa\lambda}\right)\right. \\ \nonumber
	&\left.-\left(1+\ldu^2[\mathcal{B}] \ \ellmbi^{-2}[\mathcal{B}]\right)\barad^{\mu\nu}\barad^{\kappa\lambda}\right\}
	\end{align}
\end{definition}

\begin{remark}
	We named this system the ``evolution part'' because it is composed of 6 independent equations which evolve the 6 components of $\fara$. The remaining 2 equations are constraint equations. We first proceed to prove that there are solutions to the evolution system, and then we will prove that the constraint equations are transported by the evolution.
\end{remark}

\subsection{Divergence of the canonical stress}

We prove an important lemma on the divergence of the canonical stress for the linearized system \eqref{linmbi}.

\begin{lemma}\label{lem:qdivecan}
	Given a two-form $\fara$ satisfying the $\bara$-linearized variation system
	\begin{equation}\label{eq:var}
	\left\{
	\begin{array}{ll}
	\Pi_{\nu \alpha} \nabla_\mu \farad^{\mu \nu} = \mathcal{J}_\alpha, & \text{on } \regio{}{}\\
	H^{\mu\nu\kappa\lambda}[\bara] \ \Pi_{\nu \alpha}  \nabla_\mu \fara_{\kappa\lambda} = \mathcal{I}_\alpha, & \text{on } \regio{}{}\\
	\fara =\fara_0, & \text{on } \regio{}{} \cap \widetilde \Sigma_{t_0^*}
	\end{array}\right.
	\end{equation}
	The canonical stress is defined as
	\begin{equation}
	\dot{Q} \indices{^\mu _\nu}[\bara] := H^{\mu\zeta\kappa\lambda}[\bara] {\fara}_{\kappa\lambda} {\fara}_{\nu\zeta} - \frac 1 4 \delta_\nu^\mu H^{\eta\zeta\kappa\lambda}[\bara] {\fara}_{\eta\zeta}{\fara_{\kappa\lambda}}.
	\end{equation}
	Here, $H$ is as in Equation~(\ref{hform1}) and Equation~(\ref{hform2}).
	Under these assumptions, $\dot Q$ satisfies the following relation:
	\begin{equation}\label{eq:stressdivl}
	\begin{aligned}
	n_{\widetilde{\Sigma}_{t^*}}^\nu \nabla_\mu (\dot{Q} \indices{^\mu _\nu}[\bara])=
	n_{\widetilde{\Sigma}_{t^*}}^\nu &\left(
	(\nabla_\mu H^{\mu\zeta\kappa\lambda}[\bara]) {\fara}_{\kappa\lambda} {\fara}_{\nu\zeta}- \frac 1 4 (\nabla_\nu H^{\zeta \eta \kappa\lambda}[\bara] \fara_{\zeta \eta} \fara_{\kappa\lambda})+
	\mathcal{I}^\zeta \fara_{\nu \zeta}\right.\\
	&\left.- \frac 1 2 H^{\zeta \eta \kappa\lambda}[\bara] \fara_{\kappa\lambda}( \mathcal{J}^\alpha \varepsilon_{\alpha \nu \zeta \eta})\right).
	\end{aligned}
	\end{equation}
\end{lemma}

\begin{proof}[Proof of Lemma]
	Let us consider the normal $W_0 := n_{\widetilde{\Sigma}_{t^*}}$, and complete it to a smooth, positively oriented global orthonormal frame of $T (\mathcal{S} \cap \{r \geq r_{\text{in}}\})$. Let us call the resulting frame $(W_0, W_1, W_2, W_3)$. By use of this frame, we notice that the first equation of display (\ref{eq:var}) implies the following:
	\begin{equation}\label{eq:cycderproj}
	W^\eta_a W^\zeta_b W^\beta_c (\nabla_\eta \fara_{\zeta \beta} + \nabla_\beta \fara_{\eta \zeta}+\nabla_\zeta \fara_{\beta \eta}) = W^\eta_a W^\zeta_b W^\beta_c ( \mathcal{J}^\alpha \varepsilon_{\alpha\eta\zeta\beta}),
	\end{equation}
	whenever $a, b, c \in \{0,1,2,3\}$, and $\{a,b,c\}$ is not a permutation of $\{1,2,3\}$.
	Then, by virtue of (\ref{eq:cycderproj}) and the system (\ref{eq:var}), we have
	\begin{equation*}
	\begin{aligned}
	&W_0^\nu \nabla_\mu   \dot{Q} \indices{^\mu _\nu}[\bara]\\ =&
	W_0^\nu\left((\nabla_\mu H^{\mu\zeta\kappa\lambda}[\bara]) {\fara}_{\kappa\lambda} {\fara}_{\nu\zeta} + H^{\mu\zeta\kappa\lambda}[\bara]\nabla_\mu {\fara_{\kappa\lambda}}{\fara}_{\nu\zeta} + H^{\mu \zeta \kappa\lambda}[\bara] {\fara}_{\kappa\lambda} \nabla_\mu \fara_{\nu\zeta} \right.\\
	&\left.- \frac 1 4 (\nabla_\nu H^{\zeta \eta \kappa\lambda}[\bara] \fara_{\zeta \eta} \fara_{\kappa\lambda})- \frac 1 2 \delta^\mu_\nu H^{\zeta \eta \kappa \lambda}[\bara] \nabla_\mu \fara_{\zeta \eta} {\fara}_{\kappa\lambda}\right)\\ = &
	W_0^\nu \left(
	(\nabla_\mu H^{\mu\zeta\kappa\lambda}[\bara]) {\fara}_{\kappa\lambda} {\fara}_{\nu\zeta}- \frac 1 4 (\nabla_\nu H^{\zeta \eta \kappa\lambda}[\bara] \fara_{\zeta \eta} \fara_{\kappa\lambda})+
	\mathcal{I}^\zeta \fara_{\nu \zeta}+ H^{\mu\zeta\kappa\lambda}[\bara]\fara_{\kappa\lambda}\nabla_\mu \fara_{\nu \zeta}\right.\\
	&\left.- \frac 1 2 H^{\zeta \eta \kappa\lambda}[\bara] \fara_{\kappa\lambda}(-\nabla_\eta \fara_{\nu \zeta}- \nabla_\zeta \fara_{\eta \nu} + \mathcal{J}^\alpha \varepsilon_{\alpha \nu \zeta \eta}) \right) \\
	= & W_0^\nu \left(
	(\nabla_\mu H^{\mu\zeta\kappa\lambda}[\bara]) {\fara}_{\kappa\lambda} {\fara}_{\nu\zeta}- \frac 1 4 (\nabla_\nu H^{\zeta \eta \kappa\lambda}[\bara] \fara_{\zeta \eta} \fara_{\kappa\lambda})+
	\mathcal{I}^\zeta \fara_{\nu \zeta}\right. \\
	&\left. - \frac 1 2 H^{\zeta \eta \kappa\lambda}[\bara] \fara_{\kappa\lambda}( \mathcal{J}^\alpha \varepsilon_{\alpha \nu \zeta \eta}) \right)
	\end{aligned}
	\end{equation*}
	All in all,
	\begin{equation}
	\begin{aligned}
	&W_0^\nu \nabla_\mu  \dot{Q} \indices{^\mu _\nu}[\bara] \\
	&= W_0^\nu \left(
	(\nabla_\mu H^{\mu\zeta\kappa\lambda}[\bara]) {\fara}_{\kappa\lambda} {\fara}_{\nu\zeta}- \frac 1 4 (\nabla_\nu H^{\zeta \eta \kappa\lambda}[\bara] \fara_{\zeta \eta} \fara_{\kappa\lambda})\right.\\
	&\left.
	+\mathcal{I}^\zeta \fara_{\nu \zeta}
	- \frac 1 2 H^{\zeta \eta \kappa\lambda}[\bara] \fara_{\kappa\lambda}( \mathcal{J}^\alpha \varepsilon_{\alpha \nu \zeta \eta})\right).
	\end{aligned}
	\end{equation}
\end{proof}
\begin{remark}
	Using the definition of $\hdelta$ we have that Equation~\eqref{eq:stressdivl} can be rewritten as
	\begin{equation}\label{eq:sdivhdeltal}
	\begin{aligned}
	n_{\widetilde{\Sigma}_{t^*}}^\nu (\nabla_\mu   \dot{Q} \indices{^\mu _\nu}[\bara]) &= 
	{n}_{\widetilde{\Sigma}_{t^*}}^\nu \left((\nabla_\mu \hdelta^{\mu\zeta\kappa\lambda}[\bara]) {\fara}_{\kappa\lambda} {\fara}_{\nu\zeta}- \frac 1 4 (\nabla_\nu \hdelta^{\zeta \eta \kappa\lambda}[\bara] \fara_{\zeta \eta} \fara_{\kappa\lambda})+
	\mathcal{I}^\zeta \fara_{\nu \zeta}\right.\\
	&\left. - \frac 1 2 \hdelta^{\zeta \eta \kappa\lambda}[\bara] \ \fara_{\kappa\lambda}(\mathcal{J}^\alpha \varepsilon_{\alpha \nu \zeta \eta}) -\mathcal{J}^\alpha \farad_{\alpha \nu}\right).
	\end{aligned}
	\end{equation}
\end{remark}

\subsection{A priori energy estimates for the linear system}\label{sub:apriorilocal}
First, let us define useful spacetime regions.

\begin{definition}[Cones with spacelike boundary]
	Let $r_{\text{in}} \in (0,2M)$. Let $t^*_{0} \geq 0$. Consider a point $p_0$ whose $r$-coordinate is such that $r \geq r_{\text{in}}$, and such that its $t^*$-coordinate (which we call $T^*$) satisfies $T^* \geq t_0^*$. Then, we define the past cone $\mathcal{A}_{t_{0}^*}^{p_0}$ as the following spacetime region:
	\begin{equation*}
	\mathcal{A}_{t_{0}^*}^{p_0} := \{p \in \mathcal{S}, \underline d(p|_{T^*}, p_0) \leq 2 (T^*-t_p^*), \text{where } t_p^* \text{ is the } t^* \text{-coordinate of }p\} \cap \{ r \geq r_{\text{in}}\}.
	\end{equation*}
	Here, we denote by $p|_{T^*}$ the point $q$ having the same coordinates of $p$, except for the $t^*$-coordinate, which we set to be equal to $T^*$, and $\underline d$ is the Riemannian distance induced on the $\widetilde \Sigma_{t^*}$-surfaces. Notice that the boundary of this cone is spacelike (if $r_{\text{in}} < 2M$, the surface $r = r_{\text{in}}$ is spacelike).
\end{definition}

\begin{remark}
	We make this definition in order to have a local version of the existence theorem. By a domain of dependence argument, we can then construct a global solution to our problem, given that the time of existence for each piece is uniformly controlled. Furthermore, this definition serves the purpose that, if we have smallness of $\gara$ in $L^\infty$, we can then say that the boundary terms arising in our energy estimates are positive, given that the ``linear part'' of the canonical stress will dominate in such regime.
\end{remark}

We state the following a priori estimates for the linear system (\ref{linmbi}).
\begin{proposition}\label{energyprop}
	Let $m \in \enn$, $m \geq 5$. Let $t_0^* \geq 0$ and $r_{\text{in}} \in (M, 2M)$. There exist a small number $A(m) > 0$, a large integer $w(m) \in \N_{\geq 0}$ and a constant $C(A, m) >0$ so that the following holds. Let $p_0 \in \mathcal{S} \cap \{ r \geq r_{\text{in}}\}$, such that $0 \leq t^*|_{p_0} -t_0^* \leq A$. Let $T^* := t^*|_{p_0}$. Let $\fara$ solve the linearized MBI system (\ref{linmbi}) on $\mathcal{A}_{t_{0}^*}^{p_0}$ with smooth initial data $\fara_0 \in H^m(\Sigma_{t^*_0} \cap \mathcal{A}_{t_{0}^*}^{p_0})$. Let us suppose that $\bara$ is a smooth tensor field, which satisfies
	$$
	\norm{\bara}_{L^\infty(\mathcal{A}_{t_{0}^*}^{p_0})} \leq A.$$
	Then, for all $t^* \in (t_0^*, T^*)$,
	\begin{equation}\label{eq:apriolin}
	\norm{\fara}^2_{\widetilde H^m(\Sigma_{t^*} \cap \mathcal{A}_{t_{0}^*}^{p_0})} \leq \norm{\fara_0}^2_{\widetilde H^m(\Sigma_{t_0^*} \cap \mathcal{A}_{t_{0}^*}^{p_0} )}\exp\left(C \left(t^*-t^*_0 +\int_{t_0^*}^{t^*} \norm{\bara}^w_{\widetilde{H}^m(\mathcal{A}^{p_0}_{t_{0}^*} \cap \widetilde{\Sigma}_s)}\de s\right)\right).
	\end{equation}
\end{proposition}

\begin{remark}
	We do not provide the proof of this Proposition, the reason being that it is a standard application of the divergence of the stress tensor being of lower order, in the case of the linearized system (Lemma~\ref{lem:qdivecan}) along with commutation of the MBI system and the Gronwall inequality. 
\end{remark}

\subsection{Local-in-time existence for the linear system}\label{sub:locallin}
We now continue to deduce local existence for the linearized system \eqref{linmbi} from the estimates just stated above. This argument will make essential use of a constructive element, which in this case is the Cauchy--Kowalevskaya Theorem. Furthermore, it will use the fact that we can uniformly approximate a function of several variables and a given number of its derivatives with a polynomial on a compact set and will make use of the energy estimates to provide compactness.

\begin{proposition}[Local existence for the linearized system \eqref{linmbi}] \label{prop:linex} Let $t_0^* >0$ the initial time and $r_{\text{in}} \in (M, 2M)$ an interior radius, and let $m \in \enn$, $m \geq 5$.  There exists a small number $\varepsilon_{\text{lin}} >0$ such that the following holds.
	Let $p_0$ a point in $\mathcal S$ such that, defining $T^* := t^*|_{p_0}$, $T^*-t_0^* \leq \varepsilon_{\text{lin}}$ and $r|_{p_0} \geq r_{\text{in}}$.
	
	Let $\bara \in \Lambda^2(\mathcal{A}_{t_0^*}^{p_0})$ be a $2$-form defined on the region $\mathcal{A}_{t_0^*}^{p_0}$, which furthermore enjoys the following smallness property:
	\begin{equation}
	\norm{\bara}_{L^\infty(\mathcal{A}_{t_0^*}^{p_0})} \leq \varepsilon_{\text{lin}} \text{ on } \mathcal{A}_{t_0^*}^{p_0}.
	\end{equation}
	Let $\fara_0 \in \Lambda^2(\widetilde{\Sigma}_{t_0^*} \cap \mathcal{A}_{t_0^*}^{p_0})$ a smooth $2$-form defined on the base of the cone (the initial data for the linearized system).
	
	Under these conditions, there exists a unique smooth solution $\fara \in \Lambda^2(\mathcal{A}_{t_0^*}^{p_0})$ to the initial value problem arising from the linearized system \eqref{linmbi}:
	\begin{equation}\label{eq:linearizeda}
	\left\{
	\begin{array}{ll}
	\Pi_{\nu \alpha} \nabla_\mu \farad^{\mu \nu} = 0 & \text{on } \mathcal{A}_{t_0^*}^{p_0},\\
	H^{\mu\nu\kappa\lambda}[\bara] \ \Pi_{\nu \alpha}  \nabla_\mu \fara_{\kappa\lambda} = 0 & \text{on } \mathcal{A}_{t_0^*}^{p_0},\\
	\fara =\fara_0 & \text{on } \mathcal{A}_{t_0^*}^{p_0} \cap \widetilde \Sigma_{t_0^*}.
	\end{array}\right.
	\end{equation}
	Here, $H^{\mu\nu\kappa\lambda}[\gara]$ is the tensor field defined previously (see~\eqref{hform1}). Furthermore, $\Pi_{\nu\alpha}$ is the orthogonal projection on the $\widetilde{\Sigma}_\tau$.
	
	In addition, we have the estimate:
	\begin{equation}\label{eq:toiter}
	\norm{\fara}^2_{\widetilde H^m(\Sigma_{t^*} \cap \mathcal{A}_{t_{0}^*}^{p_0})} \leq \norm{\fara_0}^2_{\widetilde H^m(\Sigma_{t_0^*} \cap \mathcal{A}_{t_{0}^*}^{p_0} )}\exp\left(C \left(t^*-t^*_0 +\int_{t_0^*}^{t^*} \norm{\bara}^w_{\widetilde{H}^m(\mathcal{A}^{p_0}_{t_{0}^*} \cap \widetilde{\Sigma}_s)}\de s\right)\right).
	\end{equation}
\end{proposition}

\begin{remark}
	We stress that the solution $\fara$ is defined on the whole set $\mathcal{A}_{t_0^*}^{p_0}$. 
\end{remark}
\begin{remark}	
	We follow the exposition in the book by Fritz John~\cite{fritzj}. An alternative proof could be obtained by writing our system in symmetric hyperbolic form (see~\cite{taylor}, Section 16.1).
\end{remark}

\begin{proof}
	We divide the proof in three \textbf{Steps}. First, we will express our system in the form required by the \emph{Cauchy--Kowalevskaya} Theorem. We will then proceed to apply said theorem to deduce existence of solutions to the linearized system in analysis when both all the components of $\fara_0$ and of $\gara$ are polynomials in regular coordinates for the set $\mathcal{A}_{t_0^*}^{p_0}$. We finally pass to the limit employing Stone--Weierstrass and the a priori estimates.
	
	\subsection*{Step 1}
	Without loss of generality, let us suppose that the cone $\mathcal{A}_{t_0^*}^{p_0} \subset \{r \geq 3M\}$. The case in which $\mathcal{A}_{t_0^*}^{p_0}$ is close to $\{r = 2M \}$ can be dealt with by an appropriate choice of coordinates. Let us furthermore, in this step, assume that $\gara$ is analytic on the cone $\mathcal{A}_{t_0^*}^{p_0}$.
	
	Under this assumption, the linearized system, using the electric--magnetic decomposition, can be written in Cauchy--Kowalevskaya form on $\mathcal{A}_{t_0^*}^{p_0}$.
	
	We start from the usual coordinates (irregular at $r = 2M$ and at the north and south poles): $(t, r, \theta, \varphi)$. We then define coordinates regular on $\mathcal{A}_{t_0^*}^{p_0}$ by the formula:
	\begin{equation}
	\begin{aligned}
	x^0 :&= t,\\
	x^1 :&= r \sin \theta \cos \varphi,\\
	x^2 :&= r \sin \theta \sin \varphi,\\
	x^3 :&= r \cos \theta.
	\end{aligned}
	\end{equation}
	In these coordinates, it is then clear that the metric $g$ will be represented by the matrix
	\begin{equation*}
	g_1 = \left(
	\begin{array}{cccc}
	-(1-\mu) & 0 & 0 & 0 \\
	0 & h_{11} & h_{12} & h_{13} \\
	0 & h_{11} & h_{12} & h_{13} \\
	0 & h_{11} & h_{12} & h_{13}
	\end{array}
	\right)
	\end{equation*}
	with $H := (h_{ij})_{i,j =1, 2, 3}$ being a positive definite matrix. Each $h_{ij}$ is an analytic function of $x^1, x^2, x^3$. Furthermore, there exists a number $b_1 > 0$ such that, for all points $(x^0, x^1, x^2, x^3)$ in the set $r \geq 3M$, we have that the matrix $H - b_1 I$ is positive definite.
	
	The electric--magnetic decomposition of $\fara$ is then defined in the usual way:
	\begin{align*}
	E_i := \fara(\partial_{x^0}, \partial_{x^i}), \qquad B_i := \farad(\partial_{x^0}, \partial_{x^i}),\\
	E^{\bara}_i := \bara(\partial_{x^0}, \partial_{x^i}), \qquad B^{\bara}_i := \barad(\partial_{x^0}, \partial_{x^i}).
	\end{align*}
	Assemble the first row of numbers into a vector: $D := (E_1, E_2, E_3, B_1, B_2, B_3)$.
	
	We now write the first equations
	$$
	\Pi_{\nu \lambda} \nabla_\mu \farad^{\mu \nu} = 0
	$$
	with respect to the $(x^i)_{i = 0, \ldots, 3}$ coordinates. These are the equations relative to the $B$ part of the field.
	
	Contracting successively with $\partial_{x^i}^\lambda$, $i =1, 2, 3$, and expanding the terms arising from the Christoffel symbols, we obtain
	\begin{equation}\label{eq:ckone}
	-(1-\mu)^{-1} \partial_{x^0} B_i = -(h^{-1})^{kl} \partial_{x^k} \farad(\partial_{x^l}, \partial_{x^i})
	+ F_1(x_1, x_2, x_3, D).
	\end{equation}
	Here, $F_1$ is an analytic function of its arguments. Together with the fact that
	$$
	\farad(\partial_{x^l}, \partial_{x^i}) = \sum_{i=1}^6 f_i(x^1, x^2, x^3) E_i,
	$$
	the $f_i$'s being analytic in their arguments, we have that Equation~\eqref{eq:ckone} is in Cauchy--Kowalevskaya form.
	
	Similarly, we analyze the equation
	$$
	H^{\mu\nu\kappa\lambda}[\bara] \ \Pi_{\nu \alpha}  \nabla_\mu \fara_{\kappa\lambda} = 0.
	$$
	This can be rewritten as
	$$
	\Pi_{\nu \alpha} (\nabla_\mu \fara^{\mu\nu} + H_{\Delta}^{\mu\nu\kappa\lambda}[\bara] \ \nabla_\mu \fara_{\kappa\lambda}) = 0.
	$$
	By an analogous reasoning, we contract successively with $\partial_{x^i}^\alpha$, $i =1, 2, 3$, and expand the terms arising from the Christoffel symbols. We obtain
	\begin{equation}
	\label{eq:cktwo}
	\begin{aligned}
	(\mathcal{M}_1)^{ij} \partial_{x^0} E_j +	(\mathcal{M}_2)^{ij} \partial_{x^0} B_j = (h^{-1})^{kl} \partial_{x^k} \fara(\partial_{x^l}, \partial_{x^i}) + (\mathcal{M}_3)^{ijkl} \partial_{x^i} \fara(\partial_{x^j}, \partial_{x^k})\\
	+ F_2(x_1, x_2, x_3, D^{\mathcal{G}}, D).
	\end{aligned}
	\end{equation}
	Here, $(\mathcal{M}_1)^{ij}$ and  $(\mathcal{M}_2)^{ij}$ are matrices depending analytically on $x^1, x^2, x^3$ and $\bara$, such that, when $\bara = 0$, $(\mathcal{M}_1)^{ij} = -(1-\mu)^{-1} \delta^{ij}$ and $(\mathcal{M}_2)^{ij} = 0$. Similarly, 
	$(\mathcal{M}_3)^{ijkl}$ is a $4$-matrix depending analytically on $x_1, x_2, x_3$ and $\bara$, such that, when $\bara =0$, $\mathcal{M}_3 = 0$. Finally, $F_2$ is a smooth function of its arguments.
	
	We then substitute in \eqref{eq:cktwo} for $\partial_{x^0}B_i$ from \eqref{eq:ckone}, and invert the matrix $\mathcal{M}_1$: it is possible to invert it because we assume that $\bara$ is small in $L^\infty$ norm, by possibly restricting the size of $\varepsilon_\text{lin}$. With these observations, and the analyticity of the functions involved in Equations \eqref{eq:ckone} and \eqref{eq:cktwo}, we conclude that those two equations, combined, are a Cauchy--Kowalevskaya form for the system \eqref{linmbi}.
	
	\subsection*{Step 2} \textbf{Claim:} Suppose the conditions in the statement of the Proposition, and furthermore suppose that $\bara$ is analytic. Then, the linearized system \eqref{eq:linearizeda} admits a solution on the whole $\mathcal{A}_{t_0^*}^{p_0}$.
	
	\textbf{Proof of Claim.} Let $m \in \enn$, $m \geq 5$. By Stone--Weierstrass, we can approximate the $E$, $B$ components of $\fara_0$ and their derivatives up to order $m$ uniformly by polynomials in $x^1, x^2, x^3$ on the set $\widetilde \Sigma_{t^*} \cap \mathcal{A}_{t_0^*}^{p_0}$. Let us call such an approximating sequence $\{\fara_0^{(j)}\}_{j \in \enn}$. We denote by $E^{(j)}_0$, $B^{(j)}_0$ respectively the $E$, $B$ components of $\fara_0$.
	
	Let us now suppose without loss of generality that the $E$, $B$ components of $\bara$ all belong to $\mathscr{A}_{M, c_0}(\mathcal{A}_{t_0^*}^{p_0})$ for some $M, c_0 >0$. Due to linearity of Equations \eqref{eq:linearizeda}, by dividing the initial data by a large constant (depending on the index $j$), we can also suppose that $E^{(j)}_0$, $B^{(j)}_0$ all belong to $\mathscr{A}_{M, c_0}(\mathcal{A}_{t_0^*}^{p_0} \cap \widetilde{\Sigma}_{t_0^*})$. Then, the Cauchy--Kowalevskaya Theorem (Theorem~\ref{thm:ck}) provides us with a time $t_1^*$ depending only on $c_0, M$ and a sequence of solutions $\fara^{(j)}$ to the system \eqref{eq:linearizeda} on $\mathcal{A}_{t_0^*}^{p_0} \cap \{t^* \leq t_1^*\}$.
	
	By the estimate \eqref{eq:toiter} and Ascoli-Arzel\`a we now have compactness of the sequence $\fara^{(j)}$ in $\mathcal{C}^1$, and hence, passing to the limit and using again the estimates \eqref{eq:toiter}, we obtain a smooth solution to the problem \eqref{eq:linearizeda} on $\mathcal{A}_{t_0^*}^{p_0} \cap \{t^* \leq t_1^*\}$ with the further assumption that $\bara$ is analytic. Since $t_1^*$ does not depend on $\fara_0$, and since the linearized system \eqref{eq:linearizeda} enjoys the domain of dependence property, we have that the solution $\fara$ can be extended to $\mathcal{A}_{t_0^*}^{p_0}$.

	\subsection*{Step 3}
	Finally, we have to remove the analyticity assumption on the components of $\bara$. We just suppose the components of $\bara$ to be smooth. We find an approximating sequence $(\bara_j)_{j \in \enn}$ of tensors whose components $E^{\bara_j}$ and $B^{\bara_j}$ are analytic. Furthermore, we suppose that the approximation is uniform in $\mathcal{C}^m$. This means that for all multi-indices $\underline{\gamma}$ of order less or equal than $m$, we have the following:
	\begin{equation}
	\begin{aligned}
	\sup_{x \in \mathcal{A}_{t_0^*}^{p_0}} |\partial^{\underline{\gamma}}(E^{\bara_j}- E^{\bara})| \to 0 \text{ as } j \to \infty,\\
	\sup_{x \in \mathcal{A}_{t_0^*}^{p_0}} |\partial^{\underline{\gamma}}(B^{\bara_j}- B^{\bara})| \to 0 \text{ as } j \to \infty,
	\end{aligned}
	\end{equation}
	The previous \textbf{step} yields a sequence of smooth solutions $\fara^{(j)}$ to \eqref{eq:linearizeda}, defined on all of $\mathcal{A}_{t_0^*}^{p_0}$. The estimates \eqref{eq:toiter} and Sobolev embedding yield uniform boundedness of the derivatives of $\mathcal{F}^{(j)}$ up to order $3$. The Theorem of Ascoli-Arzel\`a then lets us conclude the existence of a smooth solution in the conditions of the Proposition.
\end{proof}

\subsection{Proof of the local existence Theorem for the MBI system}\label{sub:proofloc}

We are now in shape to prove the local existence Theorem~\ref{thm:local}.

\begin{proof}[Proof of Theorem~\ref{thm:local}]  Proposition~\ref{prop:linex}, and an iteration argument where we set $\bara := \fara^{(n)}$ and we solve for $\fara^{(n+1)}$, combined with the a priori estimates of Proposition~\ref{energyprop} give the existence of a smooth solution to the evolution part of the MBI system on $\mathcal{A}_{t_0^*}^{p_0}$. We recall that the evolution part of the MBI system is the following:
	
	\begin{equation}
	\left\{
	\begin{array}{ll}
	\Pi_{\nu \alpha} \nabla_\mu \farad^{\mu \nu} = 0,\\
	H^{\mu\nu\kappa\lambda}[\bara] \ \Pi_{\nu \alpha}  \nabla_\mu \fara_{\kappa\lambda} = 0.
	\end{array}\right.
	\end{equation}
	
	Hence, we only need to prove propagation of the constraints. We prove the following \textbf{Claim}. Let $N := n_{\widetilde{\Sigma}_{t^*}}$, the future-directed unit normal vector to the foliation $\widetilde{\Sigma}_{t^*}$.
	
	\textbf{Claim.} If $\ipsi$ is a $2$-covariant antisymmetric tensor on the spacetime slab $\regio{}{}$, which further satisfies $\nabla^\mu \sipsi_{\mu X} = 0$, for all $X$ such that $g(X, N) = 0$, then $\nabla_N (\nabla^\mu \sipsi_{\mu N}) = 0$ on $\regio{}{}$.
	
	To prove the \textbf{claim}, let's calculate
	\begin{equation*}
	\nabla^N \nabla^\mu \sipsi_{\mu N} = \nabla^\mu \nabla^N \sipsi_{\mu N} + \rie\indices{^\eta_{NN}^\mu} \sipsi_{\mu\eta} + \rie\indices{^\eta_{\mu N}^\mu} \sipsi_{\eta N} =  \nabla^\mu \nabla^N \sipsi_{\mu N}.
	\end{equation*}
	The second term in the last display vanishes because of antisymmetry of $\ipsi$, and the third term vanishes because of the Ricci-flatness of Schwarzschild. Then
	\begin{equation}\label{eq:propaga}
	\begin{aligned}
	&\nabla^\mu \nabla^N \sipsi_{\mu N} = \frac 1 2 \nabla^\mu \varepsilon_{\mu N\alpha\beta} \nabla^N \ipsi^{\alpha \beta} = - \nabla^\mu \varepsilon_{\mu N \alpha \beta} \nabla^\beta \ipsi^{N \alpha} \\&=
	-\varepsilon_{\mu N \alpha \beta} \nabla^\mu \nabla^\beta \ipsi^{N \alpha}
	- \varepsilon_{\mu\delta\alpha\beta} (\nabla^\mu N^\delta )\ipsi^{N \alpha}.
	\end{aligned}
	\end{equation}
	Here, we used the fact that $\ipsi$ satisfies $\nabla^\mu \sipsi_{\mu X} = 0$, for all $X$ such that $g(X, N) = 0$.
	We also notice that, since $N^\alpha = g^{\alpha\beta} (dt^*)_\beta$, and from the symmetry of the covariant Hessian, we have
	$$
	\varepsilon_{\mu\delta\alpha\beta} \nabla^\mu N^\delta = 0.
	$$
	
	Therefore, we proceed to calculate
	\begin{equation}\label{16}
	-  \varepsilon_{\mu N \alpha \beta} (\nabla^\mu \nabla^\beta \ipsi^{N \alpha}-\nabla^\beta \nabla^\mu \ipsi^{N \alpha} )= \varepsilon_{\beta N \alpha \mu} (\rie\indices{_{\gamma N}^{\mu\beta}} \ipsi^{\gamma\alpha} + \rie\indices{_\gamma^{\alpha \mu \beta}}\ipsi^{N \gamma}).
	\end{equation}
	By the first Bianchi identity, we see that the second term in the last display (\ref{16}) is necessarily zero.
	
	We notice the following facts, the second implied by the form of the Riemann tensor in Schwarzschild in Appendix~\ref{sec:riemann}:
	\begin{itemize}
		\item Away from the horizon, $N = A \partial_t + B \partial_r$, with $A$ and $B$ smooth and bounded functions;
		\item In the frame $\partial_t, \partial_r, \partial_\phi, \partial_\theta$, the nonvanishing components of the Riemann tensor are: $\rie(\partial_{x^\alpha},\partial_{x^\beta},\partial_{x^\beta},\partial_{x^\alpha})$.
	\end{itemize}
	Define now, abusing notation, the shorthand $T := \partial_t$, $R:= \partial_r$, $\theta := \partial_\theta$, $\varphi := \partial_\varphi$, the coordinate fields in the irregular coordinate system $(t,r,\theta, \varphi)$. We have the following:
	\begin{align*}
	&\varepsilon_{\beta T \alpha \mu} (\rie\indices{_{\gamma T}^{\mu\beta}} \ipsi^{\gamma\alpha}) = 2 \varepsilon_{\beta T \alpha T} (\rie\indices{_{\gamma T}^{T\beta}} \ipsi^{\gamma\alpha}) = 0, \\
	&\varepsilon_{\beta R \alpha \mu} (\rie\indices{_{\gamma R}^{\mu\beta}} \ipsi^{\gamma\alpha}) = 2 \varepsilon_{\beta R \alpha R} (\rie\indices{_{\gamma R}^{R\beta}} \ipsi^{\gamma\alpha}) = 0, \\
	&\varepsilon_{\beta R \alpha \mu} \rie\indices{_{\gamma T}^{\mu \beta}} \ipsi^{\gamma \alpha} + \varepsilon_{\beta T \alpha \mu} \rie\indices{_{\gamma R}^{\mu\beta}} \ipsi^{\gamma \alpha} \\
	& = 2 \varepsilon_{TR \alpha\mu} \rie\indices{_{\gamma T}^{\mu T}} \ipsi^{\gamma \alpha} + 2 \varepsilon_{R T \alpha \mu} \rie\indices{_{\gamma R }^{\mu R}} \ipsi^{\gamma \alpha} \\
	&= 2 \varepsilon_{TR \theta\varphi} \rie\indices{_{\varphi T}^{\varphi T}} \ipsi^{\varphi \theta} + 
	2 \varepsilon_{TR \varphi\theta} \rie\indices{_{\theta T}^{\theta T}} \ipsi^{\theta \varphi} +
	2 \varepsilon_{R T \theta \varphi} \rie\indices{_{\varphi R }^{\varphi R}} \ipsi^{\varphi \theta} +
	2 \varepsilon_{R T \theta \varphi} \rie\indices{_{\theta R }^{\theta R}} \ipsi^{\varphi \theta} \\
	& =2 \varepsilon_{TR \theta \varphi} (- \rie\indices{_{RT}^{RT}}) \ipsi^{\varphi \theta} - 2 \varepsilon_{TR \theta \varphi}(- \rie\indices{_{TR}^{TR}}) \ipsi^{\varphi \theta} = 0.
	\end{align*}
	This proves that the right hand side in Equation~\eqref{eq:propaga} is zero, hence proving propagation of constraints and the local existence theorem.
\end{proof}
\section{The bootstrap assumptions}\label{sec:bootstrap}
Recall: $$\tau := (1-\chi_1(r))(u^+ +1)+ \chi_1(r)t^*,$$  $\chi_1$ being a smooth cutoff function such that
\begin{equation*}
	\chi_1(r) = 1 \text{ for } r \in [3/2 M, 3 M], \qquad \chi_1(r) = 0 \text{ for } r \in (0, M] \cup [4M, \infty).
\end{equation*}
Here, $u^+ := \max\{u, 0\}$.

\begin{assumpt}[Bootstrap assumptions]\label{as:bootstrap}
	Let $A > 0$, $N \in \N_{\geq 0}$, $\varepsilon > 0$, $r_{\text{in}} \in (M, 2M)$, and $\mathcal R \subset \mathcal{S}$. We say that $\mathcal{\fara}$ satisfies the bootstrap assumptions $\text{BA}(\mathcal R,A,N, \varepsilon)$ if the following two sets of $L^\infty$ bounds hold:
	\begin{equation}\label{bootstrap}
	\boxed{
		\begin{aligned}
		(\text{B}1,1) \qquad \sum_{i=0}^N|\partial^i \rho| &\leq A \varepsilon^{3/4} \tau^{-1} r^{-3/2}, 
		\qquad &(\text{B}1,2) \qquad \sum_{i=0}^N|\partial^i \rho| &\leq A \varepsilon^{3/4} \tau^{-1/2} r^{-2},\\
		(\text{B}2,1) \qquad \sum_{i=0}^N|\partial^i \sigma| & \leq  A \varepsilon^{3/4} \tau^{-1} r^{-3/2}, 
		\qquad  &(\text{B}2,2) \qquad \sum_{i=0}^N|\partial^i \sigma| &\leq A \varepsilon^{3/4} \tau^{-1/2} r^{-2},\\
		(\text{B}3,1) \qquad \sum_{i=0}^N|\partial^i \alphabar| &\leq A \varepsilon^{3/4} \tau^{-1} r^{-1}, \\
		(\text{B}4,1) \qquad \sum_{i=0}^N|\partial^i \alpha| &\leq  A \varepsilon^{3/4} \tau^{-1} r^{-\frac 3 2}, \qquad &(\text{B}4,2) \qquad \sum_{i=0}^N|\partial^i \alpha| &\leq  A \varepsilon^{3/4} \tau^{-\frac 1 2} r^{-2}
		\end{aligned}
	}
	\end{equation}
	
	\textbf{ in the region } $\mathcal{R} \cap \{r \geq 2M \}$, and
	\begin{equation}\label{bootstrap2}
	\boxed{
		\begin{aligned}
		(\text{B}'1) \qquad \sum_{i=0}^N|\partial^i \rho| \leq A \varepsilon^{3/4}, \qquad \qquad
		(\text{B}'2) \qquad \sum_{i=0}^N|\partial^i \sigma| \leq  A \varepsilon^{3/4},\\
		(\text{B}'3) \qquad \sum_{i=0}^N|\partial^i \alphabar| \leq A \varepsilon^{3/4}, \qquad \qquad
		(\text{B}'4) \qquad \sum_{i=0}^N|\partial^i \alpha| \leq  A \varepsilon^{3/4}
		\end{aligned}
	}
	\end{equation}
	\textbf{ in the region } $\mathcal{R} \cap \{r \in [r_{\text{in}}, 2M) \}$.
\end{assumpt}
\begin{remark}
	Notice that, in the interior region, we only have boundedness.
\end{remark}
\begin{remark}
	Notice that, in particular, the bootstrap assumptions give additional  $r$-decay for angular derivatives only.
\end{remark}
\begin{remark} \label{decaytrmk}
	Also notice that, since in the region $r^* \geq 1$ $u \geq 1$, either
	\begin{equation*}
	\begin{aligned}
	u &\geq \frac 1 3 v, \qquad \text{ or } \qquad
	r^* &\geq \frac 1 3 v,
	\end{aligned}
	\end{equation*}
	we have that the bootstrap assumptions imply a bound $\leq C' \varepsilon^{3/4} (v)^{-1}$ for all the components of the MBI field.
\end{remark}

\begin{remark}
	Proposition~\ref{norm.equiv} gives us a way to interchange these assumptions with assumptions on repeated Lie differentiation.
\end{remark}

\section{Commutation of the system}\label{sec:commut}
In this section, we find the form of the commuted MBI system.
Recall the definition of $\iindicess^b$ from Section~\ref{sec:defs}: it is the set of multi-indices of length $b$ formed from elements of $\mathcal{K}$. Let $I \in \iindicess^b$.
Given $I \in \iindicess^b$, and $a \in \N_{\geq 0}$, recall the variation
$$
\faram_{,I,a} := \lie^I_{\mathcal{K}} \nabla^a_Y \fara.
$$
\begin{remark}
	In this section and in the following, unless otherwise stated, we abbreviate
	$$
	H^{\mu\nu\kappa\lambda} = H^{\mu\nu\kappa\lambda}[\fara], \ \ {H_{\Delta}}^{\mu\nu\kappa\lambda} = H\indices{_\Delta^\mu^\nu^\kappa^\lambda}[\fara]
	$$
\end{remark}
\begin{proposition}\label{prop:commut}
Fix $r_{\text{in}} \in (M,2M)$, and a number $t_1^* \geq t_0^*$. Let the smooth tensor $\fara$ satisfy the MBI system~\eqref{MBI} on $\tregio{t_0^*}{t_1^*}$.
Then, we have
\begin{equation}\label{eq:commutzero}
\boxed{
\begin{aligned}
\nabla^\gamma (\faramd_{,I,0})_{\gamma \nu} &= 0,\\
H^{\mu\nu\kappa\lambda}\nabla_\mu  (\faram_{,I,0})_{\kappa\lambda} &= H_\Delta^{\mu\nu\kappa\lambda} \nabla_\mu (\faram_{,I,0})_{\kappa\lambda}  -\mathcal{L}^I_\mathcal{K}(H_\Delta^{\mu\nu\kappa\lambda} \nabla_\mu (\fara_{\kappa\lambda})),
\end{aligned}
}
\end{equation}
in the region $\tregio{t_0^*}{t_1^*}$.
Furthermore, for $a \geq 1$,
\begin{equation}\label{eq:commut}
\boxed{
\begin{aligned}
\nabla^\gamma (\faramd_{,I,a})_{\gamma \nu} &= a (\nabla^\mu Y^\alpha) \nabla_\alpha(\faramd_{,I, a-1})_{\mu\nu} +  (\lie^I_\mathcal{K } \text{OT}^{(a)}_1)_\nu,\\
H^{\mu\nu\kappa\lambda}\nabla_\mu ( \faram_{,I,a})_{\kappa\lambda} &= a (\nabla^\kappa Y^\alpha) \nabla_\alpha (\faram_{,I,a-1})\indices{_\kappa^\nu}  \\&- \sum_{K+L =I, |K|\geq 1}(\lie_\mathcal{K}^K H^{\mu\nu\kappa\lambda})\nabla_\mu ( \faram_{,L,a})_{\kappa\lambda}  + (\lie_\mathcal{K}^I \text{OT}^{(a)}_1)^\nu,
\end{aligned}}
\end{equation}
in the region $\tregio{t_0^*}{t_1^*} \cap \{r_1 \in (3/2 M, 5/2 M)\}$. In this formula, we used the definition of sum of multi-indices in $\iindicess^b$, given in Equation~\eqref{eq:defsumind}.

Here, the terms $\text{OT}_{1}^{(a)}$ satisfy the schematic equation, in the sense of Definition~\ref{def:short}, for $a \geq 1$:
\begin{equation*}
(\text{OT}_{1}^{(a)})_\nu= \sum_{\substack{(m_1, m_2) \in \N_{\geq 0}^2 \\ m_1+m_2 \leq a+1 \\ m_1, m_2 \leq a}} (\nabla^{m_1} H_\Delta \nabla^{m_2} \fara + \nabla^{\min\{m_1, a-1\}}\fara)_\nu.
\end{equation*}
The terms $(\text{OT}_{1}^{(a)})_\nu$ satisfy also the following bound: there exist a constant $C = C_{M, N_0}$ such that
\begin{equation*}
|(\text{OT}_{1}^{(a)})_\nu |\leq C_{M, N_0}  \sum_{\substack{(m_1, m_2) \in \N_{\geq 0}^2 \\ m_1+m_2 \leq a+1 \\ m_1, m_2 \leq a}} (|\nabla^{m_1} H_\Delta^{\mu\nu\kappa\lambda}| |\nabla^{m_2} \fara_{\kappa \lambda}|+|\nabla^{\min\{m_1, a-1\}}\fara_{\kappa\lambda}|).
\end{equation*}
\end{proposition}

\begin{proof}
We divide the proof in three subsections.

\subsection{Commutation with Killing fields}

We commute the system with repeated Lie differentiation in direction $\mathcal{K}$. We obtain the equations of variation
\begin{equation}\label{eqvar1}
\begin{aligned}
\nabla^\gamma (\faramd_{,I,0})_{\gamma \nu} &= 0,\\
H^{\mu\nu\kappa\lambda}\nabla_\mu  (\faram_{,I,0})_{\kappa\lambda} &=(\mathcal{I}^{,I,0})^\nu.
\end{aligned}
\end{equation}
Here,
\begin{equation*}
(\mathcal{I}^{,I,0})^\nu = H_\Delta^{\mu\nu\kappa\lambda} \nabla_\mu ((\faram_{,I,0})_{\kappa\lambda})  -\mathcal{L}^I_\mathcal{K}(H_\Delta^{\mu\nu\kappa\lambda} \nabla_\mu (\fara_{\kappa\lambda})).
\end{equation*}
\subsection{Commutation with the $Y$ vector field}

We commute the system with the operator $\nabla_Y$, where we recall that $Y$ is
$$
Y:= \chi_1(r) (1-\mu)^{-1} \lbar.
$$
Here, $\chi_1$ is a smooth cutoff function such that $\chi_1(r) = 1$ for $r \in [3/2 M, 3M]$, and $\chi_1(r) = 0$ for $r \in (0,M] \cup [4M, \infty)$. We introduce the cutoff so that we do not have to deal with the geometry for large values of $r$, and we can bound the Riemann tensor and its derivatives by smoothness.

\begin{remark}
We recall that the contractions here are performed in the following way: 
$$
\nabla^j_Y \fara := \nabla_Y (\nabla_Y (\cdots \nabla_Y \fara)),
$$
and \textbf{not} as
$$
(\nabla^j\fara)(Y, \ldots, Y).
$$
\end{remark}
Let us then recall the system satisfied by $\fara$
\begin{equation*}\left\{
\begin{array}{ll}
\nabla_{[\mu} \fara_{\nu \eta]} = 0,\\
H^{\mu\nu\kappa\lambda} \nabla_{\mu} \fara_{\kappa\lambda} = 0.
\end{array}\right.
 \iff
 \left\{
\begin{array}{ll}
\nabla_\mu \farad\indices{^\mu _\eta} = 0,\\
\nabla_\mu \fara\indices{^\mu ^\nu}+{H_\Delta}^{\mu\nu\kappa\lambda} \nabla_{\mu} \fara_{\kappa\lambda} = 0.
\end{array}\right.
\end{equation*}
We prove the case $I = 0$ by induction.
\\
We first prove the base step, $a =1$. To do that, we commute the above equations with $Y$, and obtain
\begin{equation*}
\begin{aligned}
&H^{\mu\nu\kappa\lambda} \nabla_\mu {\nabla_Y \fara}_{\kappa \lambda} = H^{\mu\nu\kappa\lambda} \nabla_\mu (Y^\alpha \nabla_\alpha \fara_{\kappa\lambda})\\
&=H^{\mu\nu\kappa\lambda} (\nabla_\mu Y^\alpha) \nabla_\alpha \fara_{\kappa\lambda} + H^{\mu\nu\kappa\lambda} Y^\alpha \nabla_\mu \nabla_\alpha \fara_{\kappa\lambda} \\ 
&=H^{\mu\nu\kappa\lambda} (\nabla_\mu Y^\alpha) \nabla_\alpha \fara_{\kappa\lambda} + H^{\mu\nu\kappa\lambda} Y^\alpha (\nabla_\alpha \nabla_\mu \fara_{\kappa\lambda} + \rie\indices{_{\mu \alpha}^\gamma_\kappa}\fara_{\gamma \lambda}+\rie\indices{_{\mu\alpha}^\gamma_\lambda}\fara_{\kappa \gamma} ) \\ 
&=\left(\frac 1 2 (g^{\mu\kappa}g^{\nu \lambda}-g^{\mu\lambda}g^{\nu\kappa})+H_\Delta^{\mu\nu\kappa\lambda} \right) (\nabla_\mu Y^\alpha) \nabla_\alpha \fara_{\kappa\lambda} + H^{\mu\nu\kappa\lambda} Y^\alpha \nabla_\alpha \nabla_\mu \fara_{\kappa\lambda}+ (\text{OT}^{(1)}_1)^\nu\\
&=\nabla^\kappa Y^\alpha \nabla_\alpha \fara\indices{_\kappa^\nu} + Y^\alpha\nabla_\alpha (H^{\mu\nu\kappa\lambda}\nabla_\mu \fara_{\kappa\lambda}) - Y^\alpha (\nabla_\alpha H^{\mu\nu\kappa\lambda}) \nabla_\mu \fara_{\kappa \lambda}+ (\text{OT}^{(1)}_1)^\nu.
\end{aligned}
\end{equation*}
The terms $(\text{OT}^{(1)}_1)^\nu$ (``other terms of level 1'') can be written as (using the shorthand notation of Definition~\ref{def:short}):
\begin{equation*}
(\text{OT}^{(1)}_1)^\nu = (H_\Delta \nabla \fara + \nabla H_\Delta \nabla \fara + \fara + H_\Delta \fara )^\nu.
\end{equation*}
Recall that $H_\Delta$ is $H$ without its linear part.
Furthermore, the smoothness of the Riemann tensor and the fact that $Y$ is compactly supported, imply that $(\text{OT}^{(1)}_1)^\nu$ satisfies the inequality
\begin{equation*}
|(\text{OT}^{(1)}_1)^\nu| \leq C (|H_\Delta^{\mu\nu\kappa\lambda}| (|\nabla_\alpha \fara_{\beta \gamma}| + |\fara_{\kappa \lambda}|)+|\nabla_\delta H_\Delta^{\mu\nu\kappa\lambda}| |\nabla_\alpha \fara_{\beta \gamma}|+ |\fara_{\kappa\lambda}|).
\end{equation*}
We obtain:
\begin{equation*}
H^{\mu\nu\kappa\lambda} \nabla_\mu {\nabla_Y \fara}_{\kappa \lambda} = \nabla^\kappa Y^\alpha \nabla_\alpha \fara\indices{_\kappa^\nu} + (\text{OT}^{(1)}_1)^\nu
\end{equation*}
and
\begin{equation*}
\begin{aligned}
\nabla^\mu (\nabla_Y \farad)_{\mu\nu} = (\nabla^\mu Y^\alpha) \nabla_\alpha \farad_{\mu\nu} + \rie\indices{^\mu_\alpha^\gamma_\nu} \farad\indices{_\mu_\gamma} Y^\alpha = (\nabla^\mu Y^\alpha) \nabla_\alpha \farad_{\mu\nu} + (\text{OT}^{(1)}_1)_\nu.
\end{aligned}
\end{equation*}
In conclusion, the first order commuted system is
\begin{equation}\label{eq:varequno}
\boxed{\begin{aligned}
\nabla^\gamma (\nabla_Y \farad)_{\gamma \nu} = (\mathcal{J}^{,0,1})_\nu,\\
H^{\mu\nu\kappa\lambda}\nabla_\mu (\nabla_Y \fara)_{\kappa\lambda} =({\mathcal{I}^{,0,1}})^\nu,
\end{aligned}}
\end{equation}
where we defined
\begin{equation*}
\begin{aligned}
(\mathcal{J}^{,0,1})_\nu := (\nabla^\mu Y^\alpha) \nabla_\alpha \farad_{\mu\nu} + (\text{OT}^{(1)}_1)_\nu, \\
({\mathcal{I}^{,0,1}})^\nu := \nabla^\kappa Y^\alpha \nabla_\alpha \fara\indices{_\kappa^\nu} + (\text{OT}^{(1)}_1)^\nu.
\end{aligned}
\end{equation*}
We now recall that the $(\text{OT}_{1}^{(a)})_\nu$ terms satisfy the schematic equation (Definition~\ref{def:short}), for $a \geq 1$,
\begin{equation*}
(\text{OT}_{1}^{(a)})_\nu= \sum_{\substack{(m_1, m_2) \in \N_{\geq 0}^2 \\ m_1+m_2 \leq a+1 \\ m_1, m_2 \leq a}} (\nabla^{m_1} H_\Delta \nabla^{m_2} \fara + \nabla^{\min\{m_1, a-1\}}\fara)_\nu,
\end{equation*}
along with the bound
\begin{equation*}
\left|(\text{OT}_{1}^{(a)})_\nu\right| \leq C  \sum_{\substack{(m_1, m_2) \in \N_{\geq 0}^2 \\ m_1+m_2 \leq a+1 \\ m_1, m_2 \leq a}} (|\nabla^{m_1} H_\Delta^{\mu\nu\kappa\lambda}| |\nabla^{m_2} \fara_{\kappa \lambda}|+|\nabla^{\min\{m_1, a-1\}}\fara_{\kappa\lambda}|).
\end{equation*}
Here, $C = C_{M, N_0}$.

Let us now prove the inductive step. By the reasoning which we have just concluded, we know that, when $a = 1$,
\begin{equation*}
\begin{aligned}
H^{\mu\nu\kappa\lambda}\nabla_\mu ( \faram_{,0,1})_{\kappa\lambda} &=(\nabla^\kappa Y^\alpha) \nabla_\alpha \fara\indices{_\kappa^\nu} + (\text{OT}^{(1)}_1)^\nu,\\
\nabla^\gamma (\faramd_{,0,1})_{\gamma \nu} &= (\nabla^\mu Y^\alpha) \nabla_\alpha \farad_{\mu\nu} + (\text{OT}^{(1)}_1)_\nu.
\end{aligned}
\end{equation*}
Let us suppose that this equality holds up to level $a-1$, i.e. that there holds
\begin{equation}\label{eq:inducthyp}
\begin{aligned}
H^{\mu\nu\kappa\lambda}\nabla_\mu ( \faram_{,0,a-1})_{\kappa\lambda} = (a-1)\nabla^\kappa Y^\alpha \nabla_\alpha (\faram_{,0,a-2})\indices{_\kappa^\nu} + (\text{OT}^{(a-1)}_1)^\nu = (I^{,0,a-1})^\nu,\\
\nabla^\gamma (\faramd_{,0,a-1})_{\gamma \nu} = (a-1) \nabla^\mu Y^\alpha \nabla_\alpha( \faramd_{,0, a-2})_{\mu\nu} + (\text{OT}^{(a-1)}_1)_\nu = (\mathcal{J}^{,0,a-1})_\nu.
\end{aligned}
\end{equation}
\begin{remark}
Notice that, by our definition, and by the fact that the volume form is parallel,
\begin{equation*}
\nabla_Y \faram_{,0,a} = \faram_{,0,a+1}, \qquad \nabla_Y \faramd_{,0,a} = \faramd_{,0,a+1}.
\end{equation*}
\end{remark}
Then, as before
\begin{equation*}
\begin{aligned}
&H^{\mu\nu\kappa\lambda} \nabla_\mu {(\nabla_Y \faram_{,0,a-1})}_{\kappa \lambda} = H^{\mu\nu\kappa\lambda} \nabla_\mu (Y^\alpha \nabla_\alpha (\faram_{,0,a-1})_{\kappa\lambda})\\ 
&=
H^{\mu\nu\kappa\lambda} (\nabla_\mu Y^\alpha) \nabla_\alpha (\faram_{,0,a-1})_{\kappa\lambda} + H^{\mu\nu\kappa\lambda} Y^\alpha \nabla_\mu \nabla_\alpha  (\faram_{,0,a-1})_{\kappa\lambda} \\ & =
\left(\frac 1 2 (g^{\mu\kappa}g^{\nu \lambda}-g^{\mu\lambda}g^{\nu\kappa})+H_\Delta^{\mu\nu\kappa\lambda} \right) (\nabla_\mu Y^\alpha)  \nabla_\alpha (\faram_{,0,a-1})_{\kappa\lambda} + H^{\mu\nu\kappa\lambda} Y^\alpha \nabla_\alpha \nabla_\mu  (\faram_{,0,a-1})_{\kappa\lambda}+ (\text{OT}^{(a)}_1)^\nu \\ & =
\nabla^\kappa Y^\alpha  \nabla_\alpha (\faram_{,0,a-1})\indices{_\kappa^\nu} + Y^\alpha\nabla_\alpha (H^{\mu\nu\kappa\lambda}\nabla_\mu (\faram_{,0,a-1})_{\kappa\lambda}) - Y^\alpha (\nabla_\alpha H^{\mu\nu\kappa\lambda}) \nabla_\mu (\faram_{,0,a-1})_{\kappa\lambda}+ (\text{OT}^{(a)}_1 )^\nu\\ & =
\nabla^\kappa Y^\alpha  \nabla_\alpha (\faram_{,0,a-1})\indices{_\kappa^\nu} + (a-1)Y^\alpha\nabla_\alpha (\nabla^\kappa Y^\beta \nabla_\beta (\faram_{,0,a-2})\indices{_\kappa^\nu} + \text{OT}^{(a-1)}_1)+  (\text{OT}^{(a)}_1)^\nu \\ & =
\nabla^\kappa Y^\alpha  \nabla_\alpha (\faram_{,0,a-1})\indices{_\kappa^\nu}
+ (a-1)Y^\alpha \nabla^\kappa Y^\beta \nabla_\alpha \nabla_\beta (\faram_{,0,a-2})\indices{_\kappa^\nu}
+ (\text{OT}^{(a)}_1)^\nu \\ & =
\nabla^\kappa Y^\alpha  \nabla_\alpha (\faram_{,0,a-1})\indices{_\kappa^\nu}
+ (a-1)Y^\alpha \nabla^\kappa Y^\beta \nabla_\beta \nabla_\alpha (\faram_{,0,a-2})\indices{_\kappa^\nu}
+ (\text{OT}^{(a)}_1)^\nu \\ & =
a \nabla^\kappa Y^\alpha  \nabla_\alpha (\faram_{,0,a-1})\indices{_\kappa^\nu}
+ (\text{OT}^{(a)}_1)^\nu.
\end{aligned}
\end{equation*}
This reasoning, along with the reasoning for the equation satisfied by $\faramd$, shows then, when $a \geq 1$,
\begin{equation}\label{eq:jvariation}
\boxed{
\begin{aligned}
\nabla^\gamma (\faramd_{,0,a})_{\gamma \nu} &= a (\nabla^\mu Y^\alpha) \nabla_\alpha( \faramd_{,0, a-1})_{\mu\nu} + (\text{OT}^{(a)}_1)_\nu = (\mathcal{J}^{,0,a})_\nu,\\
H^{\mu\nu\kappa\lambda}\nabla_\mu ( \faram_{,0,a})_{\kappa\lambda} &= a (\nabla^\kappa Y^\alpha) \nabla_\alpha (\faram_{,0,a-1})\indices{_\kappa^\nu} + (\text{OT}^{(a)}_1)^\nu = (\mathcal{I}^{,0,a})^\nu.
\end{aligned}}
\end{equation}
The proposition is then proved for the pure $Y$-derivatives in the proximity of $\mathcal{H}^+$.
\subsection{Commutation with mixed derivatives}
We finally complete the reasoning for mixed derivatives. We commute our previous Equation~\eqref{eq:jvariation} with $\lie^I_\mathcal{K}$.
We recall that, if $I \in \iindicess^b$ is a multi-index consisting only of Killing fields,
\begin{equation}
\faram_{,I,a} := \lie^I_{\mathcal{K}} \nabla_Y^a \fara.
\end{equation}

We now commute previous Equation~\eqref{eq:jvariation} with $\lie^I_\mathcal{K}$. We have the following:
\begin{itemize}
\item in virtue of~\cite{globalnon}, Equation 3.25, if $K \in \mathcal{K}$, $[\mathcal{L}_K, \nabla] = 0$.
\item If $K$ is either the vector field $\partial_t$ or $\partial_\varphi$ (or any other rotation Killing field), we have
$$
\lie_K \varepsilon = 0.
$$
This implies that $\lie_K \farad_{\kappa \lambda} = {}^\star(\lie_K \fara)_{\kappa \lambda}$.
\item We have that, for $K \in \mathcal{K}$,
$$
\lie_K Y = [K, Y] = 0.
$$
\end{itemize}
Using these properties, we obtain:
\begin{equation}\label{eq:mixvariation1}
\begin{aligned}
\nabla^\gamma (\faramd_{,I,a})_{\gamma \nu} = a (\nabla^\mu Y^\alpha) \nabla_\alpha(\faramd_{,I, a-1})_{\mu\nu} +  \lie^I_\mathcal{K } \text{OT}^{(j)}_1 = \lie^I_\mathcal{K }(\mathcal{J}^{,0,a})_\nu,
\end{aligned}
\end{equation}

\begin{equation}\label{eq:mixvariation2}
\begin{aligned}
H^{\mu\nu\kappa\lambda}\nabla_\mu ( \faram_{,I,a})_{\kappa\lambda} \\ & = 
a \nabla^\kappa Y^\alpha \nabla_\alpha (\faram_{,I,a-1})\indices{_\kappa^\nu} - \sum_{K+L =I, |K|\geq 1}(\lie_\mathcal{K}^K H^{\mu\nu\kappa\lambda})\nabla_\mu ( \faram_{,L,a})_{\kappa\lambda}  + \lie_\mathcal{K}^I \text{OT}^{(a)}_1 \\ & = 
\lie_\mathcal{K}^I (\mathcal{I}^{,0,a})^\nu
- \sum_{K+L =I}(\lie_\mathcal{K}^K H^{\mu\nu\kappa\lambda})\nabla_\mu ( \faram_{,L,a})_{\kappa\lambda}  .
\end{aligned}
\end{equation}
Hence, the equations of variation with mixed derivatives are
\begin{equation}
\boxed{
\begin{aligned}
\nabla^\gamma (\faramd_{,I,a})_{\gamma \nu} &= a (\nabla^\mu Y^\alpha) \nabla_\alpha(\faramd_{,I, a-1})_{\mu\nu} +  (\lie^I_\mathcal{K } \text{OT}^{(a)}_1)_\nu,\\
H^{\mu\nu\kappa\lambda}\nabla_\mu ( \faram_{,I,a})_{\kappa\lambda} &= a( \nabla^\kappa Y^\alpha) \nabla_\alpha (\faram_{,I,a-1})\indices{_\kappa^\nu}  \\&- \sum_{K+L =I, |K|\geq 1}(\lie_\mathcal{K}^K H^{\mu\nu\kappa\lambda})\nabla_\mu ( \faram_{,L,a})_{\kappa\lambda}  + (\lie_\mathcal{K}^I \text{OT}^{(a)}_1)^\nu.
\end{aligned}}
\end{equation}
This proves the claim and concludes the proof of Proposition~\ref{prop:commut}.
\end{proof}

\begin{remark}\label{rmk:formij}
The fact that $Y$ is supported away from spatial infinity lets us write, {\bf for $a \geq 1$,}
\begin{equation*}
(\lie^I_\mathcal{K } \text{OT}^{(a)}_1)_\nu = \sum_{\substack{(m_1, m_2) \in \N_{\geq 0}^2 \\ m_1+m_2 \leq a+|I|+1 \\ m_1, m_2 \leq a+|I|}} (\nabla^{m_1} H_\Delta \nabla^{m_2} \fara + \nabla^{\min\{m_1, a-1\}}\fara)_\nu = (\text{OT}^{(a+|I|)}_1)_\nu.
\end{equation*}
Similarly, {\bf for $a \geq 1$,}
\begin{equation*}
\sum_{K+L =I, |K|\geq 1}(\lie_\mathcal{K}^K H^{\mu\nu\kappa\lambda})\nabla_\mu ( \faram_{,L,a})_{\kappa\lambda}= (\text{OT}^{(a+|I|)}_1)_\nu.
\end{equation*}
This implies that, {\bf when $a \geq 1$,}
\begin{equation}\label{eq:formij}
\boxed{
\begin{aligned}
\nabla^\gamma (\faramd_{,I,a})_{\gamma \nu} &= a (\nabla^\mu Y^\alpha) \nabla_\alpha(\faramd_{,I, a-1})_{\mu\nu} +  ( \text{OT}^{(a+|I|)}_1)_\nu,\\
H^{\mu\nu\kappa\lambda}\nabla_\mu ( \faram_{,I,a})_{\kappa\lambda} &= a( \nabla^\kappa Y^\alpha) \nabla_\alpha (\faram_{,I,a-1})\indices{_\kappa^\nu}  + ( \text{OT}^{(a+|I|)}_1)^\nu.
\end{aligned}}
\end{equation}
\end{remark}
\section{Canonical stress: positivity properties and divergence}\label{sec:canstress}

\subsection{Positivity properties of the canonical stress}
\begin{definition}
Let $\fara \in \Lambda^2 (S_e)$. Let $\gara \in \Lambda^2 (S_e)$ (we think of it as a variation).
Define the canonical stress
\begin{equation}\label{eq:qdef}
\dot Q \indices{^\mu _\nu}[\gara] := H^{\mu\zeta\kappa\lambda}[\fara] {\gara}_{\kappa\lambda} {\gara}_{\nu\zeta} - \frac 1 4 \delta_\nu^\mu {\gara}_{\zeta \eta} {\gara}^{\zeta \eta }.
\end{equation}
Here, the tensor field $H$ (depending on $\fara$) is as in Equation~(\ref{hform1}) and Equation~(\ref{hform2}).
\end{definition}
\begin{remark}
The canonical stress can also be written as
\begin{equation}\label{eq:qform}
\begin{aligned}
\dot Q_{\mu\nu}[\gara] &= \underbrace{\gara\indices{_\mu^\zeta}\gara_{\nu\zeta}- \frac 1 4 g_{\mu\nu} \gara_{\zeta \eta} \gara^{\zeta\eta} }_{\text{linear terms}}+ \frac 1 2 \ellmbi^{-2} \left\{- \fara\indices{_\mu^\zeta} \gara_{\nu\zeta} \fara^{\kappa\lambda} \gara_{\kappa\lambda} + \frac 1 4 g_{\mu\nu}(\fara^{\kappa\lambda}\gara_{\kappa\lambda})^2 \right\} \\
&+ \frac 1 2 \ellmbi^{-2} (1+\lun) \left\{- \farad\indices{_\mu^\zeta} \gara_{\nu\zeta} \farad^{\kappa\lambda} \gara_{\kappa\lambda} + \frac 1 4 g_{\mu\nu}(\farad^{\kappa\lambda}\gara_{\kappa\lambda})^2 \right\}  \\
&+ \frac 1 2 \ldu \ellmbi^{-2} \left\{\fara\indices{_\mu^\zeta} \gara_{\nu\zeta} \farad^{\kappa\lambda} \gara_{\kappa\lambda} - \frac 1 4 g_{\mu\nu}\fara^{\kappa\lambda}\gara_{\kappa\lambda}\farad^{\kappa\lambda}\gara_{\kappa\lambda} \right\}  \\
&+ \frac 1 2 \ldu \ellmbi^{-2} \left\{\farad\indices{_\mu^\zeta} \gara_{\nu\zeta} \fara^{\kappa\lambda} \gara_{\kappa\lambda} - \frac 1 4 g_{\mu\nu}\fara^{\kappa\lambda}\gara_{\kappa\lambda}\farad^{\kappa\lambda}\gara_{\kappa\lambda}\right\}.
\end{aligned}
\end{equation}
\end{remark}

\begin{remark}
The canonical stress has the property that its divergence is of lower order in the derivatives of $\gara$, as proved in Lemma~\ref{lem:qdive}.
\end{remark}

We recall the definition of the redshift vector field. Recall that $\chi_1$ is a smooth cut-off function which satisfies
\begin{equation*}
\chi_1(r) = \left\{
\begin{array}{ll}
1 & \text{if } r \in [3/2 M, 3M], \\
0 & \text{if } r \in (0,M] \cup [4M, \infty).
\end{array}
\right.
\end{equation*} 
Then
\begin{equation*}
\redsh :=2 \desude{}{t^*}+ \{(1-\mu)(1+\mu)+ 5\chi_1(r)(1-\mu)+\chi_1(r)\}\desude{}{t^*} + \{(1-\mu)^2 - 5\chi_1(r)(1- \mu)-\chi_1(r)\}\desude{}{r_1}.
\end{equation*}
With respect to the null vector fields $L$ and $\lbar$, we have
\begin{equation*}
V_{\text{red}} = L + \lbar + 5 \chi_1(r) \left( (1-\mu)L +\lbar\right)+\chi_1(r) (1-\mu)^{-1}\lbar.
\end{equation*}

We prove the following lemma.

\begin{lemma}[Coercivity of canonical stress at the horizon]\label{lem:coercivity}
There exists a constant $C >0$ such that there exists a small number $\varepsilon_{\text{red}} > 0$, such that the following holds. If $|\fara| < \varepsilon_{\text{red}}$ on the set $$r \in [15/8  M, 17/8 M],$$
then we have the following,
again on the set $r \in [15/8  M, 17/8 M]$:
\begin{equation}\label{eq:coercivity}
\dot Q\indices{_\mu_\nu}[\gara] \ {n}_{\widetilde \Sigma}^\mu {V_{\text{red}}}^\nu \geq C |\gara|^2, \qquad {}^{(\redsh)} \pi^{\mu\nu} \dot Q_{\mu\nu}[\gara] \geq C |\gara|^2.
\end{equation}
 Here, ${n}_{\widetilde \Sigma}$ is the future-directed unit Lorentz normal to the foliation $\widetilde{\Sigma}_{t^*}$, and $${}^{(\redsh)} \pi^{\mu\nu} = \frac 1 2 (\nabla^\mu \redsh^\nu + \nabla^\nu \redsh^\mu)$$ is the deformation tensor relative to $\redsh$.
\end{lemma}

\begin{proof}
By our assumptions, if $\varepsilon_{\text{red}}$ is small enough, we have the following:
\begin{align*}
\ellmbi^{-2} &\leq (1 - \redeps^2 - \redeps^4)^{-1} \leq 2.
\end{align*}
By the fact that both $\redsh$ and $n_{\widetilde \Sigma_{t^*}}$ are strictly timelike, we also have
\begin{equation*}
\dot Q^{(\text{MW})}_{\mu\nu}[\gara] \ {n}_{\widetilde \Sigma_{t^*}}^\mu {V_{\text{red}}}^\nu \geq C |\gara|^2,
\end{equation*}
where 
\begin{equation*}
\dot Q^{(\text{MW})}_{\mu\nu}[\gara] := \gara_{\mu\alpha} \gara\indices{_\nu^\alpha}  - \frac 1 4 g_{\mu\nu} \gara^{\alpha \beta}\gara_{\alpha \beta},
\end{equation*}
and $C$ does not depend on $\redeps$.
An example of nonlinear perturbation term on the right hand side is then
\begin{align*}
&\left|\ellmbi^{-2} \left(-\fara\indices{_\mu^\zeta} \gara_{\nu \zeta} \fara^{\kappa \lambda} \gara_{\kappa\lambda}  + \frac 1 4 g_{\mu\nu}(\gara_{\kappa \lambda} \fara^{\kappa\lambda})\right)\right| \\
&\leq 2 |\fara\indices{_\mu^\zeta} \gara_{\nu \zeta}| |\fara^{\kappa \lambda} \gara_{\kappa\lambda}|  + \frac 1 4 |g_{\mu\nu}| |\gara_{\kappa \lambda} \fara^{\kappa\lambda}|^2 \leq C \redeps^2 |\gara|^2.
\end{align*}
We notice that the other terms in $\dot Q$ are similar, and conclude by taking $\redeps$ small and absorbing the terms coming from the nonlinear part into the linear part. This proves the first claim in (\ref{eq:coercivity}).

Concerning the second claim in (\ref{eq:coercivity}), it follows from the calculations below. We restrict to the region $r \in (15/8 M, 17/8 M)$. Recall that, since $\desude{}{t^*}$ is a Killing vector field,
\begin{equation*}
	{}^{(\redsh)}\pi = {}^{(V_1)}\pi,
\end{equation*}
where $V_1 = -5\mu L + (1-\mu)^{-1}\lbar$. Now,
\begin{equation*}
\begin{aligned}
	&{}^{(\redsh)}\pi^{uu} = \frac{2M}{r^2}(1-\mu)^{-2}, \qquad {}^{(\redsh)}\pi^{uv} = {}^{(\redsh)}\pi^{vu} = \frac{5M}{r^2}(1-\mu)^{-1}(2\mu-1), \\ &{}^{(\redsh)}\pi^{vv} = 5\frac{2M}{r^2}, \qquad 
	{}^{(\redsh)}\pi^{Au} = {}^{(\redsh)}\pi^{Av}  = 0,\qquad 
	{}^{(\redsh)}\pi^{AB} = -2r^{-1}(1+5 \mu(1-\mu))\gbar^{AB}.
\end{aligned}
\end{equation*}
Here, the components of the tensors are with respect to the $(u, v, \theta, \varphi)$ coordinates, and capital latin indices indicate contraction with coordinate vector fields arising from a local parametrization $(\theta^1, \theta^2)$ of the conformal sphere $\mathbb{S}^2$. Capital latin indices are raised and lowered with the projected metric $\gbar$, which is the metric $g$ projected on sphere of constant $r$-coordinate. These calculations, combined with the form of $\dot Q^{(\text{MW})}_{\mu\nu}$, yield the second claim in \eqref{eq:coercivity}.
\end{proof}

\subsection{Divergence of the canonical stress}

We now prove a lemma on the divergence of the canonical stress.

\begin{lemma}\label{lem:qdive}
Let $\gara$ be a smooth two-form on $\mathcal{S}$ satisfying the system (equations of variation)
\begin{equation}\label{eqvar}
\begin{aligned}
\nabla^\gamma {\garad}_{\gamma \nu} = \mathcal{J}_\nu,\\
H^{\mu\nu\kappa\lambda}[\fara]  \ \nabla_\mu \gara_{\kappa\lambda} =\mathcal{I}^\nu.
\end{aligned}
\end{equation}
The canonical stress is defined as
\begin{equation}
 \dot{Q} \indices{^\mu _\nu}[\gara] := H^{\mu\zeta\kappa\lambda}[\fara] {\gara}_{\kappa\lambda} {\gara}_{\nu\zeta}  - \frac 1 4 \delta_\nu^\mu  H^{\eta\zeta\kappa\lambda}[\fara] {\gara}_{\eta\zeta}{\gara_{\kappa\lambda}}. 
\end{equation}
Here, $H[\fara]$ is as in Equation~(\ref{hform1}) and Equation~(\ref{hform2}).
Under these assumptions, $\dot Q$ satisfies the following relation:
\begin{equation}\label{eq:stressdiv}
\nabla_\mu   \dot{Q} \indices{^\mu _\nu}[\gara]= 
(\nabla_\mu H^{\mu\zeta\kappa\lambda}[\fara]) {\gara}_{\kappa\lambda} {\gara}_{\nu\zeta}- \frac 1 4 (\nabla_\nu H^{\zeta \eta \kappa\lambda}[\fara] \ \gara_{\zeta \eta} \gara_{\kappa\lambda})+
\mathcal{I}^\zeta \gara_{\nu \zeta}
- \frac 1 2 H^{\zeta \eta \kappa\lambda}[\fara] \ \gara_{\kappa\lambda}( \mathcal{J}^\alpha \varepsilon_{\alpha \nu \zeta \eta}).
\end{equation}
\end{lemma}

\begin{proof}[Proof of Lemma]
In the proof of this lemma, we suppress the dependence of $H$ on $\fara$. We notice that the first equation of display~\eqref{eqvar} implies the following equation:
\begin{equation*}
\nabla_\eta \gara_{\zeta \beta} + \nabla_\beta \gara_{\eta \zeta}+\nabla_\zeta \gara_{\beta \eta} =  \mathcal{J}^\alpha \varepsilon_{\alpha\eta\zeta\beta}.
\end{equation*}
Then,
\begin{equation*}
\begin{aligned}
\nabla_\mu   \dot{Q} \indices{^\mu _\nu} [\gara]\\ & =
(\nabla_\mu H^{\mu\zeta\kappa\lambda}) {\gara}_{\kappa\lambda} {\gara}_{\nu\zeta} + H^{\mu\zeta\kappa\lambda}\nabla_\mu {\gara_{\kappa\lambda}}{\gara}_{\nu\zeta} + H^{\mu \zeta \kappa\lambda} {\gara}_{\kappa\lambda} \nabla_\mu \gara_{\nu\zeta} \\ &
- \frac 1 4 (\nabla_\nu H^{\zeta \eta \kappa\lambda} \gara_{\zeta \eta} \gara_{\kappa\lambda})- \frac 1 2 \delta^\mu_\nu H^{\zeta \eta \kappa \lambda} \nabla_\mu \gara_{\zeta \eta} {\gara}_{\kappa\lambda}\\ & =
(\nabla_\mu H^{\mu\zeta\kappa\lambda}) {\gara}_{\kappa\lambda} {\gara}_{\nu\zeta}- \frac 1 4 (\nabla_\nu H^{\zeta \eta \kappa\lambda} \gara_{\zeta \eta} \gara_{\kappa\lambda})+
\mathcal{I}^\zeta \gara_{\nu \zeta}+ H^{\mu\zeta\kappa\lambda}\gara_{\kappa\lambda}\nabla_\mu \gara_{\nu \zeta}\\ &
- \frac 1 2 H^{\zeta \eta \kappa\lambda} \gara_{\kappa\lambda}(-\nabla_\eta \gara_{\nu \zeta}- \nabla_\zeta \gara_{\eta \nu} + \mathcal{J}^\alpha \varepsilon_{\alpha \nu \zeta \eta}) \\ & =
(\nabla_\mu H^{\mu\zeta\kappa\lambda}) {\gara}_{\kappa\lambda} {\gara}_{\nu\zeta}- \frac 1 4 (\nabla_\nu H^{\zeta \eta \kappa\lambda} \gara_{\zeta \eta} \gara_{\kappa\lambda})+
\mathcal{I}^\zeta \gara_{\nu \zeta}\\ &
- \frac 1 2 H^{\zeta \eta \kappa\lambda} \gara_{\kappa\lambda}( \mathcal{J}^\alpha \varepsilon_{\alpha \nu \zeta \eta}) 
\end{aligned}
\end{equation*}
All in all,
\begin{equation}
\begin{aligned}
\nabla_\mu  \dot{Q} \indices{^\mu _\nu}[\gara]= 
(\nabla_\mu H^{\mu\zeta\kappa\lambda}) {\gara}_{\kappa\lambda} {\gara}_{\nu\zeta}- \frac 1 4 (\nabla_\nu H^{\zeta \eta \kappa\lambda} \gara_{\zeta \eta} \gara_{\kappa\lambda})+
\mathcal{I}^\zeta \gara_{\nu \zeta}
- \frac 1 2 H^{\zeta \eta \kappa\lambda} \gara_{\kappa\lambda}( \mathcal{J}^\alpha \varepsilon_{\alpha \nu \zeta \eta}).
\end{aligned}
\end{equation}
\end{proof}
\begin{remark}
Using the definition of $H_\Delta$ we have that Equation~\eqref{eq:stressdiv} can be rewritten as
\begin{equation}\label{eq:sdivhdelta}
\nabla_\mu   \dot{Q}\indices{^\mu _\nu}[\gara]= 
(\nabla_\mu H_\Delta^{\mu\zeta\kappa\lambda}) {\gara}_{\kappa\lambda} {\gara}_{\nu\zeta}- \frac 1 4 (\nabla_\nu H_\Delta^{\zeta \eta \kappa\lambda} \gara_{\zeta \eta} \gara_{\kappa\lambda})+
\mathcal{I}^\zeta \gara_{\nu \zeta}
- \frac 1 2 H_\Delta^{\zeta \eta \kappa\lambda} \gara_{\kappa\lambda}( \mathcal{J}^\alpha \varepsilon_{\alpha \nu \zeta \eta}) -\mathcal{J}^\alpha \garad_{\alpha \nu}.
\end{equation}
\end{remark}

\section{Deducing \texorpdfstring{$L^2$}{L2} bounds from the \texorpdfstring{$L^\infty$}{pointwise} bootstrap assumptions}\label{sec:l2fromlinf}

The goal of this section is to prove the following proposition.
\begin{proposition}\label{prop:energy}
 Let $N \in \N_{\geq 0}$, $N \geq 5$. There exist a small constant $\varepsilon_0 = \varepsilon_0(N)$, a radius $r_\text{in} \in (M, 2M)$ and a constant $C(N) > 0$ such that, for all $\varepsilon \leq \varepsilon_0$ and $\fara$ solutions of the MBI system \eqref{MBI} on $\tregio{t_0^*}{t_1^*}$, $t^*_1 \geq t_0^*$, with initial Sobolev norm $\norm{\fara}_{H^{2N}(\widetilde{\Sigma}_{t^*_0})} \leq \varepsilon$ (see Definition~\ref{def:sobolevnorm}), satisfying the bootstrap assumptions $BA(\tregio{t_0^*}{t_1^*}, 1, N, \varepsilon)$ for the components of $\fara$, there holds
\begin{equation}\label{eq:unifl2}
\norm{\fara}^2_{H^{2N}(\widetilde{\Sigma}_{t^*})} \leq C \varepsilon^2,
\end{equation}
where $t^* \in [t_0^*, t_1^*]$.
\end{proposition}

\begin{remark}
	For the remainder of the section, we will consider the integer $N$ as fixed.
\end{remark}

\begin{remark}
Notice moreover that we assume boundedness on initial data for up to $2N$ derivatives in the $\p_t$-direction for initial data (an alternative notation for the vector field $\p_t$ is $T$, see Section~\ref{sec:not:geo}). These derivatives can be ``converted'', using the MBI system, into derivatives intrinsic to the hypersurface $\widetilde{\Sigma}_{t_0^*}$.
\end{remark}

\begin{remark}\label{rmk:fnormnull}
Let us note that the norm on the LHS of estimate~\eqref{eq:unifl2} gives control over the following expression, in terms of the null components $\rho, \sigma, \alpha, \alphabar$. There exists a constant $C >0$ such that the following holds:
\begin{equation}\label{eq:fnormnull}
	\norm{\fara}^2_{H^{2N}(\widetilde{\Sigma}_{t^*})} \geq C \int_{\widetilde{\Sigma}_{t^*}} \big( |\partial^{\leq 2N} \rho|^2 + |\partial^{\leq 2N} \sigma|^2  + |\snabla^{\leq 2N} \alpha|^2 +|\snabla^{\leq 2N} \alphabar|^2\big) \de \widetilde \Sigma_{t^*}
\end{equation}
\end{remark}

Let us notice the following important Corollary of Proposition~\ref{prop:energy}, which follows from the usual Sobolev embedding in three space dimensions:

\begin{corollary}[Unweighted Sobolev embedding]\label{cor:unweightedsob}
Let $N \in \N_{\geq 0}$, $N \geq 5$. There exist a small constant $\varepsilon_0 = \varepsilon_0(N)$, a radius $r_\text{in} \in (M, 2M)$ and $C = C(N) > 0$ such that, for all $\varepsilon \leq \varepsilon_0$ and $\fara$ solutions of the MBI system~\eqref{MBI} on $\tregio{t_0^*}{t_1^*}$, $t^*_1 \geq t_0^*$, with initial Sobolev norm $\norm{\fara}_{H^{2N}(\widetilde{\Sigma}_{t^*_0})} \leq \varepsilon$ (see Definition~\ref{def:sobolevnorm}), satisfying the bootstrap assumptions $BA(\tregio{t_0^*}{t_1^*}, 1, N, \varepsilon)$ for the components of $\fara$, there holds
\begin{equation}\label{eq:unweghtsob}
\sum_{h=0}^{2N-2}|\partial^h \fara (t^*,r_1, \theta, \phi)| \leq C \varepsilon,
\end{equation}
where $t^* \in [t_0^*, t_1^*]$, $r_1 \geq r_{\text{in}}$.
\end{corollary}

Let us now turn to the proof of Proposition~\ref{prop:energy}. The scheme of the argument is as follows: we will show degenerate estimates at the event horizon (i.e.~, estimates which display degenerate control over derivatives in the $\lbar$ direction close to the event horizon). After that, we will improve those estimates to include non-degenerate control on the $\lbar$ derivatives close to $\mathcal{H}^+$. Subsequently, we will extend those estimates to include the case of mixed derivatives. Finally, we will integrate the resulting differential inequality by means of an application of Gr\"onwall's inequality. We divide this reasoning in several lemmas.

\subsection{Degenerate estimates at the horizon}\label{sub:degen}
In this section, we prove energy estimates to control degenerate (at $\mathcal{H}^+$) energy fluxes through the foliation $\Sigma_{t^*}$. It is not possible to apply the Gr\"onwall inequality directly on these estimates, as the right hand side of inequality \eqref{eq:degenergyfinal} contains all derivatives up to the horizon, and the left hand side only contains derivatives in the direction of Killing vector fields.

\emph{Throughout Subsection~\ref{sub:degen}, denote $\dot \fara :=\dot \fara_{,I, 0}$}, with $I \in \iindicess^j$. We let $t_2^* \geq t_1^* \geq t_0^*$.
Recall the definitions
\begin{equation*}
\Sigma_{t_1^*} := \{ t^* = t^*_1\} \subset S_e,  \qquad \regio{t_1^*}{t_2^*} := \cup_{s \in [t_1^*, t_2^*]} \Sigma_s.
\end{equation*}
Recall that $n_{\Sigma_{t^*}}$ is defined as the future unit normal to the foliation $\Sigma_{t^*}$.

\begin{lemma}\label{lem:degen} There exists a constant $C$ and a small constant $\varepsilon_{\text{deg}} > 0$ such that for every solution of the MBI system (\ref{MBI}) on $\tregio{t_0^*}{t_1^*}$ satisfying the bootstrap assumption $BA(\tregio{t_0^*}{t_1^*}, 1, N, \varepsilon)$ with $0 < \varepsilon < \varepsilon_{\text{deg}}$, and every $j \in \N_{\geq 0}$ such that $j \leq 2N$, the following holds. Denote $\dot \fara :=\dot \fara_{,I, 0}$, with $I \in \iindicess^j$ (see Section~\ref{sec:not:der} for the definition). 
Also, recall from Section~\ref{sec:not:stresses}
\begin{equation*}
Q^{(\text{MW})}_{\mu\nu}[\dot \fara] := \dot {\fara}_{\mu\alpha} \dot {\fara}\indices{_\nu^\alpha} - \frac 1 4 g_{\mu\nu} \dot{\fara}^{\alpha \beta} \dot{\fara}_{\alpha \beta}.
\end{equation*}
We then have:
\begin{equation} \label{eq:degenergyfinal}
\begin{aligned}
&\int_{\Sigma_{t^*_2}} \left( Q^{(MW)}_{\mu\nu}[\dot \fara] T^\nu n_{\Sigma_{t^*}}^\mu \right) \de{\Sigma_{t^*}} -
\int_{\Sigma_{t^*_1}} \left(Q^{(MW)}_{\mu\nu}[\dot \fara] T^\nu n_{\Sigma_{t^*}}^\mu \right) \de{\Sigma_{t^*}} \\ & \leq
 C \varepsilon^{3/2} \int_{\regio{t^*_1}{t^*_2}} (t^*)^{-2}  | {\dot \fara} | |\partial^{\leq j} \fara| \de \text{Vol} + 
 C \int_{\Sigma_{t^*_2}}|\fara|^2 |\dot \fara|^2 \de{\Sigma_{t^*}} +
C \int_{\Sigma_{t^*_1}}|\fara|^2 |\dot \fara|^2 \de{\Sigma_{t^*}}.
\end{aligned}
\end{equation}
\end{lemma}

\begin{remark}
	The proof will be carried out by using the energy estimates arising from $\dot Q[\dot \fara_{,I,0}]$, i.e. the canonical stress introduced in Section~\ref{sec:not:stresses}, where we set $\gara := \dot \fara_{,I,0}$. The difference between $\dot Q [\dot \fara_{,I,0}]$ and $Q^{(\text{MW})}[ \dot \fara_{,I,0}]$ may be slightly misleading, and we underline that these two objects are different.
\end{remark}

\begin{remark}
It is informative to see what the energy density on the LHS of Equation~\eqref{eq:degenergyfinal} looks like when expressed in null components. A straightforward calculation yields:
\begin{equation}
 Q^{(MW)}_{\mu\nu}[\dot \fara] T^\nu n_{\Sigma_{t^*}}^\mu  = \frac 1 {4\sqrt{1+\mu}} \Big(|\dot \alpha|^2 + \frac{1+\mu}{1-\mu}|\dot \alphabar|^2 + (1+\mu) (\dot \rho^2 + \dot \sigma^2) \Big).
\end{equation}
Dotted null components here denote the null components corresponding to $\dot \fara$. We see the degeneracy on the term in $\alphabar$, as for it to be non-degenerate we would need a factor of $(1-\mu)^{-2}$ in front of such term (recall that $\lbar$ vanishes linearly (like $1-\mu$) at $\mathcal{H}^+$). \end{remark}

\begin{proof}[Proof of Lemma~\ref{lem:degen}]
Throghout the proof, we denote the canonical stress as follows:
\begin{equation*}
	\dot{Q} \indices{^\mu _\nu}[\dot \fara_{,I,0}] := H^{\mu\zeta\kappa\lambda} {(\dot \fara_{,I,0})}_{\kappa\lambda} ({\dot \fara_{,I,0})}_{\nu\zeta} - \frac 1 4 \delta_\nu^\mu  H^{\mu\zeta\kappa\lambda} {{( \dot\fara_{,I,0})}_{\kappa\lambda}}{( \dot\fara_{,I,0})}_{\mu\zeta}.
\end{equation*}
For conciseness, we denote, throughout this proof,
\begin{equation*}
	\dot{Q} := \dot{Q} \indices{^\mu _\nu}[\dot \fara_{,I,0}].
\end{equation*}
Here, $H = H[\fara]$ is the tensor field defined in~\eqref{hform2}.  Furthermore,
recall the definition  $T := \partial_t$ from Section~\ref{sec:not:geo}. Let also
 $(b^{-1})^{\mu\nu}$ be the inverse MBI metric (also called Boillat metric) defined in Appendix~\ref{sec:bmetric}. The form of the metric $b^{-1}$ is the following:
$$
 (b^{-1})^{\mu\nu} := g^{\mu\nu} - (1+\lun)^{-1} \fara^{\mu\kappa}\fara\indices{^\nu_\kappa}.
$$
We carry out energy estimates with the vector field 
$$X^\nu :=(b^{-1})^{\nu\mu} g_{\mu\alpha} T^\alpha.$$
We denote 
$$J_1^\mu := \dot Q\indices{^\mu_\nu} X^\nu,$$
and we integrate its divergence over the region depicted in Figure~\ref{fig1}. The region is bounded above by $\Sigma_{t_2^*}$, below by $\Sigma_{t_1^*}$, and on the left by the surface $\mathcal{T}_{\text{max}}:=\{t = t_{\text{max}}\} \cap \{t_1^* \leq t^* \leq t_2^*\}$, for some $t_{\text{max}} > 0$.

We denote
\begin{equation*}
\begin{aligned}
\overline{\Sigma}_{t^*} &:= \Sigma_{t^*} \cap \{t \leq t_{\text{max}}\},\qquad 
\oregio{t^*_1}{t_2^*} :&= \regio{t^*_1}{t^*_2} \cap \{t \leq t_{\text{max}}\}.
\end{aligned}
\end{equation*}
We let $\text{d} \mathcal{T}_{\text{max}}$ the induced volume form on $\mathcal{T}_{\text{max}}$.
\begin{figure}\label{figure1}
\centering
\begin{tikzpicture}	

\node (I)    at ( 0,0) {};

\path
(I) +(90:4)  coordinate[label=90:$i^+$]  (top)
+(-90:4) coordinate[label=-90:$i^-$] (bot)
+(0:4)   coordinate                  (right)
+(180:4) coordinate[label=180:$i^0$] (left)
;
\path (top) +(-45:0.5) coordinate (tar);

\path 
(top) + (180:4) coordinate[label=90:$i^+$]  (acca)
+ (-45: 2) coordinate (nulluno)
+ (-45: 3) coordinate (nulldue)
;

\path 
(left) + (-45: 2) coordinate (correspuno)
+ (-45: 3) coordinate (correspdue)
;

\path (top) + (-135: 2 ) coordinate (horiuno)
+ (-135: 4) coordinate (horidue);
\draw[name path = stre, opacity = 0] (left) to (tar);
\draw [name path = ciao, thick, opacity = 0] (horiuno) to [bend left = 15] node[midway, above]{$\Sigma_{t^*_1}$} (right);
\draw [name path = bye, thick, opacity = 0] (horidue) to [bend left = 10] node[midway, above]{$\Sigma_{t^*_2}$} (right);
\path [name intersections={of=ciao and stre,by=E}];
\path [name intersections={of=bye and stre,by=D}];
\draw [thick] (D) to [bend left = 15] node[midway, above]{$\Sigma_{t^*_1}$} (right);
\draw [thick] (E) to [bend left = 10] node[midway, above]{$\Sigma_{t^*_2}$} (right);
%\coordinate (sopra) at (intersection of (ciao) and (stre));
\draw [thick] (E) to [bend right = 7] node[midway, right]{$t = t_{\text{max}}$} (D);
\draw (left) -- 
node[midway, above left, sloped]    {}
(top) --
node[midway, above, sloped] {$\mathcal{I}^+$}
(right) -- 
node[midway, above, sloped] {$\mathcal{I}^-$}
(bot) --
node[midway, above, sloped]    {$\mathcal{H}^-$}    
(left) -- cycle;

\draw[decorate,decoration=zigzag] (top) -- (acca)
node[midway, above, inner sep=2mm] {$r=0$};

\end{tikzpicture}
\caption{Penrose diagram of the integration region.}\label{fig1}
\end{figure}
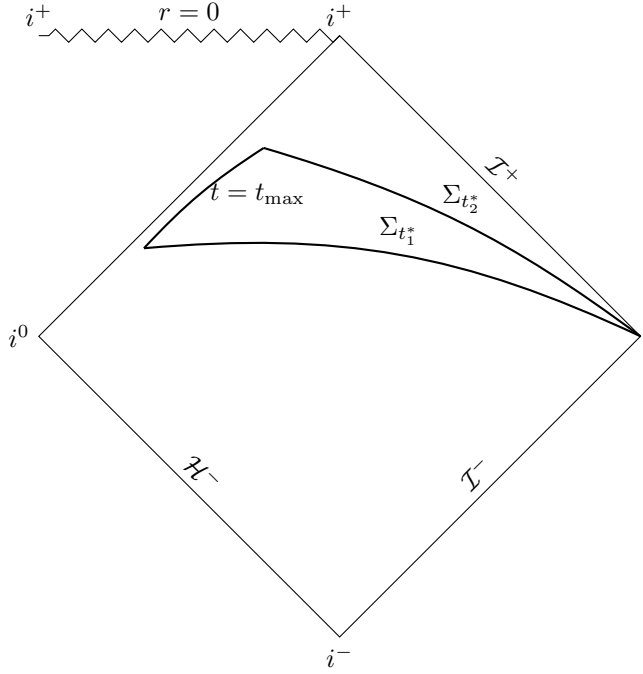
We obtain 
\begin{equation}
\begin{aligned}\label{energyt}
&\int_{\overline{\Sigma}_{{t_2^*}}} \left(\dot Q_{\mu\nu} X^\nu n_{\Sigma_{t^*}}^\mu \right) \de{\Sigma_{t^*}} -
\int_{\overline{\Sigma}_{{t_1^*}}} \left(\dot Q_{\mu\nu} X^\nu n_{\Sigma_{t^*}}^\mu \right) \de{\Sigma_{t^*}} +
\int_{\mathcal{T}_{\text{max}}} \left(\dot Q_{\mu\nu} X^\nu T^\mu \right) \de \mathcal{T}_{\text{max}}
\\
&\leq C  \int_{\oregio{{t_1^*}}{{t_2^*}}}\left(|\nabla_\mu (\dot Q\indices{^\mu_\nu})X^\nu|+ |\dot Q\indices{^\mu_\nu} \nabla_\mu X^\nu|\right) d \text{Vol}.
\end{aligned}
\end{equation}
Now, Lemma~\ref{lem:qpos} gives pointwise positivity of the third term in the LHS of Equation~\eqref{energyt}, so we can just ignore it.
We then proceed to estimate the first term in the RHS of Equation~\eqref{energyt}. We have, from Lemma~\ref{lem:qdive}, that
\begin{equation}\label{qdoteq}
\begin{aligned}
&\nabla_\mu (\dot Q\indices{^\mu_\nu}) =  {\dot \fara}_{\nu \eta} \mathcal{I}^\eta+ ( \nabla_\mu H_{\Delta}^{\mu\zeta\kappa\lambda}) {\dot \fara}_{\kappa\lambda} {\dot \fara}_{\nu\zeta}- \frac 1 4 ( \nabla_\nu H_{\Delta}^{\zeta\eta\kappa\lambda}) {\dot \fara}_{\zeta \eta} {\dot \fara}_{\kappa \lambda} =\\ &
 (\text{one}) + (\text{two}) + (\text{three}),
 \end{aligned}
\end{equation}
where
\begin{equation*}
 \mathcal{I}^\eta := (\mathcal{I}^{,I,0})^\eta= H_\Delta^{\mu\eta\kappa\lambda} \nabla_\mu (\dot \fara_{\kappa\lambda})  -\mathcal{L}^I_\mathcal{K}(H_\Delta^{\mu\eta\kappa\lambda} \nabla_\mu (\fara_{\kappa\lambda})).
\end{equation*}
Let's now write the term corresponding to $(\text{one})$ in the RHS of (\ref{energyt}) as 
\begin{align*}
	&\int_{\oregio{{t_1^*}}{{t_2^*}}} |X^\nu {\dot \fara}_{\nu \eta} \mathcal{I}^\eta| \de \text{Vol}\\ & \leq
	\int_{\oregio{{t_1^*}}{{t_2^*}}} \left|X^\nu {\dot \fara}_{\nu \eta} \left(\sum_{\substack{I \in \iindicess^{\leq j}, J \in \iindicess^{\leq j-1}	\\ |I|+|J| = j}}(\lie^{I}_\mathcal{K} H_\Delta^{\mu \eta \kappa \lambda})(\lie^J_\mathcal{K}\nabla_\mu \fara_{\kappa\lambda})\right)\right|  \de \text{Vol}\\ & \leq
	\int_{\oregio{{t_1^*}}{{t_2^*}}} \left|{\dot \fara}_{\nu \eta}\right| \left(\sum_{\substack{I \in \iindicess^{\leq j}, J \in \iindicess^{\leq j-1}	\\ |I|+|J| = j}}\left|\lie_\mathcal{K}^{I} H_\Delta^{\mu \eta \kappa \lambda}\right|\cdot \left|\lie_\mathcal{K}^{J}\nabla_\mu \fara_{\kappa\lambda}\right|\right) \de \text{Vol}\\ &
	\leq
	\int_{\oregio{{t_1^*}}{{t_2^*}}} \left|{\dot \fara}_{\nu \eta}\right| \left(\sum_{\substack{H, K \in \iindicess^{\leq j}, J \in \iindicess^{\leq j-1}	\\ |H|+|J| +|K|\leq j}}\left|\lie_\mathcal{K}^{H} \fara_{\alpha \beta} \right|\cdot \left|\lie_\mathcal{K}^{K} \fara_{\alpha \beta} \right|\cdot \left|\lie_\mathcal{K}^{J}\nabla_\mu \fara_{\kappa\lambda}\right|\right) \de \text{Vol}\\ &
	\leq C  \int_{\oregio{{t_1^*}}{{t_2^*}}} \varepsilon^{3/2} (t^*)^{-2} | {\dot \fara} | |\partial^{\leq j} \fara|  \de \text{Vol}.
\end{align*}
The second inequality is true because of the bootstrap assumptions and the form of $X^\nu$. 
The third inequality is true because of the form of $H_\Delta$ (quadratic in $\fara$).
The fourth inequality is true because at least two of $|I|, |J|, |K|+1$ are less or equal than $N$, and we can consequently apply the bootstrap assumptions, together with Remark~\ref{decaytrmk}. Finally, we used the equivalence between norms given by Lie derivatives and norms given by covariant derivatives (Proposition~\ref{norm.equiv}). Recall that the notation with $\partial$ contains weights in $r$. The terms $(\text{two})$ and $(\text{three})$ are estimated similarly.

We now examine the second term in the RHS of Equation~\eqref{energyt}. 
First, the form of $b^{-1}$ implies, together with the bootstrap assumptions, if $N \geq 1$,
\begin{equation*}
|\dot Q\indices{^\mu_\nu} \nabla_\mu X^\nu| \leq |\dot Q\indices{^\mu_\nu} \nabla_\mu T^\nu| + C |\fara| |\nabla \fara||\dot Q_{\mu\nu}| \leq |\dot Q\indices{^\mu_\nu} \nabla_\mu T^\nu| + C \varepsilon^{3/2} (t^*_1)^{-2} |\dot \fara|^2.
\end{equation*}

We now use the fact that $T$ is a Killing field. Namely,
\begin{equation*}
\dot Q \indices{^\mu_\nu} \nabla_\mu T^\nu =\qnl_{\mu\nu} \nabla^\mu T^\nu,
\end{equation*}
where, 
\begin{equation*}
\qnl_{\alpha\beta} := \frac 1 2 (\dot Q_{\alpha\beta} - \dot Q_{\beta \alpha}).
\end{equation*}
Clearly,
\begin{equation*}
|\qnl| \leq C  \varepsilon^{3/2} (t^*)^{-2} |\dot \fara|^2.
\end{equation*}
Estimating similarly the other terms, we obtain the inequality

\begin{equation} 
\begin{aligned}
\int_{\overline{\Sigma}_{t^*_2}} \left(\dot Q_{\mu\nu} X^\nu n_{\Sigma_{t^*}}^\mu \right) \de{\Sigma_{t^*}} -
\int_{\overline{\Sigma}_{t^*_1}} \left(\dot Q_{\mu\nu} X^\nu n_{\Sigma_{t^*}}^\mu \right) \de{\Sigma_{t^*}}
\leq 	C \varepsilon^{3/2} \int_{\oregio{t^*_1}{t^*_2}} (t_1^*)^{-2}  | {\dot \fara} | |\partial^{j} \fara| \de \text{Vol}.
\end{aligned}
\end{equation}
From the form of $b^{-1}$ and the form of $\dot Q$, and the bootstrap assumptions, we now obtain
\begin{equation*}
\begin{aligned}
&\int_{\overline{\Sigma}_{t^*_2}} \left( Q^{(MW)}_{\mu\nu}[\dot \fara] T^\nu n_{\Sigma_{t^*}}^\mu \right) \de{\Sigma_{t^*}} -
\int_{\overline{\Sigma}_{t^*_1}} \left( Q^{(MW)}_{\mu\nu}[\dot \fara] T^\nu n_{\Sigma_{t^*}}^\mu \right) \de{\Sigma_{t^*}} \\ 
& \leq
 C \varepsilon^{3/2} \int_{\oregio{t^*_1}{t^*_2}} (t^*)^{-2}  | {\dot \fara} | |\partial^{\leq j} \fara| \de \text{Vol} + 
C\int_{\overline{\Sigma}_{t^*_2}}|\fara|^2 |\dot \fara|^2 \de{\Sigma_{t^*}} +
C\int_{\overline{\Sigma}_{t^*_1}}|\fara|^2 |\dot \fara|^2 \de{\Sigma_{t^*}}.
\end{aligned}
\end{equation*}
Taking now $t_{\text{max}} \to \infty$, by the monotone convergence theorem, we obtain the desired estimate.

\end{proof}
\begin{remark} Note that the right hand side of~\eqref{eq:degenergyfinal}, in particular the term $|\partial^{\leq j} \fara|$, involves all derivatives of $\fara$. Hence, it is not possible to use a Gr\"onwall inequality to conclude boundedness. A further difficulty arises as we have degeneracy on the left hand side for the component $\alphabar$ at the horizon. In the next section, we proceed to commute with $Y$ at the horizon and to carry out redshift estimates in order to remove such degeneracy.
\end{remark}

\begin{remark}
The proof of previous Lemma~\ref{lem:degen} also easily yields the following inequality. Start by choosing $t_{\max}$ such that the $r$-coordinate of the points of intersection of the hypersurface $t = t_{\max}$ and the hypersurface $\Sigma_{t_2^*}$ is equal to $4M - r_{\text{in}}$, with $r_{\text{in}} \in (0,2M)$ sufficiently close to $2M$. By a similar reasoning as in the above proof, we obtain the inequality
\begin{equation}\label{eq:degenergyfinal2}
\begin{aligned}
&\int_{\Sigma_{t^*_2}\cap \{r \geq 4M - r_{\text{in}}\}} \left( Q^{(MW)}_{\mu\nu}[\dot \fara] T^\nu n_{\Sigma_{t^*}}^\mu \right) \de{\Sigma_{t^*}} \\
&\leq C \int_{\Sigma_{t^*_1}} \left(Q^{(MW)}_{\mu\nu}[\dot \fara] n_{\Sigma_{t^*}}^\nu n_{\Sigma_{t^*}}^\mu \right) \de{\Sigma_{t^*}} +
C \varepsilon^{3/2} \int_{\regio{t^*_1}{t^*_2}} (t^*)^{-2}  | {\dot \fara} | |\partial^{\leq j} \fara| \de \text{Vol}.
\end{aligned}
\end{equation}
Note that here the last two terms in display~\eqref{eq:degenergyfinal} have been incorporated resp. in the term on the LHS and in the first term in the RHS of the previous display. This is a consequence of the fact that the integral on the LHS is restricted to the region $r \geq 4M - r_{\text{in}}$, where $\p_t$ is uniformly timelike, and of the fact that $n_{\Sigma_{t^*}}$ is uniformly timelike everywhere.

Setting $t^*_1 := s_0$, and integrating Equation (\ref{eq:degenergyfinal2}) between $s_1$ and $s_2$, possibly taking $\varepsilon$ sufficiently small, and $r_\text{in}$ close to $2M$, we obtain
\begin{equation} \label{eq:intdeg}
\begin{aligned}
&\int_{s_1}^{s_2} \int_{\Sigma_{s}\cap \{r \geq 4M - r_{\text{in}}\}} \left(Q^{(MW)}_{\mu\nu}[\dot \fara] T^\nu n_{\Sigma_{t^*}}^\mu \right) \de{\Sigma_{t^*}} \de s \\
&\leq  C (s_2-s_1)\int_{\Sigma_{s_0}} \left(Q^{(MW)}_{\mu\nu}[\dot \fara] n_{\Sigma_{t^*}}^\nu n_{\Sigma_{t^*}}^\mu \right) \de{\Sigma_{t^*}} +
 C \varepsilon^{3/2} \int_{s_1}^{s_2} \int_{\regio{s_0}{s}} (t^*)^{-2}  | {\dot \fara} | |\partial^{\leq j} \fara| \de \text{Vol} \de s.
\end{aligned}
\end{equation}
\end{remark}

We now turn to proving nondegenerate estimates near the horizon $\mathcal{H}^+$. 

\subsection{Removing the degeneracy at the event horizon}\label{sub:redshcomm}
In this section, we prove estimates which control higher order derivatives in the $Y$-direction at the horizon, where $Y$, sufficiently close to $r =2M$, is defined as the transversal null derivative $(1-\mu)^{-1}\lbar$. We first prove that the bulk of the corresponding energy estimates arising from the redshift vector field $\redsh$ is positive. This is the content of Lemma~\ref{lem:pointpos}.

\begin{lemma}[Positivity of the commuted energy near $\mathcal{H}^+$]\label{lem:pointpos}
Let $I \in \iindicess^b$, $a \in \N_{\geq 0}$, such that $a + b \leq 2N$. There exist $r_{\text{in}} \in (\frac{15}{8} M, 2 M)$ and $\varepsilon_{\text{red}}> 0$  as well as collections of positive constants $C_{I,a}^{(1)},C_{I,a}^{(2)},C_{I,a}^{(3)}$ such that, if $\fara$ is a solution to the MBI system~\eqref{MBI} on $\tregio{t_1^*}{t_2^*}$ satisfying the bootstrap assumptions $BA(\tregio{t_1^*}{t_2^*}, 1, N, \varepsilon)$, with $0 < \varepsilon < \varepsilon_{\text{red}}$, there holds
\begin{align}
\label{eq:currentposz}
&\nabla^\mu({\dot Q}_{\mu\nu}[\dot \fara_{,I,0}]\redsh^\nu) \geq C^{(1)}_{I,0} |\dot \fara_{,I,0}|^2 -  C^{(2)}_{I,0} \varepsilon^{3/2} |\partial^{|I|} \fara|^2 - |\otm_2^{(I,0)}|,\\
\label{eq:currentpos}
&\nabla^\mu({\dot Q}_{\mu\nu}[\dot \fara_{,I,a}]\redsh^\nu) \\ \nonumber
& \quad \geq C^{(1)}_{I,a} |\dot \fara_{,I,a}|^2 -  C^{(2)}_{I,a} \varepsilon^{3/2} |\partial^{a+|I|} \fara|^2 
-C^{(3)}_{I,a} \left(\sum_{i = 1}^3 |\faram_{, \widetilde I (\Omega_i),a-1}|^2\right) - |\otm_2^{(I,a)}|,  \text{ if } a > 0.
\end{align}
in the region $\tregio{t_1^*}{t_2^*} \cap \{r_{\text{in}} \leq r \leq 4M - r_{\text{in}}\}$.

Here, if $K \in \mathcal K$, and if $I = (K_1, \ldots, K_b)$, the multi-index $\widetilde{I}(K)$ is given by $$\widetilde{I}(K) = (K, K_1, \ldots, K_b).$$
Moreover, we have the following schematic equation (in the sense of Definition~\ref{def:short}) satisfied by $\otm_2^{(I,0)}$:
\begin{equation}\label{eq:otm21}
\otm_2^{(I,0)} = (\nabla \hdelta )\, \faram_{,I,0}\faram_{,I,0} + \sum_{\substack{ |H| + |K| \leq |I| \\ |K|< |I| }} (\lie^H_{\mathcal{K}} \hdelta) \nabla (\lie^K_{\mathcal{K} }\fara) \, \faram_{,I,0},
\end{equation}
as well as the estimates:
\begin{equation}\label{eq:otm2est1}
|\otm_2^{(I,0)}| \leq C |\nabla \hdelta| \, |\faram_{,I,0}|^2 +  \sum_{\substack{ |H| + |K| \leq |I| \\ |K|< |I| }} |\lie^H_{\mathcal{K}} \hdelta| \, |\nabla \lie^K_{\mathcal{K} }\fara| \, |\faram_{,I,0}|.
\end{equation}
Also, $\otm_2^{(I,a)}$ satisfies the schematic equations (in the sense of Definition~\ref{def:short}), when $a \geq 1$:
\begin{equation}\label{eq:otm22}
\otm_2^{(I,a)} = \otm^{(a+|I|)}_1\faram_{,I,a},
\end{equation}
as well as the estimates:
\begin{equation}\label{eq:otm2est2}
|\otm_2^{(I,a)}| \leq C \sum_{\substack{(m_1, m_2) \in \N_{\geq 0}^2 \\ m_1+m_2 \leq a+|I|+1 \\ m_1, m_2 \leq a + |I|}} (|\nabla^{m_1} H_\Delta^{\mu\nu\kappa\lambda}| |\nabla^{m_2} \fara_{\kappa \lambda}|+|\nabla^{\min\{m_1, |I|+a-1\}}\fara_{\kappa\lambda}|) \cdot |\faram_{,I,a}|.
\end{equation}

\end{lemma}
\begin{proof} 
Let us focus on the case $a > 0$, as the case $a = 0$ is analogous (the only difference is in the definition of the error terms $\otm_2$ when $a = 0$). Moreover, let us denote
\begin{equation*}
	\dot Q\indices{^\mu_\nu} := {\dot Q}_{\mu\nu}[\dot \fara_{,I,a}] =  H^{\mu\zeta\kappa\lambda} {(\dot \fara_{,I,a})}_{\kappa\lambda} ({\dot \fara_{,I,a})}_{\nu\zeta} - \frac 1 4 \delta_\nu^\mu  H^{\mu\zeta\kappa\lambda} {{( \dot\fara_{,I,a})}_{\kappa\lambda}}{( \dot\fara_{,I,a})}_{\mu\zeta}.
\end{equation*}
We calculate
\begin{equation*}
\begin{aligned}
\nabla^\mu({\dot Q}_{\mu\nu}\redsh^\nu) = (\nabla_\mu \dot Q\indices{^\mu_\nu}) \redsh^\nu+ \dot Q_{\mu\nu} \nabla^\mu \redsh^\nu.
\end{aligned}
\end{equation*}
Now, in view of the linear theory (Lemma~\ref{lem:coercivity}), possibly restricting $r_{\text{in}}$ to be sufficiently close to $2M$, and choosing $\varepsilon_{\text{red}}$ small, there holds
\begin{equation}\label{eq:redshcoe}
\dot Q_{\mu\nu} \nabla^\mu \redsh^\nu \geq 2 C^{(1)}_{I,a} |\dot \fara_{,I,a}|^2.
\end{equation}
By Proposition~\ref{prop:commut} and Remark~\ref{rmk:formij}, we have that $\dot \fara_{,I,a}$ satisfies the following equations of variation:
\begin{equation}
\begin{aligned}
&\mathcal{J}_\nu := \nabla^\gamma (\faramd_{,I,a})_{\gamma \nu} = a (\nabla^\mu Y^\alpha) \nabla_\alpha(\faramd_{,I, a-1})_{\mu\nu} +  ( \text{OT}^{(a+|I|)}_1)_\nu,\\
&\mathcal{I}^\nu := H^{\mu\nu\kappa\lambda}\nabla_\mu ( \faram_{,I,a})_{\kappa\lambda} = a( \nabla^\kappa Y^\alpha) \nabla_\alpha (\faram_{,I,a-1})\indices{_\kappa^\nu}  + ( \text{OT}^{(a+|I|)}_1)^\nu.
\end{aligned}
\end{equation}
We now consider $(\nabla_\mu \dot Q\indices{^\mu_\nu}) \redsh^\nu$. By Lemma~\ref{lem:qdive}, 
\begin{equation*}
\begin{aligned}
\nabla_\mu   \dot{Q} \indices{^\mu _\nu}&= 
\underbrace{(\nabla_\mu H_\Delta^{\mu\zeta\kappa\lambda}) {(\dot \fara_{,I,a})}_{\kappa\lambda} {(\dot \fara_{,I,a})}_{\nu\zeta}- \frac 1 4 (\nabla_\nu H_\Delta^{\zeta \eta \kappa\lambda} (\dot \fara_{,I,a})_{\zeta \eta} (\dot \fara_{,I,a})_{\kappa\lambda})
}_{(i)} \underbrace{- \frac 1 2 H_\Delta^{\zeta \eta \kappa\lambda} (\dot \fara_{,I,a})_{\kappa\lambda}( \mathcal{J}^\alpha \varepsilon_{\alpha \nu \zeta \eta})}_{(ii)} \\ &+
\underbrace{\mathcal{I}^\zeta (\dot \fara_{,I,a})_{\nu \zeta}+\mathcal{J}^\alpha ( \faramd_{,I,a})_{\nu \alpha}}_{(iii)}.
\end{aligned}
\end{equation*}
By possibly choosing $\varepsilon_{\text{red}}$ to be smaller, we deduce, in the region $\tregio{t_1^*}{t_2^*} \cap \{r_{\text{in}} \leq r \leq 4M - r_{\text{in}}\}$,
\begin{equation*}
|(i)_\nu \redsh^\nu| \leq \frac{C^{(1)}_{I,a}}{4}  |\dot \fara_{,I,a}|^2.
\end{equation*}
Concerning $(ii)$,
\begin{equation*}
\begin{aligned}
(ii)_\nu  &= - \frac 1 2 H_\Delta^{\zeta \eta \kappa\lambda} (\dot \fara_{,I,a})_{\kappa\lambda}( \mathcal{J}^\alpha \varepsilon_{\alpha \nu \zeta \eta}) \\ &= - \frac 1 2 H_\Delta^{\zeta \eta \kappa\lambda} (\dot \fara_{,I,a})_{\kappa\lambda} \  \varepsilon_{\alpha \nu \zeta \eta}\left(
a (\nabla^\mu Y^\beta) \nabla_\beta(\faramd_{,I, a-1})\indices{_\mu^\alpha} +  ( \text{OT}^{(a+|I|)}_1)^\alpha \right).
\end{aligned}
\end{equation*}
Therefore, under the bootstrap assumptions, we have that
\begin{equation*}
|(ii)_\nu \redsh^\nu| \leq C^{(2)}_{I,a} \varepsilon^{3/2} |\partial^{a+|I|} \fara|^2 + |\otm_2^{(I,a)}|
\end{equation*}
on $\tregio{t_1^*}{t_2^*} \cap \{r_{\text{in}} \leq r \leq 4M - r_{\text{in}}\}$.

We now turn to the estimation of $(iii)$. We first notice that
\begin{align*}
\nabla_Y Y  = 0, \qquad 
\nabla_L Y = - \frac {2M} {r^2} Y, \qquad 
\nabla_{\partial_{\theta^A}} Y = - \frac 1 r \partial_{\theta^A},
\end{align*}
so that $\nabla_A Y_B = -\frac 1 r \slashed g_{AB}$.
With this in mind, we calculate
\begin{equation*}
\begin{aligned}
&( \nabla_\kappa Y^\alpha) \nabla_\alpha (\faram_{,I,a-1})\indices{^\kappa^\zeta}(\dot \fara_{,I,a})_{\nu \zeta}L^\nu\\
 &= (\nabla_L Y)^Y (\nabla_Y \faram_{,I,a-1})^{L\zeta}(\dot \fara_{,I,a})_{L \zeta} + (\nabla_B Y)^A (\nabla_A \faram_{,I,a-1})^{B\zeta}(\dot \fara_{,I,a})_{L \zeta} \\
&= - \frac{2M}{r^2} (\nabla_Y \faram_{,I,a-1})^{L\zeta}(\dot \fara_{,I,a})_{L \zeta} - \frac 1 r \slashed{g}^{AB} \nabla_A (\faram_{,I,a})\indices{_B ^\zeta}(\dot \fara_{,I,a})_{L \zeta}\\
&= -\frac{2M}{r^2} (\nabla_Y \faram_{,I,a-1})^{LY}(\dot \fara_{,I,a})_{LY}-\frac{2M}{r^2} (\nabla_Y \faram_{,I,a-1})^{LA}(\dot \fara_{,I,a})_{L A} - \frac 1 r \slashed{g}^{AB} \nabla_A (\faram_{,I,a-1})\indices{_B ^\zeta}(\dot \fara_{,I,a})_{L \zeta} \\
&= \frac {M}{2r^2} ((\faram_{,I,a})_{YL})^2 + \frac{M}{r^2} (\faram_{,I,a})\indices{_Y^A}(\faram_{,I,a})_{LA} - \frac 1 r \slashed{g}^{AB} \nabla_A (\faram_{,I,a-1})\indices{_B ^\zeta}(\dot \fara_{,I,a})_{L \zeta} + \otm_2^{(I,a)}.
\end{aligned}
\end{equation*}
Here, if $a \in \N_{\geq 0}$, $J \in \iindicess^b$ in the schematic notation of Definition~\ref{def:short},
\begin{equation*}
\otm_2^{(I,a)} = \otm^{(a+|I|)}_1\faram_{,I,a}.
\end{equation*}
The expression $\otm_2^{(I,a)}$, moreover, satisfies the following bound:
\begin{equation*}
|\otm_2^{(I,a)}| \leq C \sum_{\substack{(m_1, m_2) \in \N_{\geq 0}^2 \\ m_1+m_2 \leq a+|I|+1 \\ m_1, m_2 \leq a + |I|}} (|\nabla^{m_1} H_\Delta^{\mu\nu\kappa\lambda}| |\nabla^{m_2} \fara_{\kappa \lambda}|+|\nabla^{\min\{m_1, a-1\}}\fara_{\kappa\lambda}|) \cdot |\faram_{,I,a}|.
\end{equation*}
Similarly,
\begin{equation*}
\begin{aligned}
&( \nabla_\kappa Y^\alpha) \nabla_\alpha (\faram_{,I,a-1})\indices{^\kappa^\zeta}(\dot \fara_{,I,a})_{\nu \zeta}Y^\nu \\
&= (\nabla_L Y)^Y (\nabla_Y \faram_{,I,a-1})^{L\zeta}(\dot \fara_{,I,a})_{Y \zeta} + (\nabla_B Y)^A (\nabla_A \faram_{,I,a-1})^{B\zeta}(\dot \fara_{,I,a})_{Y \zeta} \\
&= - \frac{2M}{r^2} (\nabla_Y \faram_{,I,a-1})^{L\zeta}(\dot \fara_{,I,a})_{Y \zeta} - \frac 1 r \slashed{g}^{AB} \nabla_A (\faram_{,I,a-1})\indices{_B ^\zeta}(\dot \fara_{,I,a})_{Y \zeta}\\
&= -\frac{2M}{r^2} (\nabla_Y \faram_{,I,a-1})^{LA}(\dot \fara_{,I,a})_{Y A} - \frac 1 r \slashed{g}^{AB} \nabla_A (\faram_{,I,a-1})\indices{_B ^\zeta}(\dot \fara_{,I,a})_{Y \zeta} \\
&= \frac{M}{r^2} (\faram_{,I,a})\indices{_Y^A}(\faram_{,I,a})_{YA} - \frac 1 r \slashed{g}^{AB} \nabla_A (\faram_{,I,a})\indices{_B ^\zeta}(\dot \fara_{,I,a-1})_{Y \zeta} + \otm_2^{(I,a)}.
\end{aligned}
\end{equation*}
Using the calculations involving null components of the Hodge dual of $\fara$ in Lemma~\ref{lem:dualnull}, we have that
\begin{equation*}
\begin{aligned}
&( \nabla_\kappa Y^\alpha) \nabla_\alpha (\faramd_{,I,a-1})\indices{^\kappa^\zeta}( \faramd_{,I,a})_{\nu \zeta}L^\nu \\
&=  \frac {M}{2r^2} ((\faramd_{,I,a})_{YL})^2 + \frac{M}{r^2} (\faramd_{,I,a})\indices{_Y^A}(\faramd_{,I,a})_{LA} - \frac 1 r \slashed{g}^{AB} \nabla_A (\faramd_{,I,a-1})\indices{_B ^\zeta}(\faramd_{,I,a})_{L \zeta} + \otm_2^{(I,a)} \\
&=  \frac {M}{2r^2} ((\faramd_{,I,a})_{YL})^2 - \frac{M}{r^2} (\faram_{,I,a})\indices{_Y^A}(\faram_{,I,a})_{LA} - \frac 1 r \slashed{g}^{AB} \nabla_A (\faramd_{,I,a-1})\indices{_B ^\zeta}(\faramd_{,I,a})_{L \zeta} +  \otm_2^{(I,a)} \\
\end{aligned}
\end{equation*}
Also,
\begin{equation*}
\begin{aligned}
&( \nabla_\kappa Y^\alpha) \nabla_\alpha (\faramd_{,I,a-1})\indices{^\kappa^\zeta}(\faramd_{,I,a})_{\nu \zeta}Y^\nu 
\\ &= \frac{M}{r^2} (\faramd_{,I,a})\indices{_Y^A}(\faramd_{,I,a})_{YA} - \frac 1 r \slashed{g}^{AB} \nabla_A (\faramd_{,I,a-1})\indices{_B ^\zeta}(\faramd_{,I,a})_{Y \zeta} + \otm_2^{(I,a)}.
\end{aligned}
\end{equation*}
Concerning $(iii)$, let's write
\begin{equation}
\redsh = f_1(r) L + f_2(r) Y,
\end{equation} 
with $f_1(r) =  1 + 5 \chi_1(r) (1-\mu)$, $f_2(r) =(1-\mu)( 1+ 5\chi_1(r)) + \chi_1(r)$. Here, if $K \in \mathcal K$, and if $I = (K_1, \ldots, K_b)$, the multi-index $\widetilde{I}(K)$ is given by $$\widetilde{I}(K) = (K, K_1, \ldots, K_b).$$ Using the above calculations, with the fact that $f_1$ and $f_2$ are strictly positive on the considered region, we get the following estimates:
\begin{align*}
&(iii)_\nu \redsh^\nu \\ 
&= f_1(r)a \left(\frac {M}{2r^2} ((\faram_{,I,a})_{YL})^2 - \frac 1 r \slashed{g}^{AB} \nabla_A (\faram_{,I,a-1})\indices{_B ^\zeta}(\dot \fara_{,I,a})_{L \zeta} \right.\\
&\left. +  \frac {M}{2r^2} ((\faramd_{,I,a})_{YL})^2 - \frac 1 r \slashed{g}^{AB} \nabla_A (\faramd_{,I,a-1})\indices{_B ^\zeta}(\faramd_{,I,a})_{L \zeta}\right)  \\
&+ f_2(r) a\left(\frac{M}{r^2} (\faram_{,I,a})\indices{_Y^A}(\faram_{,I,a})_{YA} - \frac 1 r \slashed{g}^{AB} \nabla_A (\faram_{,I,a-1})\indices{_B ^\zeta}(\dot \fara_{,I,a})_{Y \zeta}\right.\\ &\left.  +  \frac{M}{r^2} (\faramd_{,I,a})\indices{_Y^A}(\faramd_{,I,a})_{YA}- \frac 1 r \slashed{g}^{AB} \nabla_A (\faramd_{,I,a-1})\indices{_B ^\zeta}(\faramd_{,I,a})_{Y \zeta} \right)
+ \otm_2^{(I,a)} \\
&\geq - C^{(3)}_{I,a} \left(|\slashed{g}^{AB} \nabla_A (\faram_{,I,a-1})\indices{_B ^\zeta}| +| \slashed{g}^{AB} \nabla_A (\faramd_{,I,a-1})\indices{_B ^\zeta}| \right)|( \faram_{,I,a})_{\nu \zeta}|- |\otm_2^{(I,a)}| \\
& \geq -C^{(3)}_{I,a} \left(\sum_{i = 1}^3 |\faram_{, \widetilde I (\Omega_i),a-1}|\right)|( \faram_{,I,a})_{\nu \zeta}| - \frac{C^{(1)}_{I,a}}{4} |( \faram_{,I,a})_{\nu \zeta}|^2- |\otm_2^{(I,a)}| \\
& \geq -C^{(3)}_{I,a} \left(\sum_{i = 1}^3 |\faram_{, \widetilde I (\Omega_i),a-1}|^2\right) - \frac{C^{(1)}_{I,a}}{2} |( \faram_{,I,a})_{\nu \zeta}|^2- |\otm_2^{(I,a)}|.
\end{align*}
Combining the estimates in Equation~\eqref{eq:redshcoe}, and the estimates for terms $(i), (ii), (iii)$ we obtain the claim.
\end{proof}

We now sum and integrate the estimate in Lemma~\ref{lem:pointpos}, and we obtain the following statement.

\begin{lemma}\label{lem:qtot}
There exist a number $r_{\text{in}} \in (0, 2M)$, sufficiently close to $2M$, a positive constant $C >0$, and finally a small $\varepsilon_{\text{int}}> 0$ such that the following holds. Let $\fara$ be a smooth 2-form which satisfies the MBI system~\eqref{MBI} on $\mathcal{R} := \tregio{t_0^*}{t_1^*}$, with $t_1^* > t_0^*$.  Let $0 < \varepsilon < \varepsilon_{\text{int}}$ and assume the bootstrap assumptions $BA(\tregio{t_0^*}{t_1^*}, 1, N, \varepsilon)$ (see Section~\ref{sec:bootstrap}). Let $t_1^* \geq s_2 \geq s_1 \geq t_0^*$, $j \in \N_{\geq 0}$, $j \leq 2N$, $a \in \N_{\geq 0}$, and $I \in \iindicess^b$. Let $a + b \leq j$, and let the variation $\dot \fara$, for notational convenience, be defined as
\begin{equation*}
\dot \fara := \dot \fara_{,I,a}.
\end{equation*}
Then, there exist positive constants $(B_{a})_{a \in \N_{\geq 0}}$ (the dependence on $j$ is suppressed in the notation) such that, defining
\begin{equation*}
\begin{aligned}
\qtot := \sum_{a+b = j} \sum_{I \in \iindicess^b} B_{a} \redsh^\nu \,  n^\mu_{\widetilde \Sigma_{t^*}} \mqtensor \indices{_\mu _\nu}[\dot \fara_{,I,a}],
\end{aligned}
\end{equation*}
the following inequality holds:
\begin{equation}\label{eq:mixfinal}
\boxed{
\begin{aligned}
&\int_{\widetilde{\Sigma}_{{s}_2}} \qtot \de {\widetilde{\Sigma}_{t^*}} -
\int_{\widetilde{\Sigma}_{{s}_1}} \qtot \de {\widetilde{\Sigma}_{t^*}}  \\ &
 + C \int_{\tregio{{s}_1}{{s}_2} \cap \{r_{\text{in}} \leq r \leq 4M - r_{\text{in}}\}}\sum_{a+b = j} \sum_{I \in \iindicess^b} |(\faram_{,I,a})_{\mu\nu}|^2  \de \text{Vol} \\ & 
\leq   C  \int_{\regio{{s}_1}{{s}_2}\cap \{4M - r_{\text{in}}\leq r \leq 3M\}} \sum_{a+b = j} \sum_{I \in \iindicess^b} |\nabla_\mu (\mqtensor \indices{^\mu _\nu}\redsh^\nu)|\, \de \text{Vol} \\ &
+C \int_{\tregio{{s}_1}{{s}_2}} \sum_{a+b = j} \sum_{I \in \iindicess^b}\left|\otm_2^{(I,a)}\right| \de \text{Vol}.
\end{aligned}
}
\end{equation}
Here, $\mqtensor\indices{_\mu_\nu} = \mqtensor\indices{_\mu_\nu}[\dot \fara_{,I,a}]$ is the canonical stress associated to $\dot \fara = \dot \fara_{,I,a}$ (and in particular it depends on the two indices $I$ and $a$). Also, the terms $\otm^{(I,a)}_2$ satisfy the schematic equations (in the sense of Definition~\ref{def:short}) of display \eqref{eq:otm21} (when $a=0$) and of display \eqref{eq:otm22} (when $a > 0$), as well as the estimates~\eqref{eq:otm2est1} and~\eqref{eq:otm2est2}. Furthermore, we recall that $\de \text{Vol}$ is the standard volume form on Schwarzschild, and $\de {\widetilde{\Sigma}_{t^*}}$ is the induced volume form on the foliation $\widetilde{\Sigma}_{t^*}$.
\end{lemma}
\begin{remark}\label{rmk:afterqtot}
	This proposition is the crucial point in our argument where we need to integrate inside the black hole region. This is in order to obtain positivity for the corresponding boundary term, as the surface $r = r_{\text{in}}$ is \textbf{strictly spacelike} inside the black hole region.
\end{remark}
\begin{proof}
	\begin{figure}\label{figure2}
		\centering
		\begin{tikzpicture}	

\node (I)    at ( 0,0) {};

\path
  (I) +(90:4)  coordinate[label=90:$i^+$]  (top)
       +(-90:4) coordinate[label=-90:$i^-$] (bot)
       +(0:4)   coordinate                  (right)
       +(180:4) coordinate[label=180:$i^0$] (left)
       +(180:4.9) coordinate[label=180:$ $] (left1)
       ;
 \path (top) +(180:0.01) coordinate (tar);

\path 
	(top) + (180:4) coordinate[label=90:$i^+$]  (acca)
		+ (-45: 2) coordinate (nulluno)
		+ (-45: 3) coordinate (nulldue)
	;

\path 
	(left) + (-45: 2) coordinate (correspuno)
		+ (-45: 3) coordinate (correspdue)
	;

\path (top) + (-150: 1.7) coordinate (horiuno)
			+ (-150: 5) coordinate (horidue);

\draw[name path = stre, opacity = 0] (left1) to (tar);
\draw [name path = ciao, thick, opacity = 0] (horiuno) to [bend left = 15] node[midway, above]{} (right);
\draw [name path = bye, thick, opacity = 0] (horidue) to [bend left = 10] node[midway, above]{} (right);
\path [name intersections={of=ciao and stre,by=E}];
\path [name intersections={of=bye and stre,by=D}];
\draw [thick] (D) to [bend left = 15] node[midway, above]{$\widetilde{\Sigma}_{t^*_1}$} (right);
\draw [thick] (E) to [bend left = 10] node[midway, above]{$\widetilde{\Sigma}_{t^*_2}$} (right);
\draw [thick] (E) to [bend left = 6] node[midway, left]{ $r = r_{\text{in}}$} (D);
\draw (left) -- 
          node[midway, above left, sloped]    {}
      (top) --
          node[midway, above, sloped] {$\mathcal{I}^+$}
      (right) -- 
          node[midway, above, sloped] {$\mathcal{I}^-$}
      (bot) --
          node[midway, above, sloped]    {$\mathcal{H}^-$}    
      (left) -- cycle;

\draw[decorate,decoration=zigzag] (top) -- (acca)
      node[midway, above, inner sep=2mm] {$r=0$};
      
\end{tikzpicture}
		\caption{Penrose diagram of the modified integration region. Notice that the surface inside the black hole region, $r = r_{\text{in}}$, is strictly spacelike.}\label{fig2}
	\end{figure}
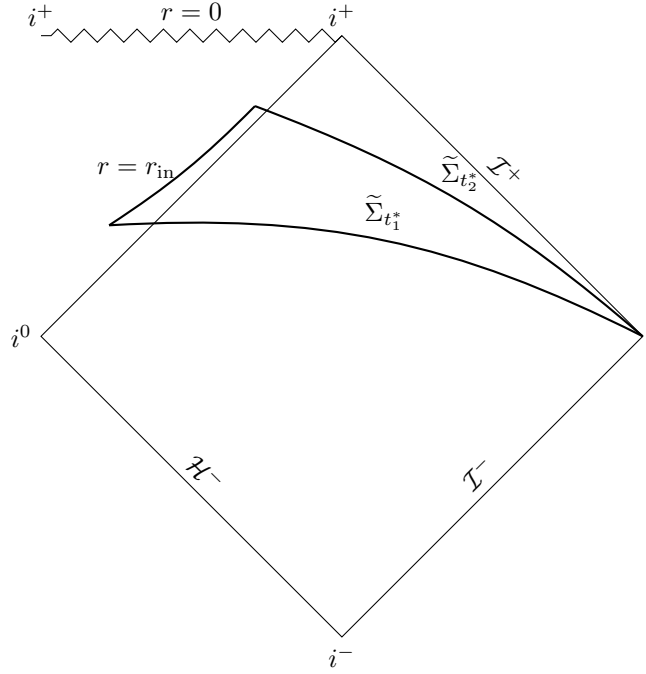
The proof of this proposition will be divided in two parts. In the first part, we will deal with the case $j=0$. In the second part, we will deal with the case $j >0$.
\begin{enumerate}[wide, labelwidth=!, labelindent=0pt]
\item \textbf{Case $j = 0$}. We integrate Equation~(\ref{eq:stressdiv}) on the region $\{r \geq r_{\text{in}}\} \cap \regio{t_1^*}{t_2^*}$ (see Figure~\ref{fig2}).
By possibly restricting to $r_{\text{in}}$ close to $2M$ and to $\varepsilon_{\text{int}}$ smaller, we obtain that the boundary term at $r_{\text{in}}$ is positive, and we get:
\begin{equation*}
\begin{aligned}
&\int_{\widetilde{\Sigma}_{s_2}} \dot Q_{\mu\nu} \redsh^\mu n^\nu_{\widetilde \Sigma_{t^*}}\de \widetilde{\Sigma}_{t^*}-\int_{\widetilde{\Sigma}_{s_1}} \dot Q_{\mu\nu} \redsh^\mu n^\nu_{\widetilde\Sigma_{t^*}}\de \widetilde{\Sigma}_{t^*} + \int_{\tregio{t_1^*}{t_2^*}} \dot Q\indices{^\mu_\nu} \nabla_\mu \redsh^\nu \de \text{Vol} \\
&\leq \int_{\tregio{t_1^*}{t_2^*}}\Big( (\nabla_\mu H^{\mu\zeta\kappa\lambda}) {\fara}_{\kappa\lambda} {\fara}_{\nu\zeta}- \frac 1 4 (\nabla_\nu H^{\zeta \eta \kappa\lambda} \fara_{\zeta \eta} \fara_{\kappa\lambda}) \Big) (\redsh)^\nu \de \text{Vol}.
\end{aligned}
\end{equation*}
The claim in the case $j=0$ then follows from the bootstrap assumptions, along with the positivity of the deformation tensor near $\mathcal{H}^+$ encoded in estimate~\eqref{eq:currentposz}, possibly restricting to a smaller value of $\varepsilon_{\text{int}}$.

\item \textbf{Case $j > 0$.} We notice that, in the notation of Lemma~\ref{lem:pointpos}, if $b \in \N_{\geq 0}$, $b \geq 1$, we have the control
\begin{equation*}
\sum_{I \in \iindicess^b} |\faram_{,I,a}|^2 \geq  \sum_{J \in \iindicess^{b-1}} \sum_{i = 1}^3 |\faram_{,\widetilde{J}(\Omega_i),a}|^2,
\end{equation*}

where, if $J = (K_1, \ldots, K_{b})$, $\widetilde{J}(\Omega_i) = (\Omega_i, K_1, \ldots, K_b)$. This follows from the definition.

Let $a+ b = j$. Let 
\begin{equation*}
\begin{aligned}
\widetilde{C}^{(1)}_{a} :&= \min_{I \in \iindicess^b} C^{(1)}_{I,a},\\
\widetilde{C}^{(3)}_{a} :&= \max_{I \in \iindicess^b} C^{(3)}_{I,a}.
\end{aligned}
\end{equation*}
Let us now sum inequality (\ref{eq:currentpos}) for all $I \in \iindicess^b$. We obtain, on $\{r_{\text{in}} \leq r \leq 4M - r_{\text{in}}\}$:
\begin{equation}\label{eq:tosum}
\begin{aligned}
&\sum_{I \in \iindicess^b} \nabla^\mu({\dot Q}_{\mu\nu}\redsh^\nu) \\
&\geq \widetilde{C}^{(1)}_a \sum_{I \in \iindicess^b} |\dot \fara_{,I,a}|^2 -  \widetilde{C}^{(2)}_a \varepsilon^{3/2} |\partial^{a+|I|} \fara|^2 
- \widetilde{C}^{(3)}_a \sum_{I \in \iindicess^b} \left(\sum_{i = 1}^3 |\faram_{, \widetilde I (\Omega_i),a-1}|^2\right) - \sum_{I \in \iindicess^b}|\otm_2^{(I,a)}|\\
& \geq \widetilde{C}^{(1)}_a \sum_{I \in \iindicess^b} |\dot \fara_{,I,a}|^2 -  \widetilde{C}^{(2)}_a \varepsilon^{3/2} |\partial^{a+|I|} \fara|^2 
- \widetilde{C}^{(3)}_a \sum_{J \in \iindicess^{b+1}}  |\faram_{,J,a-1}|^2 - \sum_{I \in \iindicess^b}|\otm_2^{(I,a)}|,
\end{aligned}
\end{equation}
for $a \geq 1$.

Let us now choose a sequence of positive real numbers $\{B_i\}_{i =1, \ldots, j}$ such that $B_j = 1$, and such that for $i \in \{1, \ldots, j\}$, we have
\begin{equation}\label{eq:posb}
B_{i-1} \widetilde{C}_{i-1}^{(1)} \geq B_{i} \widetilde{C}_{i}^{(3)} + 2.
\end{equation}
Multiply now inequality~\eqref{eq:tosum} by $B_a$ and sum for $a$ from $1$ to $j$, in order to obtain
\begin{equation}\label{eq:summed}
\begin{aligned}
&\sum_{a = 1}^j \sum_{I \in \iindicess^b} B_a \nabla^\mu({\dot Q}_{\mu\nu}\redsh^\nu) \\
&\geq 2\sum_{a =1}^j \sum_{\substack{I \in \iindicess^b\\ a+b = j}} |\dot \fara_{,I,a}|^2 -  \widetilde{C}^{(2)} \varepsilon^{3/2} |\partial^{j} \fara|^2 
- \widetilde{C}^{(3)}_1 B_1\sum_{J \in \iindicess^{j}} \sum_{i = 1}^3 |\faram_{,J,0}|^2
 - \widetilde{C}^{(4)}\sum_{a=1}^j\sum_{\substack{I \in \iindicess^b\\ a+b = j}}|\otm_2^{(I,a)}|
 \end{aligned}
 \end{equation}
 Here, $\widetilde{C}^{(2)}$ and $\widetilde{C}^{(4)}$ are positive constants which depend on $B_i$, $\widetilde{C}^{(1)}_i$, and $\widetilde{C}^{(3)}_i$.
Using inequality~\eqref{eq:currentposz}, we then have, possibly redefining the value of the constant $\widetilde C^{(2)}$,
 \begin{align*}
 &\sum_{a = 0}^j \sum_{I \in \iindicess^b} B_a \nabla^\mu({\dot Q}_{\mu\nu}\redsh^\nu) \\
&\geq 2\sum_{a =0}^j \sum_{\substack{I \in \iindicess^b\\ a+b = j}} |\dot \fara_{,I,a}|^2 -  \widetilde{C}^{(2)} \varepsilon^{3/2} |\partial^{j} \fara|^2 
   - \widetilde{C}^{(4)}\sum_{a=0}^j\sum_{\substack{I \in \iindicess^b\\ a+b = j}}|\otm_2^{(I,a)}|.
\end{align*}

 We notice finally that the term in $\partial^j \fara$ in the last line of the previous inequality (the second term) can be absorbed by the first term in the same line.

All in all, we obtain, on $\{r_{\text{in}} \leq r \leq 4M - r_{\text{in}}\}$:
\begin{equation}\label{eq:summed1}
\begin{aligned} \sum_{a = 0}^j \sum_{\substack{I \in \iindicess^b\\ a+b = j}} B_a \nabla^\mu({\dot Q}_{\mu\nu}\redsh^\nu) \geq \sum_{a =0}^j \sum_{\substack{I \in \iindicess^b\\ a+b = j}} |\dot \fara_{,I,a}|^2
  - \widetilde{C}^{(4)}\sum_{a=0}^j\sum_{\substack{I \in \iindicess^b\\ a+b = j}}|\otm_2^{(I,a)}|.
\end{aligned}
\end{equation}
Let $$\qtot := \sum_{a+b = j} \sum_{I \in \iindicess^b} B_{a} \redsh^\nu \,  n^\mu_{\widetilde \Sigma_{t^*}} \mqtensor \indices{_\mu _\nu}[\dot \fara_{,I,a}].$$

We now integrate $ \nabla^\mu(\dot Q\indices{_\mu_\nu} \redsh^\nu)$ on $\tregio{s_1}{s_2} \cap \{r_{\text{in}} \leq r\}$ (see Figure~\ref{fig2}). By restricting to $\varepsilon_{\text{int}}$ small enough and to $r_{\text{in}}$ sufficiently close to $2M$, we have that the boundary term at $\{r = r_{\text{in}}\}$ is nonnegative, and therefore
\begin{equation*}
\begin{aligned}
\int_{\widetilde{\Sigma}_{s_2}} \qtot \de \widetilde{\Sigma}_{t^*} - \int_{\widetilde{\Sigma}_{s_1}} \qtot \de \widetilde{\Sigma}_{t^*} + \sum_{a+b = j} \sum_{I \in \iindicess^b} B_a\int_{\tregio{s_1}{s_2}\cap\{r \geq r_{\text{in}} \}} \nabla^\mu(\dot Q\indices{_\mu_\nu} \redsh^\nu) \de \text{Vol} \leq 0.
\end{aligned}
\end{equation*}
We obtain estimate (\ref{eq:mixfinal}) splitting the integral between the region $\{r_{\text{in}} \leq r \leq 4M - r_{\text{in}}\}$ and $\{4M - r_{\text{in}} \leq r\}$, and considering the fact that $\redsh = T$ on the region $\{r \geq 3M\}$. Hence
\begin{equation*}
\nabla_\mu (T^\nu \dot Q\indices{^\mu_\nu}) = \otm^{(I,a)}_2
\end{equation*}
on $\{r \geq 3M\}$. This implies the claim (\ref{eq:mixfinal}).
\end{enumerate}
\end{proof}

\subsection{Concluding the proof of $L^2$ estimates}\label{sub:concludingl2}
In this section, we finish the proof of Proposition~\ref{prop:energy}. We integrate the differential inequality in Lemma~\ref{lem:qtot}, we bound the main ``bulk'' error terms using the degenerate estimate in Lemma~\ref{lem:degen}, and we finally use the Gr\"onwall inequality to obtain $L^2$ boundedness.

\begin{proof}[Proof of Proposition~\ref{prop:energy}]
We prove this Proposition by induction on $j \leq 2N$.

\begin{itemize}[wide, labelwidth=!, labelindent=0pt]
\item We start from the \textbf{induction base case}: $j = 0$. This step must be carried out considering the standard stress-energy tensor $Q$ relative to the MBI theory. Its expression is
\begin{equation}\label{eq:semtdef}
Q^{\mu\nu} = \ellmbi^{-1}(g_{\kappa\lambda}\fara^{\mu \kappa}\fara^{\nu \lambda}- \ldu^2 g^{\mu\nu}) + g^{\mu\nu}(1-\ellmbi).
\end{equation}
The remarkable property of this tensor is that, if $\fara$ satisfies the MBI system (\ref{MBI}), the divergence of $Q$ vanishes:
\begin{equation*}
\nabla_{\mu} Q^{\mu\nu} = 0.
\end{equation*}
Also, $Q_{\mu\nu} = Q_{\nu\mu}$, and it satisfies the positive energy condition. Furthermore, we have that there exist positive constants $C_1, C_2$, $\varepsilon_{\text{nul}}$ such that, for every $0 < \varepsilon < \varepsilon_{\text{nul}}$, and $|\fara| \leq \eps$,
\begin{equation*}
C_1 |\fara|^2 \leq Q_{\mu\nu} \redsh^\mu {n}_{ \widetilde \Sigma_{t^*}}^\nu \leq C_2 |\fara|^2.
\end{equation*}
Furthermore, the same inequality holds with $Q_{\mu\nu} T^\mu {n}_{\widetilde \Sigma_{t^*}}^\nu$, on the region $r \geq 4M- r_{\text{in}}$.

Applying now the divergence theorem with the current $J_1^\mu = T^\nu Q\indices{^\mu_\nu}$ gives, when $s_2 \geq s_1 \geq s_0$, upon integration on the region $\regio{s_1}{s_2}$:
\begin{equation}\label{eq:tzero}
\begin{aligned}
\int_{{\Sigma}_{s_2}} Q_{\mu\nu} T^\mu {n}_{\Sigma_{t^*}}^\nu \de {\Sigma}_{t^*} - \int_{{\Sigma}_{s_1}} Q_{\mu\nu} T^\mu {n}_{\Sigma_{t^*}}^\nu \de {\Sigma}_{t^*}\leq 0.
\end{aligned}
\end{equation}
An application of the divergence theorem with the current $J_2^\mu = {\redsh}^\nu Q\indices{^\mu_\nu}$ gives, if $s_2 \geq s_1 \geq s_0$, upon integration on the region $\tregio{s_1}{s_2}$:
\begin{equation}\label{eq:redzero}
\begin{aligned}
\int_{\widetilde{\Sigma}_{s_2}} Q_{\mu\nu} \redsh^\mu {n}_{\widetilde \Sigma_{t^*}}^\nu \de \widetilde{\Sigma}_{t^*} - \int_{\widetilde{\Sigma}_{s_1}} Q_{\mu\nu} \redsh^\mu {n}_{\widetilde \Sigma_{t^*}}^\nu \de \widetilde{\Sigma}_{t^*} +C_3 \int_{\tregio{s_1}{s_2}\cap \{r_{\text{in}}\leq r 
\leq 4M - r_{\text{in}} \}} |\fara|^2 \de \text{Vol}\\ \leq C_4\int_{\tregio{s_1}{s_2}\cap \{4M - r_{\text{in}}\leq r \leq 3M \}} |\fara|^2 \de \text{Vol}.
\end{aligned}
\end{equation}
\begin{remark}
In the previous inequality, we implicitly used the fact that, in the region ``close to the event horizon'' (where $r_{\text{in}}\leq r 
\leq 4M - r_{\text{in}}$), there exists a constant $C> 0$ such that the divergence $ \nabla_\mu (\redsh^\nu Q\indices{^\mu_\nu}) \geq C |\fara|^2$. This amounts to showing that $^{(\redsh)}\pi^{\mu\nu} Q_{\mu\nu} \geq C |\fara|^2$. By the linear theory, we know that, in this region, 
\begin{equation}\label{eq:fromlin}
Q^{(\text{MW})}_{\mu\nu}[\fara] \ ^{(\redsh)}\pi^{\mu\nu} \geq C |\fara|^2,
\end{equation}
where $Q^{(\text{MW})}_{\mu\nu} = \fara_{\mu \alpha}\fara\indices{_\nu^\alpha} - \frac 14 g_{\mu\nu} \fara_{\alpha\beta}\fara^{\alpha\beta}$ is the standard stress-energy-momentum stress for the linear Maxwell equations. We also note that there holds:
\begin{equation}
\begin{aligned}
&Q_{\mu\nu} - Q^{(\text{MW})}_{\mu\nu} = - \frac{\lun-\ldu^2}{\ellmbi (1+\ellmbi)}  \fara_{\mu \alpha} \fara\indices{_\nu^\alpha}\\
& \qquad  + g_{\mu\nu}\Big( \lun \frac{\lun - \ldu^2}{(1+\ellmbi)^2}- \frac{\ldu^2}{\ellmbi(1+\ellmbi)} \Big).
\end{aligned}
\end{equation}
It is easy to see that the error between $Q$ and $Q^{(MW)}$ is at least quartic in $\fara$. Therefore, using the bootstrap assumptions (upon possibly  restricting the value of $\eps_0$ to be smaller), plus the bound~\eqref{eq:fromlin}, we obtain the pointwise bound, in the region $\{r_{\text{in}}\leq r 
\leq 4M - r_{\text{in}}\}$: 
$$
^{(\redsh)}\pi^{\mu\nu} Q_{\mu\nu} \geq C |\fara|^2.
$$
\end{remark}
We now bound the RHS of (\ref{eq:redzero}) by (\ref{eq:tzero}), we divide the resulting inequality by $s_2-s_1$ and we take the limit $s_1 \to s_2$. We obtain:
\begin{equation}\label{eq:togronz}
F'(s) + C_3 F(s) \leq C_4 F(s_0),
\end{equation}
where we denoted
\begin{equation}
F(s) := \int_{\widetilde{\Sigma}_{s}} Q_{\mu\nu} \redsh^\mu {n}_{\widetilde \Sigma_{t^*}}^\nu \de \widetilde{\Sigma}_{t^*}.
\end{equation}
Integrating (\ref{eq:togronz}), and using the fact that we assumed
\begin{equation*}
\int_{\widetilde{\Sigma}_{t^*_0}} Q_{\mu\nu} \redsh^\mu {n}_{\widetilde \Sigma_{t^*}}^\nu \de \widetilde{\Sigma}_{t^*} \leq \varepsilon^2,
\end{equation*}
we obtain the claim when $j =0$, i.e.
\begin{equation}
	\norm{\fara}^2_{H^0(\widetilde{\Sigma}_{t^*})} \leq C \varepsilon^2.
\end{equation}

\item Let us now turn to the \textbf{induction step}. Let $j \leq 2N-1$. Let us assume that 
\begin{equation*}
\norm{\fara}^2_{H^{j}(\widetilde{\Sigma}_{t^*})} \leq \varepsilon^2.
\end{equation*}
We will deduce that
\begin{equation*}
\norm{\fara}^2_{H^{j+1}(\widetilde{\Sigma}_{t^*})} \leq C \varepsilon^2.
\end{equation*}
We now notice that we can control the right hand side of the previous estimate~\eqref{eq:mixfinal} by estimate~\eqref{eq:intdeg}. Recall that $s_0 \leq s_1 \leq s_2$. We now add a multiple of 
$$
\int_{\tregio{{s}_1}{{s}_2}\cap \{r \geq 4M - r_{\text{in}}\}}\sum_{a+b = j} \sum_{I \in \iindicess^b} |(\faram_{,I,a})_{\mu\nu}|^2  \de \text{Vol}
$$
to both sides of inequality~\eqref{eq:mixfinal}. We obtain the following estimates:
\begin{align}
&\int_{\widetilde{\Sigma}_{{s}_2}} \qtot \de {\widetilde{\Sigma}_{t^*}} -
\int_{\widetilde{\Sigma}_{{s}_1}} \qtot \de {\widetilde{\Sigma}_{t^*}} + C \int_{\tregio{{s}_1}{{s}_2} }\sum_{a+b = j} \sum_{I \in \iindicess^b} |(\faram_{,I,a})_{\mu\nu}|^2  \de \text{Vol} \nonumber \\ & 
\leq   C \int_{\regio{{s}_1}{{s}_2}\cap \{4M - r_{\text{in}}\leq r \leq 3M\}} \sum_{a+b = j} \sum_{I \in \iindicess^b} |\nabla_\mu (\mqtensor \indices{^\mu _\nu}\redsh^\nu)|\, \de \text{Vol} 
\nonumber \\
&\quad + C \int_{\tregio{{s}_1}{{s}_2} \cap \{r \geq 4M - r_{\text{in}}\} }\sum_{a+b = j} \sum_{I \in \iindicess^b} |(\faram_{,I,a})_{\mu\nu}|^2  \de \text{Vol} \nonumber \\ &
\quad +C \int_{\tregio{{s}_1}{{s}_2}} \sum_{a+b = j} \sum_{I \in \iindicess^b}\left|\otm_2^{(I,a)}\right| \de \text{Vol}\nonumber \\ &
\leq   C  \int_{\regio{{s}_1}{{s}_2}\cap \{r \geq 4M - r_{\text{in}}\}} \sum_{a+b = j} \sum_{I \in \iindicess^b}  Q^{(MW)}_{\mu\nu}[\dot \fara_{,I,a}] T^\mu n^\nu_{\Sigma_{t^*}} \de \text{Vol}\nonumber \\
&\quad +C \int_{\tregio{{s}_1}{{s}_2}} \sum_{a+b = j} \sum_{I \in \iindicess^b}\left|\otm_2^{(I,a)}\right| \de \text{Vol}\label{eq:mix1}\\ &
\leq C \int_{s_1}^{s_2} 
\int_{\Sigma_{s}} \left(\sum_{I \in \iindicess^{j}} Q^{(MW)}_{\mu\nu}[\dot \fara_{,I,0}] T^\nu n_{\Sigma_{t^*}}^\mu \right) \de{\Sigma_{t^*}} \de s \nonumber \\
& \quad +C\int_{\tregio{{s}_1}{{s}_2}} \sum_{a+b = j} \sum_{I \in \iindicess^b}\left|\otm_2^{(I,a)}\right| \de \text{Vol}+C \eps^2 (s_2-s_1)\nonumber\\
&\leq C \int_{s_1}^{s_2}  \left(\int_{\regio{s_0}{s}} (t^*)^{-2}\varepsilon^{3/2} \sum_{I \in \iindicess^{j}}| {\dot \fara_{,I,0}} | |\partial^{\leq j} \fara| d \text{Vol}\right)
\de s \nonumber\\ 
&\quad+ C(s_2-s_1) \int_{\Sigma_{s_0}} \sum_{I\in \iindicess^j}|\dot \fara_{,I,0}|^2 \de \Sigma_{t^*} + C \int_{\tregio{{s}_1}{{s}_2}}\sum_{a+b = j} \sum_{I \in \iindicess^b}  \left|\otm^{(I,a)}_2\right| \, \de \text{Vol}+C \eps^2(s_2-s_1).\nonumber
\end{align}
In order to obtain this chain of inequalities, we used the In the fact that, away from the event horizon, under the bootstrap assumptions, derivatives in the directions of Killing vector fields control all derivatives (see Proposition~\ref{prop:liecontrol}):
\begin{equation*}
	|\partial^{\leq j} \fara| \leq C \sum_{I\in \iindicess^{\leq j}} |\dot \fara_{,I,0}|^2 \qquad \text{ if } r \geq 4M - r_{\text{in}}.
\end{equation*}
Hence, in our case,
\begin{equation*}
\begin{aligned}
	 \int_{\Sigma_s \cap \{r \geq 4M - r_{\text{in}}\}} \sum_{a+b = j} \sum_{I \in \iindicess^b} Q^{(MW)}_{\mu\nu}[\dot \fara_{,I,a}] T^\mu n^\nu_{\Sigma_{t^*}} \de \Sigma_{t^*} \\
	 \leq C \int_{\Sigma_s \cap \{r \geq 4M - r_{\text{in}}\}} \sum_{I \in \iindicess^j} Q^{(MW)}_{\mu\nu}[\dot \fara_{,I,0}] T^\mu n^\nu_{\Sigma_{t^*}} \de \Sigma_{t^*}+ C \eps^2.
\end{aligned}
\end{equation*}
In the last step of inequality~\eqref{eq:mix1}, we applied Fubini's theorem and the bound~\eqref{eq:intdeg} on fluxes containing commutation only with Killing vector fields in the exterior region (notice that all the corresponding integrals are on the exterior region). We now denote
\begin{align}
F(s) := \int_{\widetilde{\Sigma}_{s}} \qtot \de {\widetilde{\Sigma}_{t^*}}.
\end{align}
In these conditions, we obtain:
\begin{equation*}
\begin{aligned}
&F({s}_2) - F({s}_1) + C_1 \int_{{s}_1}^{{s}_2} F(\bar {s}) d \bar {s} \\
&\quad \leq C({s}_2 - {s}_1) F({s}_0) + C \int_{{s}_1}^{{s}_2}\left(\int_{{s}_0}^{\bar {s}}\varepsilon^{3/2} (\bar{\bar {s}})^{-2} F(\bar {\bar {s}}) d \bar{\bar {s}}\right) d \bar {s} + C \varepsilon^2(s_2-s_1).
\end{aligned}
\end{equation*}
Dividing the previous inequality by ${s}_2 - {s}_1$ and taking the limit ${s}_2 \to {s}_1$, we get the following differential inequality for $F$
\begin{equation}\label{togron}
\begin{aligned}
F'({s}) + C_1 F({s}) \leq  C F({s}_0) + C \varepsilon^{3/2} \int_{{s}_0}^{{s}} (\bar{{s}})^{-2}F(\bar {s})d \bar {s} + C \eps^2,
\end{aligned}
\end{equation}
which is
\begin{equation}
\begin{aligned}
e^{-C_1 {s}}(e^{C_1 {s}} F({s}))' \leq C F({s}_0) + C  \varepsilon^{3/2} \int_{{s}_0}^{{s}} (\bar{{s}})^{-2}F(\bar {s})d \bar {s} + C \eps^2.
\end{aligned}
\end{equation}
We integrate the last display between ${s}_0$ and some $b \geq s_0$ to obtain
\begin{equation}\label{eq:tosup}
\begin{aligned}
&e^{C_1 b} F(b) - e^{C_1 {s}_0}F({s}_0) \leq \int_{{s}_0}^{b}e^{C_1 {s}} \left(C F({s}_0) + C \varepsilon^{3/2} \int_{{s}_0}^{{s}} (\bar{{s}})^{-2}F(\bar {s})d \bar {s} \right)d {s} + C \eps^2 e^{C_1 b}\\&\qquad = (i) +(ii) + (iii).
\end{aligned}
\end{equation}
We now define, for $b \geq s_0$, 
$$S(b):= \sup_{{s} \in [{s}_0, b]}F({s}).$$
We have then
\begin{equation*}
(i) \leq C e^{C_1 b} F(s_0).
\end{equation*}
Furthermore,
\begin{equation*}
(ii) \leq C S(b) \varepsilon^{3/2} \int_{{s}_0}^{b} e^{C_1 s}  \de {s}  \leq C S(b) \varepsilon^{3/2} e^{C_1 b}.
\end{equation*}
Dividing both sides of (\ref{eq:tosup})  by $e^{C_1 b}$, and taking the supremum of the resulting inequality for $s \in [s_0, b]$, we obtain $S(b) - C\eps^{\frac 32} S(b) \leq C \eps^2$
By restricting to small $\varepsilon$, we finally get
\begin{align}\label{enbdone}
F(s) := \int_{\Sigma_{s}} \qtot \de {\Sigma_{t^*}}\leq C \varepsilon^{2},
\end{align}
for all ${s} \geq {s}_0$.
This concludes the induction argument and proves the proposition.
\end{itemize}
\end{proof}

\section{The asymptotic behaviour of spherical averages} \label{sec:charge}
In this Section we prove that the spherical averages of $\sigma$ and $\rho$ decay, given that the charge on $\Sigma_{t_0^*}$ vanishes at spacelike infinity. Recall the definition of $\leo$, $\leo = \{\Omega_1, \Omega_2, \Omega_3\}$, where all $\Omega_i$'s are rotation Killing fields. Furthermore, recall that we defined
\begin{equation*}
Z := r^2 \rho, \qquad W = r^2 \sigma.
\end{equation*}
Also, recall that
\begin{equation*}
\mfarad_{\mu\nu} := -\ellmbi^{-1}(\fara_{\mu\nu}- \ldu \farad_{\mu\nu}).
\end{equation*}
We are going to prove the following Proposition.
\begin{proposition}\label{prop:charge}
Let $b,l,j,k \in \N_{\geq 0}$, $l \geq 1$, $b + j +k \leq l$, let $I \in \mathscr{I}_{\leo}^b$. There exist a small number $\tilde \varepsilon > 0$ and a constant $C >0$ such that the following holds. Let $\fara$ be a smooth solution of the MBI system (\ref{MBI}) on $\mathcal R := \mathcal{R}_{t_0^*}^{t_2^*} \subset S_e$. Assume the bootstrap assumptions $BA\left(\mathcal{R},1,\left\lfloor \frac{l+3}{2}\right\rfloor , \varepsilon \right)$. Assume the bound on the initial energy:
\begin{equation}
\begin{aligned}
\norm{\fara}^2_{H^{l+3}(\Sigma_{t_0^*})} &\leq \varepsilon^2.
\end{aligned}
\end{equation}
Here, $0 < \varepsilon < \tilde \varepsilon$. Assume that the electric and magnetic charge vanish asymptotically on $\Sigma_{t^*_0}$. In other words, assume that there exists a real number $\eta > 0$ such that:
\begin{equation}\label{eq:decchar}
\int_{\mathbb{S}^2} (1-\mu)^{-1}r^2\farad(\lbar, L) \desphere \leq C r^{-\eta}, \text{ on } \Sigma_{t_0^*}, \qquad
\int_{\mathbb{S}^2} (1-\mu)^{-1}r^2{}^\star \!\mfara(\lbar, L) \desphere  \leq C r^{-\eta} \text{ on } \Sigma_{t_0^*}.
\end{equation}
Then we have
\begin{align*}
&\int_{\mathbb{S}^2} \hat W_{,I,j,k}(u,v,\omega) \desphere(\omega)= 0 &\text{ for } (u,v,\omega)\in\regio{t^*_0}{t_2^*},\\
&\left|\int_{\mathbb{S}^2} \hat Z_{,I,j,k}(u,v,\omega) \desphere(\omega)\right| \leq C \varepsilon r^{-1} \tau^{- 2} &\text{ for } (u,v,\omega)\in\regio{t^*_0}{t_2^*}.
\end{align*}
\end{proposition}

\begin{proof} 
We subdivide the proof in two steps. In \textbf{Step 1}, we will analyse the case $b + j + k= 0$, and in \textbf{Step 2}, we will consider the case $b + j + k> 0$.

\subsection*{Step 1}

We recall that the MBI equations can also be written in the form
\begin{equation*}
\left\{
\begin{array}{ll}
\nabla^\mu {}^\star \! \mathcal{M}_{\mu\nu} = 0, \\
\nabla^\mu \farad_{\mu\nu} = 0.
\end{array}
\right.
\end{equation*}
Recall the form of $\mfarad$:
\begin{equation*}
\mfarad_{\mu\nu} := -\ellmbi^{-1}(\fara_{\mu\nu}- \ldu \farad_{\mu\nu}),
\end{equation*}
and the form of $\ellmbi$:
\begin{equation*}
\ellmbi^2:= 1+\lun -\ldu^2.
\end{equation*}
By plugging $\nu = L, \lbar$ in these equations, and integrating on the spheres of constant $r$, under the assumption~(\ref{eq:decchar}), there holds
\begin{equation}\label{sphavg}
\begin{aligned}
\int_{\mathbb{S}^2} (1-\mu)^{-1}r^2\farad(\lbar, L) \desphere = 0, \\
\int_{\mathbb{S}^2} (1-\mu)^{-1}r^2{}^\star \!\mfara(\lbar, L) \desphere = 0.
\end{aligned}
\end{equation}
From the first it follows immediately that the spherical average of $\sigma$ is zero everywhere.

The mean value theorem applied to the function $$(1+t)^{-\frac 1 2}$$
around $t= 0$ assures that for each $x \in [- 1/2, 1/2]$ there exists a $\xi_x \in (-|x|, |x|)$ such that
$$
(1+x)^{- \frac 1 2} = 1 - \frac 1 2 x (1+ \xi_x)^{- \frac 3 2}.
$$
Choosing $x:= \lun -\ldu^2$, we deduce the existence of a $\xi_x \in (-|x|, |x|)$ satisfying the above conditions.
Hence, the equations~(\ref{sphavg}) now give:
\begin{equation}\label{eq:sph}
\sphint{(1- \frac 1 2 (\lun -\ldu^2)(1+\xi_x)^{- \frac 3 2})\rho}=
\sphint{(1- \frac 1 2 (\lun -\ldu^2)(1+\xi_x)^{- \frac 3 2})\ldu \sigma}.
\end{equation}
From the bootstrap assumptions it now follows that
\begin{equation}\label{eq:decspher}
\left|\sphint{\rho} \right| \leq C \varepsilon^{9/4} r^{-5} \tau^{-\frac 5 2}.
\end{equation}
The decay rate follows by looking at the worst term, which is the second term in the previous formula~(\ref{eq:sph}). Indeed, term $\lun$ decays at least like $\tau^{-\frac 3 2} r^{-3}$, and also $\rho$ decays like $\tau^{-1/2} r^{-2}$. This proves the claim.

\subsection*{Step 2}
We immediately notice that, when $b >0$,
\begin{align*}
&\int_{\mathbb{S}^2} \hat W_{,I,j,k}(u,v,\omega) \desphere(\omega)= 0 &\text{ for } (u,v,\omega)\in\regio{t^*_0}{t_1^*},\\
&\int_{\mathbb{S}^2} \hat Z_{,I,j,k}(u,v,\omega) \desphere(\omega) = 0 &\text{ for } (u,v,\omega)\in\regio{t^*_0}{t_1^*}.
\end{align*}
Hence, we restrict to the case $b = 0$. We also notice that the MBI equations (\ref{MBI}) imply the following transport equations:
\begin{align}\label{eq:transpun}
&- \hat \lbar(r^2 \sigma) + (1-\mu)^{-1} r^2 \curl \alphabar= 0, &\hat \lbar(r^2 \rho) + r^2 (1-\mu)^{-1}\dive \alphabar = - r^2{H_{_\Delta}} \indices{^\mu _{\hat \lbar} ^\kappa ^\lambda} \nabla_\mu \fara_{\kappa \lambda},\\ \label{eq:transpd}
&L(r^2 \sigma)+r^2 \curl \alpha = 0, 
&-L(r^2 \rho) + r^2 \dive \alpha = - r^2{H_{_\Delta}} \indices{^\mu _L ^\kappa ^\lambda} \nabla_\mu \fara_{\kappa \lambda}.
\end{align}
Upon differentiation of the equations for $\sigma$, and subsequent integration on $\mathbb{S}^2$ with respect to the form $\desphere$, we obtain the claim for $\hat W_{,I,j,k}$ (notice that the claim is valid for any order of derivatives, as long as the solution is smooth.)

We now wish to obtain the claim for $\hat Z$. We first focus on the case $j \geq 1$. Recall that $k+j = l$. We have, since $[r \snabla, \snabla_{\hat \lbar}] = 0$,
\begin{equation*}
\hat \lbar^j T^k(Z) + \sum_{a+b=j}\hat \lbar^a T^k(r) \cdot r \dive \snabla^b_{\hat \lbar} r \alphabar= -\hat \lbar^j T^k\left(r^2 {H_{_\Delta}} \indices{^\mu _{\hat \lbar} ^\kappa ^\lambda} \nabla_\mu \fara_{\kappa \lambda}\right).
\end{equation*}
Upon integration on $\mathbb{S}^2$, the divergence term disappears, and we need to estimate
\begin{equation}
\left|\int_{\mathbb{S}^2} \hat \lbar^j T^k\left(r^2 {H_{_\Delta}} \indices{^\mu _{\hat \lbar} ^\kappa ^\lambda} \nabla_\mu \fara_{\kappa \lambda}\right)\desphere\right|.
\end{equation}
By the $L^2$ uniform estimates of Equation~(\ref{eq:unifl2}), and the ``unweighted'' Sobolev embedding, Corollary~\ref{cor:unweightedsob}, we know that
\begin{equation*}
\sum_{m=0}^{l}|\partial^m \fara| \leq C \varepsilon,
\end{equation*}
for some constant $C >0$. This remark, together with the bootstrap assumptions $BA\left(\mathcal{R},1,\left\lfloor \frac{l+3}{2}\right\rfloor , \varepsilon \right)$, and the reasoning in Lemma~\ref{lem:bsrho}, imply
\begin{equation*}
\left|\int_{\mathbb{S}^2} \hat \lbar^j T^k\left(r^2 {H_{_\Delta}} \indices{^\mu _{\hat \lbar} ^\kappa ^\lambda} \nabla_\mu \fara_{\kappa \lambda}\right)\desphere\right| \leq C \varepsilon \tau^{-2} r^{-1}.
\end{equation*}
This proves the case $j \geq 1$. If $j =0$, we have, since $L = 2T - (1-\mu)\hat \lbar$,
\begin{equation*}
- 2T (Z) + (1-\mu) \hat \lbar (Z) + r^2 \dive \alpha = - r^2{H_{_\Delta}} \indices{^\mu _L ^\kappa ^\lambda} \nabla_\mu \fara_{\kappa \lambda}.
\end{equation*}
Upon taking $k-1$ time derivatives of the last display, and integrating on $\mathbb{S}^2$, we conclude.
\end{proof}

\section{The \fackip Equations for \texorpdfstring{$\rho$}{rho} and \texorpdfstring{$\sigma$}{sigma}}\label{sec:spinredu}
In this section, we derive the \fackip Equations satisfied by the middle components $\sigma$ and $\rho$. For simplicity, we carry out the corresponding calculation for the linear Maxwell system in Appendix~\ref{app:wavemax}. The proofs in the present section are just a slight extension of those in such appendix, adding the nonlinear terms which arise from the MBI system.

\begin{lemma}[The \fackip Equation satisfied by $\rho$ in the MBI system]\label{lem:wavembir}
Let us suppose that the smooth two-form $\fara$ satisfies the MBI system~\eqref{MBI}. We have the equation:
\begin{equation}\label{eq1.3}\boxed{
- r^{-2} L \lbar (r^2 \rho) +(1-\mu) \slashed{\Delta}  \rho = \text{\bf NL}_1 + \text{\bf NL}_2.}
\end{equation}
Here,
\begin{equation*}
\text{\bf NL}_1 := 2 \frac{1-\mu}{r} \left(-\tensor{{H_{_\Delta}}}{^\mu _\lbar ^\kappa ^\lambda} \nabla_{\mu}\fara_{\kappa\lambda} \right)-2 \frac{1-\mu}{r} \left( \tensor{{H_{_\Delta}}}{^\mu _L ^\kappa ^\lambda} \nabla_{\mu}\fara_{\kappa\lambda}\right),
\end{equation*}
and
\begin{equation*}
\begin{aligned}
&\textbf{\bf NL}_2:= L^\nu \lbar^\alpha \left(-\nabla_\alpha (\nabla^\mu \fara \indices{_\mu_\nu}) +
 \nabla_\nu \nabla^\mu \fara\indices{_\mu_\alpha} \right) \\ & =
 L^\nu \lbar^\alpha \left(\nabla_\alpha  \left(\tensor{{H_{_\Delta}}}{^\mu _\nu ^\kappa ^\lambda} \nabla_{\mu}\fara_{\kappa\lambda} \right) -
  \nabla_\nu \left(\tensor{{H_{_\Delta}}}{^\mu _\alpha ^\kappa ^\lambda} \nabla_{\mu}\fara_{\kappa\lambda} \right)  \right).
 \end{aligned}
\end{equation*}
\end{lemma}

\begin{proof}
As in the proof for the linear Maxwell Equations (Lemma~\ref{lem:wavemaxwell}), we write
\begin{equation*}
\begin{aligned}
&\nabla_\mu \nabla^\mu (L^\nu \lbar^\alpha \fara_{\nu \alpha}) =  \\ &
(\nabla_\mu \nabla^\mu L^\nu) \lbar^\alpha \fara_{\nu \alpha} + (\nabla^\mu L^\nu)(\nabla_\mu \lbar^\alpha) \fara_{\nu \alpha} + (\nabla^\mu L^\nu) \lbar^\alpha \nabla_\mu \fara_{\nu \alpha}\\ & +
 (\nabla_\mu L^\nu)(\nabla^\mu \lbar^\alpha) \fara_{\nu \alpha} + L^\nu (\nabla_\mu \nabla^\mu \lbar^\alpha) \fara_{\nu \alpha} + L^\nu(\nabla^\mu \lbar^\alpha) \nabla_\mu \fara_{\nu \alpha}+\\ &
 (\nabla_\mu L^\nu) \lbar^\alpha \nabla^\mu\fara_{\nu \alpha} + L^\nu (\nabla_\mu  \lbar^\alpha) \nabla^\mu \fara_{\nu \alpha} + L^\nu \lbar^\alpha \nabla^\mu\nabla_\mu \fara_{\nu \alpha} =\\ &
 (a) +(b)+(c)+\\ &
 (d) +(e) +(f) +\\ &
 (g)+ (h) + (i).
 \end{aligned}
\end{equation*}
Comparing with the proof of Lemma~\ref{lem:wavemaxwell}, the additional contribution we get in this case comes from the terms $(c)+(f)+(g)+(h)$. In addition to the terms we had in the linear Maxwell case, we obtain here the nonlinear terms
\begin{equation*}
\text{\bf NL}_1 := 2 \frac{1-\mu}{r} \left(-\tensor{{H_{_\Delta}}}{^\mu _\lbar ^\kappa ^\lambda} \nabla_{\mu}\fara_{\kappa\lambda} \right)-2 \frac{1-\mu}{r} \left( \tensor{{H_{_\Delta}}}{^\mu _L ^\kappa ^\lambda} \nabla_{\mu}\fara_{\kappa\lambda}\right).
\end{equation*}
Furthermore, $(i)$ gives additional second order terms:
\begin{equation*}
\begin{aligned}
&\textbf{\bf NL}_2:= L^\nu \lbar^\alpha \left(-\nabla_\alpha (\nabla^\mu \fara \indices{_\mu_\nu}) +
 \nabla_\nu \nabla^\mu \fara\indices{_\mu_\alpha} \right) \\ & =
 L^\nu \lbar^\alpha \left(\nabla_\alpha  \left(\tensor{{H_{_\Delta}}}{^\mu _\nu ^\kappa ^\lambda} \nabla_{\mu}\fara_{\kappa\lambda} \right) -
  \nabla_\nu \left(\tensor{{H_{_\Delta}}}{^\mu _\alpha ^\kappa ^\lambda} \nabla_{\mu}\fara_{\kappa\lambda} \right)  \right).
 \end{aligned}
\end{equation*}
The equation satisfied by $\rho$ is then, analogously to~(\ref{eq1.2}),
\begin{equation}
- r^{-2} L \lbar (r^2 \rho) +(1-\mu) \slashed{\Delta}  \rho = \text{\bf NL}_1 + \text{\bf NL}_2.
\end{equation}
\end{proof}
We now turn to the proof of a similar Equation for $\sigma$.
\begin{lemma}[The \fackip Equation satisfied by $\sigma$ in the MBI system]\label{lem:wavembis}
Let us suppose that the smooth two-form $\fara$ satisfies the MBI Equation~\eqref{MBI}. We have the equation:
\begin{equation}\label{eq:wavesigma}\boxed{
- r^{-2} L \lbar (r^2 \sigma) + (1-\mu)\slashed{\Delta}  \sigma = \text{\bf NL}_3.}
\end{equation}
Here,
\begin{equation}\label{eq:nl3form}
\text{\bf NL}_3 := L^\nu \lbar^\alpha \nabla^\mu U_{\alpha \mu \nu},
\end{equation}
where the tensor field $U$ is defined by
\begin{equation}\label{eq:wform}
W_{\eta \zeta \delta} := -\varepsilon\indices{_\eta ^\nu _\zeta _\delta} \tensor{{H_{_\Delta}}}{^\mu _\nu ^\kappa ^\lambda} \nabla_{\mu}\fara_{\kappa\lambda}.
\end{equation}
\end{lemma}

\begin{proof}
We again recall:
\begin{align}
& \nabla_{[\mu} \fara_{\kappa \lambda]}= 0, \label{eq2.1}\\ &
\nabla^\lambda \fara_{\lambda \nu} + \tensor{{H_{_\Delta}}}{^\mu _\nu ^\kappa ^\lambda} \nabla_{\mu}\fara_{\kappa\lambda}= 0. \label{eq2.2}
\end{align}
To derive an equation for $\sigma$, we seek to perform our calculations using Hodge duality. Remember that Equation~\eqref{eq2.1} is equivalent to 
\begin{equation*}
\nabla^\lambda \farad_{\lambda \nu} = 0.
\end{equation*}
Also, the properties of the Hodge dual give:
$$
{}^\star\! \farad = - \fara.
$$
Equation~\eqref{eq2.2} is then equivalent to
\begin{equation}
-\frac 1 2 \varepsilon\indices{^\lambda _\nu ^{\alpha \beta}} \nabla_\lambda \farad_{\alpha \beta}+ \tensor{{H_{_\Delta}}}{^\mu _\nu ^\kappa ^\lambda} \nabla_{\mu}\fara_{\kappa\lambda}= 0,
\end{equation}
and, contracting with $\varepsilon\indices{_\eta ^\nu _\zeta _\delta}$, we obtain
\begin{equation}
\nabla_{[\eta} \farad_{\zeta \delta]}+\varepsilon\indices{_\eta ^\nu _\zeta _\delta} \tensor{{H_{_\Delta}}}{^\mu _\nu ^\kappa ^\lambda} \nabla_{\mu}\fara_{\kappa\lambda}= 0.
\end{equation}
Let us define
\begin{equation}
U_{\eta \zeta \delta}:= -\varepsilon\indices{_\eta ^\nu _\zeta _\delta} \tensor{{H_{_\Delta}}}{^\mu _\nu ^\kappa ^\lambda} \nabla_{\mu}\fara_{\kappa\lambda}.
\end{equation}
We therefore obtain
\begin{equation*}
\begin{aligned}
&0 = \nabla_\alpha (\nabla^\mu \farad \indices{_\mu_\nu}) = 
- \rie\indices{_\nu_\beta^\mu_\alpha} \farad \indices{_\mu^\beta} + \nabla^\mu \nabla_\alpha \farad \indices{_\mu_\nu} \stackrel{(*)}{=} \\ &
- \rie\indices{_\nu_\beta^\mu_\alpha} \farad \indices{_\mu^\beta} - \nabla^\mu \nabla_\nu \farad\indices{_\alpha_\mu}- \nabla^\mu \nabla_\mu \farad\indices{_\nu_\alpha}+ \nabla^\mu U_{\alpha \mu \nu} \stackrel{(*)}{=} \\ &
- \rie\indices{_\nu_\beta^\mu_\alpha} \farad \indices{_\mu^\beta} +
\rie\indices{_\alpha_\beta^\mu_\nu}\farad\indices{_\mu^\beta}
+  \nabla_\nu \nabla^\mu \farad\indices{_\mu_\alpha}- \nabla^\mu \nabla_\mu \farad\indices{_\nu_\alpha}+ \nabla^\mu U_{\alpha \mu \nu} \\ & =
2 \rie\indices{_\alpha_\beta^\mu_\nu}\farad\indices{_\mu^\beta}+  \nabla_\nu \nabla^\mu \farad\indices{_\mu_\alpha}- \nabla^\mu \nabla_\mu \farad\indices{_\nu_\alpha}+ \nabla^\mu U_{\alpha \mu \nu} \\ & =
2 \rie\indices{_\alpha_\beta^\mu_\nu}\farad\indices{_\mu^\beta}- \nabla^\mu \nabla_\mu \farad\indices{_\nu_\alpha}+ \nabla^\mu U_{\alpha \mu \nu} .
\end{aligned}
\end{equation*}
We again calculate
\begin{equation*}
\nabla_\mu \nabla^\mu (L^\nu \lbar^\alpha \fara_{\nu \alpha}).
\end{equation*}
The calculation goes through exactly the same as in Lemma~\ref{lem:wavemaxwell}, except for term $(i)$. In that case, we obtain an additional term on the right hand side, which is 
\begin{equation*}
\text{\bf NL}_3 := L^\nu \lbar^\alpha \nabla^\mu U_{\alpha \mu \nu}.
\end{equation*}
All in all, we obtain the claim:
\begin{equation}
- r^{-2} L \lbar (r^2 \sigma) +(1-\mu) \slashed{\Delta}  \sigma = \text{\bf NL}_3.
\end{equation}
\end{proof}

\section{Commutation of the \fackip Equations}\label{sec:commutrw}
In this section, we derive commuted versions of Equations (\ref{eq1.3}) and (\ref{eq:wavesigma}). Let us first recall a shorthand notation for the derivatives of the null components of the field, in Definition~\ref{def:dots}.
Let $b, k, j \in \N_{\geq 0}$, let $I \in \mathscr{I}^b_{\leo}$, and let $g:\mathcal{S} \to \R$ be a smooth function. Let us recall the definition of $\hat L:= (1-\mu)^{-1}\lbar$. Then, we recall
\begin{equation}\label{eq:hatdef}
\hat g_{,I,j,k} := \partial^I_{\leo}((\hat\lbar)^j (T)^k g).
\end{equation}
Furthermore, we recall
\begin{equation}\label{eq:dotdef}
\dot g_{,I,j,k} := \partial^I_{\leo}((\lbar)^j (T)^k g).
\end{equation}

\begin{lemma}\label{lem:commut} Let
\begin{equation*}
Z := r^2 \rho, \qquad W:= r^2 \sigma, \qquad R_\rho :=- r^2 (\nonl{1}+ \nonl{2}), \qquad R_\sigma := -r^2 \nonl{3}.
\end{equation*}
If $\rho$ (resp. $\sigma$) satisfy Equation (\ref{eq1.3}) (resp. Equation (\ref{eq:wavesigma})), then we have the following equations for 
$Z$ and $W$:
\begin{equation}\label{eq:ZW}
L \lbar Z - (1-\mu) \sdelta Z = R_\rho, \qquad L \lbar W - (1-\mu) \sdelta W = R_\sigma.
\end{equation}
We have the Equations for $\dot Z$ and $\dot W$, when $j = 0$:
\begin{align}\label{eq:zdot}
L \lbar \dot Z_{,I,0,k} - (1-\mu) \sdelta \dot Z_{,I,0,k} &= \partial^I_{\leo} T^k (R_\rho), \\ \label{eq:wdot}
L \lbar \dot W_{,I,0,k} - (1-\mu) \sdelta \dot W_{,I,0,k} &= \partial^I_{\leo} T^k (R_\sigma).
\end{align}
Furthermore, we have the following Equations for $\hat Z$ and $\hat W$, for $j \geq 1$:
\begin{align}\label{eq:rcommuted}
 L( \hat Z_{,I,j,k}) +j\frac{2M}{r^2} \hat Z_{,I,j,k} - {\hat \lbar}^{j-1}(\sdelta \hat Z_{,I,0,k})-\partial_{\leo}^I T^k {\hat \lbar}^{j-1}((1-\mu)^{-1}R_\rho) + \sum_{i = 1}^{j-1} f_{i}^{(j-1)} \hat Z_{,I,i,k} = 0,\\ \label{eq:scommuted}
 L( \hat W_{,I,j,k}) +j \frac{2M}{r^2} \dot W_{,I,j,k} - {\hat \lbar}^{j-1}(\sdelta \dot W_{,I,0,k})-\partial_{\leo}^I T^k  {\hat \lbar}^{j-1}((1-\mu)^{-1}R_\sigma) + \sum_{i = 1}^{j-1} f_{i}^{(j-1)} \hat W_{,I,i,k}= 0.
\end{align}
Here, for each $j \in \N_{\geq 0}$,$i \in \{1, \ldots, j\}$, $f^{(j)}_i$ is a smooth and bounded function depending only on $r$. Furthermore, we have that there exist constants $c^{(j)} > 0$ such that the following bound holds on $\mathcal{S}_e$:
\begin{equation}\label{eq:fbound}
|f^{(j-1)}_i| \leq c^{(j-1)} r^{-3}, \quad \text{all } i \in \{0, \ldots, j-1 \}.
\end{equation}
Furthermore, the equations for $\dot Z_{,I,j,k}:= \lbar^j \dot Z_{,I,0,k}$, $\dot W_{,I,j,k}:= \lbar^j \dot W_{,I,0,k}$ are as follows:
\begin{equation}\label{eq:lbarcommut}
\begin{aligned}
L \lbar \dot Z_{,I,j,k}- \sum_{a+b = j} \lbar^a \left(\frac{1-\mu}{r^2} \right) \sdelta_{\mathbb{S}^2} \lbar^b \dot Z_{,I,0,k} = \partial^I_{\leo} T^k \lbar^j R_\rho,\\
L \lbar \dot W_{,I,j,k}- \sum_{a+b = j} \lbar^a \left(\frac{1-\mu}{r^2} \right) \sdelta_{\mathbb{S}^2} \lbar^b \dot W_{,I,0,k} = \partial^I_{\leo} T^k \lbar^j R_\sigma.
\end{aligned}
\end{equation}
\end{lemma}

\begin{proof}
Equations (\ref{eq:zdot}) and (\ref{eq:wdot}) follow immediately because the $\Omega_i$'s and $T$ are Killing fields.

The second part of the Lemma can be easily proved by induction. We focus on the equation satisfied by $Z$, the other equation being analogous, and we furthermore restrict to the case $I = k = 0$ ($I$ and $k$ correspond to Killing vector fields which leave the structure of the equation unvaried).

The case $j = 1$ is evident (we adopt the convention that a sum whose upper limit is lower than the lower limit is the empty sum). Let us assume now that the conclusion holds for $j \geq 1$, and derive it for $j+1$. We take the $\lbar$ derivative of (\ref{eq:rcommuted}), and subsequently multiply both sides by $(1-\mu)^{-1}$. We have
\begin{align*}
&0 = (1-\mu)^{-1}\lbar L( \hat Z_{,0,j,0}) + (1-\mu)^{-1}j \lbar\left( \frac{2M}{r^2}\right) \hat Z_{,0,j,0} \\ &
\quad + j (1-\mu)^{-1} \frac{2M}{r^2} \lbar\left(\hat Z_{,0,j,0}\right)- {\hat \lbar}^{j}(\sdelta \hat Z_{,0,0,0})-{\hat \lbar}^{j}((1-\mu)^{-1}R_\rho) + \sum_{i = 1}^{j} g_{i}^{(j)} \hat Z_{,0,i,0}\\ &
= L \left((1-\mu)^{-1}\lbar( \hat Z_{,0,j,0})\right) + j {\hat \lbar}\left( \frac{2M}{r^2}\right) \hat Z_{,0,j,0} 
+ j \frac{2M}{r^2} {\hat \lbar}\left(\hat Z_{,0,j,0}\right)+ \frac{2M}{r^2} {\hat \lbar}\left(\hat Z_{,0,j,0}\right)\\ &
\quad - {\hat \lbar}^{j}(\sdelta \hat Z_{,0,0,0})-{\hat \lbar}^{j}((1-\mu)^{-1}R_\rho) + \sum_{i = 1}^{j} g_{i}^{(j)} \hat Z_{,0,i,0}\\ &
=  L (\hat Z_{,0,j+1,0} )
+ (j+1) \frac{2M}{r^2} \hat Z_{,0,j+1,0}- {\hat \lbar}^{j}(\sdelta \hat Z_{,0,0,0})-{\hat \lbar}^{j}((1-\mu)^{-1}R_\rho) + \sum_{i = 1}^{j} f_{i}^{(j)} \hat Z_{,0,i,0},
\end{align*}
with the appropriate definition of $g_i^{(j)}$ and $f_i^{(j)}$. It is evident that the $f^{(j)}_i$'s satisfy the bound (\ref{eq:fbound}). It is then straightforward to extend the claim to the case $I, k \neq 0$. The same reasoning holds for $\hat W$.

Finally, the proof of relations (\ref{eq:lbarcommut}) is straightforward.
\end{proof}

\section{Morawetz estimates on \texorpdfstring{$Z$}{Z} and \texorpdfstring{$W$}{W}: first step}\label{sec:premora}

Following~\cite{linearized}, we wish to obtain Morawetz estimates for solutions to the \fackip Equation, as well as a hierarchy of $r^p$-weighted estimates. We emphasize that, in this section, we will \emph{not} bound the nonlinear part in the right hand side of the \fackip Equations. We will prove estimates on the nonlinear right hand sides in Section~\ref{sec:nonlstruct}.

\subsection{Energy conservation for $Z$ and $W$}

\begin{lemma}\label{lem:encons}
Let $k \in \N_{\geq 0}$, and let $I \in \mathscr{I}^b_{\leo}$. There exists a constant $C>0$ such that the following holds. Suppose that $\dot Z_{I,0,k}$ and $\dot W_{I,0,k}$ are solutions to the following:

\begin{align}\label{eq:zdot1}
L \lbar \dot Z_{,I,0,k} - (1-\mu) \sdelta \dot Z_{,I,0,k} &= \partial^I_{\leo} T^k (R_\rho), \\ \label{eq:wdot1}
L \lbar \dot W_{,I,0,k} - (1-\mu) \sdelta \dot W_{,I,0,k} &= \partial^I_{\leo} T^k (R_\sigma)
\end{align}
on the region $\regio{}{} := \regio{t_1^*}{t_2^*}$. Suppose also that $b \geq 1$.
Consider
 $\tilde u \geq t_2^* - 2M \log M$, and recall the notation for the outgoing cones in $(u,v,\theta, \varphi)$ coordinates (see~\ref{sec:not:regfol}):
 $$
 \overline{C}_{\tilde u} = \{u = \tilde u\}.
 $$ 
 In these conditions, there holds\footnote{Note that the condition $\tilde u \geq t_2^* - 2M \log M$ implies that the surface $\{u = \tilde u \} \cap \mathcal{R}$ lies in the region $\{r \leq 3M\}$.}
\begin{align}
&\int_{\Sigma_{t_2^*}\cap\{u \leq  \tilde u\}}(|L \dot Z_{,I,0,k}|^2+(1-\mu)^{-1}|\lbar \dot Z_{,I,0,k}|^2+|\snabla \dot Z_{,I,0,k}|^2 )r^{-2}\de \Sigma_{t^*} \nonumber \\ &
+ \int_{\conplus_{\tilde u} \cap \regio{t_1^*}{t_2^*}} |L \dot Z_{,I, 0,k}|^2 + (1-\mu)|\snabla \dot Z_{,I, 0,k}|^2 \desphere \de v \nonumber \\ &
- C \int_{\Sigma_{t_1^*}} (|L \dot Z_{,I,0,k}|^2+(1-\mu)^{-1}|\lbar \dot Z_{,I,0,k}|^2+|\snabla \dot Z_{,I,0,k}|^2 )r^{-2}\de \Sigma_{t^*} \nonumber \\ &
\leq C \int_{\regio{t_1^*}{t_2^*} \cap \{ 2M \leq r \leq 4M\}} |\dot Z_{,I,0,k}| |\partial^I_{\leo} T^{k+1} (R_\rho)| \desphere \de u \de v \label{eq:enbdur} \\
&+C \int_{\regio{t_1^*}{t_2^*} \cap \{ r \geq 4M\}} |\dot Z_{,I,0,k+1}| |\partial^I_{\leo} T^{k} (R_\rho)| \desphere \de u \de v \nonumber \\ &
+ C \int_{\regio{t_1^*}{t_2^*} \cap \{2M \leq  r  \leq 3 M\}}(|\partial^I_{\leo}  {\hat \lbar} T^{k} (R_\rho)|^2 + |\partial^I_{\leo}  T^{k} ((1-\mu)^{-1}R_\rho)|^2)\de \text{Vol} \nonumber  \\ &
+ C \int_{\Sigma_{t_2^*} \cap \{2M \leq r \leq 4M\}} |\partial^I_{\leo}  T^{k} ((1-\mu)^{-1}R_\rho)|^2 \de \Sigma_{t^*} \nonumber \\
&+ C \int_{\Sigma_{t_1^*}\cap \{2M \leq r \leq 4M\}} |\partial^I_{\leo}  T^{k} ((1-\mu)^{-1}R_\rho)|^2 \de \Sigma_{t^*}.\nonumber 
\end{align}
Furthermore, the same inequality holds with every occurrence of $Z$ replaced by $W$ and every occurrence of $R_\rho$ replaced by $R_\sigma$.
\end{lemma}

\begin{proof}[Proof of Lemma~\ref{lem:encons}]
For simplicity, let $\dot Z := \dot Z_{I,0,k}$, $\dot W := \dot W_{I,0,k}$. We have the following relations:
\begin{equation}\label{eq:toencons}
\begin{aligned}
&\frac 1 2 \lbar |L \dot Z|^2 + \frac 1 2 L|\lbar \dot Z|^2 + \frac 1 2 (L + \lbar) \left((1-\mu)|\snabla \dot Z|^2 \right)\\
& \hspace{100pt} \stackrel{\mathbb{S}^2}{=} (T\dot Z)\partial^I_{\leo} T^k (R_\rho) \stackrel{\mathbb{S}^2}{=} T \left( \dot Z\partial^I_{\leo} T^k (R_\rho)\right) - \dot Z\partial^I_{\leo} T^{k +1}(R_\rho),\\
&\frac 1 2 \lbar |L \dot W|^2 + \frac 1 2 L|\lbar \dot W|^2 + \frac 1 2 (L + \lbar) \left((1-\mu)|\snabla \dot W|^2 \right) \\
& \hspace{100pt} \stackrel{\mathbb{S}^2}{=} (T\dot W)\partial^I_{\leo} T^k (R_\sigma) \stackrel{\mathbb{S}^2}{=} T \left( \dot W\partial^I_{\leo} T^k (R_\sigma)\right) - \dot W\partial^I_{\leo} T^{k +1}(R_\sigma).
\end{aligned}
\end{equation}
Here, as usual $\stackrel{\mathbb{S}^2}{=}$ meant that the equality is valid upon integration on the sphere $\mathbb{S}^2$ with respect to the form $\desphere$.

We now proceed to integrate (\ref{eq:toencons}) on $\regio{t_1^*}{t_2^*} \cap \{u \leq \tilde u\}$ with respect to the form $\desphere \de u \de v$. Let us focus on the case of $Z$, the reasoning for $W$ being analogous. Recall that we set, for simplicity, just for this proof, $\regio{}{} := \regio{t_1^*}{t_2^*}$.

We have that 
\begin{align*}
&\int_{\regio{}{} \cap\{u\leq \tilde u\}} \lbar |L \dot Z|^2 + L|\lbar \dot Z|^2 + (L + \lbar) \left((1-\mu)|\snabla \dot Z| \right)\desphere \de u \de v \\ &\geq \int_{\Sigma_{t_2^*}\cap\{u\leq \tilde u\}}(|L \dot Z|^2+(1-\mu)^{-1}|\lbar \dot Z|^2+|\snabla \dot Z|^2 )r^{-2}\de \Sigma_{t^*} \\ &
+ \int_{\conplus_{\tilde u} \cap \regio{}{}} |L \dot Z_{,I, 0,k}|^2 + (1-\mu)|\snabla \dot Z_{,I, 0,k}|^2 \desphere \de v\\ &
- C \int_{\Sigma_{t_1^*}} (|L \dot Z|^2+(1-\mu)^{-1}|\lbar \dot Z|^2+|\snabla \dot Z|^2 )r^{-2}\de \Sigma_{t^*}
\end{align*}

\begin{remark}
Note that, in the previous display, we are not discarding the boundary terms. We therefore obtain two boundary terms arising from two space-like surfaces (at $\Sigma_{t_1^*} $ and at $\Sigma_{t_2^*}$), and one arising from a null piece at $u = \tilde u$.
\end{remark}
On the other hand,
\begin{align*}
&\int_{\regio{}{}\cap \{u \leq \tilde u\}}(T\dot Z)\partial^I_{\leo} T^k (R_\rho) \desphere \de u \de v\\ 
&\qquad \leq \underbrace{\int_{\regio{}{}\cap \{u \leq \tilde u\}\cap \{2M \leq r \leq 4M \} }\left(T \left( \dot Z\partial^I_{\leo} T^k (R_\rho)\right) - \dot Z\partial^I_{\leo} T^{k +1}(R_\rho)\right)\desphere \de u \de v}_{(i)} \\
&\qquad + \underbrace{\int_{\regio{}{}\cap \{u \leq \tilde u\}\cap \{r \geq 4M \}} (T\dot Z)\partial^I_{\leo} T^k (R_\rho) \desphere\de u \de v}_{(ii)}.
\end{align*}
Now, we have, from the Cauchy--Schwarz inequality and the Poincar\'e inequality (Lemma~\ref{lem:poinca} in the Appendix, note that $b \geq 1$),
\begin{align*}
&(i) \leq C\int_{\Sigma_{t_2^*} \cap \{2M \leq r \leq 4M\}}(1-\mu)^{-1} |\dot Z| |\partial_{\leo}^I T^k R_\rho| \de \Sigma_{t^*}\\ &
+C \int_{\Sigma_{t_1^*} \cap \{2M \leq r \leq 4M\}}(1-\mu)^{-1} |\dot Z| |\partial_{\leo}^I T^k R_\rho| \de \Sigma_{t^*}
\\
&+\int_{\conplus_{\tilde u} \cap \regio{}{}}|\dot Z| |\partial^I_{\leo} T^l R_\rho| \desphere \de v\\ &
\leq \frac 1 4 \int_{\Sigma_{t_2^*} \cap \{2M \leq r \leq 4M\}}|\snabla \dot Z| r^{-2}\de \Sigma_{t^*} +\frac 1 4 \int_{\Sigma_{t_1^*} \cap \{2M \leq r \leq 4M\}}|\snabla \dot Z| r^{-2}\de \Sigma_{t^*} \\
&+ \frac 1 4 \int_{\conplus_{\tilde u} \cap \regio{}{}} (1-\mu) |\dot Z|^2 \desphere \de v\\ &
+C\int_{\Sigma_{t_2^*} \cap \{2M \leq r \leq 4M\}} |\partial^I_{\leo}  T^{k} ((1-\mu)^{-1}R_\rho)|^2 \de \Sigma_{t^*} \\
&+ C \int_{\Sigma_{t_1^*}\cap \{2M \leq r \leq 4M\}} |\partial^I_{\leo}  T^{k} ((1-\mu)^{-1}R_\rho)|^2 \de \Sigma_{t^*}\\ &
+ C  \int_{\conplus_{\tilde u} \cap \regio{}{}} (1-\mu)^{-1} |\partial^I_{\leo} T^k R_\rho|^2\desphere \de v.
\end{align*}
Finally, the fundamental theorem of calculus implies
\begin{align*}
& \int_{\conplus_{\tilde u} \cap \regio{}{}} (1-\mu)^{-1} |\partial^I_{\leo} T^k R_\rho|^2 \desphere \de v \\ &
 \leq  C\int_{\Sigma_{t_2^*} \cap \{2M \leq r \leq 4M\}} |\partial^I_{\leo}  T^{k} ((1-\mu)^{-1}R_\rho)|^2 \de \Sigma_{t^*} \\ &
+ C \int_{\regio{t_1^*}{t_2^*} \cap \{2M \leq  r  \leq 3 M\}}(|\partial^I_{\leo}  {\hat \lbar} T^{k} (R_\rho)|^2 + |\partial^I_{\leo}  T^{k} ((1-\mu)^{-1}R_\rho)|^2)\de \text{Vol}.
\end{align*}
To conclude, $(ii)$ is estimated trivially. We obtain the claim.
\end{proof}

\subsection{Morawetz estimate for $\dot Z$, $\dot W$}
In this section we prove a Morawetz estimate for the quantities $\dot Z_{,I,0,k}$ and $\dot W_{,I,0,k}$ which is degenerate at the event horizon $\mathcal{H}^+$ and at the photon sphere $r = 3M$.

\begin{proposition}\label{prop:mora}
Let $k \in \N_{\geq 0}$, and let $I \in \mathscr{I}^b_{\leo}$, with $b \geq 1$.
There exists a constant $C > 0$ such that the following holds. Suppose that $\dot Z_{,I,0,k}$ and $\dot W_{,I,0,k}$ are solutions to the following:
\begin{align}\label{eq:zdot1mor}
L \lbar \dot Z_{,I,0,k} - (1-\mu) \sdelta \dot Z_{,I,0,k} &= \partial^I_{\leo} T^k (R_\rho), \\ \label{eq:wdot1mor}
L \lbar \dot W_{,I,0,k} - (1-\mu) \sdelta \dot W_{,I,0,k} &= \partial^I_{\leo} T^k (R_\sigma)
\end{align}
in the region $\mathcal{R} := \regio{t_1^*}{t_2^*}$.
For simplicity, let $\dot Z := \dot Z_{I,0,k}$, $\dot W := \dot W_{I,0,k}$. Then, there holds
\begin{align}
 &\int_{\regio{t_1^*}{t_2^*} } 
\left\{\frac 1 {r^2}|L \dot Z - \lbar \dot Z|^2+ \frac{(r-3M)^2}{r^3}\left(|\snabla \dot Z|^2+ \frac 1 r |L \dot Z + \lbar \dot Z|^2\right)+\frac 1 {r^3} |\dot Z|^2\right\} r^{-2}\de \text{Vol}\nonumber \\ &
\leq C	\int_{\regio{t_1^*}{t_2^*}  } (r^2(1-\mu)^{-1} |\partial^I_{\leo} T^k (R_\rho)|^2 + r^{-1}(1-\mu) |\partial^I_{\leo} T^k (R_\rho)|^2 )\desphere\de u \de v \nonumber\\ &
+C \int_{\regio{t_1^*}{t_2^*} \cap \{ 2M \leq r \leq 4M\}} |\dot Z_{,I,0,k}| |\partial^I_{\leo} T^{k+1} (R_\rho)| \desphere\de u \de v \nonumber \\
&+C \int_{\regio{t_1^*}{t_2^*} \cap \{ r \geq 4M\}} |\dot Z_{,I,0,k+1}| |\partial^I_{\leo} T^{k} (R_\rho)| \desphere\de u \de v\nonumber\\ &
+ C \int_{\regio{t_1^*}{t_2^*} \cap \{2M \leq  r  \leq 3 M\}}(|\partial^I_{\leo}  {\hat \lbar} T^{k} (R_\rho)|^2 + |\partial^I_{\leo}  T^{k} ((1-\mu)^{-1}R_\rho)|^2)\de \text{Vol} \label{fullen} \\ &
+ C \int_{\Sigma_{t_2^*} \cap \{2M \leq r \leq 4M\}} |\partial^I_{\leo}  T^{k} ((1-\mu)^{-1}R_\rho)|^2 \de \Sigma_{t^*} \nonumber\\
&+ C \int_{\Sigma_{t_1^*}\cap \{2M \leq r \leq 4M\}} |\partial^I_{\leo}  T^{k} ((1-\mu)^{-1}R_\rho)|^2 \de \Sigma_{t^*}\nonumber\\ &
+  C \int_{\Sigma_{t_1^*}} (|L \dot Z|^2+(1-\mu)^{-1}|\lbar \dot Z|^2+|\snabla \dot Z|^2 )r^{-2}\de \Sigma_{t^*}.\nonumber
\end{align}
We furthermore have the same inequality for $\dot W$, with $R_\sigma$ replacing $R_\rho$.
\end{proposition}
\begin{proof}
Let us recall that, according to the definition in Section~\ref{sec:not:geo}, $\slashed{g}$ is the restriction of the ambient metric $g$ to the spheres of constant $r$-coordinate. Furthermore, recall the fact that $\snabla_L \gbar = 0, \snabla_\lbar \gbar = 0$ (see Lemma~\ref{lem:projectcomp} in the Appendix).
We consider the following identities, which can be checked via direct calculation using the previous fact about the covariant derivative of $\slashed g$,
\begin{equation}\label{en1}
\begin{aligned}
	&4 f(r) (L \dot Z - \lbar \dot Z) \partial^I_{\leo} T^k (R_\rho)
	\\ &
	\stackrel{\mathbb{S}^2}{=}(\lbar + L) \left\{f(|L \dot Z|^2-|\lbar \dot Z|^2)\right\} 
	+ (\lbar - L) \left\{f (|L \dot Z|^2+|\lbar \dot Z|^2	- 2 \frac{1-\mu}{r^2} |r \snabla \dot Z|^2 )\right\} \\ &
	+ 2f'(|L \dot Z|^2 + |\lbar \dot Z|^2)-4 \partial_{r^\star}\left(f\frac{1-\mu}{r^2} \right)|r \snabla \dot Z|^2.
\end{aligned}
\end{equation}
We also have
\begin{equation}\label{en2}
\begin{aligned}
	&4f' \dot Z \ \partial^I_{\leo} T^k (R_\rho)  \stackrel{\mathbb{S}^2}{=}(\lbar + L) \left\{f' \dot Z \cdot (\lbar + L) \dot Z \right\}\\ &
	-(\lbar - L) \left(f' \dot Z \cdot (\lbar - L)\dot Z+f''|\dot Z|^2 \right)\\ &
	-2f'''|\dot Z|^2-4f'\lbar \dot Z \cdot L \dot Z + 4 f' \left( \frac{1-\mu}{r^2} |r\snabla \dot Z|^2\right).
\end{aligned}
\end{equation}
Let us now add the previous Equations (\ref{en1}) and (\ref{en2}), to get
\begin{equation}\label{mora}
\begin{aligned}
	&4 f (L \dot Z - \lbar \dot Z) \partial^I_{\leo} T^k (R_\rho) + 4f' \dot Z \ \partial^I_{\leo} T^k (R_\rho) \\ &
	\stackrel{\mathbb{S}^2}{=} 
	(\lbar + L) \left\{f(|L \dot Z|^2-|\lbar \dot Z|^2) + f' \dot Z \cdot (\lbar + L) \dot Z \right\} \\ &
	+ (\lbar - L) \left\{f (|L \dot Z|^2+|\lbar \dot Z|^2	- 2 \frac{1-\mu}{r^2} |r \snabla \dot Z|^2 ) - f' \dot Z \cdot (\lbar - L)\dot Z - f''|\dot Z|^2 \right\}\\ &
	+ 2f'(|L \dot Z - \lbar \dot Z|^2)-4 f \partial_{r^\star}\left(\frac{1-\mu}{r^2} \right)|r \snabla \dot Z|^2 -2f'''|\dot Z|^2.
\end{aligned}
\end{equation}
Let us now choose
\begin{equation}\label{eq:fchoice}
f(r) := \left(1+\frac{M}{r} \right) \left(1-\frac{3M}{r} \right).
\end{equation} 
We now proceed to integrate Equation (\ref{mora}) on the spacetime region $\regio{t_1^*}{t_2^*} \cap \{u \leq \tilde u\}$ with respect to the form $\de u \de v$. The boundary terms relative to $\Sigma_{t_2^*}$ and to the outgoing surface $\conplus_{\tilde u} \cap \regio{t_1^*}{t_2^*}$ can be estimated by means of inequality (\ref{eq:enbdur}).

We now would like to check that the bulk term is positive. 
Due to the fact that $f'(r)>0$ for $r \geq 2M$, we see that the term in $|L \dot Z|^2$ is positive. If $b \geq 1$, by Lemma~\ref{lem:poinca}, we note that for the bulk term (in $\dot Z$ and $\snabla \dot Z$) to be positive, it suffices that there exists a $c > 0$ such that
\begin{equation}\label{posicon}
	-2 \frac{(\frac 2 {r^2} (1-\mu))'}{1-\mu}f - \frac{f'''}{1-\mu} \geq \frac c {r^3},
\end{equation}
for some positive number $c$. Let us calculate, as in~\cite{linearized},
\begin{align*}
 f'&= (1-\mu)\left(\frac{2M}{r^2}+ \frac{6M^2}{r^3}\right), \\
 f'' &= (1-\mu)\partial_r(f') = \frac{2 M \left(-48 M^3+30 M^2 r+M r^2-2 r^3\right)}{r^6},\\
 f''' &= (1-\mu)\partial_r(f'') = \frac{4 M (r-2 M) \left(144 M^3-75 M^2 r-2 M r^2+3 r^3\right)}{r^8}.
\end{align*}
 Multiplying (\ref{posicon}) by $-\frac 1 4 r^3$, we obtain that inequality (\ref{posicon}) is achieved if and only if
\begin{equation*}
	\frac{144 M^4-93 M^3 r-8 M^2 r^2+13 M r^3-2 r^4 }{r^4}< c,
\end{equation*}
which is the case for $r \geq 2M$. 

We therefore obtain the following estimate, making use of the positivity of the angular terms:

\begin{equation}\label{moraerr}
\begin{aligned}
& \int_{\regio{t_1^*}{t_2^*} \cap \{u \leq \tilde u\} } 
\left\{\frac 1 {r^2}|L \dot Z - \lbar \dot Z|^2+ \frac{(r-3M)^2}{r^3}|\snabla \dot Z|^2+\frac 1 {r^3} |\dot Z|^2\right\} r^{-2}\de \text{Vol} \\ &
\leq C	\int_{\regio{t_1^*}{t_2^*} \cap \{u \leq \tilde u\} } (|L \dot Z - \lbar \dot Z| |\partial^I_{\leo} T^k (R_\rho)| + (1-\mu)r^{-2} |\dot Z| \ |\partial^I_{\leo} T^k (R_\rho)| )\desphere\de u \de v\\ &
+C \int_{\regio{t_1^*}{t_2^*} \cap \{ 2M \leq r \leq 4M\}} |\dot Z_{,I,0,k}| |\partial^I_{\leo} T^{k+1} (R_\rho)|\desphere \de u \de v\\
&+C \int_{\regio{t_1^*}{t_2^*} \cap \{ r \geq 4M\}} |\dot Z_{,I,0,k+1}| |\partial^I_{\leo} T^{k} (R_\rho)| \desphere\de u \de v\\ &
+ C \int_{\regio{t_1^*}{t_2^*} \cap \{2M \leq  r  \leq 3 M\}}(|\partial^I_{\leo}  {\hat \lbar} T^{k} (R_\rho)|^2 + |\partial^I_{\leo}  T^{k} ((1-\mu)^{-1}R_\rho)|^2)\de \text{Vol}\\ &
+ C \int_{\Sigma_{t_2^*} \cap \{2M \leq r \leq 4M\}} |\partial^I_{\leo}  T^{k} ((1-\mu)^{-1}R_\rho)|^2 \de \Sigma_{t^*} \\
&+ C \int_{\Sigma_{t_1^*}\cap \{2M \leq r \leq 4M\}} |\partial^I_{\leo}  T^{k} ((1-\mu)^{-1}R_\rho)|^2 \de \Sigma_{t^*}\\ &
+  C \int_{\Sigma_{t_1^*}} (|L \dot Z|^2+(1-\mu)^{-1}|\lbar \dot Z|^2+|\snabla \dot Z|^2 )r^{-2}\de \Sigma_{t^*}.
\end{aligned}
\end{equation}
Notice that the last six terms in the previous display arise from estimating the future boundary terms on $\Sigma_{t_2^*}$ and on $\conplus_{\tilde u}$ using the $L^2$ estimates~\eqref{eq:enbdur}.

We can recover the missing derivative by integrating Equation~\eqref{en1} with a monotonically increasing $f$, which vanishes of third order at $r = 3M$. We obtain:

\begin{align}
& \int_{\regio{t_1^*}{t_2^*} \cap \{u \leq \tilde u\} } 
\left\{\frac 1 {r^2}|L \dot Z - \lbar \dot Z|^2+ \frac{(r-3M)^2}{r^3}\left(|\snabla \dot Z|^2+ \frac 1 r |L \dot Z + \lbar \dot Z|^2\right)+\frac 1 {r^3} |\dot Z|^2\right\} r^{-2}\de \text{Vol} \nonumber \\ &
\leq C	\int_{\regio{t_1^*}{t_2^*} \cap \{u \leq \tilde u\} } (|L \dot Z - \lbar \dot Z| |\partial^I_{\leo} T^k (R_\rho)| + (1-\mu)r^{-2} |\dot Z| \ |\partial^I_{\leo} T^k (R_\rho)| )\desphere\de u \de v \nonumber \\ &
+C \int_{\regio{t_1^*}{t_2^*} \cap \{ 2M \leq r \leq 4M\}} |\dot Z_{,I,0,k}| |\partial^I_{\leo} T^{k+1} (R_\rho)| \desphere\de u \de v \nonumber \\
&+C \int_{\regio{t_1^*}{t_2^*} \cap \{ r \geq 4M\}} |\dot Z_{,I,0,k+1}| |\partial^I_{\leo} T^{k} (R_\rho)| \desphere\de u \de v \nonumber\\ &
+ C \int_{\regio{t_1^*}{t_2^*} \cap \{2M \leq  r  \leq 3 M\}}(|\partial^I_{\leo}  {\hat \lbar} T^{k} (R_\rho)|^2 + |\partial^I_{\leo}  T^{k} ((1-\mu)^{-1}R_\rho)|^2)\de \text{Vol} \label{fullencopy} \\ &
+ C \int_{\Sigma_{t_2^*} \cap \{2M \leq r \leq 4M\}} |\partial^I_{\leo}  T^{k} ((1-\mu)^{-1}R_\rho)|^2 \de \Sigma_{t^*} \nonumber\\
&+ C \int_{\Sigma_{t_1^*}\cap \{2M \leq r \leq 4M\}} |\partial^I_{\leo}  T^{k} ((1-\mu)^{-1}R_\rho)|^2 \de \Sigma_{t^*}\nonumber\\ &
+  C \int_{\Sigma_{t_1^*}} (|L \dot Z|^2+(1-\mu)^{-1}|\lbar \dot Z|^2+|\snabla \dot Z|^2 )r^{-2}\de \Sigma_{t^*}.\nonumber
\end{align}
We first take $\tilde u \to \infty$. We then use the Cauchy--Schwarz inequality on the first term of the right hand side of the last display, absorbing the relevant terms in the left hand side, to obtain
\begin{align*}
 &\int_{\regio{t_1^*}{t_2^*}} 
\left\{\frac 1 {r^2}|L \dot Z - \lbar \dot Z|^2+ \frac{(r-3M)^2}{r^3}\left(|\snabla \dot Z|^2+ \frac 1 r |L \dot Z + \lbar \dot Z|^2\right)+\frac 1 {r^3} |\dot Z|^2\right\} r^{-2}\de \text{Vol} \\ &
\leq C	\int_{\regio{t_1^*}{t_2^*} } (r^2(1-\mu)^{-1} |\partial^I_{\leo} T^k (R_\rho)|^2 + r^{-1}(1-\mu) |\partial^I_{\leo} T^k (R_\rho)|^2 )\desphere\de u \de v\\ &
+C \int_{\regio{t_1^*}{t_2^*} \cap \{ 2M \leq r \leq 4M\}} |\dot Z_{,I,0,k}| |\partial^I_{\leo} T^{k+1} (R_\rho)|\desphere \de u \de v\\
&+C \int_{\regio{t_1^*}{t_2^*} \cap \{ r \geq 4M\}} |\dot Z_{,I,0,k+1}| |\partial^I_{\leo} T^{k} (R_\rho)| \desphere\de u \de v\\ &
+ C \int_{\regio{t_1^*}{t_2^*} \cap \{2M \leq  r  \leq 3 M\}}(|\partial^I_{\leo}  {\hat \lbar} T^{k} (R_\rho)|^2 + |\partial^I_{\leo}  T^{k} ((1-\mu)^{-1}R_\rho)|^2)\de \text{Vol}\\ &
+ C \int_{\Sigma_{t_2^*} \cap \{2M \leq r \leq 4M\}} |\partial^I_{\leo}  T^{k} ((1-\mu)^{-1}R_\rho)|^2 \de \Sigma_{t^*}\\
& + C \int_{\Sigma_{t_1^*}\cap \{2M \leq r \leq 4M\}} |\partial^I_{\leo}  T^{k} ((1-\mu)^{-1}R_\rho)|^2 \de \Sigma_{t^*}\\ &
+  C \int_{\Sigma_{t_1^*}} (|L \dot Z|^2+(1-\mu)^{-1}|\lbar \dot Z|^2+|\snabla \dot Z|^2 )r^{-2}\de \Sigma_{t^*}.
\end{align*}
This is the claim.
\end{proof}

\subsection{The redshift estimate}

Note that the Morawetz estimate~\eqref{fullen} we just proved in the previous section displayed no control over iterated derivatives in the $\hat \lbar$ direction (in fact, we were always considering the case in which $j =0$). In this section, we prove, via a redshift estimate, that we can recover iterated derivatives transversal to the event horizon $\mathcal{H}^+$, and we obtain the estimate~\eqref{eq:redsh} in terms of the bulk term appearing on the RHS of estimate~\eqref{fullen}, plus nonlinear error terms and boundary terms.

\begin{proposition}\label{prop:redsh}
Let $k, h, l \in \N_{\geq 0}$, and let $k + h \leq l$.
There exists a constant $C > 0$ such that the following holds. Suppose that $Z$ and $W$ are solutions to the following:
\begin{align}\label{eq:z2}
L \lbar Z - (1-\mu) \sdelta Z &= R_\rho, \\ \label{eq:w2}
L \lbar W - (1-\mu) \sdelta W &= R_\sigma
\end{align}
in the region $\mathcal{R} := \regio{t_1^*}{t_2^*}$.
Recall now the definition of $\hat Z$ and $\hat W$:
\begin{equation*}
\hat Z_{,I,j,k} := \partial^I_{\leo}((\hat\lbar)^j (T)^k Z), \qquad \hat W_{,I,j,k} := \partial^I_{\leo}((\hat\lbar)^j (T)^k W).
\end{equation*}
In these conditions, we have
\begin{equation}\label{eq:redsh}
\begin{aligned}
&\sum_{\substack{1 \leq |I| + j \leq h\\ I \in \iindiceo^b, b \geq 1, j \geq 1}} \left\{
\int_{\Sigma_{t_2^*}} r^{-2} |\hat Z_{,I, j, k}|^2  r^{-2} \de \Sigma_{t^*} + \int_{\regio{t_1^*}{t_2^*}}  \frac{1-\mu}{r^3}|\hat Z_{,I, j, k}|^2 \desphere \de u \de v \right\}
\\ 
& \leq C \sum_{\substack{1 \leq |J|\leq h\\ J \in \iindiceo^b, b \geq 1}} \int_{\regio{t_1^*}{t_2^*}} (1-\mu) r^{-5}  |\hat Z_{,J,0,k}|^2 \desphere \de u \de v
 +C\sum_{\substack{1 \leq |I| + j \leq h \\ I \in \iindiceo^b, b \geq 1}} \int_{\Sigma_{t_1^*}} r^{-2} |\hat Z_{,I, j, k}|^2 r^{-2} \de \Sigma_{t^*}\\ &
 +C\sum_{\substack{1 \leq |I| + j \leq h \\ I \in \iindiceo^b, b \geq 1}}\int_{\regio{t_1^*}{t_2^*} } (1-\mu)r^{-1} |\partial_{\leo}^I T^k {\hat \lbar}^{j-1}((1-\mu)^{-1}R_\rho)|^2 \desphere\de u \de v.
\end{aligned}
\end{equation}
Furthermore, we have the same inequality for $\hat W$, with $R_\rho$ replaced by $R_\sigma$.
\end{proposition}
\begin{remark}
	The estimate~\eqref{eq:redsh} could as well have been localized around $r = 2M$, without many changes to the argument. Here we formulate it on the whole exterior region in order not to deal with cut-off functions.
\end{remark}
\begin{proof}
To simplify notation, let us set $\hat Z_j := \hat Z_{,I, j, k}$. Let us start by calculating, for $j \geq 1$, using equation~\eqref{eq:rcommuted},
\begin{equation}\label{eq:mastermorared}
\begin{aligned}
& L (r^{-2}(1-\mu)|\hat Z_j|^2) = - 2\frac{(1-\mu)^2}{r^3} |\hat Z_j|^2 \\ &
+r^{-2}\left\{(1-\mu) \frac{2M}{r^2} |\hat Z_j|^2 + 2(1-\mu) \hat Z_j\left[
- j \frac{2M}{r^2} \hat Z_j + {\hat \lbar}^{j-1}(\sdelta \hat Z_0)+ \partial_{\leo}^I T^k {\hat \lbar}^{j-1}((1-\mu)^{-1}R_\rho)-\sum_{i=1}^{j-1} f^{(j-1)}_i \hat Z_i
\right]\right\}\\ &
= -2 \frac{(1-\mu)^2}{r^3} |\hat Z_j|^2-(1-\mu) (2 j-1) \frac{2M}{r^4} |\hat Z_j|^2 +\underbrace{ 2(1-\mu)r^{-2} \hat Z_j 
 {\hat \lbar}^{j-1}(\sdelta \hat Z_0)}_{(i)}\\ &
 + \underbrace{2r^{-2}(1-\mu) \hat Z_j \partial_{\leo}^I T^k {\hat \lbar}^{j-1}((1-\mu)^{-1}R_\rho)}_{(ii)}- \underbrace{2r^{-2}(1-\mu) \hat Z_j\left(\sum_{i=1}^{j-1} f^{(j-1)}_i \hat Z_i\right)}_{(iii)}.
\end{aligned}
\end{equation}
Now, we examine $(i)$, upon integration on the sphere $\mathbb{S}^2$ (note that, if $j=1$, the second term in the second line of the following display is identically $0$):
\begin{align*}
&(i)=2(1-\mu)r^{-2} \hat Z_j  {\hat \lbar}^{j-1}(\sdelta \hat Z_0) = 2 r^{-2} \lbar \hat Z_{j-1} {\hat \lbar}^{j-1}(r^{-2} \sdelta_{\mathbb{S}^2} \hat Z_0)\\ & 
= -2 r^{-4} \snabla_{\lbar} \snabla_{\mathbb{S}^2} \hat Z_{j-1}(\snabla_{\mathbb{S}^2} \hat Z_{j-1}) +2 r^{-2} \sum_{\substack{a+b =j-1 \\ a,b \in \N_{\geq 0}, b \leq j-2}} \lbar \hat Z_{j-1} ({\hat \lbar}^{a}r^{-2}) (\snabla_{\hat \lbar}^b \sdelta_{\mathbb{S}^2} \hat Z_0)\\
&= \frac 1 {r^4} \lbar |\snabla_{\hat \lbar}^{j-1} \snabla_{\mathbb{S}^2} \hat Z_0 |^2 +2 r^{-2} \sum_{\substack{a+b =j-1 \\ a,b \in \N_{\geq 0}, b \leq j-2}} \lbar \hat Z_{j-1} ({\hat \lbar}^{a}r^{-2}) (\snabla_{\hat \lbar}^b \sdelta_{\mathbb{S}^2} \hat Z_0)\\ &
\leq -\lbar \left\{\frac 1 {r^4}  |\snabla_{\hat \lbar}^{j-1} \snabla_{\mathbb{S}^2} \hat Z_0 |^2 \right\} + \lbar( r^{-4}) |\snabla_{\hat \lbar}^{j-1} \snabla_{\mathbb{S}^2} \hat Z_0 |^2 \\
&\qquad + \eta (1-\mu)r^{-3}\sum_{a = 1}^j |\hat Z_a|^2 + C (1-\mu)r^{-6} \sum_{a = 0}^{j-2} |\hat \lbar^a \sdelta_{\mathbb{S}^2}\hat Z_0|^2.
\end{align*}
Here, $\eta$ is a small positive number. Note that we applied the Cauchy--Schwarz inequality in the last term above.

Subsequently, we examine $(ii)$. Let us again choose a small parameter $\eta >0$, and let us use the Cauchy--Schwarz inequality in order to obtain:
\begin{equation*}
\begin{aligned}
|(ii)| \leq \eta (1-\mu)r^{-3} |\hat Z_j|^2 + \eta^{-1} r^{-1}(1-\mu) |\partial_{\leo}^I T^k {\hat \lbar}^{j-1}((1-\mu)^{-1}R_\rho)|^2 
\end{aligned}
\end{equation*}
Regarding $(iii)$, we estimate

\begin{equation*}
\begin{aligned}
|(iii)| = 2r^{-2}(1-\mu)\left| \hat Z_j\left(\sum_{i=1}^{j-1} f^{(j-1)}_i \hat Z_i\right)\right|\\
\leq\eta (1-\mu)r^{-3}  |\hat Z_j| + C \eta^{-1} \frac{1-\mu}{r^3} \sum_{i=1}^{j-1} |\hat Z_i|^2 
\end{aligned}
\end{equation*}

Upon integration of the previous relation on the spacetime region $\regio{t_1^*}{t_2^*} \cap \{u \leq \tilde u\}$ with respect to the form $\de u \de v$, we obtain,
\begin{align}
\nonumber & \int_{\Sigma_{t_2^*} \cap \{u \leq \tilde u\}} (r^{-2} |\hat Z_j|^2) r^{-2} \de \Sigma_{t^*}
+ \int_{\regio{t_1^*}{t_2^*} \cap \{u \leq \tilde u\}}  \frac{1-\mu}{r^3}|\hat Z_j|^2\desphere \de u \de v \\ 
\nonumber &+\int_{\conplus_{\tilde u} \cap \regio{t_1^*}{t_2^*}} r^{-4} |\snabla_{\hat \lbar}^{j-1} \snabla_{\mathbb{S}^2} \hat Z_0|^2 \desphere \de v + \int_{\Sigma_{t_2^*} \cap \{u \leq \tilde u\}} r^{-4} |\snabla_{\hat \lbar}^{j-1} \snabla_{\mathbb{S}^2} \hat Z_0|^2 r^{-2}\de \Sigma_{t^*} \\ 
\label{eq:prered}  & \leq C \int_{\Sigma_{t_1^*}} \left\{ \frac 1 {r^{2}} |\hat Z_j|^2 + \frac 1 {r^4}  |\snabla_{\hat \lbar}^{j-1} \snabla_{\mathbb{S}^2} \hat Z_0 |^2\right\} r^{-2} \de \Sigma_{t^*}\\ 
\nonumber &
+ C\int_{\regio{t_1^*}{t_2^*} \cap \{u \leq \tilde u\}}\frac{1-\mu}{r^3} \sum_{i=1}^{j-1} |\hat Z_i|^2 \desphere \de u \de v\\ 
\nonumber &
+ C\int_{\regio{t_1^*}{t_2^*} \cap \{u \leq \tilde u\}} \Big\{ (1-\mu) r^{-5} |\snabla_{\hat \lbar}^{j-1} \snabla_{\mathbb{S}^2} \hat Z_0 |^2+ (1-\mu)r^{-6} \sum_{a = 0}^{j-2} |\hat \lbar^a \sdelta_{\mathbb{S}^2}\hat Z_0|^2 \Big\} \desphere\de u \de v
\\ 
\nonumber &
+C\int_{\regio{t_1^*}{t_2^*} \cap \{u \leq \tilde u\}} (1-\mu)r^{-1} |\partial_{\leo}^I T^k {\hat \lbar}^{j-1}((1-\mu)^{-1}R_\rho)|^2 \desphere\de u \de v.
\end{align}

Now, $|\snabla_{\hat \lbar}^a \snabla_{\mathbb{S}^2} \hat Z_0| = |\snabla_{\mathbb{S}^2} \snabla_{\hat \lbar}^a \hat Z_0| \leq C \sum_{H \in \mathscr{I}_{\leo}^1} |\hat Z_{,H+I, a, k}|$, and $|\snabla_{\hat \lbar}^a \sdelta_{\mathbb{S}^2} \hat Z_0| = |\sdelta_{\mathbb{S}^2} \snabla_{\hat \lbar}^a \hat Z_0| \leq C \sum_{H \in \mathscr{I}_{\leo}^{\leq 2}} |\hat Z_{,H+I, a, k}|$. Hence, we have the following, letting $b \in \N_{\geq 0}$ such that $b +j+ k\leq l$,
\begin{align*}
&\sum_{\substack{1 \leq |I| \leq b \\ 1 \leq |J| \leq b+1}} \left\{\int_{\Sigma_{t_2^*} \cap \{u \leq \tilde u\}} (r^{-2} |\hat Z_{,I, j, k}|^2 +r^{-4} |\hat Z_{,J,j-1,k}|^2) r^{-2} \de \Sigma_{t^*}\right.
\\ &
\left.+ \int_{\regio{t_1^*}{t_2^*} \cap \{u \leq \tilde u\}}  \frac{1-\mu}{r^3}|\hat Z_{,	I, j, k}|^2 \de u \de v 
+\int_{\conplus_{\tilde u} \cap \regio{t_1^*}{t_2^*}} r^{-4} |\hat Z_{,J,j-1,k}|^2 \desphere \de v\right\}
\\
&\leq \sum_{\substack{1 \leq |I| \leq b \\ 1 \leq |J| \leq b+1 \\ 1 \leq |K| \leq b+2}} \left\{C \int_{\regio{t_1^*}{t_2^*} \cap \{u \leq \tilde u\}} (1-\mu) r^{-5}  \sum_{a=0}^{j-1}\left(  |\hat Z_{,J,a,k}|^2  + r^2 |\hat Z_{,I,a,k}|^2 \right) \desphere \de u \de v \right.\\ &
+C \int_{\regio{t_1^*}{t_2^*} \cap \{u \leq \tilde u\}} (1-\mu) r^{-6}  \sum_{a=0}^{j-2}|\hat Z_{,K,a,k}|^2 \desphere \de u \de v +C \int_{\Sigma_{t_1^*}} \left( r^{-2} |\hat Z_{,I, j, k}|^2 + r^{-4} |\hat Z_{,J,j-1,k}|^2 
 \right)r^{-2} \de \Sigma_{t^*}\\ &
 \left. +C\int_{\regio{t_1^*}{t_2^*} \cap \{u \leq \tilde u\}} (1-\mu)r^{-1} |\partial_{\leo}^I T^k {\hat \lbar}^{j-1}((1-\mu)^{-1}R_\rho)|^2\desphere \desphere \de u \de v\right\}.
\end{align*}
An induction argument then shows, if $k + h \leq l$,
\begin{align*}
&\sum_{\substack{1 \leq |I| + j \leq h\\ j \geq 1}}\left\{
\int_{\Sigma_{t_2^*} \cap \{u \leq \tilde u\}} (r^{-2} |\hat Z_{,I, j, k}|^2 +r^{-4} |\hat Z_{,J,j-1,k}|^2) r^{-2} \de \Sigma_{t^*}\right.\\
&\left.
+ \int_{\regio{t_1^*}{t_2^*} \cap \{u \leq \tilde u\}}  \frac{1-\mu}{r^3}|\hat Z_{,I, j, k}|^2 \desphere \de u \de v 
+\int_{\conplus_{\tilde u} \cap \regio{t_1^*}{t_2^*}} r^{-4} |\hat Z_{,J,j-1,k}|^2 \desphere \de v
\right\}
\\
&\leq
 C \sum_{1 \leq |J|\leq h} \int_{\regio{t_1^*}{t_2^*} \cap \{u \leq \tilde u\}} (1-\mu) r^{-5}  |\hat Z_{,J,0,k}|^2 \desphere \desphere \de u \de v
 +C\sum_{1 \leq |I| + j \leq h } \int_{\Sigma_{t_1^*}} \left\{ r^{-2} |\hat Z_{,I, j, k}|^2 \right\}r^{-2} \de \Sigma_{t^*}\\
&  +C\sum_{1\leq |I| + j \leq h }\int_{\regio{t_1^*}{t_2^*} \cap \{u \leq \tilde u\}} (1-\mu)r^{-1} |\partial_{\leo}^I T^k {\hat \lbar}^{j-1}((1-\mu)^{-1}R_\rho)|^2 \desphere \de u \de v.
\end{align*}
This implies, finally, taking $\tilde u \to \infty$,
\begin{align}
&\nonumber \sum_{\substack{1 \leq |I| + j \leq h \\ j \geq 1} }\left\{
\int_{\Sigma_{t_2^*}} r^{-2} |\hat Z_{,I, j, k}|^2  r^{-2} \de \Sigma_{t^*}
+ \int_{\regio{t_1^*}{t_2^*}}  \frac{1-\mu}{r^3}|\hat Z_{, I, j, k}|^2 \de u \de v
\right\}
\\ \label{eq:redfinal}
&\leq
 C \sum_{1 \leq |J|\leq h} \int_{\regio{t_1^*}{t_2^*}} (1-\mu) r^{-5}  |\hat Z_{,J,0,k}|^2 \desphere \de u \de v
 +C\sum_{1 \leq |I| + j \leq h } \int_{\Sigma_{t_1^*}} r^{-4} |\hat Z_{,I, j, k}|^2  \de \Sigma_{t^*}\\
 & +C\sum_{1 \leq |I| + j \leq h }\int_{\regio{t_1^*}{t_2^*}} (1-\mu) r^{-1}|\partial_{\leo}^I T^k {\hat \lbar}^{j-1}((1-\mu)^{-1}R_\rho)|^2 \desphere \de u \de v. \nonumber
\end{align}
The same reasoning holds for $\hat W$.
\end{proof}
\section{The structure of the nonlinear terms}\label{sec:nonlstruct}
We begin with a fundamental lemma.

\begin{lemma}[Form of $H^{\mu\nu\kappa\lambda}_{\Delta} \nabla_{\mu}\fara_{\kappa\lambda}$]\label{lem:fund}
If the MBI system (\ref{MBI}) holds, then we have
\begin{equation}\label{eq:forminv}
H^{\mu\nu\kappa\lambda}_{\Delta} \nabla_{\mu}\fara_{\kappa\lambda} = \frac{1}{2\ellmbi^2} \nabla_\mu(\lun-\ldu^2) \left(- \fara^{\mu\nu} + \ldu\farad^{\mu\nu} \right) - \nabla_\mu \ldu \farad^{\mu\nu}.
\end{equation}
\end{lemma}
\begin{proof}
This relation follows directly from the MBI equations (\ref{MBI}):
\begin{equation}
\nabla_\mu \mfarad^{\mu\nu} = 0, \qquad \nabla_\mu \farad^{\mu\nu} = 0,
\end{equation}
along with the form of $\mfarad = - \ellmbi^{-1} (\fara^{\mu\nu}- \ldu \farad^{\mu\nu})$, and the equation for $H_\Delta^{\mu\nu\kappa\lambda}$:
\begin{equation}
\nabla_\mu\fara^{\mu\nu} = - H_\Delta^{\mu\nu\kappa\lambda} \nabla_\mu \fara_{\kappa\lambda}.
\end{equation}
\end{proof}

We are going to use the lemma above to estimate the nonlinearities $R_\rho$ and $R_\sigma$.

\subsection{Bounds on $R_\rho$}

\begin{lemma}[Bound on $R_\rho$ and its derivatives.]\label{lem:bsrho}
Let $l, j, k, b \in \N_{\geq 0}$, let $I \in \iindicess^b$, and assume $b + j + k \leq l$. Let $\fara$ be a smooth solution of the MBI system (\ref{MBI}) on $\mathcal{R} \subset S_e$. Assume the bootstrap assumptions $BA\left(\mathcal{R},1,\left\lfloor \frac{l+2}{2}\right\rfloor , \varepsilon \right)$. Under these conditions, we have the following pointwise bound on $\mathcal{R}$:
\begin{equation}\label{eq:boundrrho}
\begin{aligned}
|\partial_{\leo}^I T^k {\hat \lbar}^{j}(R_\rho)|+|\partial_{\leo}^I T^k {\hat \lbar}^{j}((1-\mu)^{-1}R_\rho)| \leq  C \eps^{\frac 32}\tau^{-2} r^{-1}\sum_{i = 0}^{l+2} |\partial^i \fara|.
\end{aligned}
\end{equation}
\end{lemma}

\begin{proof}
We first notice that we have, by definition,
\begin{equation}\label{eq:masternlu}
(1-\mu)^{-1} R_\rho = \underbrace{(1-\mu)^{-1}r^2\nonl{1}}_{(i)} + \underbrace{(1-\mu)^{-1}r^2\nonl{2}}_{(ii)}.
\end{equation}

We divide the Proof in two \textbf{Steps.} In \textbf{Step 1}, we will analyse term $(i)$ in Equation (\ref{eq:masternlu}), whereas in \textbf{Step 2}, we will analyse term $(ii)$ in Equation (\ref{eq:masternlu}).

\subsection*{Step 1: analysis of $(i)$}
We have
\begin{equation*}
\begin{aligned}
\text{\bf NL}_1 = 2 \frac{1-\mu}{r} \left(-\tensor{{H_{_\Delta}}}{^\mu _\lbar ^\kappa ^\lambda} \nabla_{\mu}\fara_{\kappa\lambda} \right)-2 \frac{1-\mu}{r} \left( \tensor{{H_{_\Delta}}}{^\mu _L ^\kappa ^\lambda} \nabla_{\mu}\fara_{\kappa\lambda}\right).
  \end{aligned}
 \end{equation*}
Hence,
\begin{equation*}
(i) = -2 r \left(\tensor{{H_{_\Delta}}}{^\mu _\lbar ^\kappa ^\lambda} \nabla_{\mu}\fara_{\kappa\lambda} \right)-2 r\left( \tensor{{H_{_\Delta}}}{^\mu _L ^\kappa ^\lambda} \nabla_{\mu}\fara_{\kappa\lambda}\right).
\end{equation*}
We now have, from Equation (\ref{eq:forminv}),
\begin{equation*}
\begin{aligned}
&r \left(\tensor{{H_{_\Delta}}}{^\mu _\lbar ^\kappa ^\lambda} \nabla_{\mu}\fara_{\kappa\lambda} \right) =   \frac{r}{2\ellmbi^2} \nabla_\mu(\lun-\ldu^2) \left(- \fara\indices{^\mu_\lbar} + \ldu\farad\indices{^\mu_\lbar} \right) - \nabla_\mu \ldu \farad\indices{^\mu_\lbar}\\
&=  -\frac{r}{2\ellmbi^2} \snabla^A(\lun-\ldu^2) (\alphabar_A + \svol_{BA} \ldu \alphabar^B)+r \svol_{BA} \snabla^A \ldu \alphabar^B\\
&+\frac{r}{2\ellmbi^2} \lbar(\lun-\ldu^2) (-\rho + \ldu \sigma) -
r \sigma \lbar \ldu .
\end{aligned}
\end{equation*}
We then have, from the bootstrap assumptions,
\begin{equation}
\begin{aligned}
\left|-\frac{r}{2\ellmbi^2} \snabla^A(\lun-\ldu^2) (\alphabar_A + \svol_{BA} \ldu \alphabar^B)+r \svol_{BA} \snabla^A \ldu \alphabar^B\right|
\leq C (|\partial \lun| + |\partial\ldu|)|\alphabar|.
\end{aligned}
\end{equation}
We use the relations
\begin{equation}
\lun = - \rho^2 + \sigma^2 - (1-\mu)^{-1}\alpha^A \alphabar_A , \qquad \ldu = - \rho \sigma - \frac 12 (1-\mu)^{-1} \svol^{AB}\alpha_A \alphabar_B.
\end{equation}
Hence, again from the bootstrap assumptions, there holds
\begin{equation}
\begin{aligned}
&|\partial(\lun + \ldu)||\alphabar| \leq C \left(\sum_{a+b \leq 1} |\partial^a \rho| |\partial^b \rho|+|\partial^a \sigma| |\partial^b \rho|+|\partial^a \sigma| |\partial^b \sigma| + |\partial^a \alpha| |\partial^b \alphabar|\right)|\alphabar|\\
& \qquad \leq C \varepsilon^{3/2}\tau^{-2} r^{-2}(|\partial \fara|+|\fara|).
\end{aligned}
\end{equation}
Similarly,
\begin{equation}
\begin{aligned}
\left|\frac{r}{2\ellmbi^2} \lbar(\lun-\ldu^2) (-\rho + \ldu \sigma) -
r \sigma \lbar \ldu\right| \leq C r (|\partial \lun| + |\partial \ldu)|)(|\rho|+ |\sigma|)
\end{aligned}
\end{equation}
Again, using the bootstrap assumptions,
\begin{equation*}
r (|\partial(\lun)| + |\partial(\ldu)|)(|\rho|+ |\sigma|) \leq C \varepsilon^{3/2} \tau^{-2} r^{-2}(|\fara|+ |\partial \fara|).
\end{equation*}
In conclusion, we have
\begin{equation}
\left|r \left(\tensor{{H_{_\Delta}}}{^\mu _\lbar ^\kappa ^\lambda} \nabla_{\mu}\fara_{\kappa\lambda} \right)\right| \leq C \varepsilon^{3/2} \tau^{-2} r^{-2}(|\fara|+ |\partial \fara|).
\end{equation}
In a similar way, we can estimate
\begin{equation}
\begin{aligned}
&\left| \partial_{\leo}^I T^k {\hat \lbar}^{j} \left( r \tensor{{H_{_\Delta}}}{^\mu _\lbar ^\kappa ^\lambda} \nabla_{\mu}\fara_{\kappa\lambda} \right)\right| \\ &
\leq C \sum_{a + b \leq l} (|\partial^{a+1}\lun| + |\partial^{a+1}\ldu|)(|\partial^b\alphabar| + r|\partial^b\rho| + r|\partial^b\sigma|) \\
&\leq C \sum_{\substack{a + b +c\leq l
\\
a+c \geq 1
}} 
\left( |\partial^a \rho| |\partial^c \rho|+|\partial^a \sigma| |\partial^c \rho|+|\partial^a \sigma| |\partial^c \sigma| + |\partial^a \alpha| |\partial^c \alphabar|\right)
(|\partial^b\alphabar| + r|\partial^b\rho| + r|\partial^b\sigma|).
\end{aligned}
\end{equation}
Now, the key observation is that, if $a+b+c \leq l+2$, we have
$$
\max\left\{\{a,b,c\} \setminus \{\max\{a,b,c\} \}\right\} \leq \left \lfloor \frac{l+2}{2}\right \rfloor.
$$
Since we are assuming the bootstrap assumptions $BA\left(\mathcal{R},1,\left\lfloor \frac{l+2}{2}\right\rfloor , \varepsilon \right)$, we obtain
\begin{equation}
\left| \partial_{\leo}^I T^k {\hat \lbar}^{j} \left( r \tensor{{H_{_\Delta}}}{^\mu _\lbar ^\kappa ^\lambda} \nabla_{\mu}\fara_{\kappa\lambda} \right)\right| \leq C \varepsilon^{3/2} r^{-2} \tau^{-2} \sum_{i = 0}^{l+1} |\partial^i \fara|.
\end{equation}
A very similar reasoning yields the corresponding bound:
\begin{equation}
\begin{aligned}
\left| \partial_{\leo}^I T^k {\hat \lbar}^{j} \left( r \tensor{{H_{_\Delta}}}{^\mu _L ^\kappa ^\lambda} \nabla_{\mu}\fara_{\kappa\lambda} \right)\right| 
\leq C \varepsilon^{3/2} r^{-2}\tau^{-2}  \sum_{i = 0}^{l+1} |\partial^i \fara|.
\end{aligned}
\end{equation}
Combining the previous two inequalities, we obtain the bound for $(i)$:
\begin{equation}
|(i)| \leq C \varepsilon^{3/2} r^{-2}\tau^{-2}  \sum_{i = 0}^{l+1} |\partial^i \fara|.
\end{equation}
This concludes \textbf{Step 1}.

\subsection*{Step 2: analysis of $(ii)$}
Similarly, we examine the terms given by $\nonl{2}$. We first state a fundamental remark, following from the fact that $\nabla_L \lbar = \nabla_\lbar L = 0$:
 \begin{equation*}
 \begin{aligned}
 \textbf{\bf NL}_2:&=
  L^\nu \lbar^\alpha \left(\nabla_\alpha  \left(\tensor{{H_{_\Delta}}}{^\mu _\nu ^\kappa ^\lambda} \nabla_{\mu}\fara_{\kappa\lambda} \right) -
   \nabla_\nu \left(\tensor{{H_{_\Delta}}}{^\mu _\alpha ^\kappa ^\lambda} \nabla_{\mu}\fara_{\kappa\lambda} \right)  \right) \\
  & = \underbrace{\lbar \left( {H_{_\Delta}}\indices{^\mu _L ^\kappa ^\lambda} \nabla_{\mu}\fara_{\kappa\lambda}\right)}_{(a)}-\underbrace{L \left( {H_{_\Delta}}\indices{^\mu _\lbar ^\kappa ^\lambda} \nabla_{\mu}\fara_{\kappa\lambda} \right)}_{(b)}.
 \end{aligned}
 \end{equation*}
We now estimate
\begin{equation*}
\begin{aligned}
&|\partial_{\leo}^I T^k {\hat \lbar}^{j}((1-\mu)^{-1}r^2 (a))|  \\
&\leq 
C \left| \partial^{\leq l} \left((1-\mu)^{-1} r^2 \lbar \left({H_{_\Delta}}\indices{^\mu _L ^\kappa ^\lambda} \nabla_{\mu}\fara_{\kappa\lambda} \right)\right)\right| \\ &
\leq C \left|\partial^{\leq l} \left\{(1-\mu)^{-1} r\lbar  \left(r{H_{_\Delta}}\indices{^\mu _L ^\kappa ^\lambda} \nabla_{\mu}\fara_{\kappa\lambda}  \right)\right\} \right| + C \left|\partial^{\leq l} \left( r{H_{_\Delta}}\indices{^\mu _L ^\kappa ^\lambda} \nabla_{\mu}\fara_{\kappa\lambda} \right)\right|.
\end{aligned}
\end{equation*}
In view of \textbf{Step 1}, we have
\begin{equation}
\begin{aligned}
\left|\partial^{\leq l} \left( r{H_{_\Delta}}\indices{^\mu _L ^\kappa ^\lambda} \nabla_{\mu}\fara_{\kappa\lambda} \right)\right| \leq C \varepsilon^{3/2} \tau^{-2} r^{-2} \sum_{i = 0}^{l+1} |\partial^i \fara|.
\end{aligned}
\end{equation}
Furthermore, the Leibniz rule implies the bound
\begin{equation}
\begin{aligned}
\left|\partial^{\leq l} \left\{(1-\mu)^{-1} r\lbar  \left(r{H_{_\Delta}}\indices{^\mu _L ^\kappa ^\lambda} \nabla_{\mu}\fara_{\kappa\lambda}  \right)\right\} \right| \leq
C r \sum_{i = 1}^{l+1} \left|\partial^i \left(  r{H_{_\Delta}}\indices{^\mu _L ^\kappa ^\lambda} \nabla_{\mu}\fara_{\kappa\lambda}  \right) \right|.
\end{aligned}
\end{equation}
 Since we are assuming the bootstrap assumptions $BA\left(\mathcal{R},1,\left\lfloor \frac{l+2}{2}\right\rfloor , \varepsilon \right)$, we finally obtain, again from \textbf{Step 1}:
\begin{equation}
r \sum_{i = 1}^{l+1} \left|\partial^i \left(  r{H_{_\Delta}}\indices{^\mu _L ^\kappa ^\lambda} \nabla_{\mu}\fara_{\kappa\lambda}  \right) \right| \leq C \eps^{\frac 32} \tau^{-2} r^{-1}\sum_{i = 0}^{l+2} |\partial^i \fara|.
\end{equation}
A similar reasoning holds for the term $(b)$. We eventually obtain
\begin{equation}
|(ii)| \leq C \eps^{\frac 32}\tau^{-2} r^{-1}\sum_{i = 0}^{l+2} |\partial^i \fara|.
\end{equation}
This concludes \textbf{Step 2}.

The combination of \textbf{Step 1} and \textbf{Step 2} now yields the claim.
\end{proof}

\subsection{Bounds on $R_\sigma$}

\begin{lemma}[Bound on $R_\sigma$ and its derivatives.]\label{lem:bssigma}
Let $l, j, k, b \in \N_{\geq 0}$, let $I \in \iindicess^b$, and assume $b + j + k \leq l$. Let $\fara$ be a smooth solution of the MBI system (\ref{MBI}) on $\mathcal{R} \subset S_e$. Assume the bootstrap assumptions $BA\left(\mathcal{R},1,\left\lfloor \frac{l+2}{2}\right\rfloor , \varepsilon \right)$. Under these conditions, we have the following pointwise bound on $\mathcal{R}$:
\begin{equation}\label{eq:boundrsigma}
|\partial_{\leo}^I T^k {\hat \lbar}^{j}(R_\sigma)|+|\partial_{\leo}^I T^k {\hat \lbar}^{j}((1-\mu)^{-1}R_\sigma)| \leq C \eps^{\frac 32} \tau^{-2} r^{-1}\sum_{i = 0}^{l+2} |\partial^i \fara|.
\end{equation}
\end{lemma}

\begin{proof}
We start by recalling the definition of $\nonl{3}$ and $R_\sigma$:
\begin{equation}
\begin{aligned}
R_\sigma = - r^2 \nonl{3}, \qquad \nonl{3} = L^\nu \lbar^\alpha \nabla^\mu U_{\alpha \mu \nu}, \qquad U_{\alpha \mu\nu} =- \varepsilon\indices{_\alpha^\lambda_\mu_\nu} {H_{_\Delta}}\indices{^\beta_\lambda_\gamma^\delta} \nabla_\beta \fara_{\gamma \delta}.
\end{aligned}
\end{equation}
Let us now fix a smooth local frame field $\{e_1, e_2\}$ such that $g(e_i, e_j) = \delta_{ij}$, $e_1$, $e_2$ are both tangent to the spheres of constant $r$-coordinate, and furthermore
$$
(e_2)^B= \svol^{CB} \gbar_{CA} (e_1)^A
$$
Here, we used the definition of $\svol_{AB}$ and $\gbar_{AB}$, they are resp. the induced volume form and the metric on the spheres of constant $r$. In other words, $e_2$ is the oriented (with respect to $\svol$) normal vector to $e_1$ on the spheres of constant $r$.

We will also assume that, locally around a point $\bar \omega \in \mathbb{S}^2$, $e_1$ and $e_2$ have the following expression, after an appropriate change of coordinates:
\begin{equation*}
e_1 = r^{-1}\partial_\theta, \qquad e_2 = r^{-1} \frac{\partial_\varphi}{\sin \theta}.
\end{equation*}
With these choices, we have that, at $\bar \omega$,
\begin{equation}
\begin{aligned}
\nonl{3} = - 2 (1-\mu) \nabla_{\alpha} \left({H_{_\Delta}}\indices{^\beta_\eta^\gamma^\delta} \nabla_\beta \fara_{\gamma \delta}   \right) (e_1)^\alpha (e_2)^\eta + 2 (1-\mu) \nabla_{\alpha} \left({H_{_\Delta}}\indices{^\beta_\eta^\gamma^\delta} \nabla_\beta \fara_{\gamma \delta}   \right) (e_2)^\alpha (e_1)^\eta.
\end{aligned}
\end{equation}
Since $[e_1, e_2] = - r^{-1} (\cos \theta)(\sin \theta)^{-1}e_2$, we have
\begin{equation}
\begin{aligned}
\nonl{3} = &- \underbrace{2 (1-\mu) \nabla_{\alpha} \left({H_{_\Delta}}\indices{^\beta_\eta^\gamma^\delta} \nabla_\beta \fara_{\gamma \delta}  (e_2)^\eta \right) (e_1)^\alpha }_{(a)} +\underbrace{ 2 (1-\mu) \nabla_{\alpha} \left({H_{_\Delta}}\indices{^\beta_\eta^\gamma^\delta} \nabla_\beta \fara_{\gamma \delta}  (e_1)^\eta \right) (e_2)^\alpha }_{(b)}\\
&+ \underbrace{2 (1-\mu)r^{-1} {H_{_\Delta}}\indices{^\beta_{e_2}^\gamma^\delta} \nabla_\beta \fara_{\gamma \delta} }_{(c)}.
\end{aligned}
\end{equation}
Let us focus our attention on $(a)$, the reasoning for $(b)$ and $(c)$ being analogous. We have
\begin{equation}\label{eq:13.23}
\begin{aligned}
r^2(1-\mu)^{-1}(a) = -2 r \partial_\theta \left({H_{_\Delta}}\indices{^\beta_{e_2}^\gamma^\delta} \nabla_\beta \fara_{\gamma \delta} \right) = -2 \partial_\theta \left( r{H_{_\Delta}}\indices{^\beta_{e_2}^\gamma^\delta} \nabla_\beta \fara_{\gamma \delta} \right)
\end{aligned}
\end{equation}
We now use Equation~\eqref{eq:forminv} to get
\begin{align*}
{H_{_\Delta}}\indices{^\mu_{e_2}^\kappa^\lambda} \nabla_{\mu}\fara_{\kappa\lambda} &= \frac{1}{2\ellmbi^2} \nabla_\mu(\lun-\ldu^2) \left(- \fara\indices{^\mu_{e_2}} + \ldu\farad\indices{^\mu_{e_2}} \right) - \nabla_\mu \ldu \farad\indices{^\mu_{e_2}}\\
&= \frac{1}{2\ellmbi^2} \nabla_{e_1}(\lun-\ldu^2) \left(- \fara\indices{^{e_1}_{e_2}} + \ldu\farad\indices{^{e_1}_{e_2}} \right) - \nabla_{e_1} \ldu \farad\indices{^{e_1}_{e_2}}\\
&+ \frac{1}{2\ellmbi^2} \nabla_{L}(\lun-\ldu^2) \left(- \fara\indices{^{L}_{e_2}} + \ldu\farad\indices{^{L}_{e_2}} \right) - \nabla_L \ldu \farad\indices{^{L}_{e_2}}\\
&+ \frac{1}{2\ellmbi^2} \nabla_{\lbar}(\lun-\ldu^2) \left(- \fara\indices{^{\lbar}_{e_2}} + \ldu\farad\indices{^{\lbar}_{e_2}} \right) - \nabla_{\lbar} \ldu \farad\indices{^{\lbar}_{e_2}}.
\end{align*}
Hence, as before, we have, using the bootstrap assumptions as well as the expression of the invariants $\lun$ and $\ldu$, in view of~\eqref{eq:13.23},
\begin{align*}
&\left|(1-\mu)^{-1}r^2 (a)\right| \\ &
\leq C r \sum_{a+b \leq 1} (|\partial^a e_1(\lun)|+|\partial^a e_1(\ldu)|)(|\partial^b \fara\indices{^{e_1}_{e_2}}|+ |\partial^b \farad\indices{^{e_1}_{e_2}}|| )\\ &
+ C r  \sum_{a+b \leq 1} (|\partial^a L(\lun)|+|\partial^a L(\ldu)|)(|\partial^b \fara\indices{^{L}_{e_2}}|+ |\partial^b \farad\indices{^{L}_{e_2}}|| )\\ &
+C r  \sum_{a+b \leq 1} (|\partial^a \lbar(\lun)|+|\partial^a \lbar(\ldu)|)(|\partial^b \fara\indices{^{\lbar}_{e_2}}|+ |\partial^b \farad\indices{^{\lbar}_{e_2}}|| ) \\ &
\leq C \sum_{a+b+c \leq 2} (|\partial^a \rho| |\partial^b \rho|+ |\partial^a \sigma| |\partial^b \rho| + |\partial^a \sigma| |\partial^b \sigma|+ |\partial^a \alpha| |\partial^b \alphabar|)(|\partial^c\sigma|+ |\partial^c\rho|)\\ &
+C r\sum_{a+b+c \leq 2} (|\partial^a \rho| |\partial^b \rho|+ |\partial^a \sigma| |\partial^b \rho| + |\partial^a \sigma| |\partial^b \sigma|+ |\partial^a \alpha| |\partial^b \alphabar|)|\partial^c\alphabar|\\ &
+ C r \sum_{a+b+c \leq 2} (|\partial^a \rho| |\partial^b \rho|+ |\partial^a \sigma| |\partial^b \rho| + |\partial^a \sigma| |\partial^b \sigma|+ |\partial^a \alpha| |\partial^b \alphabar|)|\partial^c\alpha|\\ &
\leq C \eps^{\frac 32}\tau^{-2} r^{-1} (|\fara|+ |\partial\fara| + |\partial^2 \fara|).
\end{align*}
Similarly, we have, under the bootstrap assumptions,
\begin{equation}
\begin{aligned}
&\left| \partial_{\leo}^I T^k {\hat \lbar}^{j}((1-\mu)^{-1}r^2 (a))  \right|\\ &
\leq C \sum_{a+b+c \leq l+2} (|\partial^a \rho| |\partial^b \rho|+ |\partial^a \sigma| |\partial^b \rho| + |\partial^a \sigma| |\partial^b \sigma|+ |\partial^a \alpha| |\partial^b \alphabar|)(|\partial^c \sigma|+ |\partial^c \rho|)\\ &
+C r\sum_{a+b+c \leq l+2} (|\partial^a \rho| |\partial^b \rho|+ |\partial^a \sigma| |\partial^b \rho| + |\partial^a \sigma| |\partial^b \sigma|+ |\partial^a \alpha| |\partial^b \alphabar|)|\partial^c \alphabar|\\ &
+ C r \sum_{a+b+c \leq l+2} (|\partial^a \rho| |\partial^b \rho|+ |\partial^a \sigma| |\partial^b \rho| + |\partial^a \sigma| |\partial^b \sigma|+ |\partial^a \alpha| |\partial^b \alphabar|)|\partial^c \alpha|\\ &
\leq C \eps^{\frac 32}\tau^{-2} r^{-1} \sum_{i = 0}^{l+2}|\partial^i \fara|.
\end{aligned}
\end{equation}
This is the claim.
\end{proof}
\section{Morawetz estimates for \texorpdfstring{$Z$}{Z} and \texorpdfstring{$W$}{W}: second step}\label{sec:morawetz}

We now combine the results in Sections~\ref{sec:premora} and~\ref{sec:nonlstruct} in order to obtain, under the bootstrap assumptions, a non-degenerate Morawetz estimate for $\hat Z_{,I,j,k}$ and $\hat W_{,I,j,k}$ when the number of derivatives is high. Here we combine the structure of the nonlinear terms to obtain estimates on the RHS of the inequalities of Section~\ref{sec:premora}.

\subsection{Degenerate Morawetz estimate}

 \begin{proposition}\label{prop:moradegeps} Let $l, k, b \in \N_{\geq 0}$, let furthermore $I \in \iindiceo^b$, and assume $b + k \leq l$ and $b \geq 1$. 
 There exist a small number $\tilde \varepsilon >0$ and a constant $C >0$ such that the following holds. Let $t_2^* \geq t_0^*$, and let $\fara$ be a smooth solution of the MBI system (\ref{MBI}) on $\mathcal{R}:=\regio{t^*_0}{t^*_2} \subset S_e$. Assume the bootstrap assumptions $BA\left(\mathcal{R}, 1,\left\lfloor \frac{l+3}{2}\right\rfloor , \varepsilon \right)$. Assume furthermore the smallness of initial energy (\ref{eq:ensmall}) 
\begin{equation*}\label{eq:ensmallref}
\norm{\fara}^2_{H^{l+3}}(\Sigma_{t_0^*}) \leq \varepsilon^2.
\end{equation*}
Here, $0 < \varepsilon < \tilde \varepsilon$. For ease of notation, denote, throughout this Proposition, $\dot Z := \dot Z_{,I, 0, k}$. Under these conditions, we have the following degenerate Morawetz estimate on $\mathcal R$:
\begin{equation}\label{eq:moradegeps}
\begin{aligned}
 \int_{\regio{t_1^*}{t_2^*} } 
 \left\{\frac 1 {r^2}|L \dot Z - \lbar \dot Z|^2+ \frac{(r-3M)^2}{r^3}\left(|\snabla \dot Z|^2+ \frac 1 r |L \dot Z + \lbar \dot Z|^2\right)+\frac 1 {r^3} |\dot Z|^2\right\} r^{-2}\de \text{Vol}  \leq C \varepsilon^2.
\end{aligned}
\end{equation}
Furthermore, the same estimate holds with $\dot Z$ replaced by $\dot W$. 
\end{proposition}

\begin{remark}
We notice that this estimate ``loses two derivatives'' in the error terms.
\end{remark}

\begin{proof}
We recall the estimate~\eqref{fullen}:
\begin{align*}
& \int_{\regio{t_1^*}{t_2^*} } 
\left\{\frac 1 {r^2}|L \dot Z - \lbar \dot Z|^2+ \frac{(r-3M)^2}{r^3}\left(|\snabla \dot Z|^2+ \frac 1 r |L \dot Z + \lbar \dot Z|^2\right)+\frac 1 {r^3} |\dot Z|^2\right\} r^{-2}\de \text{Vol} \\ &
\leq C	\int_{\regio{t_1^*}{t_2^*}  } (r^2(1-\mu)^{-1} |\partial^I_{\leo} T^k (R_\rho)|^2 + r^{-1}(1-\mu) |\partial^I_{\leo} T^k (R_\rho)|^2 )\de u \de v\\ &
+C \int_{\regio{t_1^*}{t_2^*} \cap \{ 2M \leq r \leq 4M\}} |\dot Z_{,I,0,k}| |\partial^I_{\leo} T^{k+1} (R_\rho)| \de u \de v\\
& +C \int_{\regio{t_1^*}{t_2^*} \cap \{ r \geq 4M\}} |\dot Z_{,I,0,k+1}| |\partial^I_{\leo} T^{k} (R_\rho)| \de u \de v\\ &
+ C \int_{\regio{t_1^*}{t_2^*} \cap \{2M \leq  r  \leq 3 M\}}(|\partial^I_{\leo}  {\hat \lbar} T^{k} (R_\rho)|^2 + |\partial^I_{\leo}  T^{k} ((1-\mu)^{-1}R_\rho)|^2)\de \text{Vol}\\ &
+ C \int_{\Sigma_{t_2^*} \cap \{2M \leq r \leq 4M\}} |\partial^I_{\leo}  T^{k} ((1-\mu)^{-1}R_\rho)|^2 \de \Sigma_{t^*} \\
& + C \int_{\Sigma_{t_1^*}\cap \{2M \leq r \leq 4M\}} |\partial^I_{\leo}  T^{k} ((1-\mu)^{-1}R_\rho)|^2 \de \Sigma_{t^*}\\ &
+  C \int_{\Sigma_{t_1^*}} (|L \dot Z|^2+(1-\mu)^{-1}|\lbar \dot Z|^2+|\snabla \dot Z|^2 )r^{-2}\de \Sigma_{t^*}.
\end{align*}
This estimate, combined with the Cauchy--Schwarz inequality, Propositions (\ref{lem:bsrho}) and (\ref{lem:bssigma}), an application of Fubini's theorem to the spacetime terms, and the initial smallness of energy (\ref{eq:ensmallref}) yield the claim.
\end{proof}

\subsection{Removing the degeneracies}

Communting with angular momentum operators once gives the following improved version of Proposition~\ref{prop:moradegeps}

\begin{proposition}\label{prop:moraeps}
Let $l, k, b \in \N_{\geq 0}$, let $t_2^* \geq t_0^*$, let furthermore $I \in \iindiceo^b$, and assume $b + k \leq l$ and $b \geq 1$. There exist a small number $\tilde \varepsilon >0$ and a constant $C >0$ such that the following holds. Let $\fara$ be a smooth solution of the MBI system (\ref{MBI}) on $\mathcal{R}:=\regio{t^*_0}{t^*_2} \subset S_e$. Assume the bootstrap assumptions $BA\left(\mathcal{R}, 1,\left\lfloor \frac{l+4}{2}\right\rfloor , \varepsilon \right)$. Assume furthermore the smallness of initial energy (\ref{eq:ensmall}):
\begin{equation*}\label{eq:ensmall4}
\norm{\fara}^2_{H^{l+4}(\Sigma_{t_0^*})} \leq \varepsilon^2.
\end{equation*}
Here, $0 < \varepsilon < \tilde \varepsilon$. For ease of notation, denote, throughout this Proposition, $\dot Z := \dot Z_{,I, 0, k}$. Under these conditions, we have the following non-degenerate estimate on $\mathcal R$:
\begin{equation}\label{eq:morandegeps}
\begin{aligned}
 \int_{\regio{t_1^*}{t_2^*} } 
 \left\{\frac 1 {r^2}|L \dot Z|^2+ \frac{1}{r}|\snabla \dot Z|^2 +\frac 1 {r^3} |\dot Z|^2+ \frac 1 {r^2} |\hat \lbar \dot Z|^2\right\} r^{-2}\de \text{Vol}  \leq C \varepsilon^2.
\end{aligned}
\end{equation}
Furthermore, the same estimate holds with $\dot Z$ replaced by $\dot W$.
\end{proposition}

\begin{proof}
First, we notice that, if we assume the bootstrap assumption $BA\left(\mathcal{R}, 1,\left\lfloor \frac{l+4}{2}\right\rfloor, \varepsilon \right)$, as well as the smallness of energy up to $4+l$ derivatives, we have 
\begin{equation}
\begin{aligned}
\int_{\regio{t_1^*}{t_2^*} } 
\left\{\frac 1 {r^2}|L \dot Z - \lbar \dot Z|^2+ \frac{1}{r}|\snabla \dot Z|^2 +\frac{\left(1-\frac{3M}{r}\right)^2}{r^2}|L \dot Z + \lbar \dot Z|^2+\frac 1 {r^3} |\dot Z|^2\right\} r^{-2}\de \text{Vol}  \leq C \varepsilon^2.
\end{aligned}
\end{equation}
We then combine this estimate with inequality~\eqref{eq:redsh}, in which we set $l =1$. We notice that the error terms on the right hand side of~\eqref{eq:redsh} are all bounded by the assumptions of this Proposition. We have:
\begin{equation}
\begin{aligned}
\int_{\regio{t_1^*}{t_2^*} } 
\left\{\frac 1 {r^2}|L \dot Z - \lbar \dot Z|^2+ \frac{1}{r}|\snabla \dot Z|^2 +\frac{\left(1-\frac{3M}{r}\right)^2}{r^2}|L \dot Z + \lbar \dot Z|^2+\frac 1 {r^3} |\dot Z|^2+ \frac 1 {r^3} |\hat \lbar \dot Z|^2\right\} r^{-2}\de \text{Vol}  \leq C \varepsilon^2.
\end{aligned}
\end{equation}
This easily implies:
\begin{equation}
\begin{aligned}
\int_{\regio{t_1^*}{t_2^*} } 
\left\{\frac 1 {r^2}|L \dot Z|^2+ \frac{1}{r}|\snabla \dot Z|^2 +\frac 1 {r^3} |\dot Z|^2+ \frac 1 {r^2} |\hat \lbar \dot Z|^2\right\} r^{-2}\de \text{Vol}  \leq C \varepsilon^2.
\end{aligned}
\end{equation}
\end{proof}

Combining the previous Propositions, we finally obtain

\begin{proposition}\label{prop:moradef}
Let $l \in \N_{\geq 0}$. There exist a constant $C > 0$ and a small number $\tilde \varepsilon > 0$ such that the following holds. Let $t_2^* \geq t_0^*$. Let $\fara$ be a smooth solution of the MBI system (\ref{MBI}) on $\mathcal{R}:=\regio{t^*_0}{t^*_2} \subset S_e$. Assume the bootstrap assumptions $BA\left(\mathcal{R},1,\left\lfloor \frac{l+4}{2}\right\rfloor , \varepsilon \right)$. Assume furthermore the smallness of initial energy (\ref{eq:ensmall}):
\begin{equation*}
\norm{\fara}^2_{H^{l+4}(\Sigma_{t_0^*})} \leq \varepsilon^2.
\end{equation*}
Here, $0 < \varepsilon < \tilde \varepsilon$. Under these conditions, we have the following estimate on $\mathcal R$:
\begin{equation}\label{eq:moraall}
\begin{aligned}
 \sum_{\substack{b + j + k \leq l\\
		I \in \iindiceo^b, b \geq 1	}} \int_{\regio{t_1^*}{t_2^*} } 
 \left\{\frac 1 {r^2}\left(|L \hat Z_{,I, j, k}|^2+|{\hat \lbar} \hat Z_{,I, j, k}|^2\right)+ \frac 1 r |\snabla \hat Z_{,I, j, k}|^2  +\frac 1 {r^3} |\hat Z_{,I, j, k}|^2 \right\} r^{-2}\de \text{Vol}  \leq C \varepsilon^2.
\end{aligned}
\end{equation}
Furthermore, the same estimate holds with $\hat Z$ replaced by $\hat W$.
\end{proposition}

\begin{proof}[Proof of Proposition~\ref{prop:moradef}]
We have, from Proposition~\ref{prop:redsh}, in the conditions given by the hypotheses,
\begin{equation*}
\begin{aligned}
& \sum_{\substack{1 \leq |I| + j \leq h\\
	I \in \iindiceo^b, b \geq 1}} \left\{
\int_{\Sigma_{t_2^*}} r^{-2} |\hat Z_{,I, j, k}|^2  r^{-2} \de \Sigma_{t^*}
+ \int_{\regio{t_1^*}{t_2^*}}  \frac{1-\mu}{r^3}|\hat Z_{,	I, j, k}|^2 \de u \de v
\right\}
\\ &
\leq
 C \sum_{\substack{1 \leq |H| \leq h-1\\ H \in \iindiceo^b, b\geq 1}} \int_{\regio{t_1^*}{t_2^*}} (1-\mu) r^{-3}  |\snabla \hat Z_{,H,0,k}|^2  \de u \de v
 +C\sum_{\substack{1 \leq |I| + j \leq h \\ I \in \iindiceo^b, b \geq 1}} \int_{\Sigma_{t_1^*}} r^{-2} |\hat Z_{,I, j, k}|^2 r^{-2} \de \Sigma_{t^*}\\ &
  +C\sum_{\substack{1 \leq |I| + j \leq h \\ I \in \iindiceo^b, b \geq 1}}\int_{\regio{t_1^*}{t_2^*} } (1-\mu)r^{-1} |\partial_{\leo}^I T^k {\hat \lbar}^{j-1}((1-\mu)^{-1}R_\rho)|^2 \de u \de v.
\end{aligned}
\end{equation*}
The terms on the right hand side are now controlled by data, the bootstrap assumptions, and by the estimate in Proposition~\ref{prop:moraeps}.
\end{proof}

We will use the following version of Proposition~\ref{prop:moradef}, valid in a region of bounded $r$:
\begin{proposition}
\label{prop:morabdd}
For all $R^* > 0$, there exist a constant $C > 0$, a small number $\tilde \varepsilon >0$, such that the following holds. Let $t_2^* \geq t_0^*$. Let $\fara$ be a smooth solution of the MBI system (\ref{MBI}) on $\mathcal{R}:=\regio{t^*_0}{t^*_2} \subset S_e$. Assume the bootstrap assumptions $BA\left(\mathcal{R},1,\left\lfloor \frac{l+4}{2}\right\rfloor , \varepsilon \right)$. Assume furthermore the smallness of initial energy~\eqref{eq:ensmall}:
\begin{equation*}
\norm{\fara}^2_{H^{l+4}(\Sigma_{t_0^*})} \leq \varepsilon^2.
\end{equation*}
Here, $0 < \varepsilon < \tilde \varepsilon$. Under these conditions, for all $\Omega \in \leo$, we have the following estimate on $\mathcal R$:
\begin{equation}\label{eq:morabdd}
\begin{aligned}
 \int_{\regio{t_1^*}{t_2^*} \cap \{r^* \leq R^*\}}   \sum_{i \leq l}
|\partial^i \Omega Z|^2\de \text{Vol}  \leq C \varepsilon^2.
\end{aligned}
\end{equation}
Furthermore, the same estimate holds with $Z$ replaced by $W$.
\end{proposition}

\begin{remark}
	Note that, in the above Proposition, we are estimating $\Omega Z$, since we have to exclude the zeroth mode. Recall that, throughout all the propositions in Section~\ref{sec:morawetz}, we always have to assume $I \in \iindiceo^b$, with $b \geq 1$.
\end{remark}

\section{\texorpdfstring{$r^p$}{rp} estimates: first improvement of the \texorpdfstring{$r$}{r}-weight } \label{sec:pweighted1}

The goal of this section is to prove the following proposition.

\begin{proposition}\label{prop:pprel}
Let $l \in \N_{\geq 0}$. There exists a constant $C >0$ such that the following holds true. Let $t_3^* \geq t_2^* \geq t_1^* \geq  t_0^*$. Let $\fara$ satisfy the MBI system (\ref{MBI}) on $\regio{}{}:= \regio{t_0^*}{t_3^*}$, assume furthermore the bootstrap assumptions $BA\left(\mathcal{R},1,\left\lfloor \frac{l+4}{2}\right\rfloor , \varepsilon \right)$. Assume that the initial Sobolev norm satisfies the following bound:
\begin{equation*}
\norm{\fara}^2_{H^{l+4}(\Sigma_{t_0^*})} \leq \varepsilon^2.
\end{equation*}
Let $\dot Z := \dot Z_{,I,0,k}$, with $|I| + k \leq l$, $I \in \iindiceo^b$, with $b \geq 0$.
Then, the following inequality holds true:
\begin{equation}\label{eq:punoprel}
\begin{aligned}
& \int_{\Sigma_{t^*_2}\cap \{r \geq 5M\}} \left[r |L \dot Z|^2 + r|\snabla \dot Z |^2\right]r^{-2}\de \Sigma_{t^*} + \int_{\regio{t_0^*}{t_2^*}\cap \{r \geq 5M\}} \left[|L \dot Z|^2 + |\snabla \dot Z|^2 \right] (1-\mu) \de u \de v \desphere \\ &
\leq C \int_{\Sigma_{t_0^*}} \left[r |L \dot Z|^2 + r|\snabla \dot Z |^2\right] r^{-2}\de \Sigma_{t^*}+ C \varepsilon^2.
\end{aligned}
\end{equation}
The same bound holds true correspondingly for $\dot W$.
\end{proposition}

\begin{remark}
Note that in this proposition $b$ is allowed to be $0$, differently from before. The reason is that these estimates do not ``see'' the zeroth mode.
\end{remark}

\begin{proof}[Proof of Proposition~\ref{prop:pprel}]
Let us consider the identity, which follows from Equation (\ref{eq:zdot}), valid upon integration on a sphere of constant $r$:
\begin{equation}\label{eq2.15}
 \begin{aligned}
& \lbar\left\{f(r) r^p|L \dot Z|^2 \right\} + L\left\{f(r)(1-\mu)r^p|\snabla \dot Z|^2 \right\} \\ &
- \lbar\left(f(r) r^p\right)|L \dot Z|^2 - r^2 L \left(f(r) (1-\mu)r^{p-2} \right)|\snabla \dot Z|^2 \\ & \stackrel{\mathbb{S}^2}{=}
2r^p f(r) (L \dot Z) \partial^I_{\leo} T^k (R_\rho).
 \end{aligned}
\end{equation}
We let $p =1$, and we choose $f(r)$ smooth such that $0 \leq f(r) \leq 1$, $f(r) = 0$ for $r \in [2M, 4M]$, $\frac{d}{dr} f(r) \geq 0$ for $r \in [2M, \infty)$, and $f(r)=1$ for $r \geq 5M$. We integrate identity (\ref{eq2.15}) on $\regio{t_0^*}{t_2^*} \cap \{u \leq \tilde u\} \cap \{v \leq \tilde v\}$.
Upon discarding the positive boundary terms, after an application of inequality~\eqref{eq:moradegeps} to bound the error term in the region $r \in [4M, 5M]$, we obtain
\begin{equation}\label{eq:rpuno}
\begin{aligned}
& \int_{\Sigma_{t^*_2}\cap \{r \geq 5M\}} \left[r |L \dot Z|^2 + r|\snabla \dot Z |^2\right]r^{-2}\de \Sigma_{t^*}\\ &
+\int_{\regio{t_0^*}{t_2^*} \cap \{r \geq 5M\}} \left[|L \dot Z|^2 + |\snabla \dot Z|^2 \right] (1-\mu) \de u \de v \desphere \leq 
C \int_{\Sigma_{t_0^*}} (r |L \dot Z|^2) r^{-2}\de \Sigma_{t^*}\\ &
+ C \int_{\regio{t_0^*}{t_2^*}\cap \{r \geq 4M\}} \left[ r |L\dot Z| \cdot |\partial^I_{\leo} T^k (R_\rho)| \right] (1-\mu) \de u \de v \desphere\\
&  + C \int_{\regio{t_0^*}{t_2^*}\cap \{ 4M \leq r \leq 5M\}} |\snabla \dot Z|^2 \de u \de v \desphere + C \varepsilon^2.
\end{aligned}
\end{equation}
We now use the Cauchy--Schwarz inequality with a small parameter $\eta > 0$ to bound the second term in the RHS of~\eqref{eq:rpuno}, ``hiding'' the term in $L \dot Z$ into the first term of the second integral in the LHS:
\begin{equation*}
\begin{aligned}
& \int_{\regio{t_0^*}{t_2^*}\cap \{r \geq 4M\}} \left[ r |L\dot Z| \cdot |\partial^I_{\leo} T^k (R_\rho)| \right] (1-\mu) \de u \de v \desphere \\
 &\leq C \int_{\regio{t_0^*}{t_2^*}\cap \{4M \leq r \leq 5M\}} |L\dot Z|^2  (1-\mu) \de u \de v \desphere +\\
 &\qquad + \eta  \int_{\regio{t_0^*}{t_2^*}\cap \{r \geq 5M\}}  |L\dot Z|^2(1-\mu) \de u \de v \desphere + C  \int_{\regio{t_0^*}{t_2^*}\cap \{r \geq 4M\}}r^2 |\partial^I_{\leo} T^k (R_\rho)|^2  (1-\mu) \de u \de v \desphere.
\end{aligned}
\end{equation*}
We then bound the resulting error term by use of the inequality 
\begin{equation}
 \int_{\regio{t_0^*}{t_2^*}\cap \{r \geq 4M\}}r^2 |\partial^I_{\leo} T^k (R_\rho)|^2 \de u \de v \desphere \leq C \varepsilon^2,
\end{equation}
The last bound indeed follows from Lemma~\ref{lem:bsrho}, plus Lemma~\ref{lem:bssigma}, plus the $L^2$ bounds of Proposition~\ref{prop:energy}. Finally, the terms localized to $\{4M \leq r \leq 5M\}$ are bounded by means of the degenerate Morawetz estimate of Proposition~\ref{prop:moradegeps}.
\end{proof}

\section{Improving the \texorpdfstring{$r$}{r}-weights on the \texorpdfstring{$L^2$}{L2} norms of \texorpdfstring{$\rho$}{rho}, \texorpdfstring{$\sigma$}{sigma} and \texorpdfstring{$\alpha$}{alpha} on \texorpdfstring{$\Sigma_{t^*}$}{Sigma}}\label{sec:l2improved}

In this section we are going to prove stronger (in terms of $r$-weights) integrated estimates for the ``good'' components, i.e. $\alpha$, $\rho$, $\sigma$. This will enable us to prove the correct decay rates in order to close the bootstrap argument.

\subsection{Improved \texorpdfstring{$L^2$}{L2} estimates for \texorpdfstring{$\alpha$}{alpha}}
In this section, we wish to use the equation satisfied by $\alpha$ (\ref{eq:alpha}) plus the $r^p$ estimates (with $p=1$) obtained in the previous section in order to improve the weight on the $L^2$ estimates for $\alpha$.
Recall the definition of $\hat \lbar := (1-\mu)^{-1}\lbar$.

\begin{proposition}[Improved weights on fluxes of $\alpha$]\label{prop:impalpha} Let $l \in \N_{\geq 0}$. There exist a constant $C$ and a small number $\tilde \varepsilon > 0$ such that the following holds. Let $t_3^* \geq t_2^* \geq t_1^* \geq  t_0^*$. Let $\fara$ satisfy the MBI system (\ref{MBI}) on $\regio{}{}:= \regio{t_0^*}{t_3^*}$. Assume furthermore the bootstrap assumptions $BA\left(\mathcal{R},1,\left\lfloor \frac{l+4}{2}\right\rfloor , \varepsilon \right)$. Assume that the initial Sobolev norm satisfies the following bound:
\begin{equation*}
\norm{\fara}^2_{H^{l+4}(\Sigma_{t_0^*})} \leq \varepsilon^2.
\end{equation*}
Here, $0 < \varepsilon < \tilde \varepsilon$. Then, the following inequality holds true:
\begin{equation}\label{eq:alphal2def}
\begin{aligned}
&\int_{\Sigma_{t^*_2}} r |\partial^m \alpha|^2 \de \Sigma_{t^*} \leq C \varepsilon^2 + C\sum_{h = 0}^m \int_{\Sigma_{t^*_0}} r |\partial^h \alpha|^2 \de \Sigma_{t^*}\\
&+ C \sum_{\substack{|I|+k \leq m \\ I \in \iindiceo^b, b \geq 0}} \int_{\Sigma_{t_0^*}} \left[r |L \dot Z_{,I,0,k}|^2 + r|\snabla \dot Z_{,I,0,k} |^2\right] r^{-2}\de \Sigma_{t^*} \\
&+ C \sum_{\substack{|I|+k \leq m\\ I \in \iindiceo^b, b \ge 0}} \int_{\Sigma_{t_0^*}} \left[r |L \dot W_{,I,0,k}|^2 + r|\snabla \dot W_{,I,0,k} |^2\right] r^{-2}\de \Sigma_{t^*} 
\end{aligned}
\end{equation}
for all $m \leq l$, $t_3^* \geq t_2^* \geq t_0^*$.
\end{proposition}

\begin{remark}
Estimate~\eqref{eq:alphal2def} ``loses one derivative'', as we have $m$ derivatives of $\alpha$ on the LHS, but $m+1$ derivatives of $Z$ and $W$ on the RHS. This is acceptable because the boundary terms in $Z$ and $W$ on the RHS has been estimated using inequality~\eqref{eq:punoprel}, which is in turn essentially obtained by an $r^p$ estimate performed between two spacelike slices. In particular, we obtained control on \emph{high-order} derivatives of spacelike fluxes of $Z$ and $W$, with weights in $r$, from estimate~\eqref{eq:punoprel}. Using this inequality, combined with inequality~\eqref{eq:alphal2def}, gives control on fluxes of $\alpha$ with strong $r$-weights and \emph{high-order} derivatives. Since the control is high-order in the number of derivatives, loss of only one derivative will not be an issue when closing the estimates.
\end{remark}

\begin{proof}
First of all, in \textbf{Step 1}, we derive the zeroth-order estimates, which are also valid for derivatives in direction of the Killing fields. Then, in \textbf{Step 2}, we will commute the equation by $\snabla_\lbar$ to obtain higher-order control on the flux of $\snabla_\lbar \alpha$.

Let us recall that $\alpha$ satisfies the following equation:
\begin{equation}\label{eq:alpharef}
\snabla_\lbar(r \alpha_A) = r(1-\mu)(\snabla_A \rho + \svol_{AB}\snabla^B \sigma ) + r(1-\mu)(\hdelta)_{\mu A \kappa \lambda} \nabla^{\mu} \fara^{\kappa\lambda}.
\end{equation}

\subsection*{Step 1}
We apply the operator $\slie_{\leo}^I \snabla_T^k$, with $|I|+k \leq l$, to both sides of the previous display. We have, denoting $\dot \alpha := \slie^I_{\leo} \snabla_T^k \alpha$,
\begin{equation}\label{eq:alpharefcom}
\snabla_\lbar(r \dot \alpha_A) = r(1-\mu)(r^{-2 }\snabla_A Z_{,I,0,k} + r^{-2}\svol_{AB}\snabla^B W_{I,0,k} ) + (1-\mu)\slie^I_{\leo} \snabla_T^k \left(r(\hdelta)_{\mu A \kappa \lambda} \nabla^{\mu} \fara^{\kappa\lambda}\right).
\end{equation}
Let $f(r)$ be smooth, nondecreasing, such that $0 \leq f(r) \leq 1$, $f(r) = 0$ for $r \in [2M, 5M]$, and $f(r)=1$ for $r \geq 6M$. We then have, from (\ref{eq:alpharef}), letting $\eta >0$ be a small parameter,
\begin{equation*}
\begin{aligned}
&\lbar(r f(r) |r \alpha|^2) = \lbar( r f(r) ) r^2 |\alpha|^2 + 2 f(r) r (r \alpha^A) \snabla_\lbar (r \alpha_A) \\ & =
\lbar( r f(r)) r^2 |\alpha|^2 + 2 f(r) r^2 \alpha^A \left(r(1-\mu)(\snabla_A \rho + \svol_{AB}\snabla^B \sigma ) + r(1-\mu)(\hdelta)_{\mu A \kappa \lambda} \nabla^{\mu} \fara^{\kappa\lambda} \right) \\ &
\leq - (1-\mu) f(r)r^2 |\alpha|^2 + \eta (1-\mu) r^2 f(r) |\alpha|^2 + \eta^{-1} f(r) (1-\mu)  (|\snabla Z|^2 + |\snabla W|^2)\\
& + r^{4} f(r) (1-\mu)   |(\hdelta)_{\mu A \kappa \lambda} \nabla^{\mu} \fara^{\kappa\lambda}|^2.
\end{aligned}
\end{equation*}
Similarly, from (\ref{eq:alpharefcom}), we have, again with $\eta >0 $ a small parameter,
\begin{equation*}
\begin{aligned}
& \lbar(r f(r) |r \dot \alpha|^2) = \lbar( r f(r) ) r^2 |\dot \alpha|^2 + 2 f(r) r (r \dot \alpha^A) \snabla_\lbar (r \dot \alpha_A) \\ & 
\leq - (1-\mu) f(r)r^2 |\dot \alpha|^2 +2 \eta (1-\mu) r^2 f(r) |\dot \alpha|^2 \\ &
+ \eta^{-1} f(r) (1-\mu)  (|\snabla \dot Z_{,I,0,k}|^2 + |\snabla \dot W_{,I,0,k}|^2) + \eta^{-1}r^{4} f(r) (1-\mu)   |\slie_{\leo}^I \snabla_T^k \left((\hdelta)_{\mu A \kappa \lambda} \nabla^{\mu} \fara^{\kappa\lambda}\right)|^2.
\end{aligned}
\end{equation*}
Integrating such equation on $\regio{t^*_1}{t^*_2}$, using the volume form $\de u \de v \desphere$, we have then

\begin{equation*}
\begin{aligned}
& \int_{\Sigma_{t^*_2} \cap \{r \geq 6M\}} r |\dot \alpha|^2 \de \Sigma_{t^*} + \int_{\regio{t^*_1}{t^*_2}\cap\{r \geq 6M\}}   |\dot \alpha|^2\de \text{Vol}\\ &
\leq C \int_{\regio{t^*_1}{t^*_2}\cap \{r \geq 5M\}} \left( \underbrace{(1-\mu)  (|\snabla \dot Z_{,I,0,k}|^2 + |\snabla \dot W_{,I,0,k}|^2)}_{(*)} + r^2 (1-\mu) |\slie_{\leo}^I \snabla_T^k \left(r(\hdelta)_{\mu A \kappa \lambda} \nabla^{\mu} \fara^{\kappa\lambda}\right)|^2\right)\de u \de v \desphere\\ &
+ C\int_{\Sigma_{t^*_1}\cap \{r \geq 5M\}} r |\alpha|^2 \de \Sigma_{t^*}.
\end{aligned}
\end{equation*}
By (\ref{eq:punoprel}), we can bound the terms $(*)$. The other terms can be estimated by the uniform $L^2$ control of Equation~(\ref{eq:unifl2}):
\begin{equation*}
\norm{\fara}^2_{H^{l+4}(\Sigma_{t^*})} \leq \varepsilon^2, \qquad \text{for } t^* \in [t_0^*, t_3^*],
\end{equation*}
plus the bootstrap assumptions.

We finally obtain:
\begin{equation}\label{eq:alphal2pre}
\int_{\Sigma_{t^*_2}} r |\dot \alpha|^2 \de \Sigma_{t^*} \leq C \varepsilon^2 + \int_{\Sigma_{t^*_0}} r |\dot \alpha|^2 \de \Sigma_{t^*}.
\end{equation}
\subsection*{Step 2}
We would now like to estimate the $\lbar$-derivatives. We let $j \leq l$. We use again Equation (\ref{eq:alpharef}) in order to deduce
\begin{equation*}
\snabla^j_\lbar \alpha_A = \snabla_\lbar^{j-1}\left(\frac{1-\mu}{r} \alpha_A \right)
+ \snabla^{j-1}_\lbar \left( \frac{1-\mu}{r}  \left[(r \snabla)_A \rho + r\svol_{AB} \snabla^B \sigma \right]\right) + \snabla^{j-1}_\lbar\left( (1-\mu)(\hdelta)_{\mu A \kappa \lambda} \nabla^{\mu} \fara^{\kappa\lambda}\right).
\end{equation*}
This implies, taking absolute values, on the region $\{r \geq 5M\}$,
\begin{equation*}
|\snabla^j_\lbar \alpha|^2 \leq C \frac 1 {r^2} \sum_{h=0}^{j-1} |\snabla_\lbar^h \alpha|^2 + \frac C {r^2} \sum_{h=0}^{j-1} (|r\snabla \snabla_\lbar^h \rho|^2+|r\snabla \snabla_\lbar^h \sigma|^2) + |\snabla_\lbar^{j-1} \left((1-\mu)(\hdelta)_{\mu A \kappa \lambda} \nabla^{\mu} \fara^{\kappa\lambda}\right)|^2.
\end{equation*}
Integration on the surface $\Sigma_{t_2^*}$ yields
\begin{equation*}
\begin{aligned}
&\int_{\Sigma_{t^*_2}} r|\snabla^j_\lbar \alpha|^2 \de \Sigma_{t^*} 
\leq C\int_{\Sigma_{t^*_2}} \frac 1 {r} \sum_{h=0}^{j-1} |\snabla_\lbar^h \alpha|^2 \de \Sigma_{t^*}\\
&+C \int_{\Sigma_{t^*_2}}  \frac 1 {r} \sum_{h=0}^{j-1} (|r\snabla \snabla_\lbar^h \rho|^2+|r\snabla \snabla_\lbar^h \sigma|^2) \de \Sigma_{t^*}
+ C \int_{\Sigma_{t^*_2}} r|\snabla_\lbar^{j-1} \left((1-\mu)(\hdelta)_{\mu A \kappa \lambda} \nabla^{\mu} \fara^{\kappa\lambda}\right)|^2 \de \Sigma_{t^*}.
\end{aligned}
\end{equation*}
We now use the bootstrap assumptions and the uniform $L^2$ bounds on $\sigma$, $\rho$ arising from inequality~(\ref{eq:unifl2}):
\begin{equation}\label{eq:l2bdref}
\norm{\fara}^2_{H^{l+4}(\Sigma_{t^*})} \leq \varepsilon^2, \qquad \text{for } t^* \in [t_0^*, t_3^*],
\end{equation}
to deduce
\begin{align}
\label{eq:withflux1} \sum_{h = 0}^{l} \int_{\Sigma_{t^*_2}} r(|\snabla \snabla_\lbar^h \rho|^2 + |\snabla \snabla_\lbar^h \sigma|^2 )\de\Sigma_{t^*}  &\leq C \varepsilon^2,\\
\label{eq:withflux2} \sum_{h = 1}^{l} \int_{\Sigma_{t^*_2}} r|\snabla_\lbar^{h-1} \left((1-\mu)(\hdelta)_{\mu A \kappa \lambda} \nabla^{\mu} \fara^{\kappa\lambda}\right)|^2 \de \Sigma_{t^*} &\leq C \varepsilon^2.
\end{align}
Estimate~\eqref{eq:withflux1} follows from the fact that
$$
\sum_{h = 0}^{l} \int_{\Sigma_{t^*_2}} r(|\snabla \snabla_\lbar^h \rho|^2 + |\snabla \snabla_\lbar^h \sigma|^2 )\de\Sigma_{t^*} \leq C  \int_{\Sigma_{t^*_2}} (|\partial^{\leq l+1} \rho|^2 +|\partial^{\leq l+1} \sigma|^2 )\de\Sigma_{t^*}  \leq C \norm{\fara}^2_{H^{l+4}(\Sigma_{t^*})}.
$$
The last inequality in the display above follows from the fact that $ \norm{\fara}^2_{H^{l+4}(\Sigma_{t^*})}$ clearly controls the norms expressed in terms of null components (see inequality~\eqref{eq:fnormnull} in Remark~\ref{rmk:fnormnull}). On the other hand, estimate~\eqref{eq:withflux2} follows from the bootstrap assumptions and the bound~\eqref{eq:l2bdref}.

An induction argument then implies, together with the uniform $L^2$ estimates (\ref{eq:l2bdref}), for $0 \leq m \leq l$:
\begin{equation}\label{eq:alphal2der}
\sum_{h = 0}^{m} \int_{\Sigma_{t^*_2}} r |\snabla^h_\lbar \alpha|^2 \de \Sigma_{t^*} \leq C \varepsilon^2 + C\sum_{h = 0}^{m}  \int_{\Sigma_{t^*_0}} r |\partial^h \alpha|^2 \de \Sigma_{t^*}.
\end{equation}
Recall the definition of $\hat{\lbar} := (1-\mu)^{-1} \lbar$. By the $L^2$ bounds (\ref{eq:l2bdref}) and estimate (\ref{eq:alphal2der}), we remove the degeneracy at $r = 2M$ and obtain
\begin{equation}\label{eq:alphal2higher}
\sum_{h = 0}^{m}  \int_{\Sigma_{t^*_2}} r |\snabla^h_{\hat \lbar} \alpha|^2 \de \Sigma_{t^*} \leq C \varepsilon^2.
\end{equation}
It is now standard to extend the bound to all the mixed derivatives of $\alpha$ of order $m \leq l$. Finally, the equivalence of norms stated in Proposition~\ref{norm.equiv}, implies the claim.
\end{proof}

\subsection{Improved \texorpdfstring{$L^2$}{L2} estimates for \texorpdfstring{$\sigma$}{sigma} and \texorpdfstring{$\rho$}{rho}}
In this subsection, we wish to use inequality \eqref{eq:punoprel} in order to improve the $r$-weight on the $L^2$ estimates for $\rho$ and $\sigma$ on the spacelike foliation $\Sigma_{t^*}$. First of all, notice that estimate~\eqref{eq:punoprel} already constituted an improvement on the $L^2$ estimates for $\sigma$ and $\rho$ in Proposition~\ref{prop:energy}. In fact, estimate~\eqref{eq:unifl2} in Proposition~\ref{prop:energy} gives, assuming the bootstrap assumptions $BA\left(\mathcal{R},1,\left\lfloor \frac{l+4}{2}\right\rfloor , \varepsilon \right)$,
$$
 \int_{\widetilde{\Sigma}_{t^*}} \big( |\partial^{\leq l+4} \rho|^2 + |\partial^{\leq l+4} \sigma|^2 \big) \de \widetilde \Sigma_{t^*} \leq C \varepsilon^2.
 $$
On the other hand, estimate~\eqref{eq:punoprel} gives us, schematically,
$$
 \sum_{\substack{b+k \leq  l \\ I \in \mathscr{I}^b_{\leo}}} \int_{\widetilde{\Sigma}_{t^*}} r |L \partial_{\leo}^I  T^k \rho|^2 \de \widetilde \Sigma_{t^*}  + \sum_{\substack{b+k \leq l+1 \\ b \geq 1 \\ I \in \mathscr{I}^b_{\leo}}} \int_{\widetilde{\Sigma}_{t^*}} r|\partial_{\leo}^I  T^k \rho|^2  \de \widetilde \Sigma_{t^*}  \leq C \varepsilon^2.
 $$
 Hence, we see the improvement in the $r$-weight, and we also see that the estimate~\eqref{eq:punoprel} only controls higher-order derivatives in the directions of $T$ and the angular directions. Therefore, the purpose of Proposition~\ref{prop:imprhosigma} will be to extend this (improved in $r$) control to all high-order derivatives of $\rho$ (and $\sigma$), including the directions parallel to $\lbar$.

\begin{proposition}[Improved weights on fluxes of $\rho$ and $\sigma$]\label{prop:imprhosigma}
There exist a small number $\tilde \varepsilon > 0$ and a constant $C >0$ such that the following holds. Let $t_3^* > t_2^*\geq t_1^* \geq t_0^*$. Let $\regio{}{}:= \regio{t_0^*}{t_3^*}$. Let $\fara$ satisfy the MBI system (\ref{MBI}) on $\regio{}{}$, assume furthermore the bootstrap assumptions $BA\left(\mathcal{R},1,\left\lfloor \frac{l+4}{2}\right\rfloor , \varepsilon \right)$. Assume also that the initial Sobolev norm satisfies the following bound:
\begin{equation*}
\norm{\fara}^2_{H^{l+4}(\Sigma_{t_0^*})} \leq \varepsilon^2.
\end{equation*}
Here, $0 < \varepsilon < \tilde \varepsilon$. Then, the following inequality holds true:
\begin{equation}\label{eq:rhosimproved}
\begin{aligned}
&\int_{\Sigma_{t^*_2}} r |\partial^m \rho|^2 \de \Sigma_{t^*} + \int_{\Sigma_{t^*_2}} r |\partial^m \sigma|^2 \de \Sigma_{t^*} \leq C \varepsilon^2 \\
&\qquad + \sum_{\substack{|I| + k \leq m-1\\ I \in \iindiceo^b, b \geq 0}} \int_{\Sigma_{t_0^*}} \left[r |L \dot Z_{,I,0,k}|^2 + r|\snabla \dot Z_{,I,0,k} |^2\right] r^{-2} \de \Sigma_{t^*} \\
&\qquad + \sum_{\substack{|I| + k \leq m-1\\ I \in \iindiceo^b, b \geq 0}} \int_{\Sigma_{t_0^*}} \left[r |L \dot W_{,I,0,k}|^2 + r|\snabla \dot W_{,I,0,k} |^2\right]r^{-2}\de \Sigma_{t^*},
\end{aligned}
\end{equation}
for all $1 \leq m \leq l$.
\end{proposition}
 
\begin{remark}
Notice that we do not have the weighted bound on the $0$\textsuperscript{th} order terms (note that $m \geq 1$). On the other hand, we do not require the $0$\textsuperscript{th} order weighted integral term to be bounded (this is because, when we ish to estimate the nonlinear error terms, and we encounter a term with zero derivatives falling on it, we are allowed to just use the bootstrap assumptions to estimate that particular term).
\end{remark}

\begin{proof}[Proof of Proposition~\ref{prop:imprhosigma}]
We divide the proof in two steps. In \textbf{Step 1}, we will prove the claim for $\rho$, whereas in \textbf{Step 2}, we will prove the claim for $\sigma$.
\subsection*{Step 1}
We immediately notice that, if $b + k \leq l$, and $I \in \iindiceo^b$, we have, from Proposition~\ref{prop:pprel},
\begin{equation*}
\int_{\Sigma_{t^*_2}\cap \{r \geq 5M\}} \left[r |L \dot Z_{,I,0,k}|^2 + r|\snabla \dot Z_{,I,0,k} |^2\right]r^{-2}\de \Sigma_{t^*} \\
\leq C \int_{\Sigma_{t_0^*}} \left[r |L \dot Z_{,I,0,k}|^2 + r|\snabla \dot Z_{,I,0,k} |^2\right] r^{-2}\de \Sigma_{t^*}+ C \varepsilon^2.
\end{equation*}
This already proves the claim for angular derivatives and mixed angular and $T$-derivatives, when there is at least one angular derivative (looking at the second term in the LHS of the above display).
We now would like to use the transport equation satisfied by $\rho$ to extend this statement to all other derivatives. Let us recall that $\rho$ satisfies the following equation:
\begin{equation*}
-L(r^2 \rho) + r^2 \dive \alpha = - r^2{H_{_\Delta}} \indices{^\mu _L ^\kappa ^\lambda} \nabla_\mu \fara_{\kappa \lambda}.
\end{equation*}
This implies
\begin{equation*}
L(r \rho)+ (1-\mu) \rho - r \dive \alpha = r{H_{_\Delta}} \indices{^\mu _L ^\kappa ^\lambda} \nabla_\mu \fara_{\kappa \lambda}.
\end{equation*}
Let us apply $h$ times the differential operator $L$ to both sides of the previous display. We obtain
\begin{equation*}
L^{h+1}(r \rho) + L^h((1-\mu)\rho) - L^h(r\dive \alpha) = L^h \left(r{H_{_\Delta}} \indices{^\mu _L ^\kappa ^\lambda} \nabla_\mu \fara_{\kappa \lambda} \right)
\end{equation*}
Hence:
\begin{equation*}
\underbrace{\left|L^{h+1}(r \rho)\right|^2 }_{(1)} \leq \underbrace{\left|L^h((1-\mu)\rho)\right|^2}_{(2)}+ \underbrace{\left|  L^h(r\dive \alpha)\right|^2}_{(3)} + \underbrace{\left|L^h \left(r{H_{_\Delta}} \indices{^\mu _L ^\kappa ^\lambda} \nabla_\mu \fara_{\kappa \lambda} \right) \right|^2}_{(4)}
\end{equation*}
Let us multiply the previous equation by $r$ and integrate on $\Sigma_{t_2^*}$ with respect to the form $r^{-2} \de \Sigma_{t^*}$. Regarding term $(1)$, we have, by an application of the Leibniz rule,
\begin{equation*}
r\left|L^{h+1}(r \rho)\right|^2 \geq r^3|L^{h+1}\rho|^2 - C\Big(\sum_{i = 1}^{h+1}r |L^{i}(r^2) L^{h+1-i} \rho|^2 \Big) \geq r^3|L^{h+1}\rho|^2 - C\underbrace{r^2 \Big(\sum_{i = 1}^{h+1}|L^{h+1-i} \rho|^2\Big)}_{(1')}.
\end{equation*}

The terms $(1')$ and the terms arising from $(2)$ can be estimated by the uniform $L^2$ estimates of Equation~(\ref{eq:unifl2}) in Proposition~\ref{prop:energy} in the following way:
\begin{equation*}
\int_{\Sigma_{t^*_2}} \Big\{\left|L^h((1-\mu)\rho)\right|^2 + C r^2 \Big(\sum_{i = 1}^{h+1}|L^{h+1-i} \rho|^2\Big)\Big\} r^{-2}  \de \Sigma_{t^*} \leq C \norm{\fara}^2_{H^l(\Sigma_{t_2^*})} \leq C \varepsilon^2.
\end{equation*}

Similarly, we can estimate the terms in $(3)$ again by the $L^2$ estimates \eqref{eq:unifl2} (recall that those estimates give control over $r$-\emph{weighted} angular derivatives). We also estimate $(4)$ by the bootstrap assumptions and the uniform $L^2$ bounds of Equation~(\ref{eq:unifl2}), as in Lemma~\ref{lem:bsrho}.

We therefore obtain, for $m \leq l$,
\begin{equation*}
\sum_{h=1}^m \int_{\Sigma_{t_2^*}} r |L^h \rho|^2\de \Sigma_{t^*}\leq C \varepsilon^2.
\end{equation*}
The remaining derivatives in the $\lbar$ direction are obtained by the same method, given the fact that $\rho$ satisfies the transport equation:
\begin{equation*}
	\hat \lbar(r^2 \rho) + r^2 (1-\mu)^{-1}\dive \alphabar = - r^2{H_{_\Delta}} \indices{^\mu _{\hat \lbar} ^\kappa ^\lambda} \nabla_\mu \fara_{\kappa \lambda}.
\end{equation*}
Upon summation of the resulting inequalities with the uniform $L^2$ estimates of Equation~(\ref{eq:unifl2}) (in order to deal with the region of bounded $r$-coordinate), we obtain the claim for $\rho$.

\subsection*{Step 2} Let us now focus on the estimates for $\sigma$. As in the previous step, we have, from Proposition~\ref{prop:pprel}, if $b + k \leq l$, and $I \in \iindiceo^b$,
\begin{equation*}
\int_{\Sigma_{t^*_2}\cap \{r \geq 5M\}} \left[r |L \dot W_{,I,0,k}|^2 + r|\snabla \dot W_{,I,0,k} |^2\right]r^{-2}\de \Sigma_{t^*} \\
\leq C \int_{\Sigma_{t_0^*}} \left[r |L \dot W_{,I,0,k}|^2 + r|\snabla \dot W_{,I,0,k} |^2\right] r^{-2}\de \Sigma_{t^*}+ C \varepsilon^2.
\end{equation*}
Upon summation with the uniform $L^2$ estimates of Equation~(\ref{eq:unifl2}), we obtain the claim for angular derivatives. 

As in the reasoning for $\rho$, we now use the transport equation satisfied by $\sigma$:
\begin{equation*}
L(r^2 \sigma)+r^2 \curl \alpha = 0.
\end{equation*}
It follows, as before, that
\begin{equation*}
\sum_{h=1}^m \int_{\Sigma_{t_2^*}} r |L^h \sigma|^2\de \Sigma_{t^*}\leq C \varepsilon^2.
\end{equation*}
To extend the bound on derivatives in direction $\lbar$, we proceed as in the previous step, using the transport equation satisfied by $\sigma$:
\begin{equation*}
- \hat \lbar(r^2 \sigma) + (1-\mu)^{-1} r^2 \curl \alphabar= 0.
\end{equation*}
Finally, it is standard to extend the bound to all mixed derivatives. This concludes the proof.
\end{proof}

\section{Improved \texorpdfstring{$L^2$}{L2} bounds on \texorpdfstring{$R_\rho$}{Rr} and \texorpdfstring{$R_\sigma$}{Rs}}\label{sec:improvedrs}

Combining the reasoning appearing in the proof of Lemma~\ref{lem:bsrho} and Lemma~\ref{lem:bssigma} with the improved flux estimates of Section~\ref{sec:l2improved} we obtain the following.

\begin{lemma}[Improved $L^2$ bounds on $R_\rho$ and $R_\sigma$.]\label{lem:bsrhoimp}
Let $l \in \N_{\geq 0}$, $l \geq 2$. There exist a small number $\tilde \varepsilon >0$ and a constant $C >0$ such that the following holds. Let $\fara$ be a smooth solution of the MBI system (\ref{MBI}) on $\mathcal R := \mathcal{R}_{t_0^*}^{t_2^*} \subset S_e$. Assume the bootstrap assumptions $BA\left(\mathcal{R},1,\left\lfloor \frac{l+4}{2}\right\rfloor , \varepsilon \right)$. Assume the following weighted bounds on the initial energy:
\begin{equation}\label{eq:bdinit}
\begin{aligned}
&\sum_{|I| + k \leq l-1} \int_{\Sigma_{t_0^*}} \left[r |L \dot Z_{,I,0,k}|^2 + r|\snabla \dot Z_{,I,0,k} |^2\right] r^{-2} \de \Sigma_{t^*} \leq \varepsilon^2, \\
&\sum_{|I| + k \leq l-1} \int_{\Sigma_{t_0^*}} \left[r |L \dot W_{,I,0,k}|^2 + r|\snabla \dot W_{,I,0,k} |^2\right]r^{-2}\de \Sigma_{t^*} \leq \varepsilon^2,\\
&\int_{\Sigma_{t_0^*}} r |\partial^{\leq l} \alpha |^2\de \Sigma_{t^*} \leq \varepsilon^2, \\
&\norm{\fara}^2_{H^{l+4}(\Sigma_{t_0^*})} \leq \varepsilon^2.
\end{aligned}
\end{equation}
Here, $0 < \varepsilon < \tilde \varepsilon$. Under these conditions, we have the following integrated bounds, for $t_1^* \in [t_0^*, t_2^*]$, with $|I|+j+k \leq l-2$:
\begin{align}\label{eq:boundrrhoimp2}
&\int_{\Sigma_{t_1^*}} r ( |\partial_{\leo}^I T^k  {\hat \lbar}^{j}(R_\rho)|^2+|\partial_{\leo}^I T^k  {\hat \lbar}^{j}((1-\mu)^{-1}R_\rho)|^2) \de \Sigma_{t^*} \leq C \varepsilon^3 (t_1^*)^{-2},\\ \label{eq:boundrrhoimp1}
&\int_{\Sigma_{t_1^*}}  ( |\partial_{\leo}^I T^k  {\hat \lbar}^{j}(R_\rho)|^2+|\partial_{\leo}^I T^k  {\hat \lbar}^{j}((1-\mu)^{-1}R_\rho)|^2) \de \Sigma_{t^*} \leq C \varepsilon^3 (t_1^*)^{-3},\\ \label{eq:boundrrhoimp0}
&\int_{\Sigma_{t_1^*}} r^{-1} ( |\partial_{\leo}^I T^k  {\hat \lbar}^{j}(R_\rho)|^2+|\partial_{\leo}^I T^k  {\hat \lbar}^{j}((1-\mu)^{-1}R_\rho)|^2) \de \Sigma_{t^*} \leq C \varepsilon^3 (t_1^*)^{-3},
\end{align}
\begin{align}
\label{eq:boundrsigmaimp2}
&\int_{\Sigma_{t_1^*}} r ( |\partial_{\leo}^I T^k  {\hat \lbar}^{j}(R_\sigma)|^2+|\partial_{\leo}^I T^k  {\hat \lbar}^{j}((1-\mu)^{-1}R_\sigma)|^2) \de \Sigma_{t^*} \leq C \varepsilon^3 (t_1^*)^{- 2},\\ \label{eq:boundrsigmaimp1}
&\int_{\Sigma_{t_1^*}}  ( |\partial_{\leo}^I T^k  {\hat \lbar}^{j}(R_\sigma)|^2+|\partial_{\leo}^I T^k  {\hat \lbar}^{j}((1-\mu)^{-1}R_\sigma)|^2) \de \Sigma_{t^*} \leq C \varepsilon^3 (t_1^*)^{-3},\\ \label{eq:boundrsigmaimp0}
&\int_{\Sigma_{t_1^*}} r^{-1} ( |\partial_{\leo}^I T^k  {\hat \lbar}^{j}(R_\sigma)|^2+|\partial_{\leo}^I T^k  {\hat \lbar}^{j}((1-\mu)^{-1}R_\sigma)|^2) \de \Sigma_{t^*} \leq C \varepsilon^3 (t_1^*)^{-3}.
\end{align}
\end{lemma}

\begin{proof}
The proof is a straightforward adaptation of the proof of Lemma~\ref{lem:bsrho}, using the fact that now we can ``incorporate more $r$-weight'' in the $L^2$-norm of $\alpha, \rho, \sigma$ and their derivatives. We focus on the estimates for the integrals in $R_\rho$, as the remaining three estimates ((\ref{eq:boundrsigmaimp0}), (\ref{eq:boundrsigmaimp1}), (\ref{eq:boundrsigmaimp2})) are proved analogously.

\subsection*{Step 1: proof of (\ref{eq:boundrrhoimp2}) and (\ref{eq:boundrrhoimp1})}
As just stated, we focus on the case of $R_\rho$, and in particular on the estimates relative to the term $(b)$ in $\nonl{2}$, the reasoning for the other terms being analogous. We perform the same calculations as in \textbf{Step 2} of the proof of Lemma~\ref{lem:bsrho}, until we arrive at the expression
\begin{equation*}
\begin{aligned}
& |\partial_{\leo}^I T^k {\hat \lbar}^{j}\left((1-\mu)^{-1}r^2 L \left( {H_{_\Delta}}\indices{^\mu _\lbar ^\kappa ^\lambda} \nabla_{\mu}\fara_{\kappa\lambda}\right)\right)|  \\ &
\leq 
C \left| \partial^{l-2} \left((1-\mu)^{-1} r^2 L \left({H_{_\Delta}}\indices{^\mu _\lbar ^\kappa ^\lambda} \nabla_{\mu}\fara_{\kappa\lambda} \right)\right)\right| \\ &
\leq C\underbrace{\left|\partial^{l-2} \left\{(1-\mu)^{-1} rL  \left(r{H_{_\Delta}}\indices{^\mu _\lbar ^\kappa ^\lambda} \nabla_{\mu}\fara_{\kappa\lambda}  \right)\right\} \right|}_{(x)} + C\underbrace{ \left|\partial^{\leq l-2} \left( r{H_{_\Delta}}\indices{^\mu _\lbar ^\kappa ^\lambda} \nabla_{\mu}\fara_{\kappa\lambda} \right)\right|}_{(y)},
\end{aligned}
\end{equation*}
Now,
\begin{equation*}
(x) + (y) \leq C r \left|\partial^{\leq l-1} \left(r{H_{_\Delta}}\indices{^\mu _\lbar ^\kappa ^\lambda} \nabla_{\mu}\fara_{\kappa\lambda}  \right)\right|+Cr \left|\partial^{\leq l-2} \left( r{H_{_\Delta}}\indices{^\mu _\lbar ^\kappa ^\lambda} \nabla_{\mu}\fara_{\kappa\lambda} \right)\right|.
\end{equation*}
Since the structure of the two terms in the last display is the same, we can just focus on $r \left|\partial^{\leq l-1} \left(r{H_{_\Delta}}\indices{^\mu _\lbar ^\kappa ^\lambda} \nabla_{\mu}\fara_{\kappa\lambda}  \right)\right|$. We have
\begin{equation}
\begin{aligned}
& r\left| \partial^{\leq l-1} \left( r \tensor{{H_{_\Delta}}}{^\mu _\lbar ^\kappa ^\lambda} \nabla_{\mu}\fara_{\kappa\lambda} \right)\right| \\ &
\leq C r\sum_{a + b \leq l} (|\partial^{a+1}\lun| + |\partial^{a+1}\ldu|)(|\partial^b\alphabar| + r|\partial^b\rho| + r|\partial^b\sigma|) \\ &
\leq Cr \sum_{\substack{a + b +c\leq l
\\ 
a+c \geq 1
}} 
\left( |\partial^a \rho| |\partial^c \rho|+|\partial^a \sigma| |\partial^c \rho|+|\partial^a \sigma| |\partial^c \sigma| + |\partial^a \alpha| |\partial^c \alphabar|\right)
(|\partial^b\alphabar| + r|\partial^b\rho| + r|\partial^b\sigma|).
\end{aligned}
\end{equation}
Let us analyse the different terms in the last display. The bootstrap assumptions now imply that $\rho$, $\sigma$ and $\alpha$ have the same decay rates. It is therefore enough to bound the following expression, the remaining terms being treated analogously:
\begin{equation*}
r \sum_{\substack{a + b +c\leq l
\\
a+c \geq 1
}} 
\left( |\partial^a \rho| |\partial^c \rho|+ |\partial^a \alpha| |\partial^c \alphabar|\right)
(|\partial^b\alphabar| + r|\partial^b\rho|).
\end{equation*}
Let us consider each term in the product:
\begin{equation*}
\begin{aligned}
r \sum_{\substack{a + b +c\leq l
\\
a+c \geq 1
}} |\partial^a \rho| |\partial^c \rho| |\partial^b \alphabar| & \leq C\eps \tau^{-2}r^{-2}|\partial^{\leq l} \fara| + C\eps\tau^{-3/2}r^{-2}|\partial^{\leq l} \fara|,\\
r^2 \sum_{\substack{a + b +c\leq l
\\
a+c \geq 1
}} |\partial^a \rho| |\partial^c \rho| |\partial^b \rho| & \leq C \eps\tau^{-3/2}r^{-3/2} |\partial^{\leq l} \fara|,\\
r \sum_{\substack{a + b +c\leq l
\\
a+c \geq 1
}} |\partial^a \alpha| |\partial^c \alphabar| |\partial^b \alphabar| &\leq C \eps\tau^{-2}r^{-3/2} r^{1/2}|\partial^{\leq l} \alpha|+ C \eps\tau^{-3/2}r^{-2}|\partial^{\leq l} \fara|,\\
r^2 \sum_{\substack{a + b +c\leq l
\\
a+c \geq 1
}} |\partial^a \alpha| |\partial^c \alphabar| |\partial^b \rho|& \leq C \eps \tau^{-3/2}r^{-
3/2} r^{1/2}|\partial^{\leq l} \alpha|\\
&+ C \eps\tau^{-3/2}r^{-3/2} |\partial^{\leq l} \fara| + C\eps \tau^{-3/2}r^{-3/2} \sum_{h=1}^{l}r^{1/2}|\partial^{h} \rho|.
\end{aligned}
\end{equation*}
We therefore obtain
\begin{equation*}
r\left|\partial_{\leo}^I T^k {\hat \lbar}^{j}\left((1-\mu)^{-1}r^2 L \left( {H_{_\Delta}}\indices{^\mu _\lbar ^\kappa ^\lambda} \nabla_{\mu}\fara_{\kappa\lambda}\right)\right)\right|^2 \leq C \eps^2\tau^{-3}r^{-2} \left( |\partial^{\leq l} \fara|^2+r|\partial^{\leq l} \alpha|^2 + \sum_{h=1}^l r|\partial^h \rho|^2 \right).
\end{equation*}
An analogous reasoning gives us:
\begin{equation}\label{eq:toboundrrhoimp}
r ( |\partial_{\leo}^I T^k  {\hat \lbar}^{j}(R_\rho)|^2+|\partial_{\leo}^I T^k  {\hat \lbar}^{j}((1-\mu)^{-1}R_\rho)|^2)\leq C \eps^2 \tau^{-3}r^{-2} \left( |\partial^{\leq l} \fara|^2+r|\partial^{\leq l} \alpha|^2 + \sum_{h=1}^{l}r|\partial^h \rho|^2 \right).
\end{equation}
Integrating this bound on $\Sigma_{t_1^*}$, with $t_1^* \in [t_0^*, t_2^*]$, combining it with the uniform $L^2$ estimates of Equation~(\ref{eq:unifl2}), plus the weighted bounds in Proposition~\ref{prop:impalpha} and Proposition~\ref{prop:imprhosigma}, yields the claims (\ref{eq:boundrrhoimp2}) and (\ref{eq:boundrrhoimp1}).

\subsection*{Step 2: proof of (\ref{eq:boundrrhoimp0})} We repeat the same reasoning as in \textbf{Step 1}, until we get to inequality (\ref{eq:toboundrrhoimp}). Due to the different weight in $r$, it becomes
\begin{equation*}
r^{-1}( |\partial_{\leo}^I T^k  {\hat \lbar}^{j}(R_\rho)|^2+|\partial_{\leo}^I T^k  {\hat \lbar}^{j}((1-\mu)^{-1}R_\rho)|^2)\leq C \eps^2 \tau^{-3}r^{-3} \left( |\partial^{\leq l} \fara|^2+r|\partial^{\leq l} \alpha|^2 + \sum_{h=1}^{l}r|\partial^h \rho|^2 \right).
\end{equation*}
It is now straightforward to conclude the bound (\ref{eq:boundrrhoimp0}), simply by the fact that there exists a positive constant $C_1$ such that $\tau \cdot r \geq C_1 v$, in the region $\regio{t_0^*}{\infty}$.
\end{proof}

\section{The full hierarchy of \texorpdfstring{$r^p$}{rp}-estimates}\label{sec:previsited}
The improved bounds in the previous section let us extend the hierarchy of $p$-weighted estimates in Section~\ref{sec:pweighted1} to the weight $p = 2$. This in turn enables us to close the estimates on the nonlinear terms.

Fix a number $R^* > 0$ and a spacetime region $\regio{}{}\subset \mathcal{S}_e$. Given $Z$ a smooth function, $Z: \mathcal{S}_e \mapsto \R$, let us define the following quantities:
\begin{equation}
\begin{aligned}
	F^T_{\conplus_u, \regio{}{}}[Z](v_1, v_2) &:= \int_{\conplus_{u}\cap \{v_1 \leq v \leq v_2\} \cap \regio{}{}}[|L Z|^2+(1-\mu)|\snabla Z|^2]\de v  \desphere ,\\
	F^T_{\conminus_v, \regio{}{}}[Z](u_1, u_2) &:= \int_{\conminus_{v}\cap \{u_1 \leq u \leq u_2\} \cap \regio{}{}} [|\lbar Z|^2+(1-\mu)|\snabla Z|^2]\de u  \desphere, \\
	F^T_{\Sigma_{t^*_1}, \regio{}{}}[Z](r_1, r_2) &:= \int_{\Sigma_{t^*_1} \cap\{r_1 \leq r \leq r_2\} \cap \regio{}{}} [|L Z|^2+(1-\mu)|\hat \lbar Z|^2+|\snabla Z|^2] r^{-2 }\de \Sigma_{t^*}, \\ 
	F^N_{\Sigma_{t^*_1}, \regio{}{}}[Z](r_1, r_2) &:= \int_{\Sigma_{t^*_1} \cap\{r_1 \leq r \leq r_2\} \cap \regio{}{}} [|L Z|^2+|\hat \lbar Z|^2+|\snabla Z|^2] r^{-2 }\de \Sigma_{t^*}.
\end{aligned}
\end{equation}
Here, we used the definition of outgoing and ingoing cones from Section~\ref{sec:not:regfol}: 
\begin{equation}
\overline{C}_{\tilde u} := \{u = \tilde u\}, \qquad \underline{C}_{\tilde v} := \{v = \tilde v\},
\end{equation}
where $\tilde u$ and $\tilde v$ are real numbers. We furthermore define the following fluxes, which are useful for application of the $r^p$ method. We consider $P=(u_P,v_P)$ a point in $(u,v)$--coordinates, we let $t_P := \frac 1 2 (v_P + u_P)$, $r_P^* := \frac 1 2 (v_P - u_P)$, and $t^*_P := t_P + 2M \log(r_P -2M)$, where $r_P$ is such that $r^*_P = r_P + 2M \log(r_P-2M)$. Then, we define the quantities (degenerate at $\mathcal{H}^+$):
\begin{equation}\label{eq:fgiustodefdeg}
\begin{aligned}
 F_{\regio{}{}, \text{deg}}^\infty[Z] (P) :&= F^T_{\Sigma_{t_P^*}, \regio{}{}} [Z] (2M, r_P) + F^T_{\conplus_{u_P}, \regio{}{}} [Z] (v_P, \infty), \\
 \fgiustordeg[Z] (u) :&= F^\infty_{\regio{}{}, \text{deg}}[Z](u, u+2R^*).
\end{aligned}
\end{equation}
Moreover, we define the non-degenerate fluxes:
\begin{equation}\label{eq:fgiustodef}
\begin{aligned}
 F_{\regio{}{}}^\infty[Z] (P) :&= F^T_{\Sigma_{t_P^*}, \regio{}{}} [Z] (2M, r_P) + F^T_{\conplus_{u_P}, \regio{}{}} [Z] (v_P, \infty), \\
 \fgiustor[Z] (u) :&= F^\infty_{\regio{}{}}[Z](u, u+2R^*).
\end{aligned}
\end{equation}
\begin{remark}
The flux $F_{\regio{}{}, \text{deg}}^\infty[Z] (P)$ is the natural energy flux associated to the \fackip equation with multiplier $T = \p_t$, through a spacelike piece (the $F^T_{\Sigma_{t_P^*}, \regio{}{}} $ term) which originates from $\mathcal{H}^+$ and terminates to the right at $P$, plus a null piece (the $ F^T_{\conplus_{u_P}, \regio{}{}}$ term) which originates from $P$ and terminates at future null infinity. The flux $F_{\regio{}{}}^\infty[Z] (P)$ is the analogous object, where we take $V_1$ as a multiplier instead. Here, $V_1$ is defined as
\begin{equation*}
\begin{aligned}
&V_1:=
&\begin{cases}
&n_{\Sigma_{t^*}}  \ \ \text{ when } 2M \leq r \leq R,\\
&\p_t \hspace{18pt}  \text{ when } r \geq R.
\end{cases}
\end{aligned}
\end{equation*}
\end{remark}
We finally recall the spacetime regions corresponding to the \textbf{semi-null foliation}, which were already introduced in Section~\ref{sec:not:regfol}. Let $u_2 \geq u_1$ real numbers, and let $R > 2M$. We define the following spacetime regions:
\begin{align}
 \mathfrak{D}_{u_1}^{u_2} :&= \left\{r \geq R, u \in [u_1, u_2] \right\},\\
\mathcal{Z}_{u_1} :&= \left(\{t^* = u_1 + R + 4M \log(R-2M)\} \cap \{2M \leq r \leq R \}\right) \cup (\{u = u_1\} \cap \{r \geq R\}), \\
\mathfrak{Z}_{u_1}^{u_2} :&= \cup_{u \in [u_1, u_2]} \mathcal{Z}_u.
\end{align}
It is helpful to represent these regions on a Penrose diagram: see Figure~\ref{fig4}, in which these regions are depicted. As already noted in the last bullet point of Section~\ref{sec:not:regfol}, the region $ \mathfrak{D}_{u_1}^{u_2}$ is the region bounded by the timelike curve $r = R$ and the two outgoing null hypersurfaces $\{u = u_1\}$ and $\{u = u_2\}$. The region $\mathcal{Z}_{u_1}$ coincides with a constant-$t^*$ hypersurface for $r \leq R$, and coincides with the outgoing null cone $\{u = u_1\}$ for $r \geq R$. The value 
$$
t^* = u_1 + R + 4M \log(R-2M)
$$
is chosen so that the intersection of these two pieces is located at $R$ (in $(t,r)$ coordinates). Finally, $\mathfrak{Z}_{u_1}^{u_2}$ is the region bounded above by $\mathcal{Z}_{u_2}$ and below by $\mathcal{Z}_{u_1}$. 

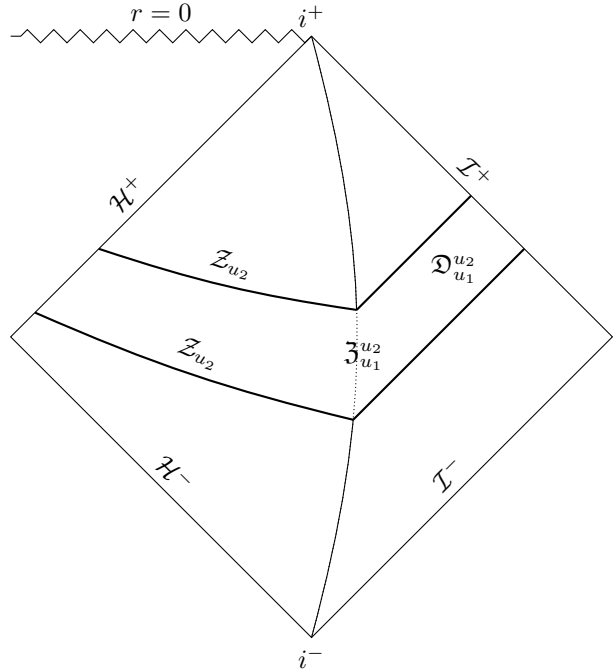
\begin{figure}[H]
	\centering
	\begin{tikzpicture}	

\node (I)    at ( 0,0) {};

\path 
  (I) +(90:4)  coordinate[label=90:$i^+$]  (top)
       +(-90:4) coordinate[label=-90:$i^-$] (bot)
       +(0:4)   coordinate                  (right)
       +(180:4) coordinate (left)
       ;

\path 
	(top) + (180:4) coordinate  (acca)
		+ (-45: 3) coordinate (nulluno)
		+ (-45: 4) coordinate (nulldue)
	;

\path 
	(left) + (-45: 3) coordinate (correspuno)
		+ (-45: 4) coordinate (correspdue)
	;

\draw [name path = rconst] (top) to [bend left = 15](bot);
\draw [name path = nc, opacity= 0] (nulluno) to (correspuno);
\draw [name path = nc1, opacity= 0] (nulldue) to (correspdue);
\draw [name intersections={of=rconst and nc}, opacity= 0](nulluno) to node[midway, sloped, above]{ \tiny $ $} (intersection-1);
\draw [name intersections={of=rconst and nc1}, opacity= 0] (nulldue) to node[midway, sloped, below]{\tiny $ $} (intersection-1);
\path (top) + (-135: 4) coordinate (horiuno)
			+ (-135: 5.2) coordinate (horidue);

\draw [name path = taudue, name intersections={of=rconst and nc}, opacity= 0] (horiuno) to node[midway,sloped, above ]{\tiny $ $}(intersection-1);
\draw [name path = tauno, name intersections={of=rconst and nc1}, opacity= 0] (horidue) to node[midway, sloped, below]{\tiny $ $} (intersection-1);

\coordinate [name intersections={of=tauno and nc}] (bam) at (intersection-1);
\coordinate [name intersections={of=rconst and tauno}] (bamm) at (intersection-1);
\coordinate [name intersections={of=rconst and taudue}] (rbam) at (intersection-1);
\coordinate [name intersections={of=nc and rconst}] (topz) at (intersection-1);
\fill[white] (bam) -- (bamm) -- (nulluno) -- (nulldue) -- cycle;
\path
	(nulluno) + (-100:1.3) coordinate[label=90:$\mathfrak{D}_{u_1}^{u_2}$] (d12);
\draw [name path = taudue, name intersections={of=rconst and nc}, bend right = 5, thick]  (horiuno) to node[midway,sloped, above ]{ $\mathcal{Z}_{u_2}$}(intersection-1);
\draw [name path = tauno, name intersections={of=rconst and nc1}, bend right = 5, thick] (horidue) to node[midway, sloped, above]{ $\mathcal{Z}_{u_2}$} (intersection-1);
\draw [name intersections={of=rconst and nc}, thick] (nulluno) -- (intersection-1);
\draw [name intersections={of=rconst and nc1}, thick] (nulldue) -- (intersection-1);
\draw (left) -- 
          node[midway, above left, sloped]    {$\mathcal{H}^+$}
      (top) --
          node[midway, above, sloped] {$\mathcal{I}^+$}
      (right) -- 
          node[midway, above, sloped] {$\mathcal{I}^-$}
      (bot) --
          node[midway, above, sloped]    {$\mathcal{H}^-$}    
      (left) -- cycle;

\draw [name path = rconst, densely dotted] (top) to [bend left = 15](bot);
\draw[decorate,decoration=zigzag] (top) -- (acca)
      node[midway, above, inner sep=2mm] {$r=0$};

\path (horiuno) +(-22:3.8)  coordinate[label=center: $\mathfrak{Z}_{u_1}^{u_2}$]  (top);
\end{tikzpicture}
	\caption{Penrose diagram of the regions in the semi-null foliation.}\label{fig4}
\end{figure}

We proceed to state and prove the $r^p$-weighted estimates.

\begin{proposition}\label{prop:previsited}
Let $l \in \enn$, $l \geq 2$. There exists a small number $\tilde \varepsilon > 0$, a constant $C >0$ and a number $R^* >0$ such that the following holds. Let $\fara$ be a smooth solution of the MBI system (\ref{MBI}) on $\mathcal R := \mathcal{R}_{t_0^*}^{t_2^*} \subset \mathcal{S}_e$. Assume the bootstrap assumptions $BA\left(\mathcal{R},1,\left\lfloor \frac{l+4}{2}\right\rfloor , \varepsilon \right)$. Assume the following weighted bounds on the initial energy:
\begin{align*}
&\sum_{|I| +j+ k \leq l-1} \int_{\Sigma_{t_0^*}} \left[r^2 |L \dot Z_{,I,j,k}|^2 + r^2|\snabla \dot Z_{,I,j,k} |^2\right] r^{-2} \de \Sigma_{t^*} \leq \varepsilon^2, \\
&\sum_{|I| +j+ k \leq l-1} \int_{\Sigma_{t_0^*}} \left[r^2 |L \dot W_{,I,j,k}|^2 + r^2|\snabla \dot W_{,I,j,k} |^2\right]r^{-2}\de \Sigma_{t^*} \leq \varepsilon^2,\\
&\int_{\Sigma_{t_0^*}} r |\partial^{\leq l} \alpha |^2\de \Sigma_{t^*} \leq \varepsilon^2, \\
&\norm{\fara}^2_{H^{l+4}(\Sigma_{t_0^*})} \leq \varepsilon^2.
\end{align*}
Here, $0 < \varepsilon < \tilde \varepsilon$. Let us define, in the context of this Proposition, for ease of notation, $\dot Z := \dot Z_{,I,j,k} = \lbar^j \dot Z_{,I,0,k}$, and similarly $\dot W := \dot W_{,I,j,k} = \lbar^j \dot W_{,I,0,k}$.

Let $u_0 = t_0 - R^*$, where $t_0 = t_0^* - 2M \log(R-2M)$, and $R^* = R + 2M \log(R-2M)$.

Under these conditions, we have the following integrated bounds on $\mathcal{R}$:
\begin{equation}\label{eq:pdueimpinit}
\boxed{
\begin{aligned}
 \int_{\mathbb{S}^2}\int_{\{u = u_2\}\cap\{r \geq R\}\cap \regio{}{}} \sum_{\substack{|I| +j+ k \leq l-2\\|I|\geq 1}}(r^{2} |L \dot Z|^2) \de v \desphere 
\leq  C \int_{\Sigma_{t^*_0}} \sum_{\substack{|I|+j+k \leq l-2\\|I|\geq 1}} (r^{2} |L \dot Z|^2) r^{-2}\de \Sigma_{t^*} + C \varepsilon^2,
\end{aligned}}
\end{equation}
valid for any $u_2 \in \R$.
\begin{equation}\label{eq:pdueimp}
\boxed{
\begin{aligned}
& \int_{\mathbb{S}^2}\int_{\{u = u_2\}\cap\{r \geq R\}\cap \regio{}{}}\sum_{\substack{|I|+j+k \leq l-2\\|I|\geq 1}} (r^{2} |L \dot Z|^2) \de v \desphere\\ & 
+\int_{\mathfrak{D}_{u_1}^{u_2} \cap \regio{}{}} \sum_{\substack{|I|+j+k \leq l-2\\|I|\geq 1}}\left[r|L \dot Z|^2 + |\snabla \dot Z|^2 + r^{-2}|\dot Z|^2\right] (1-\mu) \de u \de v \desphere \\ &
\leq   \int_{\mathbb{S}^2}\int_{\{u = u_1\}\cap\{r \geq R\}\cap \regio{}{}} \sum_{\substack{|I|+j+k \leq l-2\\|I|\geq 1}}(r^{2} |L \dot Z|^2) \de v \desphere + C \varepsilon^2,
\end{aligned}}
\end{equation}
valid for $u_2 \geq u_1 \geq u_0$.
\begin{equation}\label{eq:punoimp}
\boxed{
\begin{aligned}
& \int_{\mathbb{S}^2}\int_{\{u = u_2\}\cap\{r \geq R\}\cap \regio{}{}} \sum_{|I|+j+k \leq l-2}(r|L \dot Z|^2) \de v \desphere\\ & 
+\int_{\duu \cap \regio{}{}} \sum_{\substack{|I|+j+k \leq l-2\\|I|\geq 1}}\left[|L \dot Z|^2 + |\snabla \dot Z|^2 + r^{-2}|\dot Z|^2\right] (1-\mu) \de u \de v \desphere \\ &
\leq 	C \sum_{\substack{|I|+j+k \leq l-2\\|I|\geq 1}} \fgiustor[\dot Z](u_1) + \int_{\mathbb{S}^2}\int_{\{u = u_1\}\cap\{r \geq R\}\cap \regio{}{}} r \sum_{\substack{|I|+j+k \leq l-2\\|I|\geq 1}}|L \dot Z|^2 \de v \desphere + C \varepsilon^2 (1+|u_1|)^{-1}\\
&+ C \int_{\regio{}{} \cap \mathfrak{Z}_{u_1}^{u_2} \cap \{4M \leq r \leq R\}} \sum_{\substack{|I|+j+k \leq l-2\\|I|\geq 1}}(|L \dot Z|^2 + |\snabla \dot Z|^2) \de u \de v \de \mathbb{S}^2,
\end{aligned}}
\end{equation}
valid again for $u_2 \geq u_1 \geq u_0$.

Furthermore, the same inequalities hold when $\dot Z$ is replaced by $\dot W$.
\end{proposition}

\begin{proof}
We will divide the proof in two \textbf{Steps}. In \textbf{Step 1}, we will prove estimates (\ref{eq:pdueimpinit}), (\ref{eq:pdueimp}), (\ref{eq:punoimp}) with the restriction that, in the sums appearing in such estimates, $j$ always be equal to $0$. In \textbf{Step 2}, we will remove the restriction $j = 0$ in the sums.
\begin{itemize}[wide, labelwidth=!, labelindent=0pt]
\item[\textbf{Step 1.}] Let us consider the identity, which follows from Equation~(\ref{eq:zdot}), valid upon integration on a sphere of constant $r$:
\begin{equation}\label{eq2.15again}
 \begin{aligned}
 & \lbar\left\{f(r) r^p|L \dot Z|^2 \right\} + L\left\{f(r)(1-\mu)r^p|\snabla \dot Z|^2 \right\} \\ &
- \lbar\left(f(r) r^p\right)|L \dot Z|^2 - r^2 L \left(f(r) (1-\mu)r^{p-2} \right)|\snabla \dot Z|^2 \\ & \stackrel{\mathbb{S}^2}{=}
2r^p f(r) (L \dot Z) \partial^I_{\leo} T^k (R_\rho).
 \end{aligned}
\end{equation}
We choose $f(r)$ smooth such that $3 \geq f(r) \geq 0$, $f(r) = 0$ for $r \in [2M, 4M]$, $\frac{d}{dr}f(r) \geq 0$ for all $r \in [2M, \infty)$, and $f(r)=1$ for $r \geq 5M$. We furthermore choose $p = 2$. We integrate identity (\ref{eq2.15again}) on $\regio{}{} \cap \{u \leq u_2\}$. We discard the positive boundary terms, and we apply inequality (\ref{eq:moradegeps}) to bound the angular error term in the region $r \in [4M, 5M]$ plus the error arising from the term $- \lbar\left(f(r) r^p\right)|L \dot Z|^2 $ in the region $\regio{}{} \cap \{4M \leq r \leq R\}$. We obtain, upon choosing $R$ big enough,
\begin{equation}\label{eq:pduefirst}
\begin{aligned}
&\int_{\mathbb{S}^2}\int_{\{u = u_2\}\cap\{r \geq R\}\cap \regio{}{}} (r^2|L \dot Z|^2) \de v \desphere \\ &
+ \int_{\regio{}{} \cap \{r \geq R\}} \left[r|L \dot Z|^2 + |\snabla \dot Z|^2 \right] (1-\mu) \de u \de v \desphere \leq 
C \int_{\Sigma_{t_0^*}} (r^2 |L \dot Z|^2) r^{-2}\de \Sigma_{t^*}\\ &
+ C \int_{\regio{}{} \cap \{r \geq 4M\}} \left[ r^2  |L\dot Z| \cdot |\partial^I_{\leo} T^k (R_\rho)| \right] (1-\mu) \de u \de v \desphere
+ C \varepsilon^2.
\end{aligned}
\end{equation}
We divide the integral in the last line of the previous display as follows:
\begin{equation*}
\begin{aligned}
&\int_{\regio{}{} \cap \{r \geq 4M\}} \left[ r^2  |L\dot Z| \cdot |\partial^I T^k (R_\rho)| \right] (1-\mu) \de u \de v \desphere \\
&= \underbrace{\int_{\regio{}{} \cap \{4M \leq r \leq R\}} \left[ r^2  |L\dot Z| \cdot |\partial^I_{\leo} T^k (R_\rho)| \right] (1-\mu) \de u \de v \desphere}_{(x)}\\
&\quad  + \underbrace{\int_{\regio{}{} \cap \{r \geq R\}} \left[ r^2  |L\dot Z| \cdot |\partial^I_{\leo} T^k (R_\rho)| \right] (1-\mu) \de u \de v \desphere}_{(y)}.
\end{aligned}
\end{equation*}
The term $(y)$ in the display above can be incorporated in the LHS of~\eqref{eq:pduefirst}, by means of the Cauchy--Schwarz inequality with a small parameter $\eta > 0$. The term $(x) $, on the other hand, can be dealt with using the bootstrap assumptions and the uniform $L^2$ control given by Proposition~\ref{prop:energy}. All in all, we obtain
\begin{equation*}
\begin{aligned}
&\int_{\mathbb{S}^2}\int_{\{u = u_2\}\cap\{r \geq R\}\cap \regio{}{}} (r^2|L \dot Z|^2) \de v \desphere \\ &
+ \int_{\regio{}{}} \left[r|L \dot Z|^2 + |\snabla \dot Z|^2 \right] (1-\mu) \de u \de v \desphere \leq 
C \int_{\Sigma_{t_0^*}} (r^2 |L \dot Z|^2) r^{-2}\de \Sigma_{t^*}\\ &
+ C \int_{\regio{}{} \cap \{r \geq 4M\}} r^3 |\partial^I_{\leo} T^k (R_\rho)|^2 (1-\mu) \de u \de v \desphere
+ C \varepsilon^2.
\end{aligned}
\end{equation*}
Combining this with bound (\ref{eq:boundrrhoimp2}), we obtain the claim~(\ref{eq:pdueimpinit}), when $j$ is restricted to be $0$ in the sum. A similar reasoning gives the bound~\eqref{eq:pdueimp}.

Finally, in order to prove (\ref{eq:punoimp}) in the special case $j =0$, we carry out the same estimate~\eqref{eq:pduefirst} with the choice $p =1$, and $f: [2M, \infty) \to \R$ smooth, $f(r) = 0$ for $r \in [2M, 4M]$, $f(r) = 1$ for $r \geq 5M$. We integrate the resulting identity on the region $\regio{}{} \cap \mathfrak{Z}_{u_1}^{u_2}$. We then proceed with the same estimates as in the previous case, being careful this time to use estimate (\ref{eq:boundrrhoimp1}) in order to get the $|u|^{-1}$ decay appearing in inequality~\eqref{eq:punoimp}. 

\item[\textbf{Step 2.}] Let us again consider the identity, which follows from Equation~(\ref{eq:lbarcommut}), valid upon integration on a sphere of constant $r$-coordinate:
\begin{align}
&\int_{\mathbb{S}^2}  \lbar\left\{f(r) r^p|L \dot Z|^2 \right\} \desphere+ \int_{\mathbb{S}^2}L\left\{f(r)(1-\mu)r^p|\snabla \dot Z|^2 \right\} \desphere \nonumber \\ &
- \int_{\mathbb{S}^2}\lbar\left(f(r) r^p\right)|L \dot Z|^2\desphere 
- \int_{\mathbb{S}^2}r^2 L \left(f(r) (1-\mu)r^{p-2} \right)|\snabla \dot Z|^2 \desphere \nonumber \\ & 
-\int_{\mathbb{S}^2} f(r)r^p(L \dot Z) \sum_{\substack{a+b = j\\ a \geq 1}} \lbar^a \left(\frac{1-\mu}{r^2} \right) \sdelta_{\mathbb{S}^2} \lbar^b \dot Z_{,I,0,k}\desphere \label{eq2.15againlb}
\\ &
= \int_{\mathbb{S}^2}
2r^p f(r) (L \dot Z) \partial^I_{\leo} T^k \lbar^j (R_\rho)\desphere.\nonumber
 \end{align}
Now,
\begin{equation*}
\begin{aligned}
&\left|\int_{\mathbb{S}^2} f(r)r^p(L  \dot Z) \sum_{\substack{a+b = j\\a \geq 1}} \lbar^a \left(\frac{1-\mu}{r^2} \right) \sdelta_{\mathbb{S}^2} \lbar^b \dot Z_{,I,0,k}\desphere \right|\\&
\leq \frac 1 2 \int_{\mathbb{S}^2} pf(r)r^{p-1}|L \dot Z|^2 \desphere 
+ C  \int_{\mathbb{S}^2} r^{p-5} \sum_{b=0}^{j-1}|\sdelta_{\mathbb{S}^2}\lbar^b \dot Z_{,I,0,k}|^2 \desphere \\&
\leq \frac 1 2 \int_{\mathbb{S}^2} pf(r)r^{p-1}|L \dot Z|^2 \desphere 
+ C  \int_{\mathbb{S}^2} r^{p-3}\sum_{\Omega \in \leo} \sum_{b=0}^{j-1}|\snabla \lbar^b \Omega \dot Z_{,I,0,k}|^2 \desphere
\end{aligned}
\end{equation*}

We choose $f(r)$ smooth such that $3 \geq f(r) \geq 0$, $f(r) = 0$ for $r \in [2M, 4M]$, $\frac{d}{dr} f(r) \geq 0$ for all $r \in [2M, \infty)$, and $f(r)=1$ for $r \geq 5M$. We furthermore choose $p = 2$. We integrate identity (\ref{eq2.15againlb}) on $\regio{}{} \cap \{u \leq u_2\}$.
Upon discarding the positive boundary terms, we apply inequality (\ref{eq:morabdd}) to bound the error term in the region $r \in [4M, 5M]$. We obtain
\begin{equation*}
\begin{aligned}
& \int_{\mathbb{S}^2}\int_{\{u = u_2\}\cap\{r \geq R\}\cap \regio{}{}} (r^2|L \dot Z|^2) \de v \desphere \\ &
+ \int_{\regio{}{} \cap \{r \geq 4M\}\cap \{u \leq u_2 \}}   \left[r|L \dot Z|^2 + |\snabla \dot Z|^2 \right] (1-\mu) \de u \de v \desphere \leq 
C \int_{\Sigma_{t_0^*}} (r^2 |L \dot Z|^2) r^{-2}\de \Sigma_{t^*}\\ &
+ C \int_{\regio{}{} \cap \{r \geq 4M\}} \left[ r^2  |L\dot Z| \cdot |\partial^I_{\leo} T^k (R_\rho)| \right] (1-\mu) \de u \de v \desphere\\ &
+ C \int_{\regio{}{} \cap \{r \geq 4M\}\cap \{u \leq u_2 \}} r^{-1}\sum_{\Omega \in \leo} \sum_{b=0}^{j-1}|\snabla \lbar^b \Omega \dot Z_{,I,0,k}|^2 \desphere \de u \de v
+ C \varepsilon^2.
\end{aligned}
\end{equation*}
An application of the Cauchy--Schwarz inequality yields
\begin{equation*}
\begin{aligned}
&\int_{\mathbb{S}^2}\int_{\{u = u_2\}\cap\{r \geq R\}\cap \regio{}{}} (r^2|L \dot Z|^2) \de v \desphere \\ &
+ \int_{\regio{}{} \cap \{r \geq 4M\}\cap \{u \leq u_2\}} \left[r|L \dot Z|^2 + \underbrace{|\snabla \dot Z|^2}_{(i)} \right] (1-\mu) \de u \de v \desphere \leq 
C \int_{\Sigma_{t_0^*}} (r^2 |L \dot Z|^2) r^{-2}\de \Sigma_{t^*}\\ &
+ C \int_{\regio{}{} \cap \{r \geq 4M\}\cap \{u \leq u_2 \}} r^3 |\partial^I_{\leo} T^k (R_\rho)|^2 (1-\mu) \de u \de v \desphere\\ &
+  \int_{\regio{}{} \cap \{r \geq 4M\}\cap \{u \leq u_2 \}} r^{-1} \sum_{b=0}^{j-1}\underbrace{\sum_{\Omega \in \leo}|\snabla \lbar^b \Omega \dot Z_{,I,0,k}|^2}_{(ii)} \desphere \de u \de v
+ C \varepsilon^2.
\end{aligned}
\end{equation*}
We conclude now by an induction argument, using the fact that, upon summation, we have the following relation between terms $(i)$ and $(ii)$, if $a \in \enn$:
\begin{equation*}
\sum_{|J| = a+1} |\snabla \dot Z_{,J,j,k}|^2 \geq \sum_{\substack{\Omega \in \leo \\ |I|= a}} |\snabla \Omega \dot Z_{,I,j,k}|^2.
\end{equation*}
\end{itemize}
This concludes the proof.
\end{proof}

\section{Energy conservation and Morawetz estimate in terms of \texorpdfstring{$\fgiustor$}{FR}}\label{sub:moragiustoests}

In order to close our estimates, we need a Morawetz estimate and a bound on the energy involving only terms containing $\fgiustor$. Namely, the Morawetz estimate will be used to absorb error terms localized near the hypersurface $\{r = R\}$, whereas the energy bound will be used to eliminate the restriction to a dyadic sequence which follows from an application of the $r^p$ method. The proof of those two bounds is the content of this section.

This section is structured as follows. First, we will state and prove Proposition~\ref{prop:morafgiusto}, in which we derive two Morawetz estimates (one degenerate at $\mathcal{H}^+$ and at $r = 3M$, the other without degeneracies). These estimates, however, display boundary terms in their respective right hand sides. We are then going to deal with such boundary terms in Propostion~\ref{prop:encfgiusto}, in which we are going to do a $\p_t$-energy estimate. This will allow us to prove, always in Proposition~\ref{prop:encfgiusto}, an inequality which is both a Morawetz estimate and an energy conservation statement in terms of the fluxes $\fgiustor$, and displays only past boundary terms in its right hand side.

\begin{proposition}[Morawetz estimate in terms of $\fgiustor$ with boundary terms] \label{prop:morafgiusto}
Let $l \in \enn$, $l \geq 5$. There exist a small number $\tilde \varepsilon > 0$, a constant $C >0$ and a number $R^* > 0$ such that the following holds. Let $t_0^* \leq t_1^* \leq t_2^* \leq t_3^*$. Moreover, let $\fara$ be a smooth solution of the MBI system (\ref{MBI}) on $\mathcal{R}_{t_0^*}^{t_3^*} \subset \mathcal{S}_e$. Let $\mathcal R := \regio{t_0^*}{t_3^*}$. Assume the bootstrap assumptions $BA\left(\regio{t_0^*}{t_3^*},1,\left\lfloor \frac{l+4}{2}\right\rfloor, \varepsilon \right)$. Assume the following weighted bounds on the initial energy, similar to (\ref{eq:bdinit}):
\begin{equation}
\begin{aligned}
&\sum_{|I| + k \leq l-1} \int_{\Sigma_{t_0^*}} \left[r |L \dot Z_{,I,0,k}|^2 + r|\snabla \dot Z_{,I,0,k} |^2\right] r^{-2} \de \Sigma_{t^*} \leq \varepsilon^2, \\
&\sum_{|I| + k \leq l-1} \int_{\Sigma_{t_0^*}} \left[r |L \dot W_{,I,0,k}|^2 + r|\snabla \dot W_{,I,0,k} |^2\right]r^{-2}\de \Sigma_{t^*} \leq \varepsilon^2,\\
&\int_{\Sigma_{t_0^*}} r |\partial^{\leq l} \alpha |^2\de \Sigma_{t^*} \leq \varepsilon^2, \\
&\norm{\fara}^2_{H^{l+4}(\Sigma_{t_0^*})} \leq \varepsilon^2.
\end{aligned}
\end{equation}
Here, $0 < \varepsilon < \tilde \varepsilon$. Let $u_i = t_i - R^*$, where $t_i = t_i^* - 2M \log(R-2M)$, and $R^* = R + 2M \log(R-2M)$, with $i \in \{0,1,2, 3\}$.

Let $\tilde u \in \R$, with $\tilde u \geq u_3$. Then, we have the following inequality:
\begin{equation}\label{eq:morafgiustodegstat}
\begin{aligned}
& \sum_{\substack{|I|+k \leq l-3\\|I|\geq 1}} \int_{\mathfrak{Z}_{u_1}^{u_2} \cap \regio{}{} \cap \{u \leq \tilde u\}} 
\Big\{\frac 1 {r^2}|L \dot Z_{I,0,k} - \lbar \dot Z_{I,0,k}|^2\\
&\qquad + \frac{(r-3M)^2}{r^3}\left(|\snabla \dot Z_{I,0,k}|^2+ \frac 1 r |L \dot Z_{I,0,k} + \lbar \dot Z_{I,0,k}|^2\right)+\frac 1 {r^3} |\dot Z_{I,0,k}|^2\Big\} r^{-2}\de \text{Vol} \\ 
& \leq C   \sum_{\substack{|I|+k \leq l-3\\|I|\geq 1}} \Big\{\fgiustordeg[\dot Z_{I,0,k}](u_1) + C  \fgiustordeg[\dot Z_{I,0,k}](u_2)  \\
&\quad + C  \int_{\conplus_{\tilde u} \cap \regio{t_1^*}{t_2^*}}|L \dot Z_{I,0,k}|^2 + (1-\mu)|\snabla\dot Z_{I,0,k}|^2 \desphere \de v\\
& \quad + C \int_{\Sigma_{t^*_3} \cap \mathfrak{Z}_{u_1}^{u_2}} (|L \dot Z_{I,0,k}|^2 + |\snabla \dot Z_{I,0,k}|^2 + |\hat \lbar \dot Z_{I,0,k}|^2)   r^{-2}\de \Sigma_{t^*}\Big\}
+ C\varepsilon^3 (|u_1|+1)^{-2},
\end{aligned}
\end{equation}
Moreover, the following holds:
\begin{equation}\label{eq:morahats}
\begin{aligned}
& \sum_{\substack{|I|+j+k \leq l-3\\|I| \geq 1}} \Big\{ \int_{\Sigma_{t_2^*} \cap \{r \leq R\}} |\hat Z_{,I,j,k}|^2\de \Sigma_{t^*}+ \int_{\mathfrak{Z}_{u_1}^{u_2} \cap \regio{}{}} r^{-2} |\hat Z_{,I,j,k}|^2 r^{-2} \de \text{Vol} \Big\} \\
&\quad \leq C \varepsilon^3(|u_1|+1)^{-2} + C \sum_{\substack{|I|+j+k \leq l-4\\ |I|\geq 1}} \fgiustor[\hat Z_{,I,j,k}](u_1)+ C\sum_{\substack{|I|+k \leq l-3\\|I| \geq 1}} \int_{\mathfrak{Z}_{u_1}^{u_2} \cap \regio{}{}} r^{-2} |\hat Z_{,I,0,k}|^2 r^{-2} \de \text{Vol} .
\end{aligned}
\end{equation}
Furthermore, the same inequalities hold true if $Z$ is replaced by $W$.
\end{proposition}

\begin{proof}[Proof of Proposition~\ref{prop:morafgiusto}]
Let us first set $\dot Z:= \dot Z_{,I,0,k}$, with $|I|+k \leq l-3$.

By a similar reasoning to the one leading to (\ref{fullen}), we obtain the following inequality (notice that, in particular, the reasoning here depends on deriving an estimate analogous to (\ref{eq:enbdur}) in the region $\mathfrak{Z}_{u_1}^{u_2} \cap \regio{}{} \cap \{u \leq \tilde u\}$):
\begin{align*}
& \int_{\mathfrak{Z}_{u_1}^{u_2} \cap \regio{}{} \cap \{ u \leq \tilde u\}} 
\left\{\frac 1 {r^2}|L \dot Z - \lbar \dot Z|^2+ \frac{(r-3M)^2}{r^3}\left(|\snabla \dot Z|^2+ \frac 1 r |L \dot Z + \lbar \dot Z|^2\right)+\frac 1 {r^3} |\dot Z|^2\right\} r^{-2}\de \text{Vol} \\ &
\leq C	\int_{\mathfrak{Z}_{u_1}^{u_2}   \cap \regio{}{} \cap \{ u \leq \tilde u\} } (r^2(1-\mu)^{-1} |\partial^I_{\leo} T^k (R_\rho)|^2 + r^{-1}(1-\mu) |\partial^I_{\leo} T^k (R_\rho)|^2 )\de u \de v \desphere\\ &
+C \int_{\mathfrak{Z}_{u_1}^{u_2} \cap \{ 2M \leq r \leq 4M\}  \cap \regio{}{} \cap \{ u \leq \tilde u\}} |\dot Z_{,I,0,k}| |\partial^I_{\leo} T^{k+1} (R_\rho)| \de u \de v \desphere\\ &
+C \int_{\mathfrak{Z}_{u_1}^{u_2} \cap \{ r \geq 4M\}  \cap \regio{}{} \cap \{ u \leq \tilde u\}} |\dot Z_{,I,0,k+1}| |\partial^I_{\leo} T^{k} (R_\rho)| \de u \de v \desphere\\ &
+ C \int_{\mathfrak{Z}_{u_1}^{u_2} \cap \{2M \leq  r  \leq 3 M\}  \cap \regio{}{} \cap \{ u \leq \tilde u\}}(|\partial^I_{\leo}  {\hat \lbar} T^{k} (R_\rho)|^2 + |\partial^I_{\leo}  T^{k} ((1-\mu)^{-1}R_\rho)|^2)\de \text{Vol}\\ &
+ C \int_{\Sigma_{t_2^*} \cap \{2M \leq r \leq 4M\}  \cap \regio{}{} \cap \{ u \leq \tilde u\}} |\partial^I_{\leo}  T^{k} ((1-\mu)^{-1}R_\rho)|^2 \de \Sigma_{t^*}\\&
 + C \int_{\Sigma_{t_1^*}\cap \{2M \leq r \leq 4M\} \cap \regio{}{} \cap \{ u \leq \tilde u\}} |\partial^I_{\leo}  T^{k} ((1-\mu)^{-1}R_\rho)|^2 \de \Sigma_{t^*}\\ 
& + C   \fgiustordeg[\dot Z](u_1) + C  \fgiustordeg[\dot Z](u_2)  + C  \int_{\conplus_{\tilde u} \cap \regio{t_1^*}{t_2^*}}|L \dot Z|^2 + (1-\mu)|\snabla\dot Z|^2 \desphere \de v\\
&+ C \int_{\Sigma_{t^*_3} \cap \mathfrak{Z}_{u_1}^{u_2}} (|L \dot Z|^2 + |\snabla \dot Z|^2 + |\hat \lbar \dot Z|^2)   r^{-2}\de \Sigma_{t^*}.
\end{align*}
We use the Cauchy--Schwarz inequality with a small parameter $\eta > 0$ on the terms in the second and third line on the RHS of the previous display to incorporate the terms in $Z$ into the LHS. We then have a total maximum number of $|I| + k + 1$ derivatives on $R_{\rho}$ in the RHS, and the hypotheses of Lemma~\ref{lem:bsrhoimp} let us only estimate $l-2$ of them, so we need to have $|I|+k\leq l-3$. Hence, using the decay estimate (\ref{eq:boundrsigmaimp0}), this implies
\begin{equation}\label{eq:morafgiustodegpf}
\begin{aligned}
& \int_{\mathfrak{Z}_{u_1}^{u_2} \cap \regio{}{}} 
\left\{\frac 1 {r^2}|L \dot Z - \lbar \dot Z|^2+ \frac{(r-3M)^2}{r^3}\left(|\snabla \dot Z|^2+ \frac 1 r |L \dot Z + \lbar \dot Z|^2\right)+\frac 1 {r^3} |\dot Z|^2\right\} r^{-2}\de \text{Vol} \\ 
& \leq C   \fgiustordeg[\dot Z](u_1) + C  \fgiustordeg[\dot Z](u_2)  + C  \int_{\conplus_{\tilde u} \cap \regio{t_1^*}{t_2^*}}|L \dot Z|^2 + (1-\mu)|\snabla\dot Z|^2 \desphere \de v\\
& \quad + C \int_{\Sigma_{t^*_3} \cap \mathfrak{Z}_{u_1}^{u_2}} (|L \dot Z|^2 + |\snabla \dot Z|^2 + |\hat \lbar \dot Z|^2)   r^{-2}\de \Sigma_{t^*}
+ C\varepsilon^3 (|u_1|+1)^{-2},
\end{aligned}
\end{equation}
with $|I|+k \leq l-3$. The last display, in particular, implies inequality (\ref{eq:morafgiustodegstat}).

We now consider again relation (\ref{eq:mastermorared}). It follows that, letting $\hat Z_j := \hat Z_{,I,j,k}$, with $|I| + j + k \leq l-3$,
\begin{equation}\label{eq:masterfgiustomora}
\begin{aligned}
&L (r^{-1}(1-\mu)|\hat Z_j|^2) + \frac{(1-\mu)^2}{r^2} |\hat Z_j|^2+(1-\mu) (2 j-1) \frac{2M}{r^3} |\hat Z_j|^2 \\
&= \underbrace{ 2(1-\mu)r^{-1} \hat Z_j 
 {\hat \lbar}^{j-1}(\sdelta \hat Z_0)}_{(i)} + \underbrace{2r^{-1}(1-\mu) \hat Z_j \partial_{\leo}^I T^k {\hat \lbar}^{j-1}((1-\mu)^{-1}R_\rho)}_{(ii)}\\
 &\quad - \underbrace{2r^{-1}(1-\mu) \hat Z_j\left(\sum_{i=1}^{j-1} f^{(j-1)}_i \hat Z_i\right)}_{(iii)}.
\end{aligned}
\end{equation}
We have 
\begin{align*}
&(i) \stackrel{\mathbb{S}^2}{=} 2 r^{-1} (\lbar \hat Z_{j-1}) r^{-2} \sdelta_{\mathbb{S}^2}\hat Z_{j-1}+2 r^{-1} \lbar \hat Z_{j-1} \sum_{\substack{a+b = j-1\\ a \geq 1}} \hat \lbar^a (r^{-2}) \sdelta_{\mathbb{S}^2}\hat Z_b \\ &
\stackrel{\mathbb{S}^2}{=} -2 r^{-1} (\snabla_\lbar (\snabla_{\mathbb{S}^2})_A \hat Z_{j-1}) ((\snabla_{\mathbb{S}^2})^A \hat Z_{j-1}) + 2 r^{-1} \lbar \hat Z_{j-1} \sum_{\substack{a+b = j-1\\ a \geq 1}} \hat \lbar^a (r^{-2}) \sdelta_{\mathbb{S}^2}\hat Z_b \\ &
\leq -\lbar \left(r^{-1}|\snabla \hat Z_{j-1}|^2 \right) +\lbar(r^{-3}) |\snabla_{\mathbb{S}^2} \hat Z_{j-1}|^2 + \frac 1 {10} r^{-2} (1-\mu) |\hat Z_j|^2 + C(1-\mu)\sum_{b = 0}^{j-2} r^{-6} |\sdelta_{\mathbb{S}^2}\hat Z_b|^2\\ &
\leq- \lbar \left(r^{-1}|\snabla \hat Z_{j-1}|^2 \right) +3\frac{(1-\mu)}{r^4} |\snabla_{\mathbb{S}^2} \hat Z_{j-1}|^2 + \frac 1{10} r^{-2}(1-\mu) |\hat Z_j|^2 + C(1-\mu)\sum_{b = 0}^{j-2} \sum_{\Omega \in \leo}r^{-4} |\snabla \Omega \hat Z_b|^2.
\end{align*}
Moreover,
\begin{align*}
|(ii)|\leq \frac 1 {10}r^{-2}(1-\mu)|\hat Z_j|^2 + C (1-\mu) \big(\partial_{\leo}^I T^k {\hat \lbar}^{j-1}((1-\mu)^{-1}R_\rho) \big)^2.
\end{align*}
Finally,
\begin{align*}
|(iii)| \leq  \frac 1 {10}r^{-2}(1-\mu)|\hat Z_j|^2 + Cr^{-6} (1-\mu)\sum_{i =1}^{j-1}|\hat Z_i|^2.
\end{align*}
Let us proceed to integrate the relation (\ref{eq:masterfgiustomora}) on the spacetime region $\mathfrak{Z}_{u_1}^{u_2} \cap \regio{}{}\cap \{u \leq \tilde u\}$, and take the limit as $\tilde u \to \infty$. An induction argument, together with the bound (\ref{eq:morafgiustodegpf}) then yields (note that we discarded some of the future boundary terms due to their positivity):
\begin{equation*}
\begin{aligned}
& \sum_{\substack{|I|+j+k \leq l-3\\|I| \geq 1}} \Big\{ \int_{\Sigma_{t_2^*} \cap \{r \leq R\}} |\hat Z_{,I,j,k}|^2\de \Sigma_{t^*}+ \int_{\mathfrak{Z}_{u_1}^{u_2} \cap \regio{}{}} r^{-2} |\hat Z_{,I,j,k}|^2 r^{-2} \de \text{Vol} \Big\} \\
&\leq C \varepsilon^2(|u_1|+1)^{-2} + C \sum_{\substack{|I|+j+k \leq l-4\\ |I|\geq 1}} \fgiustor[\hat Z_{,I,j,k}](u_1)+ \sum_{\substack{|I|+k \leq l-3\\|I| \geq 1}} \int_{\mathfrak{Z}_{u_1}^{u_2} \cap \regio{}{}} r^{-2} |\hat Z_{,I,0,k}|^2 r^{-2} \de \text{Vol} .
\end{aligned}
\end{equation*}
This is the claim (\ref{eq:morahats}).
\end{proof}

\begin{proposition}[Energy conservation and Morawetz estimate in terms of $\fgiustor$-fluxes] \label{prop:encfgiusto} Let $l \in \enn$, $l \geq 4$. There exist a small number $\tilde \varepsilon >0$, a number $R^* >0$ and a constant $C >0$ such that the following holds. Let $t_3^* \geq t_0^*$. Let $\fara$ be a smooth solution of the MBI system (\ref{MBI}) on $\mathcal{R}_{t_0^*}^{t_3^*} \subset \mathcal{S}_e$. Let $\mathcal R := \regio{t_0^*}{t_3^*}$.  Assume the bootstrap assumptions $BA\left(\regio{t_0^*}{t_3^*},1,\left\lfloor \frac{l+4}{2}\right\rfloor, \varepsilon \right)$. Assume the following bounds on the initial energy:
\begin{equation}
\begin{aligned}
&\sum_{|I| + k \leq l-1} \int_{\Sigma_{t_0^*}} \left[r |L \dot Z_{,I,0,k}|^2 + r|\snabla \dot Z_{,I,0,k} |^2\right] r^{-2} \de \Sigma_{t^*} \leq \varepsilon^2, \\
&\sum_{|I| + k \leq l-1} \int_{\Sigma_{t_0^*}} \left[r |L \dot W_{,I,0,k}|^2 + r|\snabla \dot W_{,I,0,k} |^2\right]r^{-2}\de \Sigma_{t^*} \leq \varepsilon^2,\\
& \int_{\Sigma_{t_0^*}} r |\partial^{\leq l} \alpha |^2\de \Sigma_{t^*} \leq \varepsilon^2, \\
&\norm{\fara}^2_{H^{l+4}(\Sigma_{t_0^*})} \leq \varepsilon^2.
\end{aligned}
\end{equation}
Here, $0 < \varepsilon < \tilde \varepsilon$. Let $u_i = t_i - R^*$, where $t_i = t_i^* - 2M \log(R-2M)$, and $R^* = R + 2M \log(r-2M) $, $i \in \{0,1,2,3\}$.
Under these conditions, we have the energy conservation statement:
\begin{equation}\label{eq:enconsZu}
\boxed{
\begin{aligned}
&\sum_{\substack{|I|+j+k \leq l-4\\ |I|\geq 1}}\fgiustor[\hat Z_{,I,j,k}](u_2) +  \sum_{\substack{|I|+j+k \leq l-3\\ |I|\geq 1}}\int_{\mathfrak{Z}_{u_1}^{u_2} \cap \regio{}{}} r^{-2} |\hat Z_{,I,j,k}|^2 r^{-2} \de \text{Vol} \\
& \qquad  \leq \sum_{\substack{|I|+j+k \leq l-3\\ |I|\geq 1}}\fgiustor[\hat Z_{,I,j,k}](u_2) + C \varepsilon^2 (|u_1|+1)^{-2}.
\end{aligned}}
\end{equation}
valid for $u_0 \leq u_1 \leq u_2 \leq u_3$.

Furthermore, the same inequality holds when $\hat Z$ is replaced by $\hat W$.
\end{proposition}

\begin{proof}[Proof of Proposition~\ref{prop:encfgiusto}]
We divide the proof in several {\bf Steps}. In {\bf Step 1}, we are going to prove the claim~\eqref{eq:enconsZu} under the restriction $j =0$, and replacing all occurrences of $\fgiustor$ with $\fgiustordeg$ in the claim. In {\bf Step 2}, with the aid of inequality~\eqref{eq:morahats}, we are going to remove these restrictions.

\subsection*{Step 1}
In this step, let us restrict to $j=0$ and $|I|+k \leq l -3$. Let us also recall the energy conservation identity, Equation (\ref{eq:toencons}), valid upon integration on the sphere $\mathbb{S}^2$:
\begin{equation}\label{eq:tointnear}
\frac 1 2 \lbar |L \dot Z_{,I,0,k}|^2 + \frac 1 2 L|\lbar \dot Z_{,I,0,k}|^2 + \frac 1 2 (L + \lbar) \left((1-\mu)|\snabla \dot Z_{,I,0,k}|^2 \right) \stackrel{\mathbb{S}^2}{=} (T\dot Z_{,I,0,k})\partial^I_{\leo} T^k (R_\rho).
\end{equation}
We proceed to integrate the last display on $\{ u \leq \tilde u\} \cap \regio{t_0^*}{t_3^*} \cap \mathfrak{Z}_{u_1}^{u_2}$, where $\tilde u$ is thought of as being very large:
$$
\tilde u \gg u_3,
$$ 
so that the hypersurface $\{u = \tilde u\}$ is ``very close'' to the event horizon $\mathcal{H}^+$. Note that we have the inequality $t^*_0 \leq t_1^* \leq t_2^* \leq t_3^*$. This means that the boundary of the region $\{ u \leq \tilde u\} \cap \regio{t_0^*}{t_3^*} \cap \mathfrak{Z}_{u_1}^{u_2}$ is composed of the following hypersurfaces:
\begin{itemize}
\item The future boundary is composed, from left to right, of a null hypersurface close to the horizon ($\mathfrak{Z}_{u_1}^{u_2} \cap \{ u = \tilde u\}$), followed by the spacelike surface $\Sigma_{t_2^*}   \cap \{ u \leq \tilde u\} \cap \{r \leq R\}$, followed by the null hypersurface $\{u = u_2\} \cap \{r \geq R\} \cap \{t^* \leq t_3^*\}$, followed finally by the spacelike hypersurface $\Sigma_{t_3^*} \cap \mathfrak{Z}_{u_1}^{u_2}$.
\item The past boundary is composed, from left to right, by the spacelike hypersurface $\Sigma_{t^*_1}\cap \{u \leq \tilde u\}$ followed by the null hypersurface $\{u = u_1\} \cap \{r \geq R\} \cap \{t^* \leq t_3^*\}$.
\end{itemize}

We also note that the restriction $\{ u \leq \tilde u\}$ is just technical, and we will take the limit as $\tilde u \to \infty$, discarding the positive boundary terms introduced at the hypersurface  $\{ u = \tilde u\}$. We also recall the notation for outgoing and ingoing null cones introduced in Section~\ref{sec:not:regfol}:
$$
\conplus_{\tilde u} := \{ u = \tilde u\}, \qquad  \conminus_{\tilde v} := \{v = \tilde v\}.
$$

We obtain, denoting $\dot Z := \dot Z_{I, 0, k}$, 
\begin{align*}
&\fgiustordeg[\dot Z](u_2)\\
&\quad +\int_{\conplus_{\tilde u} \cap \regio{t_1^*}{t_2^*}}|L \dot Z|^2 + (1-\mu)|\snabla\dot Z|^2 \desphere \de v +\int_{\Sigma_{t^*_3} \cap \mathfrak{Z}_{u_1}^{u_2}} \left(|L \dot Z|^2 + |\snabla \dot Z|^2 + |\hat \lbar \dot Z|^2  \right) r^{-2}\de \Sigma_{t^*} \\
&
\leq C\fgiustordeg[\dot Z](u_1) + C \Big|\int_{\regio{}{} \cap \mathfrak{Z}_{u_1}^{u_2} \cap \{u \leq \tilde u \} }(T\dot Z_{,I,0,k})\partial^I_{\leo} T^k (R_\rho) \de u \de v \desphere\Big|.
\end{align*}
We now integrate the last term in the previous inequality by parts \emph{only in the region} $\{5/2 M \leq r \leq 4M\}$ (this is to take into account the fact that the Morawetz estmate ``loses derivatives'' in this region). We obtain:
\begin{align}
& \Big|\int_{\regio{}{} \cap \mathfrak{Z}_{u_1}^{u_2} \cap \{u \leq \tilde u \} }(T\dot Z_{,I,0,k})\partial^I_{\leo} T^k (R_\rho) \de u \de v \desphere\Big| \nonumber \\
& \leq \underbrace{ \Big|\int_{\regio{}{} \cap \mathfrak{Z}_{u_1}^{u_2} \cap \{u \leq \tilde u \} \cap \{5/2M \leq  r \leq 4M\}}(\dot Z_{,I,0,k})\partial^I_{\leo} T^{k+1} (R_\rho) \de u \de v \desphere\Big|}_{(I)} \nonumber\\
&\quad + \underbrace{\Big|\int_{\regio{}{} \cap \mathfrak{Z}_{u_1}^{u_2} \cap \{u \leq \tilde u \} \cap \big(\{2M \leq  r \leq 5/2 M \} \cup \{r \geq  4M \} \big)}(T \dot Z_{,I,0,k})\partial^I_{\leo} T^{k} (R_\rho) \de u \de v \desphere\Big|}_{(II)} \nonumber  \\
&\quad + \underbrace{\int_{\Sigma_{t_1^*} \cap \{5/2M \leq r \leq 4M\}} |\dot Z_{,I,0,k}| |\partial^I_{\leo} T^k (R_\rho)|\de \Sigma_{t^*} + \int_{\Sigma_{t_2^*} \cap \{5/2M \leq r \leq 4M\}} |\dot Z_{,I,0,k}| |\partial^I_{\leo} T^k (R_\rho)|\de \Sigma_{t^*}}_{(III)}.\nonumber 
\end{align}
By our assumptions, $(III) \leq C \eps^2(1+|u_1|)^{-2}$. Moreover, we have, using inequality~\eqref{eq:morafgiustodegstat}:
\begin{align*}
& (II) 
\leq C \eps \int_{\regio{}{} \cap \mathfrak{Z}_{u_1}^{u_2} \cap \{u \leq \tilde u \} \cap \{2M \leq  r \leq 5/2M\}}|\dot Z_{,I,0,k+1}|^2\de \text{Vol} \\
&\quad + C \eps^{-1}\int_{\regio{}{} \cap \mathfrak{Z}_{u_1}^{u_2} \cap \{u \leq \tilde u \} \cap \{2M \leq  r \leq 5/2M\}}|\partial^I_{\leo} T^k (R_\rho)|^2 \de \text{Vol}\\
&\quad  +\eps \int_{\regio{}{} \cap \mathfrak{Z}_{u_1}^{u_2} \cap \{u \leq \tilde u \} \cap \{r \geq 4M\}}r^{-2}|\dot Z_{,I,0,k+1}|^2\de \text{Vol}\\& \quad  + \eps^{-1}\int_{\regio{}{} \cap \mathfrak{Z}_{u_1}^{u_2} \cap \{u \leq \tilde u \} \cap \{r \geq 4M\}}|\partial^I_{\leo} T^k (R_\rho)|^2 \de \text{Vol}
\\ 
&\leq C \eps \sum_{\substack{|I'|+ k' \leq l-3\\ |I'|\geq 1}}  \fgiustordeg[\dot Z_{,I',0,k'}](u_1) +  \fgiustordeg[\dot Z_{,I',0,k'}](u_2)\\ 
&\quad + C \eps \sum_{\substack{|I'|+ k' \leq l-3\\ |I'|\geq 1}}\Big\{ \int_{\conplus_{\tilde u} \cap \regio{t_1^*}{t_2^*}}|L \dot Z_{I', 0, k'}|^2 + (1-\mu)|\snabla\dot Z_{,I',0,k'}|^2 \desphere \de v\\
&\qquad + \int_{\Sigma_{t^*_3} \cap \mathfrak{Z}_{u_1}^{u_2}} |L \dot Z_{,I',0,k'}|^2 + |\snabla \dot Z_{,I',0,k'}|^2 + | \lbar \dot Z_{,I',0,k'}|^2   r^{-2}\de \Sigma_{t^*}\Big\} \\
&\quad + C\varepsilon^2 (|u_1|+1)^{-2}.
\end{align*}
Here, we used the Cauchy--Schwarz inequality with parameter $\eps$ to deduce the second line from the first. This is because we seek to absorb the boundary terms which arise from an application of the Morawetz estimate~\eqref{eq:morafgiustodegstat} into the left hand side of the resulting inequality. Moreover, we used the bounds~\eqref{eq:boundrrhoimp2} to bound the nonlinear error terms.
Similarly,
\begin{align*}
&(I) \leq C \eps \sum_{\substack{|I'|+ k' \leq l-3\\ |I'|\geq 1}}  \fgiustordeg[\dot Z_{,I',0,k'}](u_1) +  \fgiustordeg[\dot Z_{,I',0,k'}](u_2)\\ 
&\quad + C \eps \sum_{\substack{|I'|+ k' \leq l-3\\ |I'|\geq 1}}\Big\{ \int_{\conplus_{\tilde u} \cap \regio{t_1^*}{t_2^*}}|L \dot Z_{I', 0, k'}|^2 + (1-\mu)|\snabla\dot Z_{,I',0,k'}|^2 \desphere \de v\\
&\qquad + \int_{\Sigma_{t^*_3} \cap \mathfrak{Z}_{u_1}^{u_2}} |L \dot Z_{,I',0,k'}|^2 + |\snabla \dot Z_{,I',0,k'}|^2 + | \lbar \dot Z_{,I',0,k'}|^2   r^{-2}\de \Sigma_{t^*}\Big\} \\
&\quad + C\varepsilon^2 (|u_1|+1)^{-2}.
\end{align*}
All in all, we obtain, recalling that we defined $\dot Z := \dot Z_{I, 0, k}$,
\begin{align}
&\fgiustordeg[\dot Z](u_2)\nonumber\\
&\quad +\int_{\conplus_{\tilde u} \cap \regio{t_1^*}{t_2^*}}|L \dot Z|^2 + (1-\mu)|\snabla\dot Z|^2 \desphere \de v +\int_{\Sigma_{t^*_3} \cap \mathfrak{Z}_{u_1}^{u_2}} \left(|L \dot Z|^2 + |\snabla \dot Z|^2 + |\hat \lbar \dot Z|^2  \right) r^{-2}\de \Sigma_{t^*} \nonumber\\
&
\leq C\fgiustordeg(u_1) + C \eps \sum_{\substack{|I'|+ k' \leq l-3\\ |I'|\geq 1}}  \fgiustordeg[\dot Z_{,I',0,k'}](u_1) +  \fgiustordeg[\dot Z_{,I',0,k'}](u_2) \label{eq:masterzendeg} \\ 
&\quad + C \eps \sum_{\substack{|I'|+ k' \leq l-3\\ |I'|\geq 1}}\Big\{ \int_{\conplus_{\tilde u} \cap \regio{t_1^*}{t_2^*}}|L \dot Z_{I', 0, k'}|^2 + (1-\mu)|\snabla\dot Z_{,I',0,k'}|^2 \desphere \de v\nonumber\\
&\qquad  \qquad + \int_{\Sigma_{t^*_3} \cap \mathfrak{Z}_{u_1}^{u_2}} |L \dot Z_{,I',0,k'}|^2 + |\snabla \dot Z_{,I',0,k'}|^2 + | \lbar \dot Z_{,I',0,k'}|^2   r^{-2}\de \Sigma_{t^*}\Big\} \nonumber\\
&\quad + C\varepsilon^2 (|u_1|+1)^{-2}.\nonumber
\end{align}
We then proceed to sum inequality~\eqref{eq:masterzendeg} for all $I$, $k$ such that $|I|+k \leq l-3$, and $|I| \geq 1$. Upon absorbing the error terms arising from $(I)$ and $(II)$ in the left hand side of the resulting inequality, this gives the following restricted version of the claim:
\begin{equation}
\boxed{
\begin{aligned}
\sum_{\substack{|I|+ k \leq l-3\\|I|\geq 1}} \fgiustordeg[\hat Z_{,I,0,k}](u_2) \leq C\sum_{\substack{|I|+ k \leq l-3\\|I|\geq 1}} \fgiustordeg[\hat Z_{,I,0,k}](u_1) + C \varepsilon^2 (|u_1|+1)^{-2}.
\end{aligned}}
\end{equation}
Furthermore, revisiting estimate~\eqref{eq:morafgiustodegstat}, we obtain the following Morawetz estimate (note that a slight modification of the above reasoning gives control on all the boundary terms on the right hand side of~\eqref{eq:morafgiustodegstat}):
\begin{equation}
\boxed{
\begin{aligned}
& \sum_{\substack{|I|+k \leq l-3\\|I|\geq 1}} \int_{\mathfrak{Z}_{u_1}^{u_2} \cap \regio{}{}} 
\Big\{\frac 1 {r^2}|L \dot Z_{I,0,k} - \lbar \dot Z_{I,0,k}|^2\\
&\qquad + \frac{(r-3M)^2}{r^3}\left(|\snabla \dot Z_{I,0,k}|^2+ \frac 1 r |L \dot Z_{I,0,k} + \lbar \dot Z_{I,0,k}|^2\right)+\frac 1 {r^3} |\dot Z_{I,0,k}|^2\Big\} r^{-2}\de \text{Vol} \\ 
& \leq C\sum_{\substack{|I|+ k \leq l-3\\|I|\geq 1}} \fgiustordeg[\hat Z_{,I,0,k}](u_1) + C \varepsilon^2 (|u_1|+1)^{-2}.
\end{aligned}}
\end{equation}
This concludes {\bf Step 1} of the proof.

\subsection*{Step 2} We are now going to remove the degeneracies. The last display, in combination with inequality~\eqref{eq:morahats}, gives:
\begin{equation}\label{eq:morall}\boxed{
\begin{aligned}
& \sum_{\substack{|I|+j+k \leq l-3\\|I| \geq 1}} \Big\{ \int_{\Sigma_{t_2^*} \cap \{r \leq R\}} |\hat Z_{,I,j,k}|^2\de \Sigma_{t^*}+ \int_{\mathfrak{Z}_{u_1}^{u_2} \cap \regio{}{}} r^{-2} |\hat Z_{,I,j,k}|^2 r^{-2} \de \text{Vol} \Big\} \\
&\qquad \leq C\sum_{\substack{|I|+j+ k \leq l-3\\|I|\geq 1}} \fgiustor[\hat Z_{,I,j,k}](u_1) + C \varepsilon^2 (|u_1|+1)^{-2}.
\end{aligned}}
\end{equation}

Let us impose the following requirement throughout the rest of the proof: $|I|+j+k \leq l-4$. Recall the commuted equation:
\begin{equation}\label{eq:zhatcomref}
L \lbar \dot Z_{,I,j,k}- \sum_{a+b = j} \lbar^a \left(\frac{1-\mu}{r^2} \right) \sdelta_{\mathbb{S}^2} \lbar^b \dot Z_{,I,0,k} = \partial^I_{\leo} T^k \lbar^j R_\rho.
\end{equation}
Let us now choose a smooth nondecreasing radial function $g(r)$ which satisfies the following: $g(r) = 0$ for $r \in [2M, R-M]$, $g(r) =1$ for $r \in [R, \infty)$.
We now multiply relation (\ref{eq:zhatcomref}) by $2 g(r) \p_t\dot Z_{,I,j,k}$, and integrate on $\mathbb{S}^2$. We obtain:
\begin{equation*}
\begin{aligned}
&\frac 12 \lbar\left( g(r)|L \dot Z_{,I,j,k}|^2\right)+\frac 12 L\left( g(r)|\lbar \dot Z_{,I,j,k}|^2\right)-\frac 12 \lbar g(r) |L \dot Z_{,I,j,k}|^2 -\frac 12 L g(r) |\lbar \dot Z_{,I,j,k}|^2
\\ & +\frac 12 (L + \lbar) (g(r) |\snabla \dot Z_{,I,j,k}|^2)
-2 g(r) \sum_{\substack{a + b = j\\ a \geq 1}}\p_t \dot Z_{,I,j,k} \, \lbar^a \left(\frac{1-\mu}{r^2}\right) \sdelta_{\mathbb{S}^2} \dot Z_{,I,b,k} \\
& \qquad \stackrel{\mathbb{S}^2}{=} 2 g(r) \p_t \dot Z_{,I,j,k} \partial_{\leo}^I T^k \lbar^j R_\rho.
\end{aligned}
\end{equation*}
We now proceed to integrate the previous display on $\mathfrak{Z}_{u_1}^{u_2}\cap \regio{}{} \cap \{r \geq R-M\}$ with respect to the form $\de u \de v$. We obtain the following inequality, discarding some of the positive boundary terms:
\begin{equation}\label{eq:secondadd}
\begin{aligned}
& \int_{\conplus_{u_2} \cap \regio{}{} \cap \{r \geq R\}} |L \dot Z_{,I,j,k}|^2 +|\snabla \dot Z_{,I,j,k}|^2  \desphere \de v \\ &
\leq C \fgiustor[\dot Z_{,I,j,k}](u_1) + C \int_{\regio{}{} \cap \mathfrak{Z}_{u_1}^{u_2}  \cap \{R -M \leq r \leq R\} } |L \dot Z_{,I,j,k}|^2 +|\lbar \dot Z_{,I,j,k}|^2  \desphere \de v\\ &
\quad + C\int_{\regio{}{} \cap \mathfrak{Z}_{u_1}^{u_2} \cap \{r \geq R -M\}} \left(\sum_{b=0}^{j-1}r^{-4}|\sdelta_{\mathbb{S}^2} \dot Z_{,I,b,k}|^2 + r^2|\partial_{\leo}^I T^k \lbar^j R_\rho|^2\right)\desphere \de u \de v \\
&\quad +C\int_{\regio{}{} \cap \mathfrak{Z}_{u_1}^{u_2} \cap \{r \geq R -M\}} r^{-2} |\p_t \dot Z_{,I,j,k}|^2\desphere \de u \de v 
\end{aligned}
\end{equation}
Summing the previous inequality over all $I, j, k$ such that $|I|+j+k \leq l-4$, and such that $|I|\geq 1$, we obtain, bounding the spacetime terms with the aid of~\eqref{eq:morall},
\begin{equation}\label{eq:claimmoraZ}
\boxed{
\begin{aligned}
&\sum_{\substack{|I|+j+k \leq l-4\\ |I|\geq 1}}\fgiustor[\hat Z_{,I,j,k}](u_2) +  \sum_{\substack{|I|+j+k \leq l-3\\ |I|\geq 1}}\int_{\mathfrak{Z}_{u_1}^{u_2} \cap \regio{}{}} r^{-2} |\hat Z_{,I,j,k}|^2 r^{-2} \de \text{Vol} \\
& \qquad  \leq \sum_{\substack{|I|+j+k \leq l-3\\ |I|\geq 1}}\fgiustor[\hat Z_{,I,j,k}](u_2) + C \varepsilon^2 (|u_1|+1)^{-2}.
\end{aligned}}
\end{equation}
This is the claim, and concludes {\bf Step 2} of the proof.\footnote{Note that, in the region where $r \geq R$, it is straightforward to bound fluxes of $\hat Z_{,I,j,k}$ in terms of fluxes involving $\hat Z_{,I,j,k}$. This is tacitly used deriving~\eqref{eq:claimmoraZ} from~\eqref{eq:secondadd}.}
\end{proof}

\begin{remark}
In view of Proposition~\ref{prop:encfgiusto}, we can estimate the bulk term on the right hand side of estimate~\eqref{eq:punoimp}, to obtain (under the same assumptions of Proposition~\ref{prop:previsited})
\begin{align}
	& \int_{\mathbb{S}^2}\int_{\{u = u_2\}\cap\{r \geq R\}\cap \regio{}{}} \sum_{\substack{|I|+j+k \leq l-4\\|I|\geq 1}}(r|L \dot Z_{,I,j,k}|^2) \de v \desphere \nonumber \\ & 
	+\int_{\duu \cap \regio{}{}} \sum_{\substack{|I|+j+k \leq l-4\\|I|\geq 1}}\left[|L \dot Z_{,I,j,k}|^2 + |\snabla \dot Z_{,I,j,k}|^2 + r^{-2}|\dot Z_{,I,j,k}|^2\right] (1-\mu) \de u \de v \desphere \label{eq:punowithdecay} \\ &
	\leq 	C \sum_{\substack{|I|+j+k \leq l-4\\|I|\geq 1}} \fgiustor[\dot Z_{,I,j,k}](u_1) +\sum_{\substack{|I|+j+k \leq l-4\\|I|\geq 1}} \int_{\mathbb{S}^2}\int_{\{u = u_1\}\cap\{r \geq R\}\cap \regio{}{}} r |L \dot Z_{,I,j,k}|^2 \de v \desphere + C \varepsilon^2 (1+|u_1|)^{-1}. \nonumber
	\end{align}
\end{remark}

\section{Application of the \texorpdfstring{$r^p$}{rp} method: decay of null fluxes}\label{sec:pdecay}

Now that we have all the necessary estimates, we can apply the $r^p$ method of Dafermos and Rodnianski to prove integrated decay for $Z$ and $W$. The argument will rely on exploiting the $r^p$-hierarchy, proving integrated decay on dyadic sequences and then using the results of section~\ref{sub:moragiustoests} to remove the restriction to sequences.

\begin{proposition}\label{prop:fluxdecay}
Let $l \in \enn$, $l \geq 1$. There exist a small number $\tilde \varepsilon > 0$, a number $R^* >0$ and a constant $C >0$ such that the following holds. Let $\fara$ be a smooth solution of the MBI system~\eqref{MBI} on $\mathcal{R}_{t_0^*}^{t_3^*} \subset \mathcal{S}_e$. Let $\mathcal R := \mathcal{R}_{t_0^*}^{t_3^*}$. Assume the bootstrap assumptions $BA\left(\regio{t_0^*}{t_3^*},1,\left\lfloor \frac{l+9}{2}\right\rfloor, \varepsilon \right)$. Assume the following bounds on the initial energy:
\begin{equation}
\begin{aligned}
&\sum_{|I| + j + k \leq l+4} \int_{\Sigma_{t_0^*}} \left[r^2 |L \dot Z_{,I,j,k}|^2 + r|\snabla \dot Z_{,I,j,k} |^2\right] r^{-2} \de \Sigma_{t^*} \leq \varepsilon^2, \\
&\sum_{|I|+j + k \leq l+4} \int_{\Sigma_{t_0^*}} \left[r^2 |L \dot W_{,I,j,k}|^2 + r|\snabla \dot W_{,I,j,k} |^2\right]r^{-2}\de \Sigma_{t^*} \leq \varepsilon^2,\\
&\int_{\Sigma_{t_0^*}} r |\partial^{\leq l + 5} \alpha |^2\de \Sigma_{t^*} \leq \varepsilon^2, \\
&\norm{\fara}^2_{H^{l+9}(\Sigma_{t_0^*})} \leq \varepsilon^2.
\end{aligned}
\end{equation}
Here, $0 < \varepsilon < \tilde \varepsilon$. Let $u_0 = t_0 - R^*$, where $t_0 = t_0^* - 2M \log(R-2M)$, and $R^* = R + 2M \log(R-2M)$.

Recall that $Z := r^2 \rho$, and it satisfies Equation (\ref{eq:zdot}).
Under these assumptions, we have the inequality:
\begin{equation}\label{eq:unibdsecond}
\boxed{
\begin{aligned}
\sum_{\substack{|I|+j+k \leq l+3\\ |I|\geq 1}}  \int_{\mathbb{S}^2}\int_{\conplus_{u_1} \cap \regio{}{} \cap \{r \geq R\}} r^{2} |L \hat Z_{,I,j,k}|^2 \de v \desphere \leq C \varepsilon^2,
\end{aligned}
}
\end{equation}
valid for all $u_1 \in \R$. Furthermore, we have the decay of fluxes, valid for $u \geq u_0$,
\begin{equation}\label{eq:decayfuno}
\boxed{
\begin{aligned}
\sum_{\substack{|I|+j+k \leq l+1\\|I|\geq 1}} \fgiustor[\hat Z_{,I,j,k}](u) + \int_{\conplus_{u} \cap \regio{}{} \cap \{r \geq R\}} r|L \hat Z_{,I,j,k}|^2 \desphere \de v \leq
C (|u|+1)^{-1} \varepsilon^2,
\end{aligned}
}
\end{equation}

\begin{equation}\label{eq:decayfdue}
\boxed{
\begin{aligned}
\sum_{\substack{|I|+j+k \leq l\\|I|\geq 1}} \fgiustor[\hat Z_{,I,j,k}](u) \leq C (|u|+1)^{-2} \varepsilon^2.
\end{aligned}
}
\end{equation}
\end{proposition}

\begin{remark}
Notice that in this Proposition we require five plus the number of derivatives needed in all propositions of Section~\ref{sec:previsited}. Hence, all results of previous propositions require five additional derivatives.
\end{remark}

\begin{proof}
Let us set $\dot Z := \dot Z_{,I,j,k}$, and $\hat Z := \hat Z_{,I,j,k}$.
It clearly follows, from (\ref{eq:pdueimpinit}) and the assumptions on the initial energy, that we have the uniform bound
\begin{equation*}
 \int_{\mathbb{S}^2}\int_{\{u = u_2\}\cap\{r \geq R\}\cap \regio{}{}} \sum_{|I|+j+k \leq l+3}(r^{2} |L \dot Z|^2) \de v \desphere 
\leq  C \int_{\Sigma_{t^*_0}} \sum_{|I|+j+k \leq l+3} (r^{2} |L \dot Z|^2) r^{-2}\de \Sigma_{t^*} + C \varepsilon^2 \leq C \varepsilon^2,
\end{equation*}
for all $u_1 \in \R$. The hypotheses let us use Proposition~\ref{prop:previsited}. Furthermore, adding a multiple of the Morawetz estimate (\ref{eq:enconsZu}), we have, from (\ref{eq:pdueimp}), letting $u_1 \geq u_0$, recalling that $\hat \lbar = (1-\mu)^{-1} \lbar$,
\begin{equation}
\begin{aligned}
&\int_{\regio{\underline{t}^*_1}{\underline{t}^*_2} \cap \{r \leq R\}} \sum_{\substack{|I|+j+k \leq l+3\\|I|\geq 1}} (|L \hat Z|^2 + |\hat \lbar \hat Z|^2 + |\snabla \hat Z|^2 + |\hat Z|^2) \de \text{Vol}\\
&+\int_{\mathfrak{D}_{u_1}^{u_2} \cap \regio{}{}} \sum_{\substack{|I|+j+k \leq l+3\\|I|\geq 1}}\left[r|L \dot Z|^2 + |\snabla \dot Z|^2 + r^{-2}|\dot Z|^2\right] (1-\mu) \de u \de v \desphere 
\leq    C \varepsilon^2.
\end{aligned}
\end{equation}
Here, $\underline{t}_i^* := u_i + R^* + 2M \log(R-2M)$, where $R^* = R + 2M\log(R-2M)$, for $i =1,2$.

Since we can easily estimate the weight $(1-\mu)$ in the region $\{r \geq R\}$, the last display in particular implies
\begin{equation}
\begin{aligned}
&\int_{\regio{\underline{t}^*_1}{\underline{t}^*_2} \cap \{r \leq R\}} \sum_{\substack{|I|+j+k \leq l+3\\|I|\geq 1}} (|L \hat Z|^2 + |\hat \lbar \hat Z|^2  +|\snabla \hat Z|^2+ |\hat Z|^2) \de \text{Vol}\\
&+\int_{\mathfrak{D}_{u_1}^{u_2} \cap \regio{}{}} \sum_{\substack{|I|+j+k \leq l+3\\|I|\geq 1}}\left[r|L \hat Z|^2 + |\snabla \hat Z|^2 + r^{-2}|\hat Z|^2\right] (1-\mu) \de u \de v \desphere 
\leq    C \varepsilon^2.
\end{aligned}
\end{equation}
We consider the sequence $\underline{u}_n := 2^n u_0$. We deduce the existence of a sequence $\tilde u_n$ such that $ \underline{u}_n \leq \tilde u_n \leq  \underline{u}_{n+1}$ and 
\begin{equation}\label{eq:decsequno}
\sum_{\substack{|I|+j+k \leq l+3\\|I|\geq 1}} \left(\fgiustor[\hat Z] (\tilde u_n) + \int_{\conplus_{\tilde u_n} \cap \regio{}{} \cap \{r \geq R\}} r |L \hat Z|^2 \de v \desphere \right)\leq C \varepsilon^2 u_n^{-1}.
\end{equation}
We eliminate the restriction to the dyadic sequence via estimate (\ref{eq:punowithdecay}). We have, for $u \geq u_0$, (note that we add $5$ derivatives, so $l-4+5 = l+1$)
\begin{equation}
\sum_{\substack{|I|+j+k \leq l+1\\|I|\geq 1}} \left(\fgiustor[\hat Z] (u) + \int_{\conplus_{\tilde u} \cap \regio{}{} \cap \{r \geq R\}} r |L \hat Z|^2 \de v \desphere \right)\leq C \varepsilon^2 u^{-1}.
\end{equation}
We now use inequality (\ref{eq:punowithdecay}), added to the Morawetz estimate (\ref{eq:enconsZu}). We estimate the boundary terms in the resulting inequality using estimate~\eqref{eq:decsequno}. We finally obtain, along the sequence $\underline{u}_n = u_0 2^n$, the bound
\begin{equation}
\begin{aligned}
&\int_{\regio{\underline{t}^*_{n}}{\underline{t}^*_{n+1}} \cap \{r \leq R\}} \sum_{\substack{|I|+j+k \leq l+1\\|I|\geq 1}} (|L \hat Z|^2 + |\hat \lbar \hat Z|^2  + |\snabla \hat Z|^2 + |\hat Z|^2) \de \text{Vol}\\
&+\int_{\mathfrak{D}_{u_n}^{u_{n+1}} \cap \regio{}{}} \sum_{\substack{|I|+j+k \leq l+1\\|I|\geq 1}}\left[|L \hat Z|^2 + |\snabla \hat Z|^2 + r^{-2}|\hat Z|^2\right] (1-\mu) \de u \de v \desphere 
\leq    C \varepsilon^2 (1+u_n)^{-1}.
\end{aligned}
\end{equation}
Here, again, $\underline{t}_i^* := \underline{u}_i + R^* + 2M \log(R-2M)$, where $R^* = R + 2M\log(R-2M)$, for $i =1,2, \ldots$
Hence, along a sequence $\bar u_n$, such that $\underline{u}_n \leq \bar u_n \leq \underline{u}_{n+1}$, we have
\begin{equation}
\sum_{\substack{|I|+j+k \leq l\\|I|\geq 1}} \fgiustor[\hat Z] (\bar u_n) \leq C \varepsilon^2 \bar u_n^{-2}.
\end{equation}
We remove the restriction to the dyadic sequence using the conservation of energy inequality in Proposition~\ref{prop:encfgiusto}.
\end{proof}

\section{Pointwise decay for \texorpdfstring{$\rho$}{rho} and \texorpdfstring{$\sigma$}{sigma}}\label{sec:sobrs}
In this section we use the bounds on the fluxes to prove $L^\infty$ decay for $\rho$ and $\sigma$.
\begin{proposition} \label{prop:linfdecrs}
There exists a small number $\tilde \varepsilon > 0$, a number $R^* > 0$ and a constant $C>0$ such that the following holds. Let us require the same assumptions of Proposition~\ref{prop:fluxdecay}. We have then
\begin{equation}\label{eq:decayrs}
|\partial^{\leq l-2} \rho|, |\partial^{\leq l-2} \sigma| \leq C \varepsilon \tau^{-1} r^{-\frac 3 2}\text{ on } \regio{}{}, \qquad |\partial^{\leq l-2} \rho|, |\partial^{\leq l-2} \sigma| \leq C \varepsilon \tau^{-\frac 1 2 } r^{-2} \text{ on } \regio{}{}.
\end{equation}
\end{proposition}

\begin{proof}
We divide the proof in two steps. In \textbf{Step 1}, we will prove the claim in the region $\regio{}{} \cap \{u \geq u_0 \}$. In \textbf{Step 2}, we will accomplish the easier task to prove the claim in the region $\regio{}{} \cap \{u \leq u_0 \}$. We restrict our reasoning to the estimates for $\rho$, the estimates for $\sigma$ being analogous.
\subsection*{Step 1}
Recall the definition of the spherical average of a function:
\begin{equation*}
\bar f := \frac 1 {4 \pi}\int_{\mathbb{S}^2} f(\omega) \desphere(\omega).
\end{equation*}
Let $|I| + j + k \leq l-2$. We apply the Sobolev embedding in the Appendix, Lemma~\ref{lem:sobforrho}, with the choice of $f := r^3 \rho$, and obtain
\begin{equation}
\begin{aligned}
&\sup_{r_1 \in [2M, R]} |\hat \rho_{,I,j,k} (t_1^*, r_1, \omega)-\overline{ \hat \rho_{,I,j,k}}(t_1^*, r_1)| \leq C
\int_{2M}^R \int_{\mathbb{S}^2}(|\snabla\snabla \hat \rho_{,I,j,k}|^2 + |\snabla_{\partial_{r_1}} \snabla \snabla \hat \rho_{,I,j,k}|^2) \desphere \de r_1 \\
& \leq \sum_{\substack{|I|+j+k \leq l\\|I|\geq 1}} \fgiusto[\hat Z] (u) \leq C \varepsilon^2 u^{-2}.
\end{aligned}
\end{equation}
Together with the decay of the spherical average in Proposition~\ref{prop:charge}, we have the claim in the region $\regio{}{} \cap \{r \leq R\}$.

We now wish to derive bounds on the region $\regio{}{} \cap \{u \geq u_0\} \cap \{r \geq R\}$. 
We begin by noticing:
\begin{equation}
\begin{aligned}
&|\hat \rho_{,I,j,k} -\overline{ \hat \rho_{,I,j,k}}|^2  \leq C \int_{\mathbb{S}^2} \sum_{i,j =1}^3 |\snabla_{\Omega_i} \snabla_{\Omega_j} \hat \rho_{,I,j,k}|^2 \de \mathbb{S}^2 \\ & \leq C \int_{\mathbb{S}^2} r^4 |\snabla \snabla \hat \rho_{,I,j,k}|^2 \de \mathbb{S}^2 \leq C \int_{\mathbb{S}^2} r^{-2} \sum_{|J| = |I|+1} |\snabla \hat Z_{,J,j,k}|^2 \de \mathbb{S}^2.
\end{aligned}
\end{equation}
Also, we have that, letting $\hat Z := \hat Z_{,J,j,k}$, and $v_R$ such that $-u + v_R = 2R^*$,
\begin{equation}
\begin{aligned}
&r\int_{\mathbb{S}^2} |\snabla \hat Z|^2(u,v,\omega) \de \mathbb{S}^2 \\ 
& \leq C 
\int_{v_R}^v \int_{\mathbb{S}^2}|\snabla \hat Z|^2(u,v,\omega) \de \mathbb{S}^2 \de v + 
C\int_{v_R}^v \int_{\mathbb{S}^2}r \left(|\snabla \hat Z| |\snabla_L \snabla \hat Z| \right)\de \mathbb{S}^2 \de v
+C \int_{\mathbb{S}^2}  |\snabla \hat Z|^2(u_R,v_R, \omega) \de \mathbb{S}^2 \\ 
& \leq C
\int_{v_R}^v \int_{\mathbb{S}^2} \sum_{i=1}^3 \left(|\snabla \hat Z| |\snabla_L \snabla_{\Omega_i} \hat Z| \right)\de \mathbb{S}^2 \de v + C\int_{v_R}^v \int_{\mathbb{S}^2} \sum_{i=1}^3 |\snabla \hat Z|^2\de \mathbb{S}^2 \de v 
+C \int_{\mathbb{S}^2}  |\snabla \hat Z|^2(u_R,v_R, \omega) \de \mathbb{S}^2 \\ 
& \leq C \varepsilon^2 \tau^{-2}
+ C \left( \int_{v_R}^v\int_{\mathbb{S}^2} |\snabla \hat Z|^2 \de \mathbb{S}^2 \de v \right)^{\frac 1 2}  \left( \int_{v_R}^v \int_{\mathbb{S}^2} \sum_{i=1}^3 |\snabla_L  \snabla_{\Omega_i }\hat Z|^2 \de \mathbb{S}^2 \de v \right)^{\frac 1 2} \leq C \varepsilon^2 \tau^{-2}.
\end{aligned}
\end{equation}
Let us now restrict to $|I|+j+k \leq l-2$, and furthermore $|J| = |I| +1$, so that 
\begin{equation}
\int_{v_R}^v \int_{\mathbb{S}^2} \sum_{i=1}^3 |\snabla_L  \snabla_{\Omega_i }\hat Z|^2 \de \mathbb{S}^2 \de v\leq C\sum_{\substack{H+j'+k' \leq l\\|H|\geq 1}} \fgiustor[\hat Z_{,H,j',k'}](u) \leq C \varepsilon^2 (|u|+1)^{-2},
\end{equation}
by Proposition~\ref{prop:fluxdecay}.
This implies, together with the decay of the spherical average in Proposition~\ref{prop:charge}, that the claim holds in the region $\regio{}{} \cap \{u \geq u_0\}$:
\begin{equation}
|\partial^m \rho|, |\partial^m \sigma| \leq C \varepsilon \tau^{-1 } r^{-\frac 3 2},
\end{equation}
for $m \leq l-2$.

We now turn to the second part of the claim (decay rate: $\tau^{- \frac 1 2}r^{-2}$). We have, again letting $\hat Z := \hat Z_{,J,j,k}$, and $v_R$ such that $-u + v_R = 2R^*$,
\begin{equation}\label{eq:phifar}
\begin{aligned}
&r^2\int_{\mathbb{S}^2} |\snabla \hat Z|^2(u,v,\omega) \de \mathbb{S}^2(\omega) \leq C \int_{\mathbb{S}^2} \sum_{i = 1}^3 |\snabla_{\Omega_i}\hat Z|^2(u,v, \omega)\desphere(\omega) \\ & \leq 
C\int_{\mathbb{S}^2} \sum_{i = 1}^3 |\snabla_{\Omega_i}\hat Z|^2(u_R,v_R, \omega)\desphere(\omega)+C \int_{v_R}^v \int_{\mathbb{S}^2} \snabla_L( \sum_{i=1}^3|\snabla_{\Omega_i}\hat Z|^2) \desphere(\omega)\de v.
\end{aligned}
\end{equation}
The first term in the right hand side of the last display is estimated by averaging in $r$ and using the decay estimate~(\ref{eq:decayfdue}). For the second term, we notice that
\begin{equation}\label{eq:phifar2}
\begin{aligned}
&\left|\int_{v_R}^v \int_{\mathbb{S}^2} \snabla_L( \sum_{i=1}^3|\snabla_{\Omega_i}\hat Z|^2) \desphere(\omega)\de v\right| \\ & \leq C 
\sum_{i=1}^3 \left( \int_{v_R}^v \int_{\mathbb{S}^2} r^{-2}|\snabla_{\Omega_i}\hat Z|^2 \desphere(\omega)\de v\right)^{\frac 1 2} \times \left(\int_{v_R}^v \int_{\mathbb{S}^2} r^2|\snabla_L \snabla_{\Omega_i}\hat Z|^2 \desphere(\omega)\de v\right)^{\frac 1 2}  \leq C 
\varepsilon^2 \tau^{-1}.
\end{aligned}
\end{equation}
Here, we used both inequalities (\ref{eq:unibdsecond}) and (\ref{eq:decayfdue}). This reasoning, together with the spherical Sobolev estimate in Lemma~\ref{lem:sobsphere} and Proposition~\ref{prop:charge} on the spherical average implies the bound for $\rho$:
\begin{equation}\label{eq:decayrhor}
|\partial^m \rho| \leq C \varepsilon \tau^{-\frac 1 2 } r^{-2},
\end{equation}
for all $m \leq l-2$. This concludes the proof of \textbf{Step 1}.
\subsection*{Step 2} In the region $u \leq u_0$, it suffices to show, for some constant $C>0$, that
\begin{equation*}
|\partial^m \rho| \leq C \varepsilon r^{-2}.
\end{equation*}
This follows by the uniform estimate (\ref{eq:unibdsecond}), and the bounds on initial data on $\Sigma_{t_0^*}$, plus $r$-weighted Sobolev embedding on $\Sigma_{t_0^*}$.
\end{proof}
\section{Pointwise decay for \texorpdfstring{$\alpha$}{alpha}}\label{sec:soba}
In this section, we establish $L^\infty$ bounds for $\alpha$. These arguments are very similar in nature to those carried our for the linear case, both in~\cite{masthes} and in~\cite{linearized}. We prove the following Proposition.
\begin{proposition}\label{prop:linfalpha}
There exist a small number $\tilde \varepsilon > 0$, a number $R^* > 0$ and a constant $C > 0$ such that the following holds. Let us require the same assumptions of Proposition~\ref{prop:fluxdecay}.
Furthermore, let us assume that $\alpha$ satisfies the initial bounds
\begin{equation}\label{eq:initalpha}
\int_{\Sigma_{t_0^*}} r^2 |\partial^{\leq l} \alpha|^2 \de \Sigma_{t^*} \leq \varepsilon^2, \qquad \int_{\Sigma_{t_0^*}} r^4 |\partial^{\leq l-1} \snabla_{\p_{r_1}} \alpha|^2 \de \Sigma_{t^*} \leq \varepsilon^2.
\end{equation}
We have then
\begin{equation}\label{eq:decayalphai}
|\partial^{\leq l-5} \alpha|\leq C \varepsilon \tau^{-\frac 1 2 } r^{-2},\qquad
|\partial^{\leq l-5} \alpha| \leq C \varepsilon \tau^{-1} r^{-\frac 3 2},
\end{equation}
on the region $\regio{}{}$.
\end{proposition}

\begin{proof}
We divide the proof in two \textbf{Steps}. In \textbf{Step 1}, we will derive the uniform bound $|\partial^m \alpha| \leq C \varepsilon r^{-1}$ for $m \leq l-4$. In \textbf{Step 2}, we will improve the rate up to the one claimed in the statement of the Proposition.

\subsection*{Step 1} We have, from Equation~(\ref{eq:transpd}), that $\rho$ and $\sigma$ satisfy the following transport equations:
\begin{align}
&L(r^2 \sigma)+r^2 \curl \alpha = 0, 
&-L(r^2 \rho) + r^2 \dive \alpha = - r^2{H_{_\Delta}} \indices{^\mu _L ^\kappa ^\lambda} \nabla_\mu \fara_{\kappa \lambda}.
\end{align}
We commute these equations with the operator $\partial^I_{\leo} \hat \lbar^j T^k$, with $|I| + j + k \leq l-3$, and use the decay estimates proved in Proposition~\ref{prop:linfdecrs}. We use the ``unweighted'' Sobolev embedding of Corollary~\ref{cor:unweightedsob} to bound the containing the expression $ \nabla_\mu \fara_{\kappa \lambda}$. Finally, we use the Sobolev embedding on the sphere $\mathbb{S}^2$, plus the fact that the map $\alpha \mapsto (\dive \alpha, \curl \alpha)$ has no kernel. We have the claim:
\begin{equation}\label{eq:decalphaprel}
|\partial^{\leq l-4} \alpha| \leq C \varepsilon r^{-1}.
\end{equation}
\subsection*{Step 2}
By Equation~\eqref{eq:alpha} for $\alpha$ derived in Lemma~\ref{lem:transportalpha}, we have
\begin{equation}\label{eq:alpharef3}
\snabla_\lbar(r \alpha_A) = r(1-\mu)(\snabla_A \rho + \svol_{AB}\snabla^B \sigma ) + r(1-\mu)(\hdelta)_{\mu A \kappa \lambda} \nabla^{\mu} \fara^{\kappa\lambda}.
\end{equation}
This implies, letting $e_i := \Omega_i /r$, where $i \in \{1,2,3\}$, since $\snabla_\lbar e_i = 0$,
\begin{equation*}
\lbar (r \alpha(e_i)) = r(1-\mu)(\snabla_{e_i} \rho + \svol_{e_i B}\snabla^B \sigma ) + r(1-\mu)(\hdelta)_{\mu e_i \kappa \lambda} \nabla^{\mu} \fara^{\kappa\lambda}.
\end{equation*}
Commuting the previous display with derivatives $\partial_{\leo}^I T^k$, such that $|I|+k \leq l-5$, we then obtain, letting $A := \alpha(e_i)$,
\begin{equation*}
\lbar(r \dot A_{,I,0,k}) = r(1-\mu)\partial_{\leo}^I T^k (\snabla_{e_i} \rho + \svol_{e_i B}\snabla^B \sigma ) + r(1-\mu) \partial_{\leo}^I T^k ((\hdelta)_{\mu e_i \kappa \lambda} \nabla^{\mu} \fara^{\kappa\lambda}).
\end{equation*}
We let $p:= (u,v, \theta, \varphi)$, and we integrate the previous display on a line of constant $v$ from the initial hypersurface $\Sigma_{t_0^*}$ to $u$. We let $u_{t_0^*}$ be the coordinate of the point of intersection of the line of constant $v$-coordinate passing through $p$. The boundary term at $\Sigma_{t_0^*}$ can be estimated from initial data, by the Sobolev embedding, and the assumptions on the norm in (\ref{eq:initalpha}):
\begin{equation}\label{eq:alphaint}
\begin{aligned}
r|\dot A_{I,0,k}| \lesssim & \int_{u_{t_0^*}}^u \sum_{i = 1}^3( \underbrace{|\snabla_{\Omega_i} \dot \rho_{,I,0,k}| + |\snabla_{\Omega_i} \dot \sigma_{,I,0,k}|}_{(I)} + \underbrace{|\partial_{\leo}^I T^k( (\hdelta)_{\mu \Omega_i \kappa \lambda} \nabla^\mu \fara^{\kappa\lambda})|}_{(II)})\de u'\\
&+ C \varepsilon^2 \min\{\tau^{-1}r^{-\frac 1 2} , \tau^{- \frac 1 2}r^{-1}\}.
\end{aligned}
\end{equation}
Here, we implicitly used the fact that our definition of $\tau$ implies $v^{-1} \leq C \min\{\tau^{-1}, r^{-1}\}$, for some constant $C > 0$. We now would like to show the following \textbf{Claim}:
\begin{equation*}
(I) + (II) \leq C \varepsilon \min\{\tau^{-1}r^{-\frac 3 2}, \tau^{-\frac 1 2}r^{-2}\}.
\end{equation*}
The term $(II)$ is estimated using the ``unweighted'' Sobolev embedding of Corollary~\ref{cor:unweightedsob}, estimate~(\ref{eq:decalphaprel}) from \textbf{Step 1}, and the bootstrap assumptions. Indeed, by the form of $H_{_\Delta}$, it suffices to estimate the term
$$
|\partial^I_{\leo} T^k( \fara_{\mu \Omega_i}\nabla^\mu \lun)|.
$$
After expressing the sum in the index $\mu$, we obtain
\begin{equation}\label{eq:structhd}
|\partial^I_{\leo} T^k( \fara_{\mu \Omega_i}\nabla^\mu \lun)| \leq C |\partial^I_{\leo}T^k(\sigma \Omega_i(\lun))|+r|\partial^I_{\leo}T^k(\alpha(e_i) \lbar (\lun))|+r|\partial^I_{\leo}T^k(\alphabar(e_i)L(\lun))|.
\end{equation}
All the terms except the last in (\ref{eq:structhd}) can be estimated by the unweighted Sobolev embedding of Corollary~\ref{cor:unweightedsob}, plus the bootstrap assumptions. Regarding the last term, we again use the unweighted Sobolev embedding of Corollary~\ref{cor:unweightedsob}, plus the bootstrap assumptions, except when $\alpha$ carries the highest number of derivatives. In that case, we use \textbf{Step 1} and the bootstrap assumptions, to obtain
$$
r|\partial^a\alphabar| |\partial^b \alphabar||\partial^c \alpha| \leq \varepsilon \min\{\tau^{-1}r^{-\frac 3 2}, \tau^{-\frac 1 2}r^{-2}\},
$$
if $a+ b + c \leq l-4$. We have therefore obtained the \textbf{Claim} for term $(II)$.

Regarding term $(I)$, we recall what we proved in Proposition~\ref{prop:linfdecrs}:
\begin{equation}
|\partial^m \rho|, |\partial^m \sigma| \leq C \varepsilon \tau^{-1} r^{-\frac 3 2}\text{ on } \regio{}{}, \qquad |\partial^m \rho|, |\partial^m \sigma| \leq C \varepsilon \tau^{-\frac 1 2 } r^{-2} \text{ on } \regio{}{},
\end{equation}
for all $m \leq l-2$. Hence, we have the \textbf{Claim} also for $(I)$.

We would then like to estimate the integral appearing in (\ref{eq:alphaint}). If necessary, we split such integral in two parts:
\begin{equation}\label{eq:split}
\int_{u_{_{\Sigma_{t^*_0}}}}^u = \int_{u_{_{\Sigma_{t^*_0}}}}^{u_0} + \int_{u_0}^u.
\end{equation}
In order to estimate the first integral in the decomposition~\eqref{eq:split}, we just use the fact that, in the region $\{u \leq u_0 \} \cap \regio{}{}$, $v \leq 2 r^* + u_0$ and $v \geq v_0$, for some constant $v_0 \in \R$. The claim follows in a straightforward manner.

To estimate the second integral in the decomposition~\eqref{eq:split}, we observe that trivially, on the region $\{u \geq u_0\}\cap \regio{}{}$, either
$$
u \geq v/3 \qquad \text{or} \qquad (u \leq v/3 \iff 3 r^* \geq v).
$$
Hence, in $\regio{}{} \cap \{u \geq v/3\} \cap \{u \geq u_0\}$,
$$
|\partial^m \rho|, |\partial^m \sigma| \leq C \varepsilon v^{-1}r^{-\frac 3 2},
$$
whereas in $\regio{}{} \cap \{u \leq v/3\} \cap \{u \geq u_0\}$,
$$
|\partial^m \rho|, |\partial^m \sigma|\leq C \varepsilon u^{-1/2}v^{-2},
$$
for $m \leq l-2$.
Then,
\begin{equation*}
\begin{aligned}
\int_{u_0}^v (u')^{-\frac 1 2}v^{-2} \de u' = \int_{u_0}^{v/3} (u')^{-\frac 1 2} v^{-2}\de u' + \int_{v/3}^u (v-u')^{-3/2}v^{-1}\de u' \\
= 2((v/3)^{\frac 1 2} - u_0^{\frac 1 2})v^{-2} -2 ((v-u)^{- \frac 1 2}- (v-v/3)^{-\frac 1 2}) v^{-1} \lesssim v^{-1} r^{-\frac 1 2}.
\end{aligned}
\end{equation*}
This implies the decay rates
\begin{equation*}
|\dot A_{,I,0,k}| \leq C \varepsilon \tau^{- \frac 1 2} r^{-2}, \qquad |\dot A_{,I,0,k}| \leq C \varepsilon \tau^{-1} r^{-\frac 3 2},
\end{equation*}
with $|I| + k \leq l-5$. Finally, commuting the transport equation for $\alpha$ (Equation~\eqref{eq:alpharef3}) multiple times with $\hat \lbar$, using the bootstrap assumptions, the unweighted Sobolev embedding of Corollary~\ref{cor:unweightedsob}, and \textbf{Step 1}, implies 
\begin{equation*}
|\hat A_{,I,j,k}| \leq C \varepsilon \tau^{- \frac 1 2} r^{-2}, \qquad |\dot A_{,I,j,k}| \leq C \varepsilon \tau^{-1} r^{-\frac 3 2},
\end{equation*}
with again $|I| + j + k \leq l-5$. This concludes the proof of the Proposition.
\end{proof}

\section{Pointwise decay for \texorpdfstring{$\alphabar$}{alphab}}\label{sec:sobb}

\begin{proposition}\label{prop:decayalphabar}
There exists a small number $\tilde \varepsilon > 0$, a number $R^* > 0$ and a constant $C>0$ such that the following holds. Let us require the same assumptions of Proposition~\ref{prop:fluxdecay}.
We have then
\begin{equation}\label{eq:decayalphabar}
|\partial^{\leq l-4} \widetilde \alphabar|\leq C \varepsilon \tau^{-1} r^{-1} \text{ on } \regio{}{}.
\end{equation}
Here, recall that $\widetilde \alphabar := (1-\mu)^{-1} \alphabar$.
\end{proposition}
\begin{proof}
This proof will be carried out looking at the transport equations satisfied by $\alphabar$, and using the $L^\infty$ decay estimates previously obtained in Proposition~\ref{prop:linfdecrs} and Proposition~\ref{prop:linfalpha}. We divide the proof in two \textbf{Steps}. In \textbf{Step 1}, we will prove the required decay estimate in the region $\regio{}{} \cap \{r \leq R\}$. In \textbf{Step 2}, we will prove the required estimate in $\regio{}{} \cap \{r \geq R\}$.

\subsection*{Step 1}
We recall the transport equations (\ref{eq:transpun}):
\begin{equation}
- \hat \lbar(r^2 \sigma) + (1-\mu)^{-1} r^2 \curl \alphabar= 0, \qquad \hat \lbar(r^2 \rho) + r^2 (1-\mu)^{-1}\dive \alphabar = - r^2{H_{_\Delta}} \indices{^\mu _{\hat \lbar} ^\kappa ^\lambda} \nabla_\mu \fara_{\kappa \lambda}.
\end{equation}
Commuting these relations with $\partial_{\leo}^I \hat \lbar^j T^k$, with $|I|+j+k \leq l-4$, using the decay rates obtained in Proposition~\ref{prop:linfalpha} and Proposition~\ref{prop:linfdecrs}, and the bootstrap assumptions on the terms containing $H_{_\Delta}$, we obtain the claim in the region $\regio{}{} \cap \{2M \leq r \leq R\}$.

\subsection*{Step 2} Look at the equation for $\alphabar$:
\begin{equation}\label{eq:alphabarreprep}
\boxed{
\snabla_L(r \alphabar_A) = r(1-\mu)(- \snabla_A \rho + \svol_{AB}\snabla^B \sigma ) +r(1-\mu)(\hdelta)_{\mu A \kappa \lambda} \nabla^{\mu} \fara^{\kappa\lambda}.
}
\end{equation}
Let $e_i := r^{-1}\Omega_i$ as before, and use the fact that $\snabla_L e_i = 0$ to obtain
\begin{equation*}
L(r \alphabar(e_i)) = (1-\mu)(- \snabla_{\Omega_i} \rho + \svol_{\Omega_i B}\snabla^B \sigma ) +(1-\mu)(\hdelta)_{\mu {\Omega_i} \kappa \lambda} \nabla^{\mu} \fara^{\kappa\lambda}.
\end{equation*}
Letting now $\underline A := \alphabar(e_i)$, we commute the previous display with the operator $\partial^I_{\leo} T^k$, with $|I|+k \leq l-4$. We obtain:
\newcommand{\Abar}{{\underline A}}
\begin{equation*}
L(r \dot {\underline A}_{,I,0,k}) = (1-\mu)(- \snabla_{\Omega_i} \dot \rho_{,I,0,k} + \svol_{\Omega_i B}( \snabla^B \dot \sigma_{,I,0,k}) ) +(1-\mu)\partial_{\leo}^I T^k((\hdelta)_{\mu {\Omega_i} \kappa \lambda} \nabla^{\mu} \fara^{\kappa\lambda}).
\end{equation*}
Hence, integrating on a line of constant $u$-coordinate, letting $v_R \in \R$ such that $v_R -u = 2R^*$, we obtain
\begin{equation}\label{eq:masterabar}
 r |\dot {\underline{A}}_{,I,0,k} (u,v,\omega)| - R |\dot \Abar(u, v_R, \omega)| \leq C \int_{v_R}^v 
 \underbrace{(|\Omega_i \dot \rho_{,I,0,k}| + \sum_{a=1}^3 |\Omega_a \dot \sigma_{,I,0,k}| )}_{(i)} +\underbrace{|\partial^I_{\leo} T^k ((\hdelta)_{\mu \Omega_i \kappa \lambda} \nabla_{\mu} \fara_{\kappa\lambda})|}_{(ii)} \de v.
\end{equation}
We now notice that term $(ii)$ can be dealt with exactly in the same way as term $(II)$ in the proof of Proposition~\ref{prop:linfalpha}. We have then, under our assumptions, that, in the region $\regio{}{} \cap \{r \geq R \}$,
\begin{equation*}
(ii) \leq C \varepsilon (|u|+1)^{-1}r^{-\frac 3 2}.
\end{equation*}
Similarly, from Proposition~\ref{prop:linfdecrs}, we have
\begin{equation*}
(i) \leq C \varepsilon (|u|+1)^{-1}r^{-\frac 3 2}.
\end{equation*}
We now notice that
\begin{equation}
\int_{v_R}^{v} (|u|+1) \  r(u,v')^{-\frac 3 2} \de v' = (|u|+1)^{-1}\left( \frac 1 {R^{\frac 1 2}}- \frac 1 {r^{\frac 1 2}}\right) \leq C (|u|+1)^{-1}.
\end{equation}
Plugging this into Equation~\eqref{eq:masterabar}, we obtain the claim for all derivatives of $\alphabar$, except those in the $\hat \lbar$ direction.

Since we are far from $\{r = 2M\}$, it suffices to show the desired claim for derivatives only in the frame $\{L, T, \Omega_1, \Omega_2, \Omega_3\}$.

From Equation~(\ref{eq:alphabarreprep}), it follows that
\begin{equation*}
L(r \dot {\underline A}_{,I,0,k}) = (1-\mu)(- \snabla_{\Omega_i} \dot \rho_{,I,0,k} + \svol_{\Omega_i B}( \snabla^B \dot \sigma_{,I,0,k}) ) +(1-\mu)\partial_{\leo}^I T^k((\hdelta)_{\mu {\Omega_i} \kappa \lambda} \nabla^{\mu} \fara^{\kappa\lambda}).
\end{equation*}
Commuting the last display with the operator $L^j$, using the decay rates obtained in Proposition~\ref{prop:linfdecrs}, the bootstrap assumptions and the unweighted Sobolev embedding of Corollary~\ref{cor:unweightedsob}, we obtain
\begin{equation}
|\partial^m \alphabar|\leq C \varepsilon \tau^{-1} r^{-1},
\end{equation}
on the region $\{r \geq R\} \cap \regio{}{}$. We combine this estimate with the bound for small $r$ obtained previously in \textbf{Step 1}, and conclude the proof of the Proposition.
\end{proof}

\section{Closing the bootstrap argument}\label{sec:close}

In this Section, we conclude the argument and prove the global existence Theorem~\ref{thm:gwp}.

\begin{proof}[Proof of Theorem~\ref{thm:gwp}]
Let $\varepsilon_0 > 0$, and consider $\fara_0$ from the statement of the Theorem. Let $0 < \varepsilon < \varepsilon_0$. Let $I \subset \R$ (depending on $\fara_0$ and $\varepsilon_0$) be the set such that a smooth solution $\fara$ exists to the MBI system~(\ref{MBI}) on ${\regio{}{}}_I:=\{r \geq r_{\text{in}}\} \cap \{t^* \in I\}$ which has initial data $\fara_0$, and satisfies the bootstrap assumptions $BA\left({\regio{}{}}_I,1,\left\lfloor \frac{l+9}{2}\right\rfloor, \varepsilon \right)$. The set $I \cap (t_0^*, \infty)$ is nonempty, because of the local existence statement Theorem~\ref{thm:local}.

Furthermore, by continuity, the set $I$ is closed.

Finally, we would like to prove that the set $I \cap (t_0^*, \infty)$ is open. Suppose then that $x \in I$. We apply Propositions~\ref{prop:energy},~\ref{prop:linfdecrs},~\ref{prop:linfalpha},~\ref{prop:decayalphabar}. We arrive to the conclusion that, possibly restricting $\varepsilon_0$ to be smaller, there exists an open interval $J \subset \R$, $x \in J$, such that $\fara$ satisfies the bootstrap assumptions $BA\left({\regio{}{}}_J,\frac 1 2,l-5, \varepsilon \right)$, where ${\regio{}{}}_J := \{r \geq r_{\text{in}}\} \cap \{t^* \in J\}$. We now just need to choose $l$ such that $l-5 \geq \left \lfloor \frac{l+9}{2}\right \rfloor$. The choice $l = 18$ serves the purpose.
\end{proof}

\appendix

\section{Computation with the projected connection}
\begin{lemma}\label{lem:projectcomp}
We have the following identities:
\begin{equation*}
\snabla_L \gbar = 0, \qquad \snabla_\lbar \gbar = 0,
\end{equation*}
as well as the commutation relations
\begin{equation*}
[\snabla_L, \snabla_\lbar] = 0, \qquad [\snabla_L, r \snabla] = 0, \qquad [\snabla_\lbar, r\snabla] = 0.
\end{equation*}
The last two equalities are meant in the following sense: if $T\indices{_{A_1, \ldots, A_N}^{B_1, \ldots,B_M}}$ is a tensor of type $(N, M)$ tangential to the spheres of constant $r$, then we have
\begin{equation*}
(\snabla_L (r \snabla T))\indices{_{C, A_1, \ldots, A_N}^{B_1, \ldots,B_M}} =( r \snabla (\snabla_L T))\indices{_{C, A_1, \ldots, A_N}^{B_1, \ldots,B_M}}
\end{equation*}
and similarly
\begin{equation*}
(\snabla_\lbar (r \snabla T))\indices{_{C, A_1, \ldots, A_N}^{B_1, \ldots,B_M}} =( r \snabla (\snabla_\lbar  T))\indices{_{C, A_1, \ldots, A_N}^{B_1, \ldots,B_M}}
\end{equation*}
\end{lemma}

\begin{lemma}\label{lem:liecomm}
We have the following identities:
\begin{equation*}
[\snabla_L, \slie_{\Omega_i}] = 0, \qquad [\snabla_\lbar, \slie_{\Omega_i}] = 0, \qquad [\snabla, \slie_{\Omega_i}] = 0.
\end{equation*}
Here, $\slie$ is the Lie derivative induced by the connection $\snabla$.
\end{lemma}

\begin{proof}
We notice that, since $\Omega_i$ are Killing for $\gbar$, the last equality holds true by standard theory, cf.~\cite{globalnon}.

We will only prove the first equality when the derivatives are acting on a spherical one-form $\omega$. Let $e_i := \Omega_i/r$. Let $\eta$ be a spherical one form, and $X:= e_j$. We compute
\begin{equation*}
(\slie_{\Omega_i} \eta)(X) = (\snabla_{\Omega_i} \eta)(X) + \eta(\snabla_{X} \Omega_i).
\end{equation*}
(see Lemma 3.2.1 in~\cite{speck1}).
Then,

\begin{equation*}
\begin{aligned}
&([\snabla_L, \slie_{\Omega_i}] \omega)(X) \\
&=
L(\slie_{\Omega_i}\omega(X))-(\slie_{\Omega_i})(\snabla_L X)-\Omega_i(\snabla_L \omega(X)) + (\snabla_L \omega)([\Omega_i, X]) \\ 
&=
L(\Omega_i (\omega(X))-\omega([\Omega_i, X]))-\Omega_i (\omega(\snabla_L X))+\omega([\Omega_i, \snabla_L X])\\
&-\Omega_i (L(\omega(X))-\omega(\snabla_L X))+L(\omega([\Omega_i, X]))-\omega(\snabla_L [\Omega_i, X]) = 0.
\end{aligned}
\end{equation*}
Since $\snabla_L X = 0$, as well as $\snabla_L [\Omega_i, X] = 0$.
This proves the lemma, as the quantity in the beginning is a tensor, so it suffices to verify is vanishes on a basis at each point.
\end{proof}

\section{Shortand notation for tensors}

Let $N \in \mathbb{N}_{\geq 0}$, let $A_1, A_2, \ldots, A_n$ be tensors of type $(k_i, l_i)$, $i \in \{1, \ldots, n\}$ on $T \mathcal{S}$. Using the Schwarzschild metric, we can suppose that such tensors are all covariant of order $k_i$, so that
\begin{equation*}
A_h = (A_h)_{\alpha^h_{1} \ldots \alpha^h_{k_h}}.
\end{equation*}

\begin{definition}[Shorthand notation for tensors on $\mathcal{S}$]\label{def:short}
Let $m, N \in \N_{\geq 0}$, let $\mathcal{R} \subset \mathcal S$. Then there exists a smooth covariant tensor field $T$ on $T \mathcal R$, such that the following holds:
\begin{equation*}
(A_1 \cdots A_n)_{\beta_1 \ldots \beta_m} := T\indices{^{(\alpha^1_{1} \cdots \alpha^1_{k_1}) \cdots (\alpha^n_{1} \cdots \alpha^n_{k_n})}_{\beta_1 \cdots \beta_m}} {(A_1)}_{\alpha^1_{1} \ldots \alpha^1_{k_1}} \cdots{(A_h)}_{\alpha^n_{1} \cdots \alpha^n_{k_n}}.
\end{equation*}
\end{definition}
\begin{remark}
We will often require $T$ to satisfy the following type of bound on $\mathcal{R}$: we require the existence of a constant $C$ depending only on $\mathcal{R}$, $M$ and $N$ such that
\begin{equation*}
|\partial^N T| \leq C_{M, N, \mathcal{R}},
\end{equation*}
in the sense of Definition~\ref{def:gennormder} (without loss of generality, we can suppose $T$ to be covariant).
\end{remark}
\begin{remark}
Fix $r_{\text{in}} \in (0, 2M)$. Notice that, if we consider the region $\mathcal{R} := \mathcal{S}_e \setminus \{r_1 \leq r_{\text{in}}\}$, the Riemann tensor $\rie$ satisfies (w.l.o.g. it is covariant)
\begin{equation*}
|\partial^N \rie| \leq C_{M, N, r_{\text{in}}}.
\end{equation*}
\end{remark}

\section{Positivity of a flux of \texorpdfstring{$\dot Q$}{Qdot} and the inverse MBI metric}\label{sec:bmetric}
This section closely follows similar calculations done by J. Speck in the paper~\cite{speck1} (Proposition 7.4.4).

Given $\fara$ solution to the MBI system~\eqref{MBI}, we define the inverse MBI metric as
\begin{equation}\label{eq:invmbiform}
(b^{-1})^{\mu\nu} := g^{\mu\nu} - (1+\lun)^{-1} \fara^{\mu\kappa}\fara\indices{^\nu_\kappa}.
\end{equation}
We prove the following lemma.
\begin{lemma}[Positivity of $\dot Q(T, X)$]\label{lem:qpos}
Let $\fara$ be a solution to the MBI system~\eqref{MBI} on a region $\mathcal{R} \subset \mathcal{S}_e$. 
Define the vectorfield $X^\nu$ to be
\begin{equation}
(b^{-1})^{\nu\mu} g_{\mu\alpha} T^\alpha.
\end{equation}
Then, as long as
\begin{equation}\label{eq:hypcond}
\ellmbi  > 0
\end{equation}
on $\mathcal R$, we have that the contraction
\begin{equation}\label{eq:positivity}
\dot Q_{\mu\nu} T^\mu X^\nu \geq 0.
\end{equation}
\end{lemma}

\begin{proof}
Let us begin by defining the shorthand notation $\dot \fara := \dot \fara_{,I,0}$.
We decompose the field in the electric and magnetic parts.
\begin{equation*}
B^\nu := - T^\mu \farad\indices{_\mu^\nu}, \qquad E^\nu := T^\mu \fara\indices{_\mu^\nu},
\end{equation*}
ans similarly we define $\dot E$, $\dot B$ using the tensor $\dot \fara$.
We note that $E, B, \dot E, \dot B$ are parallel to the foliation of $\mathcal{S}_e$ by surfaces of constant $t$.
We therefore denote by $\langle \cdot , \cdot \rangle$ the (positive definite) inner product induced by the Schwarzschild metric on that foliation.
With this notation, we calculate the contractions:
\begin{equation}
\begin{aligned}
\fara\indices{_T^\kappa}\dot \fara_{T\kappa} &= \langle E, \dot E \rangle, \\
\farad\indices{_T^\kappa}\dot \fara_{T\kappa} &= -\langle B, \dot E \rangle, \\
\fara\indices{_\kappa_\lambda}\dot \fara^{\kappa\lambda} &= 2(1-\mu)^{-1}(-\langle E, \dot E \rangle + \langle B, \dot B \rangle)\\
\farad^{\kappa\lambda}\dot\fara_{\kappa\lambda} &= 2(1-\mu)^{-1}(\langle B, \dot E \rangle  + \langle E, \dot B \rangle)\\
\lun = \frac 1 2 \fara\indices{_\kappa_\lambda} \fara^{\kappa\lambda} &= (1-\mu)^{-1}(-|E|^2 + |B|^2)\\
\ldu = \frac 1 4 \farad^{\kappa\lambda}\fara_{\kappa\lambda} &=  (1-\mu)^{-1}\langle B, E \rangle.
\end{aligned}
\end{equation}
Let us now denote by latin letters $a, b, c, \ldots$ indices relative to tensors tangent to the surfaces $\{t = \text{const}\}$.
Recall:
\begin{equation}
2\ellmbi^2(1+\lun) X^\mu = 2 \ellmbi^2 (1+\lun)(b^{-1})^{\mu\nu} T_\nu = 2 \ellmbi^2(1+\lun)T^\mu - 2 \ellmbi^2 \fara^{\mu\kappa}\fara_{T\kappa}.
\end{equation}
We then have
\begin{equation}\label{eq:qcontract}
2\ellmbi^2(1+\lun) \dot Q_{\mu\nu}T^\mu X^\nu = 2 \ellmbi^2(1+(1-\mu)^{-1}|B|^2) \dot Q_{TT} - 2 \ellmbi^2 \dot Q_{Ta} \fara^{a \kappa} \fara_{T \kappa}.
\end{equation}
Recalling the form of $\dot Q$, Equation~\eqref{eq:qform}, we have
\begin{equation}
\begin{aligned}
\dot Q_{\mu\nu} &={\dot \fara}\indices{_\mu^\zeta}{\dot \fara}_{\nu\zeta}- \frac 1 4 g_{\mu\nu} {\dot \fara}_{\zeta \eta} {\dot \fara}^{\zeta\eta} + \frac 1 2 \ellmbi^{-2} \left\{- \fara\indices{_\mu^\zeta} {\dot \fara}_{\nu\zeta} \fara^{\kappa\lambda} {\dot \fara}_{\kappa\lambda} + \frac 1 4 g_{\mu\nu}(\fara^{\kappa\lambda}{\dot \fara}_{\kappa\lambda})^2 \right\} \\
&+ \frac 1 2 \ellmbi^{-2} (1+\lun) \left\{- \farad\indices{_\mu^\zeta} {\dot \fara}_{\nu\zeta} \farad^{\kappa\lambda} {\dot \fara}_{\kappa\lambda} + \frac 1 4 g_{\mu\nu}(\farad^{\kappa\lambda}{\dot \fara}_{\kappa\lambda})^2 \right\}  \\
&+ \frac 1 2 \ldu \ellmbi^{-2} \left\{\fara\indices{_\mu^\zeta} {\dot \fara}_{\nu\zeta} \farad^{\kappa\lambda} {\dot \fara}_{\kappa\lambda} - \frac 1 4 g_{\mu\nu}\fara^{\kappa\lambda}{\dot \fara}_{\kappa\lambda}\farad^{\kappa\lambda}{\dot \fara}_{\kappa\lambda} \right\}  \\
&+ \frac 1 2 \ldu \ellmbi^{-2} \left\{\farad\indices{_\mu^\zeta} {\dot \fara}_{\nu\zeta} \fara^{\kappa\lambda} {\dot \fara}_{\kappa\lambda} - \frac 1 4 g_{\mu\nu}\fara^{\kappa\lambda}{\dot \fara}_{\kappa\lambda}\farad^{\kappa\lambda}{\dot \fara}_{\kappa\lambda}\right\}.
\end{aligned}
\end{equation}
We then compute the first term in (\ref{eq:qcontract}):
\begin{equation*}
\begin{aligned}
&2 \ellmbi^2 \dot Q_{TT} = \ellmbi^2 (|\dot E|^2 + |\dot B|^2)\\
&+ (1-\mu)^{-1}(\langle E, \dot E\rangle^2 - \langle B, \dot B\rangle^2)\\
&+(1+(1-\mu)^{-1}(|B|^2-|E|^2))(1-\mu)^{-1}(\langle B, \dot E \rangle^2 - \langle E, \dot B \rangle^2)\\
&+2(1-\mu)^{-2}\langle E, B\rangle (\langle B, \dot B\rangle \langle E, \dot B\rangle + \langle E ,\dot E \rangle \langle B, \dot E \rangle).
\end{aligned}
\end{equation*}
We compute the second term in (\ref{eq:qcontract}):
\begin{equation}
\begin{aligned}
&\ellmbi^2 \dot Q_{T a} \fara^{a\kappa} \fara_{T \kappa} \\
=& - (1-\mu)^{-2}\langle B, \dot B \rangle^2 \left(|E|^2 + (1-\mu)^{-1}\langle E, B\rangle \right)\\
&+ (1-\mu)^{-2}\langle B , \dot B\rangle \langle B, \dot E\rangle (1+(1-\mu)^{-1}|B|^2)\langle E,B \rangle\\
&+ 2(1-\mu)^{-2} \langle B, \dot B\rangle \langle E,\dot B \rangle (1+ (1-\mu)^{-1}|B|^2)\langle E,B \rangle\\
&+(1-\mu)^{-1} \langle B,\dot B \rangle \langle E, \dot E \rangle (1+(1-\mu)^{-1}|B|^2)\\
&-(1-\mu)^{-2}\langle E,\dot B \rangle \left\{(1+(1-\mu)^{-1}(|B|^2-|E|^2))|B|^2+(1-\mu)^{-1}\langle E, B\rangle^2 \right\}\\
&-(1-\mu)^{-1}\langle E, \dot B \rangle \langle B, \dot E \rangle (1+(1-\mu)^{-1}(|B|^2-|E|^2))(1+(1-\mu)^{-1}|B|^2)\\
&-(1-\mu)^{-2}\langle E,\dot B \rangle \langle E, \dot E \rangle \langle E, B\rangle (1+(1-\mu)^{-1}|B|^2)
\end{aligned}
\end{equation}
These expressions are exactly the same as the ones appearing in the paper~\cite{speck1}, with the formal substitutions
$$
E \to (1-\mu)^{-\frac 1 2} E, \qquad B \to (1-\mu)^{-\frac 1 2} B, \qquad \dot E \to \dot E, \qquad \dot B \to \dot B.
$$
We proceed to show nonnegativity (in the components of $\dot E$, $\dot B$) for the resulting quadratic form:
$$
2\ellmbi^2(1+\lun) \dot Q_{\mu\nu}T^\mu X^\nu = 2 \ellmbi^2(1+(1-\mu)^{-1}|B|^2) \dot Q_{TT} - 2 \ellmbi^2 \dot Q_{Ta} \fara^{a \kappa} \fara_{T \kappa}.
$$
We choose vectorfields $e_1$, $e_2$, $e_3$ such that $E \in \text{span} \{e_1\}$, $B \in \text{span}\{e_1, e_2\}$, $\langle e_i, e_3 \rangle = \delta_{i3}$, $|e_i| =1$. We decompose:
\begin{align*}
E &= E_1 e_1,\\
B &= B_1 e_1 + B_2 e_2,\\
\dot E &= \dot E_1 e_1 + \dot E_2 e_2 + \dot E_3 e_3,\\
\dot B &= \dot B_1 e_1 + \dot B_2 e_2 + \dot B_3 e_3.
\end{align*}
Due to the orthogonality of $e_1, e_2, e_3$, the only terms containing $E_3$ and $B_3$ are the terms arising from the first term in $2\ellmbi^2(1+\lun) \dot Q_{\mu\nu}T^\mu X^\nu$, i.e. the ``linear'' term
\begin{equation}
\ellmbi^2 (1 + (1-\mu)^{-1}|B|^2)(E_3^2+B_3^2),
\end{equation}
which is manifestly positive in $E_3, B_3$ if $\ellmbi > 0$ on $\mathcal R$.

We then proceed to calculate the components of the matrix $A$ such that
\begin{equation}
(B_1, B_2, E_1, E_2) A (B_1, B_2, E_1, E_2)^t = 2\ellmbi^2(1+\lun) \dot Q_{\mu\nu}T^\mu X^\nu  - \ellmbi^2 (1 + (1-\mu)^{-1}|B|^2)(E_3^2+B_3^2).
\end{equation}
Here, the superscript $^t$ denotes transposition. We have that $A$ is obviously a symmetric matrix, with entries
\begin{align*}
A_{11} &= \left(\frac{B_2^2-E_1^2}{1-\mu}+1\right) \left(\frac{B_2^2 E_1^2}{(1-\mu)^2}+\ellmbi^2\right)\\
A_{12} &= -\frac{B_1 B_2 \left(\frac{B_2^2 E_1^2}{(1-\mu)^2}+\ellmbi^2\right)}{1-\mu}\\
A_{13} &= \frac{B_1 B_2^2 E_1 \left(\frac{B_1^2+B_2^2}{1-\mu}+1\right)}{(\mu-1)^2}\\
A_{14} &= (1 - \mu)^{-1} B_2 E_1 \left(1 + \frac{B_1^2 + B_2^2}{1-\mu}\right) \left(1 + \frac{B_2^2 - 
      E_1^2}{1-\mu}\right)\\
A_{22} &= \left(\frac{B_1^2}{1-\mu}+1\right) \left(\frac{B_2^2 E_1^2}{(1-\mu)^2}+\ellmbi^2\right)\\
A_{23} &=-\frac{B_2 E_1 \left(\frac{B_1^2}{1-\mu}+1\right) \left(\frac{B_1^2+B_2^2}{1-\mu}+1\right)}{1-\mu} \\
A_{24} &=-\frac{B_1 B_2^2 E_1 \left(\frac{B_1^2+B_2^2}{1-\mu}+1\right)}{(1-\mu)^2} \\
A_{33} &= \left(\frac{B_1^2}{1-\mu}+1\right) \left(\frac{B_1^2+B_2^2}{1-\mu}+1\right)^2\\
A_{34} &= \frac{B_1 B_2 \left(\frac{B_1^2+B_2^2}{1-\mu}+1\right)^2}{1-\mu}\\
A_{44} &= \left(\frac{B_1^2+B_2^2}{1-\mu}+1\right)^2 \left(\frac{B_2^2-E_1^2}{1-\mu}+1\right).
\end{align*}

We now denote by $M_k$ the $k$-th principal minor of the matrix $A$: $M_k := (A_{ij})_{(i,j)\in \{1, \ldots, k\}\times\{1, \ldots, k\}}$.
Let us calculate the determinants of such minors:
\begin{align}\label{pmin}
\det(M_1) &= \left(\frac{B_2^2-E_1^2}{1-\mu}+1\right) \left(\frac{B_2^2 E_1^2}{(1-\mu)^2}+\ellmbi^2\right),\\ \label{smin}
\det(M_2) &=\frac{\ellmbi^2 \left(B_2^2 E_1^2+\ellmbi^2 (1-\mu)^2\right)^2}{(1-\mu)^4} ,\\ \label{tmin}
\det(M_3) &= \frac{\ellmbi^4 \left(B_1^2+1-\mu\right) \left(B_1^2+B_2^2+1-\mu\right)^2 \left(B_2^2 E_1^2+\ellmbi^2 (1-\mu)^2\right)}{(1-\mu)^5},\\ \label{qmin}
\det(M_4) &= \frac{\ellmbi^8 \left(B_1^2+B_2^2+1-\mu\right)^4}{(1-\mu)^4} .
\end{align}
Since $\mathcal{R} \subset \mathcal{S}_e$, and since we are assuming $\ellmbi > 0$, the expressions in (\ref{smin}), (\ref{tmin}), (\ref{qmin}) are manifestly positive.

Concerning $\det (M_1)$, we distinguish two cases.
\begin{itemize}
\item If $1-(1-\mu)^{-1}E_1 > 0$, clearly $\det(M_1) > 0$,
\item If $1-(1-\mu)^{-1}E_1 \leq 0$, we notice
\begin{equation*}
\frac{B_2^2-E_1^2}{1-\mu}+1 = \ellmbi^2 - (1-\mu)^{-1}B_1^2(1-(1-\mu)^{-1}E_1) \geq \ellmbi^2  > 0.
\end{equation*}
\end{itemize}
This concludes the proof of the Lemma.
\end{proof}
\section{Calculations to deduce the \fackip Equations}\label{app:wavemax}

In this appendix, we collect useful calculations which are used to derive the form of the \fackip Equations in Proposition~\ref{lem:wavembir} and~\ref{lem:wavembis}. We restrict here to the case of the linear Maxwell system. In the proofs of Proposition~\ref{lem:wavembir} and Proposition~\ref{lem:wavembis}, we extend the reasoning to the MBI case.

Let's first prove a simple lemma about commutation of derivatives.
\begin{lemma}
	We have the equation
	\begin{equation*}
	\nabla_\alpha (\nabla^\mu \fara \indices{_\mu_\nu}) - 
	\nabla_\nu \nabla^\mu \fara\indices{_\mu_\alpha} = 
	2 \rie\indices{_\alpha_\beta^\mu_\nu}\fara\indices{_\mu^\beta}- \nabla^\mu \nabla_\mu \fara\indices{_\nu_\alpha}.
	\end{equation*}
\end{lemma}

\begin{proof}
	Commuting derivatives,
	\begin{equation*}
	\begin{aligned}
	&\nabla_\alpha (\nabla^\mu \fara \indices{_\mu_\nu}) = 
	- \rie\indices{_\nu_\beta^\mu_\alpha} \fara \indices{_\mu^\beta} + \nabla^\mu \nabla_\alpha \fara\indices{_\mu_\nu} \stackrel{(*)}{=} \\ &
	- \rie\indices{_\nu_\beta^\mu_\alpha} \fara \indices{_\mu^\beta} - \nabla^\mu \nabla_\nu \fara\indices{_\alpha_\mu}- \nabla^\mu \nabla_\mu \fara\indices{_\nu_\alpha} \stackrel{(*)}{=} \\ &
	- \rie\indices{_\nu_\beta^\mu_\alpha} \fara \indices{_\mu^\beta} +
	\rie\indices{_\alpha_\beta^\mu_\nu}\fara\indices{_\mu^\beta}
	+  \nabla_\nu \nabla^\mu \fara\indices{_\mu_\alpha}- \nabla^\mu \nabla_\mu \fara\indices{_\nu_\alpha}\\ & =
	2 \rie\indices{_\alpha_\beta^\mu_\nu}\fara\indices{_\mu^\beta}+  \nabla_\nu \nabla^\mu \fara\indices{_\mu_\alpha}- \nabla^\mu \nabla_\mu \fara\indices{_\nu_\alpha}.
	\end{aligned}
	\end{equation*}
	This implies the claim.
\end{proof}

\subsection{The wave equation, Maxwell case}
Let's derive the wave equation for $\rho$ and $\sigma$ suppressing the nonlinear term.

\begin{lemma}\label{lem:wavemaxwell}
	Let $\fara$ satisfy the Maxwell system
	\begin{equation*}
	\nabla_\mu \fara\indices{^\mu_\nu} = 0, \qquad \nabla_{[\mu} \fara_{\nu \kappa]} = 0.
	\end{equation*}
	We then have
	\begin{align}\label{eq:rhomw}
	& - r^{-2} L \lbar (r^2 \rho) +(1-\mu) \slashed{\Delta}  \rho = 0,
	\\ &
	- r^{-2} L \lbar (r^2 \sigma) +(1-\mu) \slashed{\Delta}  \sigma = 0. \label{eq:sigmamw}
	\end{align}
	Here, $\sigma$ and $\rho$ are the middle components defined in Equation~\eqref{eq:middledef}.
\end{lemma}

\begin{proof}
	We calculate:
	\begin{equation*}
	\begin{aligned}
	&\nabla_\mu \nabla^\mu (L^\nu \lbar^\alpha \fara_{\nu \alpha}) =  \\ &
	(\nabla_\mu \nabla^\mu L^\nu) \lbar^\alpha \fara_{\nu \alpha} + (\nabla^\mu L^\nu)(\nabla_\mu \lbar^\alpha) \fara_{\nu \alpha} + (\nabla^\mu L^\nu) \lbar^\alpha \nabla_\mu \fara_{\nu \alpha}\\ & +
	(\nabla_\mu L^\nu)(\nabla^\mu \lbar^\alpha) \fara_{\nu \alpha} + L^\nu (\nabla_\mu \nabla^\mu \lbar^\alpha) \fara_{\nu \alpha} + L^\nu(\nabla^\mu \lbar^\alpha) \nabla_\mu \fara_{\nu \alpha}+\\ &
	(\nabla_\mu L^\nu) \lbar^\alpha \nabla^\mu\fara_{\nu \alpha} + L^\nu (\nabla_\mu  \lbar^\alpha) \nabla^\mu \fara_{\nu \alpha} + L^\nu \lbar^\alpha \nabla^\mu\nabla_\mu \fara_{\nu \alpha} =\\ &
	(a) +(b)+(c)+\\ &
	(d) +(e) +(f) +\\ &
	(g)+ (h) + (i).
	\end{aligned}
	\end{equation*}
	We first consider $(c) + (f) + (g) + (h)$. We obtain
	
	\begin{equation*}
	\begin{aligned}
	&(c) + (f) + (g) + (h) \\ & =
	(\nabla^\mu L^\nu) \lbar^\alpha \nabla_\mu \fara_{\nu \alpha}
	+ L^\nu(\nabla^\mu \lbar^\alpha) \nabla_\mu \fara_{\nu \alpha}+
	(\nabla_\mu L^\nu) \lbar^\alpha \nabla^\mu\fara_{\nu \alpha} + L^\nu (\nabla_\mu  \lbar^\alpha) \nabla^\mu \fara_{\nu \alpha} \\ & =
	2 (\nabla^\mu L^\nu) \lbar^\alpha \nabla_\mu \fara_{\nu \alpha}
	+ 2L^\nu(\nabla^\mu \lbar^\alpha) \nabla_\mu \fara_{\nu \alpha} =\\ &
	2 (\nabla^\mu L)^L \nabla_\mu \fara_{L \lbar}
	+ 
	2 (\nabla^\mu L)^\theta \nabla_\mu \fara_{\theta \lbar} +
	2 (\nabla^\mu L)^\varphi \nabla_\mu \fara_{\varphi \lbar} +\\ &
	2(\nabla^\mu \lbar)^\lbar \nabla_\mu \fara_{L \lbar}+
	2(\nabla^\mu \lbar)^\theta \nabla_\mu \fara_{L \theta}+
	2(\nabla^\mu \lbar)^\varphi \nabla_\mu \fara_{L \varphi} \\ & =
	2 (\nabla_L L)^L \nabla^L \fara_{L \lbar}   + 
	2(\nabla_\lbar \lbar)^\lbar \nabla^\lbar \fara_{L \lbar} + (\text{angular terms}).
	\end{aligned}
	\end{equation*}
	The first two terms in the last line of the previous equation then read:
	\begin{equation*}
	\begin{aligned}
	2 (\nabla_L L)^L \nabla^L \fara_{L \lbar}   + 
	2(\nabla_\lbar \lbar)^\lbar \nabla^\lbar \fara_{L \lbar} =
	\frac{2M}{r^2} (1-\mu)^{-1}(\nabla_L \fara_{L \lbar}- \nabla_\lbar \fara_{L \lbar}).
	\end{aligned}
	\end{equation*}
	The angular terms, instead, become
	\begin{equation*}
	\begin{aligned}
	&= 2 \frac{1-\mu} r (\nabla^\theta \fara_{\theta \lbar}+\nabla^\varphi \fara_{\varphi \lbar} )-2 \frac{1-\mu} r (\nabla^\theta \fara_{L\theta}+\nabla^\varphi \fara_{L\varphi} ) \\ & 
	\stackrel{(*)}{=} 2 \frac{1-\mu}{r} (-\nabla^L \fara_{L \lbar}+ \nabla^\lbar \fara_{L \lbar}) = \frac 1 r (\nabla_\lbar \fara_{L \lbar}- \nabla_L \fara_{L \lbar}) \\ & =
	- \frac{1-\mu} r (1-\mu)^{-1}( \nabla_L \fara_{L \lbar}-\nabla_\lbar \fara_{L \lbar}).
	\end{aligned}
	\end{equation*}
	Therefore,
	\begin{equation*}
	\begin{aligned}
	&(c) + (f) + (g) + (h) = 
	\frac{2M}{r^2} (1-\mu)^{-1}(\nabla_L \fara_{L \lbar}- \nabla_\lbar \fara_{L \lbar}) - \frac{1-\mu} r (1-\mu)^{-1}( \nabla_L \fara_{L \lbar}-\nabla_\lbar \fara_{L \lbar})\\ & =
	(1-\mu)^{-1}\left(\frac{2M}{r^2} - \frac 1 r + \frac{2M}{r^2} \right)( \nabla_L \fara_{L \lbar}-\nabla_\lbar \fara_{L \lbar}).
	\end{aligned}
	\end{equation*}
	Let us notice that
	\begin{equation}
	\begin{aligned}
	&\nabla_\theta \partial_\theta = \frac r 2 \lbar - \frac r 2 L,\\ &
	\nabla_\varphi \partial_\varphi = \sin^2 \theta \frac r 2 \lbar - \sin^2 \theta \frac r 2 L- \sin \theta \cos \theta \partial_\varphi.
	\end{aligned}
	\end{equation}
	Now, consider 
	\begin{equation*}
	\begin{aligned}
	&\nabla^\mu \nabla_\mu \lbar = g^{L \lbar} (\nabla_L \nabla_\lbar \lbar - \nabla_{\nabla_{L}\lbar}\lbar) + g^{\lbar L} (\nabla_\lbar \nabla_L \lbar - \nabla_{\nabla_{\lbar} L}\lbar) \\ & +
	g^{\theta \theta} (\nabla_\theta \nabla_\theta \lbar - \nabla_{\nabla_\theta \partial_\theta} \lbar) + g^{\varphi \varphi} \left(\nabla_\varphi \nabla_\varphi \lbar - \nabla_{\nabla_\varphi \partial_\varphi}\lbar\right) \\ & =
	- \frac 1 2 (1-\mu)^{-1} \left(-\nabla_L\left(\frac{2M}{r^2} \lbar\right)\right) \\ & +
	r^{-2}\left(-\frac{1-\mu}{r}\nabla_\theta \partial_\theta-  \nabla_{ \frac r 2 \lbar - \frac r 2 L}\lbar\right)+ r^{-2} \sin^{-2} \theta \left(-\frac{1-\mu}{r}\nabla_\varphi \partial_\varphi-  \nabla_{\sin^2 \theta \frac r 2 \lbar - \sin^2 \theta \frac r 2 L- \sin \theta \cos \theta \partial_\varphi}\lbar\right) \\ & =
	\frac 1 2 (1-\mu)^{-1} \left(-2 \frac{2M}{r^3} \right)	(1-\mu)\lbar +\\ &
	r^{-2}\left(- \frac{1-\mu}{r} \left( \frac r 2 \lbar - \frac r 2 L \right) + \frac r 2 \frac{2M}{r^2} \lbar\right)\\ &
	+
	r^{-2} \sin^{-2} \theta \left(- \frac{1-\mu} r \left( \sin^2 \theta \frac r 2 \lbar - \sin^2 \theta \frac r 2 L- \sin \theta \cos \theta \partial_\varphi\right) - \frac r 2 \sin^2 \theta \nabla_\lbar \lbar + \sin \theta \cos \theta \nabla_\varphi \lbar\right) =\\ &
	- \frac{2M}{r^3} \lbar + r^{-2}\left(- \frac{1-\mu}{2}\lbar + \frac{1-\mu} 2 L + \frac M r \lbar \right) + r^{-2} \left(- \frac{1-\mu}{2}\lbar + \frac{1-\mu}2 L + \frac r 2 \frac{2M} {r^2} \lbar\right)\\ &
	=- \frac{2M}{r^3} \lbar + 2r^{-2}\left(- \frac{1-\mu}{2}\lbar + \frac{1-\mu} 2 L + \frac M r \lbar \right) \\ & =
	- \frac{1-\mu}{r^2} \lbar + \frac{1-\mu}{r^2} L.
	\end{aligned}
	\end{equation*}
	Similarly,
	\begin{align*}
	&\nabla^\mu \nabla_\mu L = g^{L \lbar} (\nabla_L \nabla_\lbar L - \nabla_{\nabla_{L}\lbar}L) + g^{\lbar L} (\nabla_\lbar \nabla_L L - \nabla_{\nabla_{\lbar} L}L) \\ & +
	g^{\theta \theta} (\nabla_\theta \nabla_\theta L - \nabla_{\nabla_\theta \partial_\theta} L) + g^{\varphi \varphi} \left(\nabla_\varphi \nabla_\varphi L - \nabla_{\nabla_\varphi \partial_\varphi}L\right) \\ & =
	-\frac 1 2 (1-\mu)^{-1} \nabla_\lbar\left(\frac{2M}{r^2}L \right) \\ & +
	r^{-2}\left(\frac{1-\mu}{r}\nabla_\theta \partial_\theta-  \nabla_{ \frac r 2 \lbar - \frac r 2 L}L\right)+ r^{-2} \sin^{-2} \theta \left(\frac{1-\mu}{r}\nabla_\varphi \partial_\varphi-  \nabla_{\sin^2 \theta \frac r 2 \lbar - \sin^2 \theta \frac r 2 L- \sin \theta \cos \theta \partial_\varphi}L\right) \\ & =
	-\frac 1 2 (1-\mu)^{-1} \nabla_\lbar\left(\frac{2M}{r^2}L \right) +\\ &
	r^{-2}\left( \frac{1-\mu}{r} \left( \frac r 2 \lbar - \frac r 2 L \right) + \frac r 2 \frac{2M}{r^2} L\right)\\ &
	+
	r^{-2} \sin^{-2} \theta \left( \frac{1-\mu} r \left( \sin^2 \theta \frac r 2 \lbar - \sin^2 \theta \frac r 2 L- \sin \theta \cos \theta \partial_\varphi\right) + \frac r 2 \sin^2 \theta \nabla_L L + \sin \theta \cos \theta \nabla_\varphi L\right) =\\ &
	- \frac{2M}{r^3} L + r^{-2}\left( \frac{1-\mu}{2}\lbar - \frac{1-\mu} 2 L + \frac M r L \right) + r^{-2} \left( \frac{1-\mu}{2}\lbar - \frac{1-\mu}2 L + \frac r 2 \frac{2M} {r^2} L\right)\\ &
	=- \frac{2M}{r^3} \lbar + 2r^{-2}\left( \frac{1-\mu}{2}\lbar - \frac{1-\mu} 2 L + \frac M r L \right) \\ & =
	\frac{1-\mu}{r^2} \lbar - \frac{1-\mu}{r^2} L.
	\end{align*}
	Now,
	\begin{equation*}
	(a) + (e) = (\nabla_\mu \nabla^\mu L^\nu) \lbar^\alpha \fara_{\nu \alpha} +  L^\nu (\nabla_\mu \nabla^\mu \lbar^\alpha) \fara_{\nu \alpha} = -2 \frac{1-\mu}{r^2} \fara(L, \lbar)
	\end{equation*}
	Finally, 
	\begin{equation*}
	\begin{aligned}
	&(b) + (d) =2 (\nabla_\mu L^\nu)(\nabla^\mu \lbar^\alpha) \fara_{\nu \alpha} \\ & =
	2 (\nabla_L L)^L (\nabla^L \lbar)^\lbar \fara(L,\lbar) = (1-\mu)^{-1}  \frac{4M^2}{r^4} \fara(L, \lbar).
	\end{aligned}
	\end{equation*}
	Putting everything together, we obtain
	\begin{equation*}
	\begin{aligned}
	&\square_g (\fara(L, \lbar)) =(1-\mu)^{-1} \left(\frac{2M}{r^2} - \frac 1 r + \frac{2M}{r^2} \right)( \nabla_L \fara_{L \lbar}-\nabla_\lbar \fara_{L \lbar}) \\ & +
	-2 \frac{1-\mu}{r^2} \fara(L, \lbar) +(1-\mu)^{-1}  \frac{4M^2}{r^4} \fara(L, \lbar) + 2 \rie\indices{_\alpha_\beta^\mu_\nu}\fara\indices{_\mu^\beta}L^\nu \lbar^\alpha.
	\end{aligned}
	\end{equation*}
	Let us now define
	$$f := \fara(L, \lbar).$$
	Now, we calculate
	\begin{equation*}
	\begin{aligned}
	&\square_g (f) = 2 g^{L \lbar} L \lbar f + g^{\theta \theta} \nabla_\theta \nabla_\theta f - \nabla_{\nabla_\theta \partial_\theta}f +  g^{\varphi \varphi} \nabla_\varphi \nabla_\varphi f - \nabla_{\nabla_\varphi \partial_\varphi}f \\ & =
	- (1-\mu)^{-1} L \lbar f +\slashed{\Delta} f + r^{-1} L f - r^{-1} \lbar f.
	\end{aligned}
	\end{equation*}
	Also,
	\begin{equation*}
	\begin{aligned}
	\nabla_L \fara_{L \lbar}-\nabla_\lbar \fara_{L \lbar} = L f - \lbar f - \fara(\nabla_L L, \lbar) + \fara(L, \nabla_\lbar \lbar) = L f - \lbar f - \frac{4M}{r^2} f.
	\end{aligned}
	\end{equation*}
	We also calculate:
	\begin{equation*}
	\begin{aligned}
	&2 \rie\indices{_\alpha_\beta^\mu_\nu}\fara\indices{_\mu^\beta}L^\nu \lbar^\alpha = 2 \rie_{\lbar \beta \mu L} \fara^{\mu \beta} = 2 \rie_{\lbar L \lbar L} \fara^{\lbar L} = 2 (-8M (1-\mu)^2 r^{-3}) \frac 1 4 (1-\mu)^{-2} f \\ & =
	-4 M r^{-3} f.
	\end{aligned}
	\end{equation*}
	We finally have:
	\begin{equation*}
	\begin{aligned}
	&- (1-\mu)^{-1} L \lbar f +\slashed{\Delta} f + r^{-1} L f - r^{-1} \lbar f = (1-\mu)^{-1} \left(\frac{2M}{r^2} - \frac 1 r + \frac{2M}{r^2} \right) \left( L f - \lbar f - \frac{4M}{r^2} f\right) +\\ &
	-2 \frac{1-\mu}{r^2} f +(1-\mu)^{-1}  \frac{4M^2}{r^4}f -4 M r^{-3} f. 
	\end{aligned}
	\end{equation*}
	This implies
	
	\begin{equation*}
	\begin{aligned}
	&- (1-\mu)^{-1} L \lbar f +\slashed{\Delta} f  
	=  (1-\mu)^{-1}\left(\frac{2M}{r^2} - \frac 2 r + \frac{4M}{r^2} \right) \left( L f - \lbar f\right) - (1-\mu)^{-1} \left(\frac{2M}{r^2} - \frac 1 r + \frac{2M}{r^2} \right) \frac{4M}{r^2} f +\\ &
	-2 \frac{1-\mu}{r^2} f +(1-\mu)^{-1}  \frac{4M^2}{r^4}f -4 M r^{-3} f,\\
	\\
	&- (1-\mu)^{-1} L \lbar f +\slashed{\Delta} f 
	=  -(1-\mu)^{-1}\frac 2 r \left(1- \frac{3M} r\right) \left( L f - \lbar f\right) - (1-\mu)^{-1} \left(\frac{2M}{r^2} - \frac 1 r + \frac{2M}{r^2} \right) \frac{4M}{r^2} f +\\ &
	-2 \frac{1-\mu}{r^2} f +(1-\mu)^{-1}  \frac{4M^2}{r^4}f -4 M r^{-3} f,\\
	\\
	&- (1-\mu)^{-1} L \lbar f +\slashed{\Delta} f 
	=  -(1-\mu)^{-1}\frac 2 r \left(1- \frac{3M} r\right) \left( L f - \lbar f\right) - (1-\mu)^{-1} \frac{2} r \left(1- \frac{4M} r +6 \frac{M^2}{r^2}\right)f.
	\end{aligned}
	\end{equation*}
	The last display is equivalent to the equation
	\begin{equation*}
	r^{-2}(- L \lbar( r^2(1-\mu)^{-1} f)+(1-\mu) \slashed{\Delta}((1-\mu)^{-1}f)) = 0,
	\end{equation*}
	which is the equation
	\begin{equation}\label{eq1.2}
	- r^{-2} L \lbar (r^2 \rho) +(1-\mu) \slashed{\Delta}  \rho = 0,
	\end{equation}
	having defined $\rho:= \frac 1 2 (1-\mu)^{-1} \fara(\lbar, L)$.

	By Hodge duality (see Appendix~\ref{sec:dual}), the reasoning for $\sigma$ is the same.
\end{proof}

\section{Derivation of the transport equations satisfied by \texorpdfstring{$\alpha$}{alpha} and \texorpdfstring{$\alphabar$}{alphab}}\label{sec:transpder}
Recall that $(\theta^A, \theta^B)$ is a system of local coordinates on $\mathbb{S}^2$. Let $\partial_{\theta^A}, \partial_{\theta^B}$ the associated local vector fields. To shorten notation, contraction with $\partial_{\theta^A}$ is denoted by a capital subscript $A$.

\begin{lemma}\label{lem:transportalpha}
	Assume $\fara$ satisfies the MBI system~\eqref{MBI}. Then:
	\begin{equation}\label{eq:alpha}
	\boxed{
		\snabla_\lbar(r \alpha_A) = r(1-\mu)(\snabla_A \rho + \svol_{AB}\snabla^B \sigma ) + r(1-\mu)(\hdelta)_{\mu A \kappa \lambda} \nabla^{\mu} \fara^{\kappa\lambda},
	}
	\end{equation}
	and
	\begin{equation}\label{eq:alphabar}
	\boxed{
		\snabla_L(r \alphabar_A) = r(1-\mu)(- \snabla_A \rho + \svol_{AB}\snabla^B \sigma ) +r(1-\mu)(\hdelta)_{\mu A \kappa \lambda} \nabla^{\mu} \fara^{\kappa\lambda}.
	}
	\end{equation}
\end{lemma}

\begin{proof}
	By the MBI system~\eqref{MBI}, we have
	\begin{equation*}
	\nabla_{A} \fara_{\lbar  L} + \nabla_{L} \fara_{A \lbar} + \nabla_{\lbar} \fara_{L  A} = 0.
	\end{equation*}
	Now, $[\partial_{\theta^A}, L] = [\partial_{\theta^A}, \lbar] = 0$.
	Since $\nabla_L \partial_{\theta^A} = \frac {1-\mu} r \partial_{\theta^A}$, and $\nabla_\lbar \partial_{\theta^A} = - \frac{ 1-\mu} r \partial_{\theta^A}$, the previous display implies
	\begin{equation*}
	2(1-\mu)\snabla_{A} \rho + \frac{1-\mu} r \fara_{A \lbar} - \frac{1-\mu}{r} \fara_{L A} + \snabla_L \alphabar_A - \snabla_\lbar \alpha_A= 0
	\end{equation*}
	which in turn implies
	\begin{equation}\label{eq:alpha1}
	2(1-\mu)\snabla_{A} \rho + \frac 1 r \snabla_L(r \alphabar_A) - \frac 1 r \snabla_\lbar(r \alpha_A)= 0.
	\end{equation}
	Similarly, from the second equation of the MBI system (\ref{MBI}), we deduce
	\begin{equation*}
	\nabla^L \fara_{AL} + \nabla^\lbar \fara_{A\lbar} + \nabla^B \fara_{AB} = (\hdelta)_{\mu A \kappa \lambda} \nabla^{\mu} \fara^{\kappa\lambda}.
	\end{equation*}
	This implies
	\begin{equation*}
	- \frac 1 2 (1-\mu)^{-1} \left(\frac 1 r \snabla_\lbar(r \alpha_A) + \frac 1 r \snabla_L(r \alphabar_A)\right) + \svol_{AB} \snabla^B \sigma = (\hdelta)_{\mu A \kappa \lambda} \nabla^{\mu} \fara^{\kappa\lambda},
	\end{equation*}
	which implies
	\begin{equation}\label{eq:alpha2}
	-\left(\frac 1 r \snabla_\lbar(r \alpha_A) +  \frac 1 r \snabla_L(r \alphabar_A)\right) + 2 (1-\mu)\svol_{AB} \snabla^B \sigma = 2 (1-\mu)(\hdelta)_{\mu A \kappa \lambda} \nabla^{\mu} \fara^{\kappa\lambda}.
	\end{equation}
	Summing now~\eqref{eq:alpha1} and~\eqref{eq:alpha2}, we obtain
	\begin{equation*}
	\frac 1 r \snabla_\lbar(r \alpha_A) = (1-\mu)(\snabla_A \rho + \svol_{AB}\snabla^B \sigma ) + (1-\mu)(\hdelta)_{\mu A \kappa )\lambda} \nabla^{\mu} \fara^{\kappa\lambda}.
	\end{equation*}
	Which is,
	\begin{equation}
	\boxed{
		\snabla_\lbar(r \alpha_A) = r(1-\mu)(\snabla_A \rho + \svol_{AB}\snabla^B \sigma ) + r(1-\mu)(\hdelta)_{\mu A \kappa \lambda} \nabla^{\mu} \fara^{\kappa\lambda}.
	}
	\end{equation}
	On the other hand, subtracting~\eqref{eq:alpha2} from~\eqref{eq:alpha1}, we obtain:
	\begin{equation}
	\boxed{
		\snabla_L(r \alphabar_A) = r(1-\mu)(- \snabla_A \rho + \svol_{AB}\snabla^B \sigma ) +r(1-\mu)(\hdelta)_{\mu A \kappa \lambda} \nabla^{\mu} \fara^{\kappa\lambda}.
	}
	\end{equation}
\end{proof}

\section{Calculations with the dual tensor field}\label{sec:dual}
Recall the definition of the null components:
\begin{equation}
	\begin{aligned}
		\alpha_\mu :&= \gbar\indices{_\mu^\nu} F_{\nu\lambda}L^\lambda, \\
		\alphabar_\mu :&= \gbar\indices{_\mu^\nu} F_{\nu\lambda}\lbar^\lambda \\
		\rho :&= \frac 1 2 \left(1-\frac{2M}{r} \right)^{-1}F(\lbar,L), \\
		\sigma :&= \frac 1 2 \svol^{CD} F_{CD}.
	\end{aligned}
\end{equation}

\begin{lemma}\label{lem:dualnull}
The components of $\farad$ are
\begin{align}
	^\odot\alphabar = - \alphabar^B \svol_{BA}, \qquad
	^\odot\alpha = \alpha^B \svol_{BA}, \qquad ^\odot\rho = \sigma, \qquad ^\odot \sigma = - \rho.
\end{align}
\end{lemma}

\begin{proof}
	For the first two, it suffices to calculate (recall that indices $A$ and $B$ indicate contraction with resp. $\partial_{\theta^A}$, $\partial_{\theta^B}$)
	\begin{equation}
		^\odot\alpha(\partial_{\theta^A}) = \frac 1 2 \gbar\indices{_A^\nu} \varepsilon_{\alpha\beta\nu L} \fara^{\alpha\beta} = \frac 1 2 \varepsilon_{\alpha \beta A L} \fara^{\alpha \beta} = \varepsilon_{B \lbar A L} \fara^{B\lbar} = - r^2 \sin \theta 2 (1-\mu) \fara^{B\lbar}= \alpha^B \svol_{BA}.
	\end{equation}
	Similarly for the other:
	\begin{equation}
		^\odot\alphabar(\partial_{\theta^A}) = \frac 1 2 \gbar\indices{_A^\nu} \varepsilon_{\alpha\beta\nu \lbar} \fara^{\alpha\beta} = \frac 1 2 \varepsilon_{\alpha \beta A \lbar} \fara^{\alpha \beta} = \varepsilon_{B L A \lbar} \fara^{BL} = r^2 \sin \theta 2 (1-\mu) \fara^{BL}= -\alphabar^B \svol_{BA}.
	\end{equation}
	Also,
	\begin{equation}
		^\odot\rho = \frac 1 2 (1-\mu)^{-1} \farad(\lbar, L)= \frac 1 2 (1-\mu)^{-1} \frac 1 2 \varepsilon_{\alpha \beta \lbar L} \fara^{\alpha\beta}\\
		= \frac 1 4 (1-\mu)^{-1} 2 (1-\mu) \svol_{AB} \fara^{AB} = \sigma.
	\end{equation}
Finally,
	\begin{equation}
		^\odot\sigma = \frac 1 2 \varepsilon^{AB} \farad_{AB} = \frac 1 4 \svol^{AB} \varepsilon_{\alpha \beta AB} \fara^{\alpha\beta}\\
		= 2(1-\mu) (-1/2)^2 (1-\mu)^{-2} F_{L \lbar} = - \rho.
	\end{equation}
\end{proof}

\section{The nonzero components of the Riemann tensor}\label{sec:riemann}

We define the Riemann tensor $\rie$ as as $4$-covariant tensor field such that, for any four vectorfields $A,B,C,D$ we have the following:
\begin{equation}
\rie (A,B,C,D) = g(A, \nabla_C \nabla_D B - \nabla_D \nabla_C B - \nabla_{[C,D]}B).
\end{equation}
Then, we have the commutation relations (valid for two-forms $F$):
\begin{equation*}
\nabla_\mu \nabla_\nu F_{\alpha \beta} - \nabla_\nu  \nabla_\mu F_{\alpha \beta} = - F_{\alpha \lambda} \rie\indices{^\lambda_\beta_\mu_\nu} - F_{\lambda \beta} \rie\indices{^\lambda _{\alpha \mu \nu}}
\end{equation*}

We notice that, in the $(u,v,\theta, \varphi)$ coordinates on Schwarzschild, the only nonzero components of the Riemann tensor are:
\begin{align*}
\rie_{uvuv} &= - 8M (1-\mu)^2r^{-3}, \qquad
\rie_{\theta u \theta v} &= \frac{2M}{r}(1-\mu), \\
\rie_{\varphi u \varphi v} &= \frac{2M}{r} (1-\mu) \sin^2 \theta, \qquad
\rie_{\theta \varphi \theta \varphi} &= 2M r \sin^2 \theta.
\end{align*}

\section{Poincar\'e and Sobolev lemmas}\label{sec:poincasob}
We collect here a few useful results used in the paper. We begin by a standard Poincar\'e estimate on the sphere $\mathbb{S}^2$:
\begin{lemma}[Poincar\'e inequality for functions on $\mathbb{S}^2$]\label{lem:poinca}
	Let $F$ be a smooth function on $\mathbb{S}^2$ with vanishing integral average on $\mathbb{S}^2$. Then, we have the inequality
	\begin{equation*}
	\int_{\mathbb{S}^2} |d F|^2 \desphere \geq 2 \int_{\mathbb{S}^2} |F|^2 \desphere.
	\end{equation*}
\end{lemma}
We also recall the following standard result.
\begin{lemma}[Sobolev estimate for scalar functions on the sphere]\label{lem:sobsphere}
Let $(\mathbb{S^2}, g_{\mathbb{S}^2})$ be the two-sphere with the standard metric, let $\nabla$ be the associated Levi-Civita connection, and let $f$ be a smooth function $f: \mathbb{S}^2 \to \R$. Let $\bar f := \frac 1 {4 \pi} \int_{\mathbb{S}^2} f\desphere$ be the spherical average of $f$. There exists a universal constant $C$ such that
\begin{equation}
\sup_{\mathbb{S}^2}|f - \bar f|^2 \leq C \int_{\mathbb{S}^2} |\nabla \nabla f|^2\desphere.
\end{equation}
\end{lemma}
The following Lemma is also standard.
\begin{lemma}[1-d trivial inequality]
Let $[a,b] \subset \R$, with $a < b$. Let $f:[a,b] \to \R$ a smooth function with zero integral on $[a,b]$. Then there holds:
\begin{equation}
\sup_{[a,b]}|f| \leq C \int_{[a,b]} |f'(x)| \de x
\end{equation}
\end{lemma}
We furthermore recall the following Lemma:
\begin{lemma}[Sobolev inequality involving only certain derivatives]\label{lem:sobforrho}
Let $f: (\Sigma_{t_1^*} = \{t^* = t_1^*\}) \to \R$. Let $R > r_c > 2M$. Let $\bar f$ be again the mean of $f$ over the spheres:
\begin{equation}
\bar f := \frac 1 {4 \pi} \int_{\mathbb{S}^2} f(\omega) \desphere(\omega).
\end{equation}
There exists a constant $C_{r_c, R}$ such that
\begin{equation}\label{eq:sobosigma}
\sup_{r_1 \in [r_c, R], \omega \in \mathbb{S}^2}|f(t_1^*, r_1, \omega)-\bar f (t_1^*, r_1)|^2  \leq C \int_{r_c}^R \int_{\mathbb{S}^2}(|\snabla \snabla f|^2 + |\snabla_{\partial_{r_1}} \snabla \snabla f|^2) \desphere \de r_1
\end{equation}
\end{lemma}
\begin{proof}
Let 
\begin{equation*}
F(r_1) := \int_{\mathbb{S}^2}|(\snabla \snabla f)(t^*_1, r_1, \omega)|^2 \desphere(\omega)- \frac 1 {R-r_c} \int_{r_c}^R \int_{\mathbb{S}^2}|(\snabla \snabla f)(t^*_1, r_1, \omega)|^2 \desphere(\omega)\de r_1.
\end{equation*}
The preceding lemma then shows
\begin{equation*}
\sup_{r_1 \in [r_c, R]} |F(r_1)| \leq \int_{r_c}^R |\partial_{r_1}F(s)| \de s.
\end{equation*}
Since $\snabla_L \gbar = \snabla_\lbar \gbar = 0$, $f$ is smooth, and the Cauchy-Schwarz inequality, we have
$$
\left|\partial_{r_1} \int_{\mathbb{S}^2} \gbar^{AB} \gbar^{CD} (\snabla_A \snabla_C f) (\snabla_B \snabla_D f)\desphere\right| \lesssim \int_{\mathbb{S}^2} (|\snabla \snabla f|^2+|\snabla_{\partial_{r_1}} \snabla \snabla f|^2)\desphere.
$$
Combining this with the previous Sobolev lemma on spheres~\ref{lem:sobsphere}, we have the claim.
\end{proof}

\section{The theorem of Cauchy--Kowalevskaya} \label{sec:analytics}

Let $B$ be an open connected set in $\R^n$, and $f: B \to \R$ be analytic. We say that $f$ belongs to the class $\mathscr{A}_{M, c_0}(B)$ for some $M, c_0 >0$ if, on $B$, we have, for all multi-indices $\alpha$,
\begin{equation}
|D^\alpha f| \leq M {c_0}^{-k},
\end{equation}
where $|\alpha| = k$ and $D$ denotes partial differentiation.

We state here the form of the Cauchy--Kowalevskaya Theorem useful for our purposes.

\begin{theorem}[Cauchy--Kowalevskaya Theorem]\label{thm:ck}
	Let $B \subset \R^n$ be an Euclidean ball. Suppose that $W$ is the graph of an analytic function $\phi$ over $B$:
	\begin{equation*}
	W = \{x \in \R^{n+1}, x = (\phi(x_1, \ldots, x_n), x_1, \ldots, x_n), \text{ for } (x_1, \ldots, x_n) \in B\}.
	\end{equation*}
	Suppose furthermore that we have a quasilinear system of PDEs in Cauchy--Kowalevskaya form:
	\begin{equation}\label{eq:ckstat}
	\partial_{x^0} U = \sum_{i=1}^n \mathcal{M}_i(x, U) \partial_{x^i} U + F(x, U),
	\end{equation}
	where the unknown $U$ is a vector in $\R^k$, $x = (x^1, \ldots, x^n)$, and furthermore $\mathcal{M}_i : \R^{n+k} \to \R^{k^2}$ are matrices with analytic coefficients in the variables $x, U$. Also, $F: \R^{n+k} \to \R^k$ is an analytic function of all its arguments. We impose analytic initial data $U_0$ for $U$ on $W$.
	
	Then, for every point $q \in B$, let $p := (\phi(q), q) \in \R^{n+1}$. In these conditions, there is a radius $r_p > 0$ such that the system of equations
	\begin{equation}
	\left\{
	\begin{array}{ll}
	\partial_t U = \sum_{i=1}^n \mathcal{M}_i(x, U) \partial_{x^i} U + F(x, U) & \text{ on } B(p,r_p) \\
	U = U_0 & \text{ on } W \cap B(p,r_p).
	\end{array}
	\right.
	\end{equation}
	admits an analytic solution $U: B(p,r_p) \to \R^k$. Here, $B(p, r_p)$ denotes the $(n+1)$-dimensional Euclidean ball of center $p$ and radius $r_p$.
	
	Furthermore, let $c_0, M$ be positive numbers. The radius $r_p$ depends only on $c_0$, $M$ if all the following requirements are satisfied:
	\begin{equation}\label{eq:tounift}
	\begin{aligned}
	&U_0 \in \mathscr{A}_{M,c_0}(B), \\
	&\mathcal{M}_i, F \in \mathscr{A}_{M,c_0}(B \times \R^{2n}) \text{ for } i = 1, \ldots, n,\\
	&\phi \in \mathscr{A}_{M,c_0}(B), \\
	&r_p \leq d(q, \partial B).
	\end{aligned}
	\end{equation}
	Here, $d$ denotes the Euclidean distance in $\R^n$, and $\text{ diam}$ denotes the Euclidean diameter.
\end{theorem}
\begin{remark}
	The last requirement ensures that the ball $B(p, r_p)$ does not ``overshoot'' the ball $B$.
\end{remark}

\section{Addendum: stationary solutions and a heuristic calculation of the asymptotics of the spherical averages of \texorpdfstring{$\sigma$}{sigma} and \texorpdfstring{$\rho$}{rho}}\label{sec:heuristic}

The aim of this Section is to elaborate on how the spherical averages of $\sigma$ and $\rho$ evolve dynamically, both in the linear Maxwell theory and in the MBI theory. We first describe charged (stationary) solutions to the Maxwell system on Schwarzschild, then we discuss charged solutions to the MBI system on Schwarzschild. Finally, we sketch a derivation of the asymptotics of the spherical averages of $\sigma$ and $\rho$, as a function of the initial data, under some reasonable (but not justified) assumptions on the decay of the null components.

\subsection{Charged solutions to the linear Maxwell equations on Schwarzschild}\label{boundsta}
In the case of the linear Maxwell theory, there exist solutions to the Maxwell Equations on the Schwarzschild spacetime that are regular, stationary and nonzero (in particular, they do not decay in time). It can be proved that all such solutions are given by the following expression for the field tensor $\fara$:
\begin{equation*}
\fara = \frac{\rho_0}{r^2} dt \wedge dr^* + \frac{\sigma_0}{r^2}\sin \theta (d \theta \wedge  d \phi),
\end{equation*}
with $\rho_0$ and $\sigma_0$ real numbers. 

In order to prove time-decay, then, one has to exclude the presence of these charged solutions. In the case of linear Maxwell, it is easy to eliminate such issue as the Equations are linear, hence it suffices to subtract the corresponding charged components in order to obtain decay of the field tensor. See~\cite{masthes}.

\subsection{Charged solutions to the MBI system on Schwarzschild}

A similar phenomenon appears in the MBI case, and indeed there exist stationary solutions to the MBI system on Schwarzschild. Here, we limit our calculation to stationary solutions in spherical symmetry, such that $\alpha = \alphabar =0$ (cf.~the ``hairy ball theorem'').

\begin{proposition}\label{prop:statsols}
	Let $\fara$ be a spherically symmetric smooth stationary solution to the MBI system~\eqref{MBI} on Schwarzschild, such that the null components $\alpha = \alphabar = 0$ identically. Then, there exist numbers $\sigma_0, \rho_0 \in \R$ such that
	\begin{equation}\label{eq:stationarysol}
		\sigma(t^*, r_1, \theta, \varphi) = \sigma(r) = \frac{\sigma_0}{r^2}, \qquad \rho(t^*, r_1, \theta, \varphi) = \rho(r) = \frac{\rho_0}{\sqrt{\rho_0^2+\sigma_0^2+ r^4}}.
	\end{equation}
\end{proposition}

\begin{proof}
Let us start from the form of Equations (\ref{eq:transpun}) to obtain the form of the stationary solutions of the MBI system.
\begin{align}\label{eq:stat1}
	&- \hat \lbar(r^2 \sigma) + (1-\mu)^{-1} r^2 \curl \alphabar= 0, &\hat \lbar(r^2 \rho) + r^2 (1-\mu)^{-1}\dive \alphabar = - r^2{H_{_\Delta}} \indices{^\mu _{\hat \lbar} ^\kappa ^\lambda} \nabla_\mu \fara_{\kappa \lambda},\\ \label{eq:stat2}
	&L(r^2 \sigma)+r^2 \curl \alpha = 0, 
	&-L(r^2 \rho) + r^2 \dive \alpha = - r^2{H_{_\Delta}} \indices{^\mu _L ^\kappa ^\lambda} \nabla_\mu \fara_{\kappa \lambda}.
\end{align}
In these equations, set $\alpha = \alphabar = 0$, and assume that $\rho, \sigma$ are functions of $r$ only. Then, the equations in the first column readily give the existence of $\sigma_0 \in \R$ such that $\sigma = \sigma_0 r^{-2}$ identically. We then proceed to sum the equations in the right column of displays~\eqref{eq:stat1} and~\eqref{eq:stat2}, to obtain
\begin{equation*}
	L(r^2 \rho) - \lbar (r^2 \rho) = r^2 \hdelta\indices{^\mu_{L + \lbar}^\kappa^\lambda}\nabla_\mu \fara_{\kappa\lambda}.
\end{equation*}
By our choice of $\alpha = \alphabar = 0$, we have
\begin{equation}
	\begin{aligned}
	\lun &= \sigma^2 - \rho^2, \qquad
	\ldu &= -\rho\sigma.
	\end{aligned}
\end{equation}
By Equation (\ref{eq:forminv}), then,
\begin{equation}
H^{\mu\nu\kappa\lambda}_{\Delta} \nabla_{\mu}\fara_{\kappa\lambda} = \frac{1}{2\ellmbi^2} \nabla_\mu(\lun-\ldu^2) \left(- \fara^{\mu\nu} + \ldu\farad^{\mu\nu} \right) - \nabla_\mu \ldu \farad^{\mu\nu}.
\end{equation}
Hence
\begin{equation*}
	2(1-\mu)\partial_r (r^2 \rho) = \frac{r^2 (1-\mu)}{(1+\sigma^2)(1-\rho^2)} \partial_r\left( (1+\sigma^2)(1-\rho^2) \right)(1+\sigma^2)\rho- 2 r^2 (1-\mu)\partial_r(\rho\sigma) \sigma.
\end{equation*}
By the fact that $\sigma r^2 = \sigma_0$, we then have
\begin{equation*}
	2\partial_r( \rho r^2  (1+\sigma^2)) = \frac{r^2 }{(1+\sigma^2)(1-\rho^2)} \partial_r\left( (1+\sigma^2)(1-\rho^2) \right)(1+\sigma^2)\rho.
\end{equation*}
Integrating the previous display yields the existence of $\rho_0 \in \R$ such that
\begin{equation}\label{eq:statfin}
	\rho^2 r^4 (1+\sigma^2) = \rho_0^2 (1-\rho^2).
\end{equation}
Given the form of $\sigma$, this implies the claim.
\end{proof}

\subsection{Asymptotic behaviour of the spherical averages of $\rho$ and $\sigma$ when the initial charge is nonzero}

Recall that, in the above proof of global stability for the MBI system on Schwarzschild, we needed to deduce asymptotic bounds for spherical averages of $\rho$ and $\sigma$ along the evolution. To do that, we imposed the charge to vanish at spacelike infinity, on $\Sigma_{t_0^*}$. This was then propagated along the evolution in Section~\ref{sec:charge}. We remark that such is an essential element of our proof, since we make large use of Poincar\'e estimates, which in turn require information on the spherical means.

In the remaining part of this Section, we would like to address the problem of determining the asymptotic behaviour of the spherical averages of $\sigma$ and $\rho$, if we assume \emph{nonzero initial charge} and certain \emph{decay in time} for all the components of the field. This may prove useful in a proof of global stability of the MBI system on Schwarzschild with non-vanishing initial charge.

We point out that a similar problem arises in the context of the Kerr stability conjecture, the so-called \emph{final state problem}. In that case, one seeks to calculate, as a function of initial data, the parameters $a$ and $M$ such that the solution will be asymptotic to a Kerr black hole of parameters $(a,M)$.

However, there is a caveat. The calculation below points to the fact that the charge for the MBI system is conserved in the nonlinear evolution, along future null infinity. We do not expect a similar statement to hold for angular momentum and mass in the context of the Kerr stability conjecture.

We prove the following Proposition.

\begin{proposition}\label{prop:finsta}
	Let $\fara$ be a smooth solution to the MBI system~\eqref{MBI} on Schwarzschild, satisfying the following decay assumptions: there exist $a, b, C >0$, and a smooth radial function $\rho_f(r)$ such that
	\begin{equation}\label{eq:decassumpt}
	\begin{aligned}
	|\rho - \rho_f (r)| &\leq C \tau^{-a},\\
	|\sigma - \sigma_0 r^{-2}| &\leq C \tau^{-a}r^{-b},\\
	|\alpha| &\leq C \tau^{-1}r^{-1} ,\\
	|\alphabar| &\leq C \tau^{-\frac 1 2} r^{-3}.
	\end{aligned}
	\end{equation}
	Then, we let
	$$
	\rho_0 := \frac 1 {4\pi}\lim_{r \to \infty} r^2 \int_{\mathbb{S}^2} \rho(t_0^*, r) \desphere.
	$$
	(In particular, the limit appearing in the right hand side of the previous display exists). In these conditions, we have
	\begin{equation}\label{eq:fincha}
		\rho_f(r) = \frac{\rho_0}{\sqrt{r^4 + \sigma_0^2+ \rho_0^2}}.
	\end{equation}
\end{proposition}

\begin{remark}
	We remark that the assumption~\eqref{eq:decassumpt} is a reasonable one but, in a proof of stability of MBI on Schwarzschild, such assumption will need to be proved in the context of a bootstrap argument. Hence the above Proposition is very far from addressing the stability problem for MBI in the charged case.
\end{remark}

\begin{remark}
	We remark that both $\rho_0$ and $\sigma_0$ can be calculated starting from initial data on $\widetilde{\Sigma}_{t_0^*}$. Hence the Proposition gives a way of calculating the asymptotic behaviour of spherical means as a function of initial data.
\end{remark}

\begin{remark}
	Furthermore, we remark that the expression for the final charge~\eqref{eq:fincha} coincides with the form of stationary solutions found in~\eqref{eq:stationarysol}.
\end{remark}

\begin{remark}
	This Proposition is not concerned with the behaviour of the charge along null infinity, rather with the behaviour of spherical averages of $\rho$ and $\sigma$ on a region of finite $r$-coordinate. Nevertheless, very similar calculations indicate that the charge for $\rho$ is conserved along future null infinity.
\end{remark}

\begin{proof}[Proof of Proposition~\ref{prop:finsta}]
The set of equations
\begin{equation*}
	\nabla^\mu \mfarad_{\mu\nu} = 0, \qquad \nabla^\mu \farad_{\mu\nu} = 0
\end{equation*}
implies, through the null decomposition, that the quantities
\begin{equation}\label{eq:conslaw}
	\int_{\mathbb{S}^2} \frac 1 {2(1-\mu)}\mfarad(\lbar, L) r^2 \desphere = C_\rho, \qquad 
	\int_{\mathbb{S}^2} \frac{1}{2(1-\mu)}\farad(\lbar, L) r^2 \desphere = C_\sigma,
\end{equation}
where $C_\rho$, $C_\sigma$ are constant along the evolution. From the second equation, we obtain that
$$
\int_{\mathbb{S}^2} \sigma r^2 \desphere = \text{const.}
$$
The goal now is to determine $\rho_f(r)$. Recall that $$
\mfarad_{\mu\nu} = -\ellmbi^{-1}(\fara_{\mu\nu}- \ldu \farad_{\mu\nu}),$$
and that
$$
\lun = \sigma^2 - \rho^2 + 2 \alpha^A \alphabar_A, \qquad \ldu = -\rho \sigma - 2 \svol_{AB} \alpha^A \alphabar^B.
$$
From (\ref{eq:conslaw}) and the assumptions, we obtain, taking the limit as $r \to \infty$ along points of the form $(t^*_0, r)$,
\begin{equation*}
	-C_\rho = \lim_{r \to \infty} r^2 \int_{\mathbb{S}^2} \rho(t_0^*, r) \desphere.
\end{equation*}
This follows from the expression for $\mfara$, in which the only term that survives is the linear term (the one corresponding to $\fara$) and furthermore $\ellmbi \to 1$ as $r \to \infty$. We let $\rho_0 := - \frac 1 {4\pi} C_\rho$.
Now,
\begin{equation*}
	\begin{aligned}
		\frac 1 {2(1-\mu)} \ellmbi^{-1} (\fara(\lbar, L)-\ldu \farad(\lbar, L)) = \frac{\rho + \rho \sigma^2 +2 \rho \svol_{AB} \alpha^A \alphabar^B}{\sqrt{1+\lun-\ldu^2}}
	\end{aligned}
\end{equation*}
Let us now consider the limit as $r$ is fixed, and $t^* \to \infty$. Then,
\begin{equation*}
	\begin{aligned}
		\lim_{t^* \to \infty} \lun &= \sigma_0^2 r^{-4} - \rho^2_f(r), \qquad 
		\lim_{t^* \to \infty} \ldu^2 &= \sigma_0^2 r^{-4} \rho^2_f(r).
	\end{aligned}
\end{equation*}
We obtain eventually,
\begin{equation*}
	\int_{\mathbb{S}^2} ((1 + \sigma^2_0 r^{-4})(1-\rho^2_f))^{-\frac 1 2}\rho_f r^2 \desphere = 4 \pi \rho_0.
\end{equation*}
This implies
\begin{equation*}
(1 + \sigma^2_0 r^{-4})^{-\frac 1 2}(1-\rho^2_f)^{-\frac 1 2}\rho_f r^2 = \rho_0.
\end{equation*}
By inverting the last display, we get
\begin{equation*}
	\rho_f(r) = \frac{\rho_0}{\sqrt{r^4 + \sigma_0^2+ \rho_0^2}}
\end{equation*}
\end{proof}

\bibliography{max.bib}
\bibliographystyle{plain}
\end{document}